\UseRawInputEncoding
\documentclass[11pt,reqno]{amsart}
\usepackage{amsthm,amsfonts,amssymb,euscript,fancyhdr,datetime,xcolor,graphicx,float}
\usepackage[colorlinks,linkcolor=blue,citecolor=blue]{hyperref}

\newcommand\smallO{
  \mathchoice
    {{\scriptstyle\mathcal{O}}}
    {{\scriptstyle\mathcal{O}}}
    {{\scriptscriptstyle\mathcal{O}}}
    {\scalebox{.6}{$\scriptscriptstyle\mathcal{O}$}}
  }

\usepackage[font=small,labelfont=bf]{caption} 

\newcommand{\gd}{{g \mkern-8mu /\ \mkern-5mu }}
\newcommand{\di}{\mbox{$d \mkern-9.2mu /$\,}}

\def\abf{{\mathbf{a}}}

\def\Kerr{{\mathrm{Kerr}}}
\def\tth{{\tilde{\theta}}}
\def\a{{\alpha}}

\def\be{{\beta}}

\def\ga{\gamma}
\def\Ga{\Gamma}
\def\de{\delta}

\def\la{\lambda}
\def\La{\Lambda}
\def\si{\sigma}
\def\Si{\Sigma}
\def\om{\omega}
\def\Om{\Omega}
\def\th{\theta}

\def\nab{\nabla}
\def\varep{\varepsilon}
\def\pr{{\partial}}

\def\les{\lesssim}
\def\rh{{\rho}}
\def\ind{{\in \mkern-16mu /\ \mkern-4mu}}
\def\XX{{\mathcal{X}}}

\def\Th{{\Theta}}

\def\gdcd{{\dot{\gd_c}}}

\providecommand{\lrpar}[1]{\left( #1\right)}

\def\Ldo{{\overset{\circ}{\Ld}}}
\def\phid{{\dot{\phi}}}

\def\AA{{\mathcal A}}

\def\MM{{\mathcal M}}

\def\II{{\mathcal I}}

\def\EE{{\mathcal E}}
\def\HH{{\mathcal H}}

\def\OO{{\mathcal O}}

\def\KK{{\mathcal K}}

\def\DD{{\mathcal D}}

\def\RR{{\mathcal R}}
\def\QQ{{\mathcal Q}}

\def\HHb{\underline{\mathcal H}}

\def\D{{\bf D}}

\def\g{{\bf g}}

\def\SSS{{\mathbb{S}}}
\def\RRR{{\mathbb R}}

\DeclareMathOperator{\Div}{\mathrm{div}}
\DeclareMathOperator*{\Curl}{\mathrm{curl}}
\def\half{\frac{1}{2}}
\newcommand{\pd}{\pd \mkern-9mu/\ \mkern-7mu}
\newcommand{\Lied}{\mathcal{L} \mkern-9mu/\ \mkern-7mu}

\newcommand{\kd}{k \mkern-7mu /\ \mkern-7mu}
\newcommand{\kdN}{\mkern-6mu k_N \mkern-9mu /\ \mkern-7mu}

\newcommand{\DDd}{\DD \mkern-10mu /\ \mkern-5mu}

\newcommand{\Du}{\underline{D}}
\newcommand{\RRRic}{\mathrm{Ric}}

\newcommand{\Divd}{\Div \mkern-17mu /\ }
\newcommand{\Divdo}{{\overset{\circ}{\Div \mkern-17mu /\ }}}
\newcommand{\Curld}{\Curl \mkern-17mu /\ }
\newcommand{\Curldo}{{\overset{\circ}{\Curl \mkern-17mu /\ }}}
\newcommand{\Nd}{\nabla \mkern-13mu /\ }

\newcommand{\Ld}{\triangle \mkern-12mu /\ }
\newcommand{\iin}{\in \mkern-16mu /\ \mkern-5mu}

\newcommand{\trchi}{{\tr \chi}}
\newcommand{\trchib}{{\tr \chib}}

\def\ni{\noindent}

\def\Lb{{\,\underline{L}}}

\def\tr{\mathrm{tr}}

\def\chih{{\widehat \chi}}
\def\chib{{\underline \chi}}
\def\chibh{{\underline{\chih}}}

\def\etab{{\underline \eta}}
\def\omb{{\underline{\om}}}

\def\aa{{\underline{\a}}}

\def\th{\theta}

\def\Rbf{{\mathbf{R}}}

\newcommand{\gac}{{\overset{\circ}{\ga}}}

\newcommand{\ab}{{\underline{\alpha}}}
\newcommand{\beb}{{\underline{\beta}}}

\newtheorem{theorem}{Theorem}[section]
\newtheorem{lemma}[theorem]{Lemma}
\newtheorem{proposition}[theorem]{Proposition}
\newtheorem{corollary}[theorem]{Corollary}
\newtheorem{definition}[theorem]{Definition}
\newtheorem{remark}[theorem]{Remark}

\numberwithin{equation}{section}

\makeatletter
\def\@setthanks{\vspace{-\baselineskip}\def\thanks##1{\@par##1\@addpunct.}\thankses}
\makeatother
\begin{document}

\title[Characteristic gluing to Kerr]{Characteristic gluing to the Kerr family \\ and application to spacelike gluing}
\author[S. Aretakis, S. Czimek and I. Rodnianski]{Stefanos Aretakis $^{(1)}$, Stefan Czimek $^{(2)}$, and Igor Rodnianski $^{(3)}$} 

\thanks{\noindent$^{(1)}$ Department of Mathematics, University of Toronto, 40 St George Street, Toronto, ON, Canada, \texttt{aretakis@math.toronto.edu}. \\
$^{(2)}$ Institute for Computational and Experimental Research in Mathematics, Brown University, 121 South Main Street, Providence, RI 02903, USA,  \texttt{stefan\_czimek@brown.edu}. \\
$^{(3)}$ Department of Mathematics, Princeton University, Fine Hall, Washington Road, Princeton, NJ 08544, USA, \texttt{irod@math.princeton.edu}. }

\begin{abstract} This is the third paper in a series of papers adressing the characteristic gluing problem for the Einstein vacuum equations. We provide full details of our characteristic gluing (including the $10$ charges) of strongly asymptotically flat data to the data of a suitably chosen Kerr spacetime. The choice of the Kerr spacetime crucially relies on relating the $10$ charges to the ADM energy, linear momentum, angular momentum and the center-of-mass. As a corollary, we obtain an alternative proof of the Corvino-Schoen spacelike gluing construction for strongly asymptotically flat spacelike initial data. 
\end{abstract}

\maketitle
\setcounter{tocdepth}{3}
\tableofcontents
%%%%%%%%%%%%%%%%%%%%%%%%%%%%%%%%%%%%%%%%
\section{Introduction} \label{SEC99901}

\ni The gluing problem in general relativity investigates whether it is possible to join two given vacuum spacetimes. Concretely, one can approach this problem by attempting to construct a solution to the constraint equations which agrees inside a bounded domain with specified initial data, and on the complement of a large ball with other specified initial data. The geometric obstructions to solving the gluing problem provide insights into the rigidity properties of the Einstein equations.

In \cite{ACR3} we initiated the study of the characteristic gluing problem for initial data for the Einstein vacuum equations. This problem amounts to connecting two initial data sets along a truncated null hypersurface by solving the null constraint equations. There are several reasons for considering the characteristic gluing problem: 1.) the null constraint equations are of transport character (in contrast to the previously studied gluing problem for spacelike initial data which requires to analyze the elliptic Riemannian constraint equations), 2.) the null lapse function and the conformal geometry of the characteristic hypersurface can be freely prescribed, 3.) characteristic gluing of spacetimes implies spacelike gluing of the spacetimes.
 
In \cite{ACR3,ACR1} we explicitly derived a $10$-dimensional space of gauge-invariant charges on sections of null hypersurfaces that act as obstructions to the characteristic gluing problem and we showed that, modulo this $10$-dimensional space, characteristic gluing is always possible for data sets that are close to the Minkowski data. In this paper, we prove that characteristic initial data that are close to the Minkowski data can be fully glued (including the $10$ charges) to the characteristic data of a suitable Kerr spacetime. By rescaling we show that strongly asymptotically flat data can also be characteristically glued to the data of some Kerr spacetime. As a corollary, we obtain an alternative proof of the Corvino--Schoen gluing construction (for strongly asymptotically flat spacelike initial data) that relies on solving the null constraint equations instead of the Riemannian constraint equations. Our approach crucially relies on relating the $10$ charges to the ADM energy, linear momentum, angular momentum and center-of-mass.

The introduction is structured as follows. In Section \ref{SEC999011} we discuss the gluing problem in general relativity. In Section \ref{SEC999012} we outline the main results. In Section \ref{SEC999013} we give an overview of the main ideas of the proofs. 

%%%%%%%%%%%%%%%%%%%%%%%%%%%%%%%%%%%%%%%%%%%%%%%%%%%%%%%%%%%%
\subsection{The gluing problem in general relativity} \label{SEC999011} In this section we discuss the gluing problem in general relativity. We first discuss the literature for the spacelike gluing, and then we summarize our previous work on the characteristic gluing for the Einstein vacuum equations.

%%%%%%%%%%%%%%%%%%%%%%%%%%%%%%
\subsubsection{The spacelike gluing problem} \label{SEC9990112}

\ni The majority of gluing constructions in general relativity are, up to now, concerned with the gluing of \emph{spacelike} initial data subject to the elliptic constraint equations. 

On the one hand, gluing constructions based on connected-sum gluing (see also the works \cite{SchoenYauPSCM,GromovL} on codimension-$3$ surgery for manifolds of positive scalar curvature) were studied by Chru\'sciel--Isenberg--Pollack \cite{CIP1,CIP2}, Chru\'sciel--Mazzeo \cite{ChruscielMazzeo}, Isenberg--Maxwell--Pollack \cite{IMP3}, Isenberg--Mazzeo--Pollack \cite{IMP1,IMP2}.

On the other hand, in the ground breaking work of Corvino \cite{Corvino} and Corvino--Schoen \cite{CorvinoSchoen}, the \emph{geometric under-determinedness} of the spacelike constraint equations is used to study the (codimension-$1$) gluing problem. 
In particular, they proved that it is possible to glue asymptotically flat spacelike initial data across a compact region to precisely Kerr spacelike initial data; see also Corollary \ref{THMspacelikeGLUINGtoKERRv2INTRO} below. More constructions and refinements based on this approach were established by Chru\'sciel--Delay \cite{ChruscielDelay1,ChruscielDelay}, Chru\'sciel--Pollack \cite{ChruscielPollack}, Cortier \cite{Cortier}, Hintz \cite{Hintz}. Another milestone was the result \cite{CarlottoSchoen} by Carlotto--Schoen where it is shown that spacelike initial data can be glued -- along a non-compact cone -- to Minkowski spacelike initial data.

The \emph{characteristic} gluing problem was previously studied by the first author \cite{CiteGluing, CiteElliptic} in the much simpler setting of the linear homogeneous wave equation on general (but fixed) Lorentzian manifolds. Similarly to the present paper, \cite{CiteGluing} determined that the only obstructions to solving the characteristic gluing problem are conservation laws along null hypersurfaces. In the following it was shown that these conservation laws have important applications in the study of the evolution of scalar perturbations on both sub-extremal \cite{CiteSS,CitePriceLaw, CiteKerr, Ma} and extremal \cite{CiteExtremal1,CiteExtremal2,CiteExtremal3,PRL} black hole spacetimes.

%%%%%%%%%%%%%%%%%%%%%%%%%%%%%%
\subsubsection{The characteristic gluing problem} \label{SEC9990113} In this section we discuss the \emph{codimension-$10$ characteristic gluing} for the Einstein vacuum equations introduced in \cite{ACR3}. Before stating the main results of that paper, we introduce the following notation.

Let $(\MM_1,\g_1)$ and $(\MM_2,\g_2)$ be two vacuum spacetimes. Let $S_1$ and $S_2$ be two spacelike $2$-spheres in $\MM_1$ and $\MM_2$, respectively, and assume (without loss of generality) they are each intersection spheres of local double null coordinate systems, respectively. We define \emph{sphere data} $x_1$ on $S_1$ and $x_2$ on $S_2$ to be given by the respective restriction of the metric components, Ricci coefficients and components of the Riemann curvature tensor of the spacetimes to the respective spheres (see also Section \ref{SEC9990132}) with respect to the respective double null coordinate system.

One of the main insights of \cite{ACR3,ACR1} is the derivation of a family of charges on the sections of null hypersurfaces that act as obstructions to the characteristic gluing problem. The charges arise from conservation laws for the linearized constraint equations. They split into two classes: An infinite-dimensional space of \emph{gauge-dependent charges} and a 10-dimensional space of \emph{gauge-invariant charges}. The former charges can always be overcome by gauge perturbations. We will refer to the gauge-invariant charges as simply the charges. For further discussion, see Section \ref{SECgeominterINTROmain}. For precise definition of the charges, see Sections \ref{SEC9990132} and \ref{SECdefcharges99901}.

The main result of \cite{ACR3,ACR1} can be summarized as follows.\\

\ni \textbf{Perturbative codimension-$10$ characteristic gluing} \cite{ACR3,ACR1}. \emph{Let on two spheres $S_1$ and $S_2$ be given the sphere data $x_1$ and $x_2$, sufficiently close to the respective sphere data on the round spheres of radius $1$ and $2$ in Minkowski spacetime, respectively. Then there is a null hypersurface $\HH'_{[1,2]}$, connecting the sphere data $x_1$ on $S_1$ to a transversal perturbation $S_2'$ of the sphere $S_2$ with sphere data $x_2'$, solving the null constraint equations, and such that all derivatives tangential to $\HH'_{[1,2]}$ of the sphere data $x_1$ and $x'_2$ are -- up to a $10$-dimensional space of charges explicitly defined at $S_2'$ -- smoothly glued.}\\

\ni Sphere data determines all derivatives of the metric components up to order $2$, hence \emph{perturbative gluing is gluing at the level of $C^2$} of the metric components (up to the $10$-dimensional space of charges). 

We note that it is equivalently possible to perturb $S_2$ instead of $S_1$ in the above result -- this is in fact the version stated in \cite{ACR3,ACR1}.

\begin{figure}[H]
\begin{center}
\includegraphics[width=7.5cm]{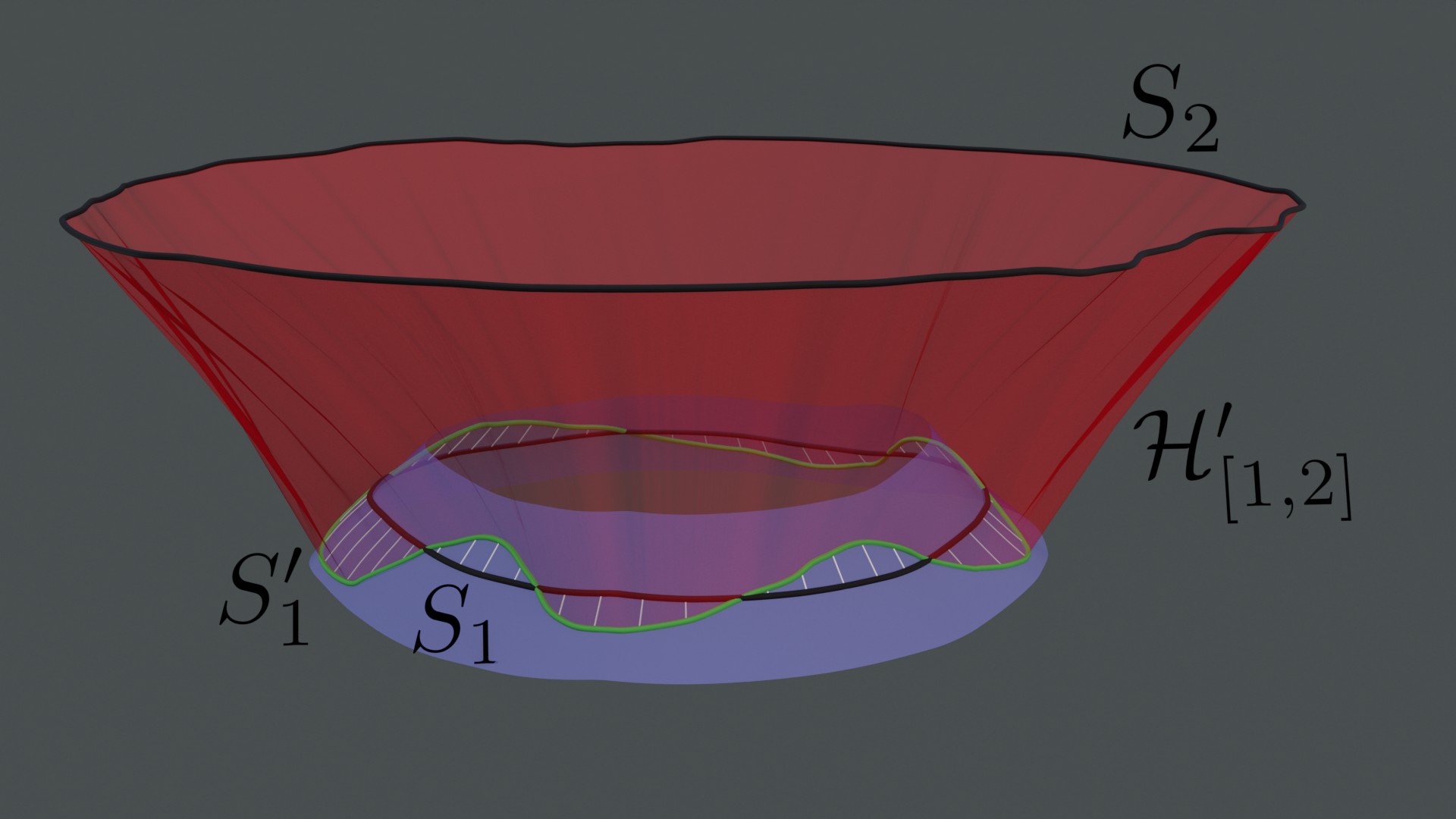} \,\,\, \includegraphics[width=7.5cm]{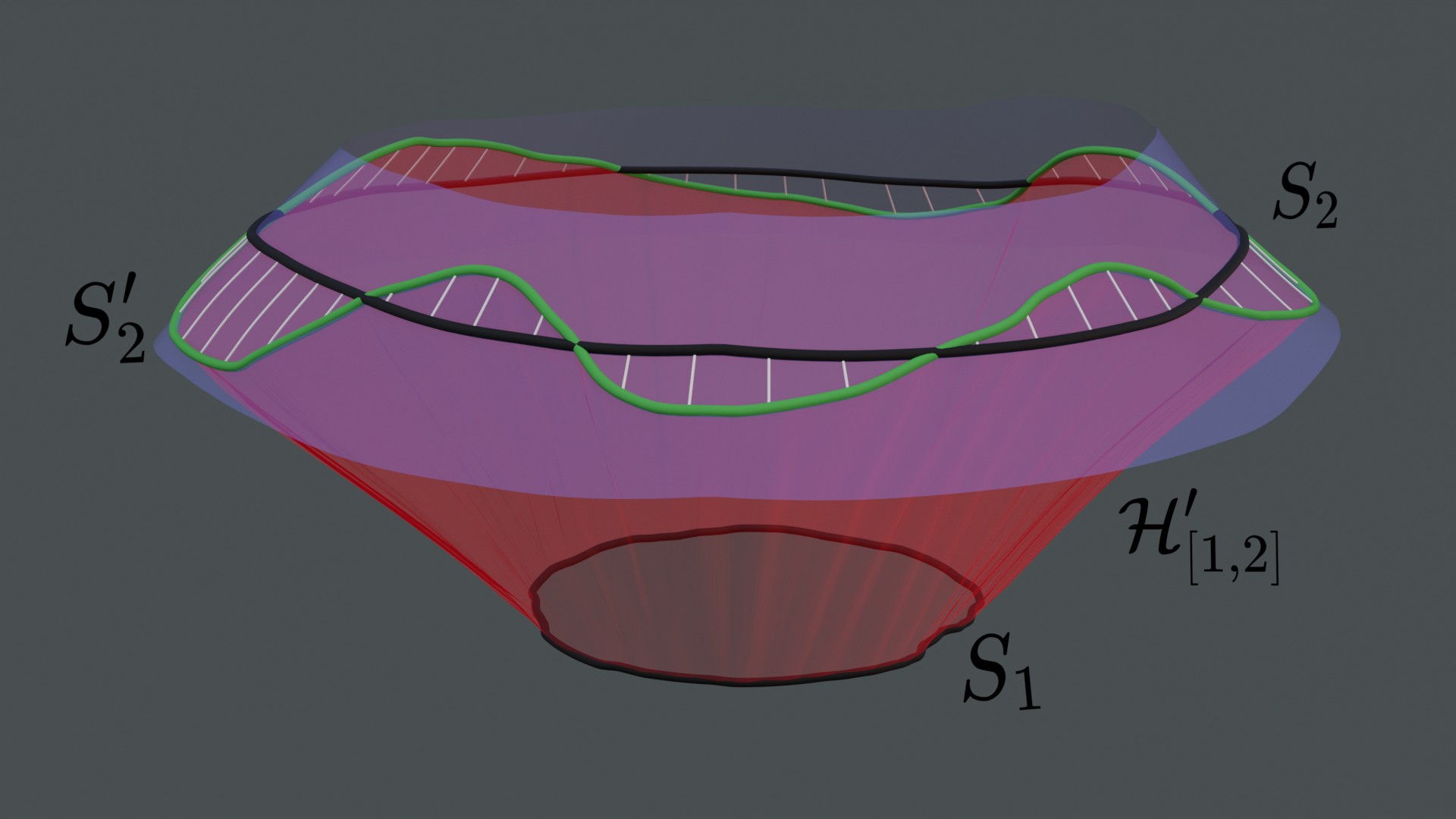} 
\vspace{0.4cm} 
\caption{The two versions of the codimension $10$ characteristic gluing along one null hypersurface.}\label{FIG5}
\end{center}
\end{figure}

\ni In \cite{ACR3,ACR1} we also consider characteristic gluing along two null hypersurfaces bifurcating from an auxiliary sphere, and prove the following result. \\

\ni \textbf{Bifurcate codimension-$10$ characteristic gluing} \cite{ACR3,ACR1}. \emph{Consider two spheres $S_1$ and $S_2$ equipped with sphere data $x_1$ and $x_2$ as well as prescribed higher-order derivatives in all directions, respectively. If this higher-order data on $S_1$ and $S_2$ is sufficiently close to the respective higher-order data on the round spheres of radius $1$ and $2$ in Minkowski spacetime, then it is possible to characteristically glue -- up to a $10$-dimensional space of charges -- the higher-order data of $S_1$ and $S_2$  along two null hypersurfaces bifurcating from an auxiliary sphere $S_{\mathrm{aux}}$.}\\ 

\ni We note that in the above result the spheres $S_1$ and $S_2$ are not perturbed. Moreover, \emph{bifurcate gluing is higher-regularity gluing}, that is, we can glue any order of derivatives of the metric components (up to the $10$-dimensional space of charges).

\begin{figure}[H]
\begin{center}
\includegraphics[width=9.5cm]{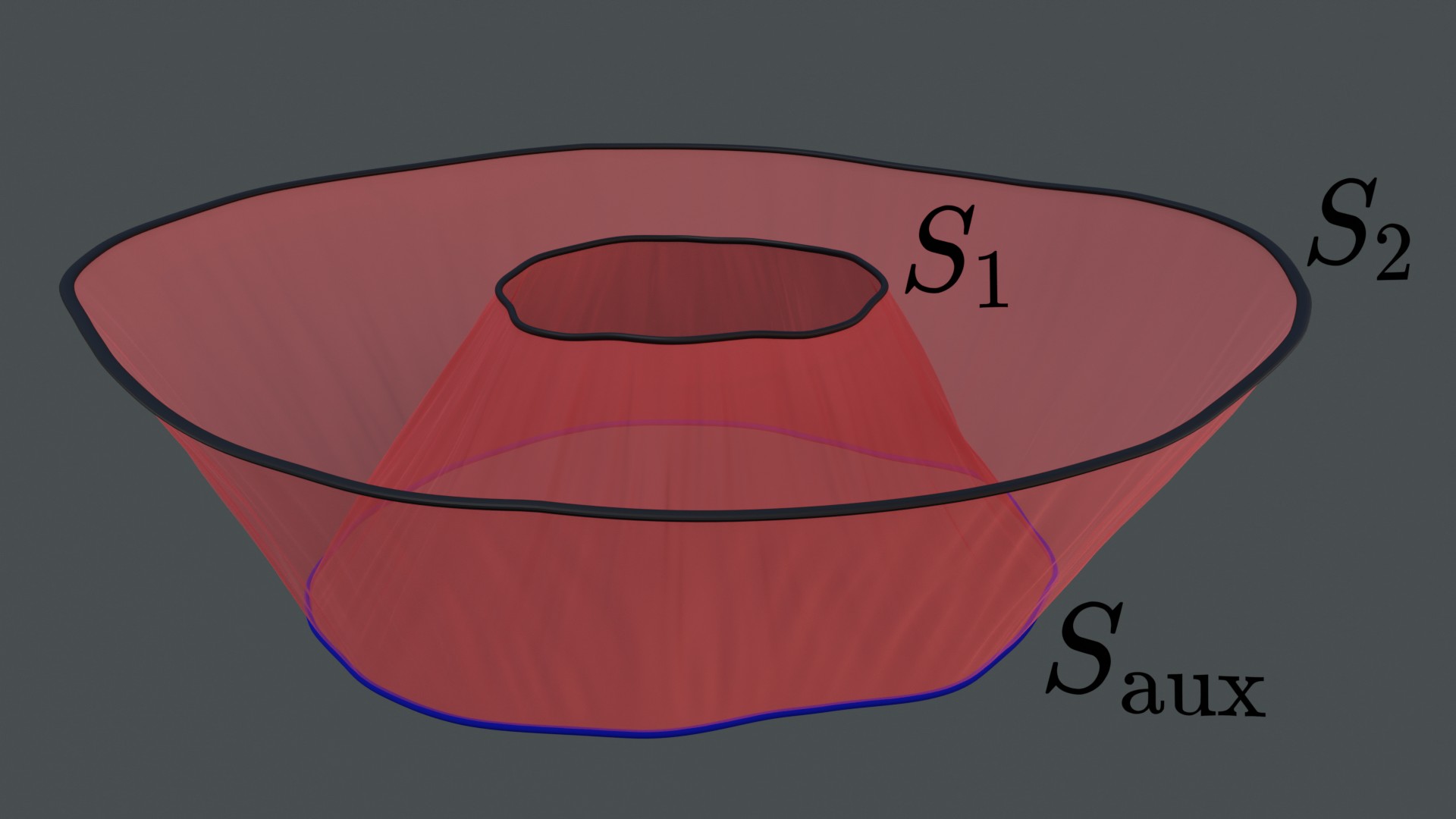} 
\vspace{0.4cm} 
\caption{Higher-order codimension-$10$ characteristic gluing along two null hypersurfaces bifurcating from an auxiliary sphere.}\label{FIG6}
\end{center}
\end{figure}

\ni We finally note that in the previous work \cite{CiteGluing} by the first author, it was shown that the only obstructions to first-order gluing of solutions to the \emph{wave equation} along null hypersurfaces are conserved charges defined on the sections of the hypersurface. It was further shown that the existence of these charges is related to the kernel of a geometric elliptic operator defined on the hypersurface \cite{CiteElliptic}. These conserved charges have subsequently been shown to have applications in the study of the asymptotics for the wave equation on black hole backgrounds, both in the sub-extremal \cite{CitePriceLaw, Ma, CiteKerr} and in the extremal case \cite{CiteExtremal1,Easymptotics,PRL,Khanna}.

%%%%%%%%%%%%%%%%%%%%%%%%%%%%%%
%%%%%%%%%%%%%%%%%%%%%%%%%%%%%%
\subsection{Main results on the characteristic gluing to the Kerr family} \label{SEC999012}
In this section we outline the main results of this paper on the characteristic gluing to Kerr.
%%%%%%%%%%%%%%%%%%%%%%%%%%%%%%
\subsubsection{Geometric interpretation of charges} \label{SECgeominterINTROmain}

As we have already discussed, the characteristic gluing of \cite{ACR3,ACR1} holds up to a $10$-dimensional space of charges. These charges are calculated as integrals over spacelike $2$-spheres and are denoted by the real number $\mathbf{E}$ and the $3$-dimensional vectors $\mathbf{P}, \mathbf{L}$ and $\mathbf{G}$. At the linear level, the charges $\mathbf{E}$ and $\mathbf{P}$ are proportional to the modes $l=0$ and $l=1$ of $\rh+ r \Divd \beta$, while $\mathbf{L}$ and $\mathbf{G}$ are proportional to the magnetic and electric parts of the mode $l=1$ of $\beta$, see Section \ref{SECdefcharges99901} for precise definitions.

\begin{theorem} \label{THMconstruction111999999111} Given a strongly asymptotically flat family of sphere data on spheres $S_R$, as defined in Section \ref{SECdefinitionAF}, the charges $(\mathbf{E}_R,\mathbf{P}_R,\mathbf{L}_R, \mathbf{G}_R)$ have a limit $(\mathbf{E}_\infty,\mathbf{P}_\infty,\mathbf{L}_\infty, \mathbf{G}_\infty)$, called the asymptotic charges. In case the spheres $S_R$ lie in a strongly asymptotically flat spacelike hypersurface, the asymptotic charges are related to the ADM asymptotic invariants of the spacelike hypersurface by
\begin{align*} 
\begin{aligned} 
(\mathbf{E}_\infty,\mathbf{P}_\infty,\mathbf{L}_\infty, \mathbf{G}_\infty)= (\mathbf{E}_{\mathrm{ADM}},\mathbf{P}_{\mathrm{ADM}},\mathbf{L}_{\mathrm{ADM}}, \mathbf{C}_{\mathrm{ADM}}),
\end{aligned} %\label{}
\end{align*}
where $\mathbf{E}_{\mathrm{ADM}}$ denotes the {energy} (often called {mass}), $\mathbf{P}_{\mathrm{ADM}}$ the {linear momentum}, $\mathbf{L}_{\mathrm{ADM}}$ the {angular momentum}, and $\mathbf{C}_{\mathrm{ADM}}$ the {center-of-mass}.
\end{theorem}

\ni Our definitions of the charges are, to leading order, consistent with previous definitions in general relativity of mass, linear and angular momentum in terms of integrals over spheres; see, for example, \cite{Komar,LandauLifshitz,Wald}. 

%%%%%%%%%%%%%%%%%%%%%%%%%%%%%%
\subsubsection{Perturbative characteristic gluing to Kerr} \label{SEC9990121}

The following is a first version of our main result for characteristic gluing to Kerr along one null hypersurface, see Theorem \ref{THMMAINPRECISE} for a precise version.

\begin{theorem}
\label{THMfirstINTROmainresult991}
Consider a strongly asymptotically flat family of sphere data $x_R$ on spheres $S_R$. For $R\geq1$ sufficiently large, there exist 1.) a perturbation $S^{\mathrm{pert}}_R$ of $S_R$ along the ingoing null hypersurface $\HHb_{R\cdot [-\de,\de]}$, 2.) a sphere $S_{2R}^{\Kerr}$ in some Kerr spacetime, and 3.) a null hypersurface $\HH_{[R,2R]}$, solving the constraint equations, and connecting $S^{\mathrm{pert}}_R$ and $S_{2R}^{\Kerr}$ and their respective sphere data.
\end{theorem}

\begin{figure}[H]
\begin{center}
\includegraphics[width=9.5cm]{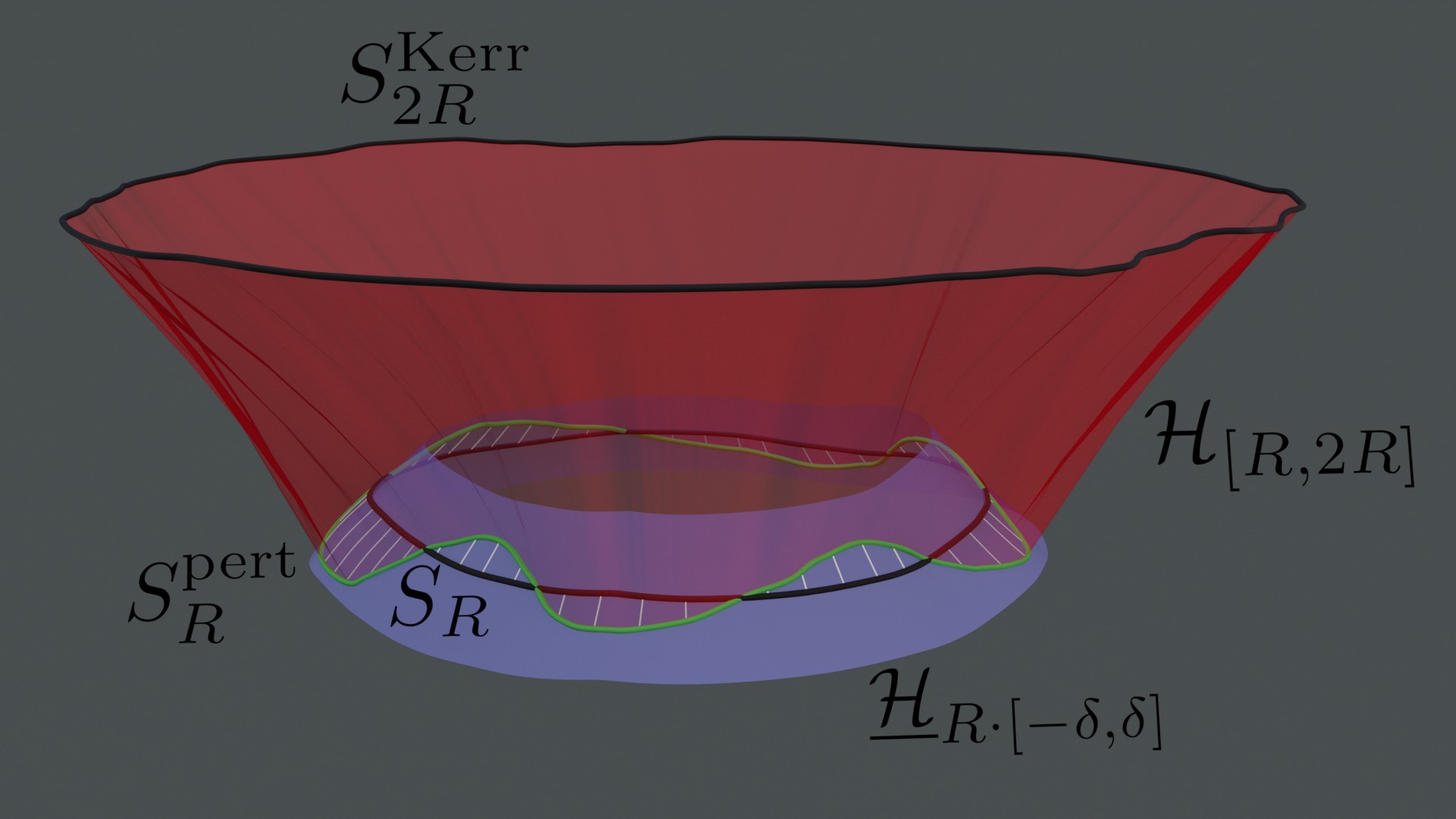} 
\vspace{0.4cm} 
\caption{Perturbative characteristic gluing to the Kerr family.}\label{THEO1}
\end{center}
\end{figure}

\ni The above perturbative characteristic gluing to Kerr is $C^2$-gluing for the metric components. In Theorem \ref{THMfirstINTROmainresult991} we glue to a reference sphere in Kerr. We could alternatively glue to a perturbation of the reference sphere in Kerr to avoid perturbing $S_R$ to $S_R^{\mathrm{pert}}$. In Theorem \ref{THMfirstINTROmainresult991} it is not necessary to have a family of sphere data data. Indeed, one can replace this family with one fixed sphere datum with sufficiently strong bounds.  

%%%%%%%%%%%%%%%%%%%%%%%%%%%%%%
\subsubsection{Bifurcate characteristic gluing to Kerr} \label{SEC9990122}

We can also characteristically glue higher-order derivatives in all directions (without perturbing any of the spheres) to Kerr by applying the bifurcate characteristic gluing of \cite{ACR3,ACR1}. This yields higher-regularity gluing of metric components. We refer to Theorem \ref{THMdoubleCHARgluingTOkerr} for a precise version of the following.

\begin{theorem} 
\label{THMcharGluingTWOfirstintroversion199901} 
On spheres $S_R$ let $x_R$ be a strongly asymptotically flat family of sphere data together with prescribed higher-order derivatives. For $R\geq1$ sufficiently large, we can characteristically glue, to higher-order, along two null hypersurfaces bifurcating from an auxiliary sphere, the sphere $S_R$ to a sphere $S_{2R}^{\Kerr}$ in some Kerr spacetime.
\end{theorem}

\begin{figure}[H]
\begin{center}
\includegraphics[width=9.5cm]{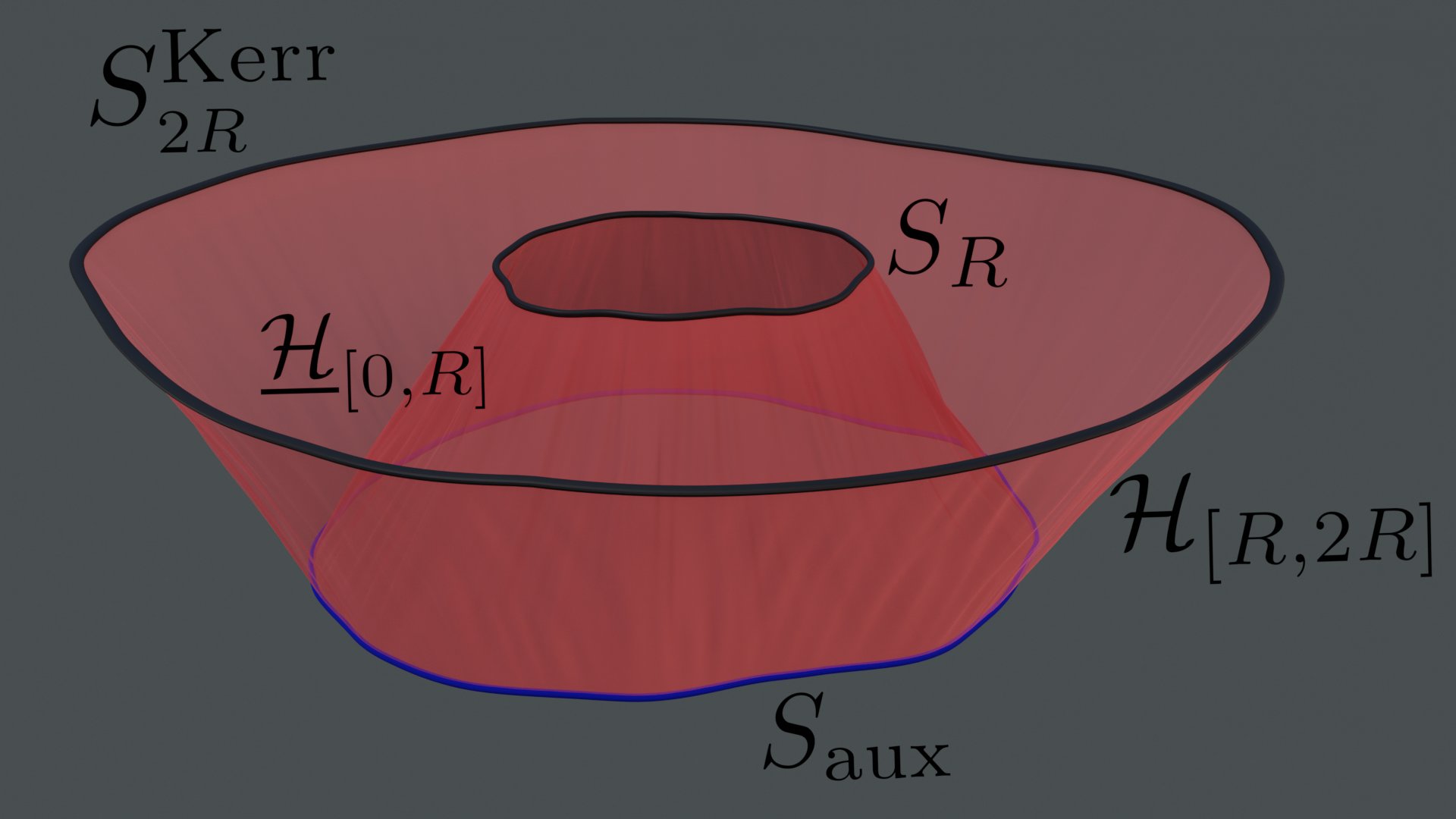} 
\vspace{0.4cm} 
\caption{Bifurcate characteristic gluing to the Kerr family.}\label{THEO2}
\end{center}
\end{figure}
%\vspace{-0.8cm}

\ni In Theorem \ref{THMcharGluingTWOfirstintroversion199901} it is not necessary to have a strongly asymptotic family of sphere data. Indeed, one can replace this family with one fixed sphere data with sufficiently strong bounds.
%%%%%%%%%%%%%%%%%%%%%%%%%%%%%%
\subsubsection{Spacelike gluing to Kerr} \label{SEC9990123} As corollary of Theorem \ref{THMcharGluingTWOfirstintroversion199901} we can deduce spacelike gluing to Kerr for strongly asymptotically flat spacelike initial data, see Corollary \ref{THMspacelikeGLUINGtoKERRv2} for a precise version.

\begin{corollary}
\label{THMspacelikeGLUINGtoKERRv2INTRO} Let $(\Si,g,k)$ be smooth strongly asymptotically flat spacelike initial data with asymptotic invariants
\begin{align*} 
\begin{aligned} 
(\mathbf{E}_{\mathrm{ADM}}, \mathbf{P}_{\mathrm{ADM}}, \mathbf{L}_{\mathrm{ADM}}, \mathbf{C}_{\mathrm{ADM}})
\end{aligned} %\label{}
\end{align*}
such that $(\mathbf{E}_{\mathrm{ADM}} )^2 > \vert \mathbf{P}_{\mathrm{ADM}} \vert^2$. Then, sufficiently far out, $(g,k)$ can be smoothly glued across a compact region to spacelike initial data for some Kerr spacetime.
\end{corollary} 

\begin{figure}[H]
\begin{center}
\includegraphics[width=9.5cm]{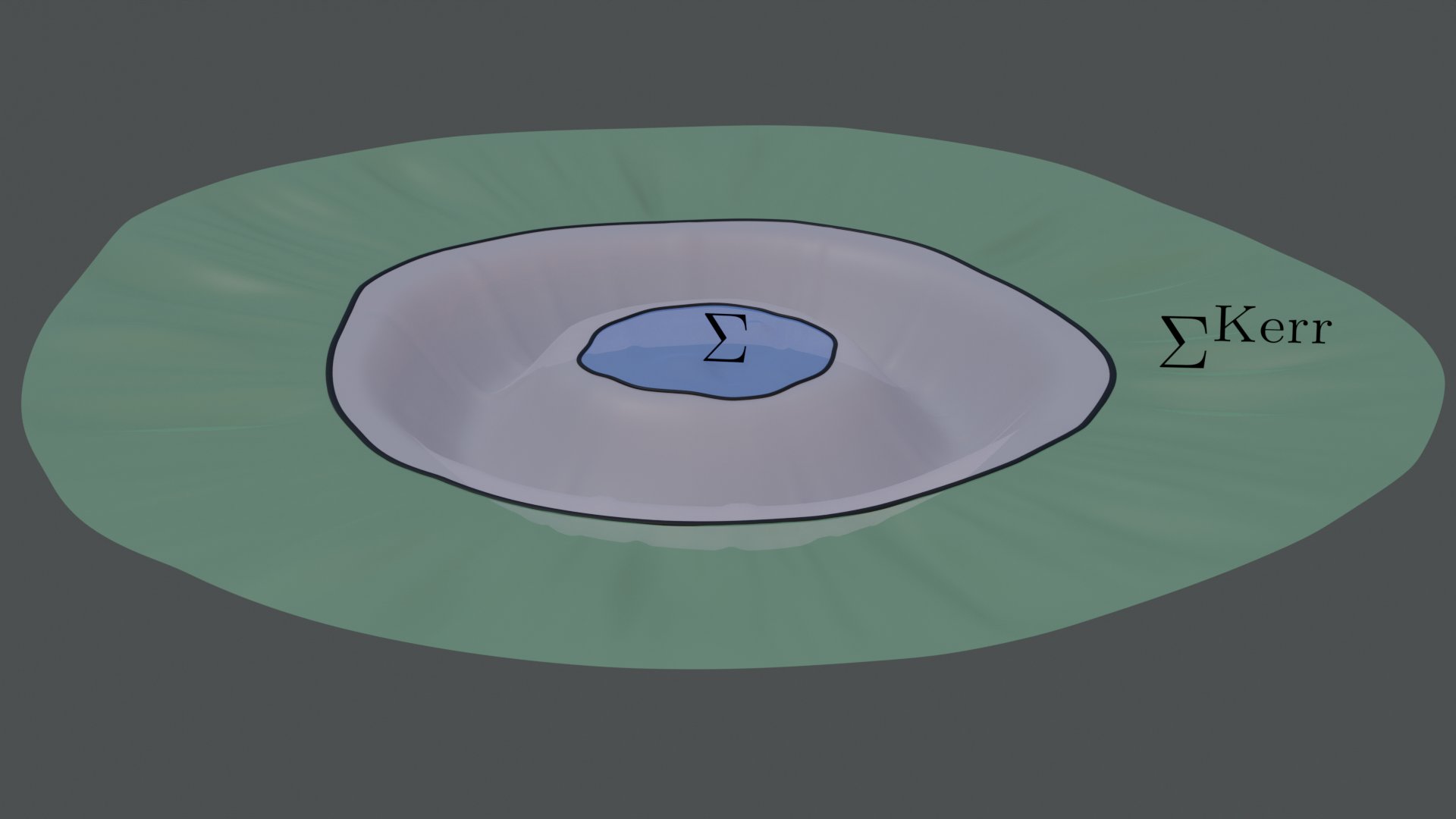} 
\vspace{0.4cm} 
\caption{Smooth spacelike gluing to Kerr.}\label{FIGspacelikegl1}
\end{center}
\end{figure}
%\vspace{-0.8cm}

\ni \emph{Remarks on Corollary \ref{THMspacelikeGLUINGtoKERRv2INTRO}.}
\begin{enumerate}
\item The assumption of strong asymptotic flatness of the spacelike initial data corresponds to working in the center-of-mass frame of the isolated gravitational system, see, for example, \cite{ChrKl93}.
\item The proof of Corollary \ref{THMspacelikeGLUINGtoKERRv2INTRO} is by combining the bifurcate gluing to Kerr with local existence \cite{LukChar,LukRod1} for the characteristic initial value problem.
\end{enumerate}

\begin{figure}[H]
\begin{center}
\includegraphics[width=9.5cm]{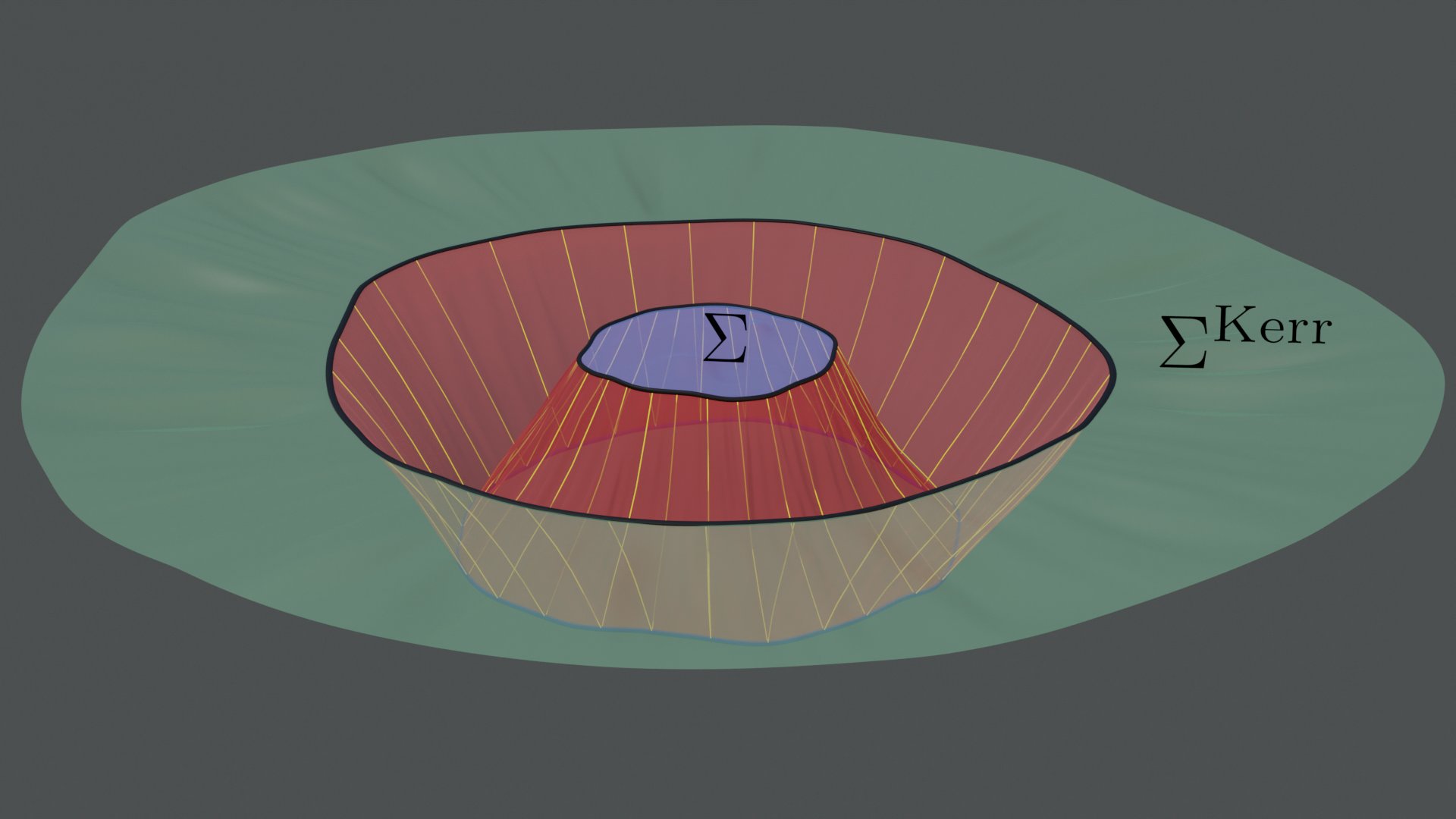} 
\vspace{0.4cm} 
\caption{Application of bifurcate characteristic gluing in the proof of spacelike gluing.}\label{FIGspacelikegl2}
\end{center}
\end{figure}
%\vspace{-0.8cm}

%%%%%%%%%%%%%%%%%%%%%%%%%%%%%%
%%%%%%%%%%%%%%%%%%%%%%%%%%%%%%
\subsection{Overview of the main ideas} \label{SEC999013}

%%%%%%%%%%%%%%%%%%%%%%%%%%%%%%
\subsubsection{Main steps}\label{SEC9990131} In this section we will outline the main steps of the proofs of Theorems \ref{THMfirstINTROmainresult991} and \ref{THMcharGluingTWOfirstintroversion199901}. We refer the reader to Figure \ref{FIGstrategy} below for an illustration of the relevant spheres and charges.

\begin{figure}[H]
\begin{center}
\includegraphics[width=12cm]{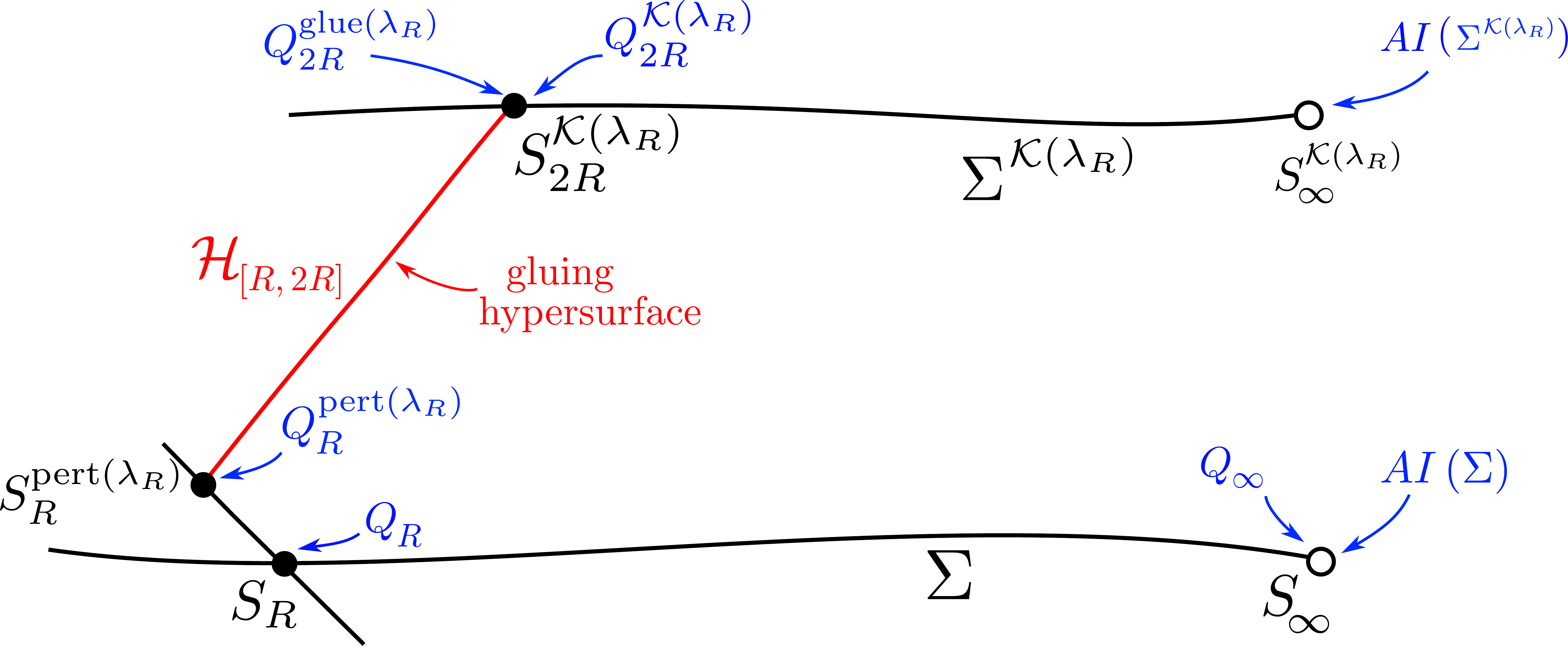} 
\vspace{0.4cm} 
\caption{The relevant spheres and charges for the characteristic gluing to the Kerr family.}\label{FIGstrategy}
\end{center}
\end{figure}

\begin{enumerate}
\item We setup the problem by distinguishing two cases:
\begin{enumerate}
\item We are given a strongly asympotically flat spacelike initial data set $(\Si,g,k)$ foliated by $2$-spheres $S_R$.
\item We are given a strongly asymptotically flat family of spheres $S_R$ with sphere data $(x_R)$ which are not lying in strongly asymptotically flat spacelike initial data.
\end{enumerate}

\item We apply the perturbative codimension-$10$ characteristic gluing of \cite{ACR3,ACR1} to glue a perturbation $S_R^{\mathrm{pert}(\la_R)}$ of $S_R$ to a sphere $S_{2R}^{\KK(\la_R)}$ in some Kerr spacetime $\KK(\la_R)$ along a null hypersurface $\HH_{[R,2R]}$. Here the $\la_R$ is a parametrization of the Kerr spacetimes through asymptotic invariants, see Step (4) below. The gluing holds up to the $10$-dimensional space of charges. We denote the associated charges on $S_R$ and $S_R^{\mathrm{pert}(\la_R)}$ by $Q_R$ and $Q_R^{\mathrm{pert}(\la_R)}$, respectively. Moreover, we denote by $Q_{2R}^{\KK(\la_R)}$ the charges on $S_{2R}^{\KK(\la_R)}$ calculated from Kerr, and by $Q_{2R}^{\mathrm{glue}(\la_R)}$ the charges on the same sphere calculated from the gluing solution on $\HH_{[R,2R]}$. We consider the charge difference
\begin{align} 
\begin{aligned} 
(\Delta Q)(\la_R):= Q_{2R}^{\KK(\la_R)} - Q_{2R}^{\mathrm{glue}(\la_R)}.
\end{aligned} \label{EQdefDIFFchargegs}
\end{align}
Our goal is to determine a Kerr parameter $\la_R$ for which $(\Delta Q)(\la_R)=0$.

\item We derive asymptotic expansions for the charges $Q_R$ for large $R\geq1$ (and denote the limits by $Q_\infty$). In case (a) of a spacelike hypersurface $\Si$, we show that these are related to the asymptotic invariants $\mathrm{AI}(\Si)$ of $(\Si,g,k)$, thus yielding a geometric interpretation of the charges. See Section \ref{SEC9990133}.

\item We make use of the $10$-dimensional parametrization $\la \in \RRR^{10}$ of Kerr spacelike initial data $(\Si^{\KK(\la)},g^{\KK(\la)},k^{\KK(\la)})$ by its asymptotic invariants $\mathrm{AI}(\Si^{\KK(\la)})$ developed in Chru\'sciel--Delay \cite{ChruscielDelay}, and construct spheres $S_{2R}^{\KK(\la)}$ with sphere data $x_{2R}^{\KK(\la)}$ lying in these spacelike initial data. In case (a) we consider the parameter $\la$ of Kerr spacelike initial data to lie in a ball $B$ of finite radius in $\RRR^{10}$. In case (b) we consider a larger set of parameters $\la=\la_R$, namely, an ellipsoid $\EE_R$ with semi-major axis proportional to $R^{1/2}$. See Sections \ref{SEC9990134} and \ref{SEC9990135}.

\item We derive a homotopy between the charge difference $(\Delta Q)(\la_R)$ and an appropriate difference between the asymptotic charges $Q_\infty$ and the asymptotic invariants $\mathrm{AI}(\Si^{\KK(\la_R)})$ of the Kerr initial data $\Si^{\KK(\la_R)}$. The latter difference is shown to always admit a root $\la'_R$. We prove uniform estimates for the homotopy that allow us to conclude by a topological degree argument that the charge difference  $(\Delta Q)(\la_R)$ also admits a root. See Section \ref{SEC9990135}.

\end{enumerate}

%%%%%%%%%%%%%%%%%%%%%%%%%%%%%%
\subsubsection{Sphere data and charges}\label{SEC9990132} 

For a given sphere $S$, the sphere data $x$ on $S$ is given by the following geometric components on $S$,
\begin{align*} 
\begin{aligned} 
x = (\Om,\gd, \trchi, \chih, \trchib, \chibh, \eta, \om, D\om, \omb, \Du\omb, \a, \be, \rh, \si, \beb, \ab).
\end{aligned} %\label{}
\end{align*}
The above components are expressed in a null frame in the context of a double null coordinate system. For the precise definitions we refer to Section \ref{SECsetupdoublenull}.  The data $x$ determines all derivatives of the spacetime time metric up to order $2$ (see also Section \ref{SECspheredataDEF2333}). Furthermore, given sphere data on $S$ we introduce the associated charges 
\begin{align*} 
\begin{aligned} 
Q= \lrpar{\mathbf{E},\mathbf{P},\mathbf{L},\mathbf{G}} \text{ on } S
\end{aligned} %\label{}
\end{align*}
to be the integrals over $S$ as defined by \eqref{EQdefcharges999919999}.

We define \emph{strongly asymptotically flat sphere data} on a family of spheres $S_R$ in accordance with the decay towards \emph{spacelike} infinity in the works \cite{ChrKl93} and \cite{KlNic}. Indeed, we show by explicit construction that each strongly asymptotically flat spacelike initial data admits strongly asymptotically flat sphere data on families of spheres (but in general we do not assume that sphere data stems from spacelike initial data). Our construction is such that the special case of Schwarzschild spacelike data (expressed in isotropic coordinates so that strong asymptotic flatness holds) leads precisely to the family of Schwarzschild reference sphere data with respect to Eddington--Finkelstein double null coordinates. To achieve this for general spacelike initial data, we rescale the coordinate sphere $S_r$ to $S_1$, apply the coordinate change from isotropic coordinates to Schwarzschild coordinates, make appropriate gauge choices for $\widehat{L}$, $\widehat{\Lb}$ and $\Om$, and use the definitions of Ricci coefficients and null curvature components. We show that the rescaled quantities are well-defined and derive estimates and then we rescale back up. See Sections \ref{SECspacelikeRescalingtoSmalldata} and \ref{SECproofExistenceConstruction}.

%%%%%%%%%%%%%%%%%%%%%%%%%%%%%%
\subsubsection{Relation of charges to asymptotic invariants}\label{SEC9990133}

For strongly asymptotically flat families of sphere data $x_R$ on spheres $S_R$, the following asymptotic expansions hold for large $R\geq1$,
\begin{align*} 
\begin{aligned} 
\mathbf{E}\lrpar{ S_{R}} =& \mathbf{E}_\infty + \smallO(R^{-1/2}), & \mathbf{P}\lrpar{ S_{R}} =&  \smallO(R^{-1/2}), \\ 
\mathbf{L}\lrpar{ S_{R}} =& \mathbf{L}_\infty + \smallO(1), & \mathbf{G}\lrpar{ S_{R}} =& \mathbf{G}_\infty + \smallO(1).
\end{aligned} %\label{}
\end{align*}
where
$\lrpar{\mathbf{E}_\infty, \mathbf{P}_\infty=0, \mathbf{L}_\infty, \mathbf{G}_\infty}$ are defined as the limits of $(\mathbf{E},\mathbf{P},\mathbf{L},\mathbf{G})$ on $S_R$ as $R\to \infty$. These asymptotic charges can also be defined for more general families of sphere data (not necessarily strongly asymptotically flat), in which case we expect $\mathbf{P}_\infty \neq 0$ and $\mathbf{G}_\infty=+\infty$. On the other hand, if the family of sphere data lies in strongly asymptotically flat spacelike initial data, then we have the following stronger decay rates for $\mathbf{E}$ and $\mathbf{P}$,
\begin{align} 
\begin{aligned} 
 \mathbf{E}(S_R) = \mathbf{E}_{\infty} + \OO(R^{-1}), \,\, \mathbf{P}(S_R) = \OO(R^{-3/2}).
\end{aligned} \label{EQdecaystrongINTRO9993}
\end{align}

\ni The charges $(\mathbf{E},\mathbf{P},\mathbf{L},\mathbf{G})$ have a well-defined connection to the ADM asymptotic invariants. Indeed, for spheres $S_R$ in asymptotically flat spacelike initial data with well-defined ADM asymptotic invariants, we can relate the charges $(\mathbf{E},\mathbf{P},\mathbf{L},\mathbf{G})$ to the local integrals $\mathbf{E}_{\mathrm{ADM}}^{\mathrm{loc}}$ of the ADM \emph{energy}, $\mathbf{P}_{\mathrm{ADM}}^{\mathrm{loc}}$ of the ADM \emph{linear momentum}, $\mathbf{L}_{\mathrm{ADM}}^{\mathrm{loc}}$ of the ADM \emph{angular momentum}, and $\mathbf{C}_{\mathrm{ADM}}^{\mathrm{loc}}$ of the ADM \emph{center-of-mass} as follows,
\begin{align} 
\begin{aligned} 
\mathbf{E}(S_R) =& \mathbf{E}_{\mathrm{ADM}}^{\mathrm{loc}}(S_R)+\smallO(1), & \mathbf{P}(S_R) =& \mathbf{P}_{\mathrm{ADM}}^{\mathrm{loc}}(S_R) +\smallO(1), \\
\mathbf{L}(S_R) =& \mathbf{L}_{\mathrm{ADM}}^{\mathrm{loc}}(S_R) +\smallO(1), & \mathbf{G}(S_R) =& \mathbf{C}_{\mathrm{ADM}}^{\mathrm{loc}}(S_R) - R \cdot \mathbf{P}_{\mathrm{ADM}}^{\mathrm{loc}}(S_R) +\smallO(1).
\end{aligned} \label{EQintro9993999}
\end{align}
Hence in that case,
\begin{align*} 
\begin{aligned} 
\mathbf{E}_\infty = \mathbf{E}_{\mathrm{ADM}}, \,\, \mathbf{P}_\infty = \mathbf{P}_{\mathrm{ADM}}, \,\, \mathbf{L}_\infty = \mathbf{L}_{\mathrm{ADM}}.
\end{aligned} %\label{}
\end{align*}

\ni For families of sphere data in an asymptotically flat spacelike hypersurface with non-vanishing total ADM linear momentum $\mathbf{P}_{\mathrm{ADM}}\neq 0$ (such hypersurfaces are not strongly asymptotically flat), \eqref{EQintro9993999} shows that $\mathbf{P}_\infty \neq0$, and subsequently, $\vert \mathbf{G}_\infty \vert=+\infty$. Importantly, the Kerr spacelike initial data satisfies in general $\mathbf{P}_{\mathrm{ADM}}\neq 0$ which has significant repercussions for our analysis of the gluing problem to the Kerr family. In particular, it forces us to consider and prove delicate estimates for spacelike initial data with very large center-of-mass $\mathbf{C}_{\mathrm{ADM}}$, see also the discussion in Sections \ref{SEC9990134} and \ref{SEC9990135} below.

On the other hand, for families of sphere data in strongly asymptotically flat spacelike hypersurfaces (in which case $\mathbf{P}_{\mathrm{ADM}}=0$ and $\mathbf{P}_{\mathrm{ADM}}^{\mathrm{loc}}(S_R)=\OO(R^{-3/2})$), we have by \eqref{EQintro9993999} that $\mathbf{G}_\infty= \mathbf{C}_{\mathrm{ADM}}$ is well-defined.

%%%%%%%%%%%%%%%%%%%%%%%%%%%%%%
\subsubsection{Parametrization of the Kerr reference family}\label{SEC9990134}
We define Kerr reference spacelike initial data by using the Kerr parameter map constructed in Chrusciel--Delay \cite{ChruscielDelay}. This map takes as input the Kerr parameters of mass $M$ and (renormalized) angular momentum $a$, as well as a rotation matrix $\Lambda'$, a Lorentz boost matrix $\Lambda$, and a translation vector $\mathbf{a}$, and yields as output spacelike initial data for Kerr. The map is constructed by taking a constant time slice in Boyer-Lindquist coordinates, and rotating, translating, and Lorentz boosting it following the given input.

We define the so-called Kerr asymptotic invariants map $\mathcal{K}$, which takes as input a set of asymptotic invariants $\lambda=(\mathbf{E}, \mathbf{P}, \mathbf{L},\mathbf{C})\in\RRR \times \RRR^3 \times \RRR^3 \times \RRR^3$ and yields as output Kerr spacelike initial data $\mathcal{K}(\lambda)=(M,a,\La, \La',\mathbf{a})$ that realizes these asymptotic invariants. In Chrusciel--Delay \cite{ChruscielDelay} it is shown that this map is well defined for $\mathbf{E}^2 > \vert \mathbf{P} \vert^2$.

We derive asymptotic convergence estimates for the local ADM integrals of spheres $S_R$ in Kerr spacelike initial data to the asymptotic ADM invariants. This is accomplished by rescaling close to the Schwarzschild data with small mass and by explicitly estimating the Kerr parameters $\KK(\la)=(M,a,\La,\La',\mathbf{a})$  through the asymptotic invariants (these estimates rely on the full class of symmetries of Minkowski spacetime). Then the estimates for the Kerr parameters yield precise control of the spacelike data which can then be used to estimate the convergence rates of the local ADM integrals. Importantly, the derived convergence rates are \emph{uniform} over a large set of parameters $\la$, see Section \ref{SEC9990135} below.

Moreover, we define sphere data $x_R^{\Kerr}$ on spheres $S_R$ in the Kerr spacelike initial data and derive asymptotic expansions for the charges $(\mathbf{E},\mathbf{P},\mathbf{L},\mathbf{G})$ by using the geometric interpretation \eqref{EQintro9993999} and the above convergence estimates for the local ADM integrals.

%%%%%%%%%%%%%%%%%%%%%%%%%%%%%%
\subsubsection{Choice of Kerr to glue to}\label{SEC9990135} 

The goal is to prove that for sufficiently large $R\geq1$ we can characteristically glue to a Kerr sphere $S_{2R}^{\KK(\la)}$. Ideally, one would like to consider a fixed set of Kerr parameters $\la$ such that the spheres $S^{\KK(\la)}_{2R}$ in the corresponding spacelike initial data have asymptotic charges $(\mathbf{E}_\infty^\Kerr, \mathbf{P}_\infty^\Kerr, \mathbf{L}_\infty^\Kerr, \mathbf{G}_\infty^\Kerr)$ close to the given $(\mathbf{E}_\infty, \mathbf{P}_\infty, \mathbf{L}_\infty, \mathbf{G}_\infty)$, and subsequently argue that there exists a $\la$ in that set which solves the gluing problem.

However, for each fixed $\la$ with $\mathbf{P}^\Kerr_{\mathrm{ADM}}\neq 0$ we have by \eqref{EQintro9993999} that as $R\to \infty$,
\begin{align} 
\begin{aligned} 
\mathbf{G}^{\Kerr}(S^{\KK(\la)}_{2R}) =& \underbrace{\mathbf{C}_{\mathrm{ADM}}^{\mathrm{loc}}(S^{\KK(\la)}_{2R})}_{\to \mathbf{C}^\Kerr_{\mathrm{ADM}}} - (2R) \cdot \underbrace{\mathbf{P}_{\mathrm{ADM}}^{\mathrm{loc}}(S^{\KK(\la)}_{2R})}_{\to \mathbf{P}_{\mathrm{ADM}}^\Kerr} +\smallO(1) \to \infty,
\end{aligned} \label{EQblowupG}
\end{align}
which in particular shows that $\mathbf{G}_\infty^\Kerr$ is far away from matching the finite $\mathbf{G}_\infty$. Thus we have to change our approach and consider an \emph{$R$-dependent} set of $\la$ which accommodates bounded $\mathbf{G}^{\Kerr}(S^{\KK(\la)}_{2R})$ by allowing for growing center-of-mass $\mathrm{C}^\Kerr_{\mathrm{ADM}} = \OO(R^{1/2})$ and small linear momentum $\mathrm{P}_{\mathrm{ADM}}^\Kerr = \OO(R^{-1/2})$ to cancel to top order. Namely, we consider the ellipsoid $\EE_R(\mathbf{E}_\infty)$ defined by
\begin{align*} 
\begin{aligned} 
\lrpar{R^{1/2} \vert \mathbf{E}(\la)- \mathbf{E}_\infty \vert}^2 + \lrpar{R^{1/2} \vert \mathbf{P}(\la) \vert}^2+\lrpar{R^{-1/4} \vert \mathbf{L}(\la) \vert}^2+\lrpar{R^{-1/2} \vert \mathbf{C}(\la) \vert}^2 \leq \lrpar{\mathbf{E}_\infty}^2.
\end{aligned} %\label{}
\end{align*}

\ni In the simpler case where the spheres $S_R$ lie in strongly asymptotically flat spacelike hypersurfaces, we have the stronger decay $\mathbf{P}(S_R)=\OO(R^{-3/2})$ which implies in the matching process that $\mathbf{G}^{\Kerr}(S^{\KK(\la)}_{2R})$ remains finite as $R\to \infty$ (unlike in \eqref{EQblowupG}).

To determine the Kerr parameter $\la_R$ which makes the charge difference $(\Delta Q)(\la_R)=0$ (see \eqref{EQdefDIFFchargegs} for definition), we use the asymptotic expansions of $\QQ_R$ and $\QQ_{2R}^{\KK(\la_R)}$ (discussed in Sections \ref{SEC9990133} and \ref{SEC9990134} above) to construct a homotopy from 
$F_1(\la_R) := (\Delta Q)(\la_R)$ to the mapping $F_0(\la_R)$ defined by
\begin{align*} 
\begin{aligned} 
(\mathbf{E}_\infty,\mathbf{P}_\infty, \mathbf{L}_\infty, \mathbf{C}_\infty) - \lrpar{\mathbf{E}_{\mathrm{ADM}}^{\Kerr},\mathbf{P}_{\mathrm{ADM}}^{\Kerr},\mathbf{L}_{\mathrm{ADM}}^{\Kerr},\mathbf{C}_{\mathrm{ADM}}^{\Kerr}-2R \cdot \mathbf{P}_{\mathrm{ADM}}^{\Kerr}}.
\end{aligned} %\label{}
\end{align*}
The mapping $F_0(\la_R)$ has a unique zero in the interior of $\EE_R(\mathbf{E}_\infty)$ for large $R\geq1$. Moreover, as indicated in the previous section, the asymptotic expansions for $\QQ_{2R}^{\KK(\la_R)}$ hold \emph{uniformly} for large $R\geq1$ and $\la \in \EE_R(\mathbf{E}_\infty)$, so that in particular we have uniform estimates for the constructed homotopy. Therefore we conclude by a topological degree argument that the charge difference $(\Delta Q)(\la_R)$ must have a zero.

%%%%%%%%%%%%%%%%%%%%%%%%%%%%%%%%%%%%%%%%
\subsection{Overview of the paper} The paper is structured as follows.
\begin{itemize}
\item In Section \ref{sec:GeometryOfNullHypersurfaces} we introduce the notation and state the definitions and preliminaries of this paper.
\item In Section \ref{SECpreciseSTATEMENTMAIN111} we precisely state the main results of this paper.
\item In Section \ref{SECproofMainTheorem1} we prove the main theorem of this paper, Theorem \ref{THMMAINPRECISE}.
\item In Section \ref{SECspacelikecorollary} we prove Corollary \ref{THMspacelikeGLUINGtoKERRv2}, the gluing of spacelike initial data to Kerr.
\item In Section \ref{SECCSgluingbasics} we introduce spacelike initial data for the Einstein vacuum equations, and asymptotic invariants.
\item In Section \ref{SECstatementConstruction} we prove construct strongly asymptotically flat families of sphere data from strongly asymptotically flat spacelike initial data. 
\item In Section \ref{SECappKerrFamilyDetails} we recall the construction of Kerr reference spacelike initial data, define Kerr reference sphere data, and prove estimates.
\item In Appendix \ref{SECcompletenessKERR} we prove that the Kerr reference spheres admit future-complete outgoing null congruences.
\item In Appendix \ref{SECellEstimatesSpheres} we recapitulate the Fourier theory for Hodge systems on Riemannian $2$-spheres, and prove geometric identities and Fourier decomposition estimates.

\end{itemize}

%%%%%%%%%%%%%%%%%%%%%%%%%%%%%%%%%%%%%%%%
\subsection{Acknowledgements} S.A. acknowledges support through the NSERC grant 502581 and the Ontario Early Researcher Award. S.C. acknowledges support through the NSF grant DMS-1439786 of the Institute for Computational and Experimental Research in Mathematics (ICERM). I.R. acknowledges support through NSF grants DMS-2005464, DMS-1709270 and a Simons Investigator Award.

%%%%%%%%%%%%%%%%%%%%%%%%%%%%%%%%%%%%%%%%
\section{Notation, definitions and preliminaries} \label{sec:GeometryOfNullHypersurfaces}

\ni In this section we introduce the notation, definitions and preliminaries of this paper. We start with the following general notation.
\begin{itemize}
\item For two real numbers $A$ and $B$, the inequality $A \les B$ means that there is a universal constant $C>0$ such that $A \leq C \, B$. 
\item Greek indices range over $\a=0,1,2,3$, lowercase Latin indices over $a=1,2,3$ and uppercase Latin indices over $A=1,2$.
\item For a real number $r>0$ and a point $x$ in a metric space $X$, denote by $B(x,r)$ the open ball in $X$ of radius $r$ centered at $x$.
\item For real numbers $\varep>0$ and $\a \geq 0$, let $\OO(\varep^\a)$ and $\smallO(\varep^\a)$ respectively denote terms such that
\begin{align*} 
\begin{aligned} 
\lim\limits_{\varep\to0} \frac{\OO(\varep^\a)}{\varep^\a} <\infty, \,\,\,\, \lim\limits_{\varep\to0} \frac{\smallO(\varep^\a)}{\varep^\a} =0.
\end{aligned} %\label{}
\end{align*}
\item For given Cartesian coordinates $(x^1,x^2,x^3)$, we define the associated spherical coordinates $(r,\th^1,\th^2)$ by 
\begin{align} 
\begin{aligned} 
x^1=r \sin\th^1\cos\th^2, \,\, x^2= r\sin\th^1\sin\th^2, \,\, x^3 = r\cos\th^1.
\end{aligned} \label{EQdefinSPHERICALAFcoord}
\end{align}
\end{itemize}
%%%%%%%%%%%%%%%%%%%%%%%%%%%%%%%%%%%%%%%%
\subsection{Double null coordinates} \label{SECsetupdoublenull} In this section we summarize the standard setup of double null coordinates, Ricci coefficients and null curvature components. We refer to Section 2.1 in \cite{ACR1} for full details.

Let $(\MM,\g)$ be a vacuum spacetime, and denote by $\D$ its covariant derivative and by $\Rbf$ its Riemann curvature tensor. Let $u$ and $v$ be two local optical functions on $\MM$, and denote
\begin{align*} 
\begin{aligned} 
\HH_{u_0} = \{u=u_0\}, \,\, \HHb_{v_0} = \{v=v_0\}, \,\, S_{u_0,v_0} = \HH_{u_0} \cap \HHb_{v_0},
\end{aligned} %\label{}
\end{align*}
where we assume the optical functions $u$ and $v$ are such that the $S_{u,v}$ are spacelike $2$-spheres. Let $\gd$ denote the induced metric on $S_{u,v}$, and $\Nd$ the induced covariant derivative.  Let $r(u,v)$ denote the \emph{area radius} of $(S_{u,v},\gd)$, defined by
\begin{align*} 
\begin{aligned} 
\mathrm{area}_\gd \lrpar{S_{u,v}} = 4\pi r^2.
\end{aligned}
\end{align*}
Define the \emph{geodesic null vectorfields} $L'$ and $\Lb'$, the \emph{null lapse} $\Om$, and the \emph{normalized null vectorfields} $\widehat{L}$ and $\widehat{\Lb}$ by
\begin{align} 
\begin{aligned} 
L' := -2 \D u, \,\, \Lb' := -2 \D v, \,\, \Om^{-2} := -\half \g\lrpar{L',\Lb'}, \,\, \widehat{L} := \Om L', \,\,\widehat{\Lb} := \Om \Lb'.
\end{aligned} \label{EQdefNullPairsINM}
\end{align}
Let $(\th^1,\th^2)$ be local coordinates on $S_{u_0,v_0}$, for some given real numbers $v_0>u_0$. We extend $(\th^1,\th^2)$ to $\MM$ by first transporting them along the null generators $L'$ of $\HH_{u_0}$ and then onto $\MM$ along the generators $\Lb'$ of the null hypersurfaces $\HHb_v$. The coordinates $(u,v,\th^1,\th^2)$ are called \emph{double null coordinates}. In particular, it holds in double null coordinates that
\begin{align} 
\begin{aligned} 
L= \pr_v + b, \,\, \Lb = \pr_u.
\end{aligned} \label{EQvectorexpr2488}
\end{align}
and the spacetime metric $\mathbf{g}$ can be written as
\begin{align} 
\begin{aligned} 
\g = -4 \Om^2 du dv + \gd_{AB} \lrpar{d\th^A + b^A dv}\lrpar{d\th^B + b^B dv},
\end{aligned} \label{EQdoublenullform1}
\end{align}
where the \emph{shift vector} $b=b^A \pr_A$ is an $S_{u,v}$-tangential vectorfield, satisfying
\begin{align*} 
\begin{aligned} 
b=0 \text{ on } \HH_{u_0}.
\end{aligned} %\label{}
\end{align*}
Through the coordinates $(\th^1,\th^2)$, we define on each $S_{u,v}$ the unit round metric 
\begin{align} 
\begin{aligned} 
\gac:= \lrpar{d\th^1}^2 + \sin^2 \th^1 \lrpar{d\th^2}^2.
\end{aligned} \label{EQdefRoundUnitMetric}
\end{align}
We also define vector spherical harmonics with respect to $\gac$ on $S_{u,v}$, see Appendix \ref{SECellEstimatesSpheres} for definitions. Furthermore, we decompose the induced metric $\gd$ into
\begin{align} 
\begin{aligned} 
\gd = \phi^2 \gd_c \text{ where } \phi^2 := \frac{\sqrt{\gd}}{\sqrt{\gac}}, \,\, \gd_c := \phi^{-2} \gd,
\end{aligned} \label{EQdefdecomposition1}
\end{align}
where $\sqrt{\gd}$ and $\sqrt{\gac}$ denote the volume forms of $\gd$ and $\gac$ with respect to $(\th^1,\th^2)$, respectively. 

We define the \emph{Ricci coefficients} as follows. For $S_{u,v}$-tangent vectorfields $X$ and $Y$, let
\begin{align} \begin{aligned}
\chi(X,Y) :=& \g(\D_X \widehat{L},Y), & \chib(X,Y) :=& \g(\D_X \widehat{\Lb},Y), \\
\zeta(X) :=& \half \g(\D_X \widehat{L}, \widehat{\Lb}), & \underline{\zeta}(X) :=& \half \g(\D_X \widehat{\Lb}, \widehat{L}), \\
\eta :=& \zeta + \di \log \Om , & \etab :=& -\zeta + \di \log \Om, \\
\om :=& D \log \Om, & \omb :=& \Du \log \Om,
\end{aligned}\label{DEFricciCoefficients} \end{align}
where $\di$ is the extrinsic derivative of $S_{u,v}$. We remark that 
\begin{align} 
\begin{aligned} 
\zeta = - \underline{\zeta}, \,\, \etab=-\eta+2\di\log\Om.
\end{aligned} \label{EQriccirelationetabeta} 
\end{align}
We define the \emph{null curvature components} as follows. For $S_{u,v}$-tangent vectorfields $X$ and $Y$, let
\begin{align} \begin{aligned}
\alpha(X,Y) :=& \Rbf(X,\widehat{L}, Y, \widehat{L}), & \beta(X) :=& \half \Rbf(X, \widehat{L},\widehat{\Lb},\widehat{L}), \\
\rh :=& \frac{1}{4} \Rbf(\widehat{\Lb}, \widehat{L}, \widehat{\Lb}, \widehat{L}), & \sigma \iin(X,Y) :=& \half \Rbf(X,Y,\widehat{\Lb}, \widehat{L}), \\
\beb(X) :=& \half \Rbf(X, \widehat{\Lb},\widehat{\Lb},\widehat{L}), & \ab(X,Y) :=& \Rbf(X,\widehat{\Lb}, Y, \widehat{\Lb}).
\end{aligned}\label{EQnullcurvatureCOMPDEF} \end{align}

%%%%%%%%%%%%%%%%%%%%%%%%%%%%%%%%%%%%%%%%
\subsection{Null structure equations} \label{SECnullstructureequations} The geometric setting and the Einstein equations imply relations between the metric components, Ricci coefficients and null curvature components, the so-called \emph{null structure equations}. Before stating them, we introduce the following notation from Chapter 1 of \cite{ChrFormationBlackHoles}.
\begin{itemize}
\item For a $S_{u,v}$-tangent tensor $W$ on $\MM$, let
\begin{align*}
DW:= \Lied_L W, \,\, \underline{D}W:= \Lied_\Lb W,
\end{align*}
where $\Lied$ denotes the projection of the Lie derivative on $\MM$ onto the tangent space of $S_{u,v}$.
\item For two $S_{u,v}$-tangential $1$-forms $X$ and $Y$, let
\begin{align} \begin{aligned}
(X,Y):=& \gd(X,Y), & ({}^\ast X)_A :=& \in_{AB}X^B, \\
(X \widehat{\otimes} Y)_{AB} :=& X_A Y_B + X_B Y_A - (X \cdot Y)\gd_{AB}, & \Divd X :=& \Nd^A X_A, \\
 (\Nd \widehat{\otimes} Y)_{AB} :=& \Nd_A Y_B + \Nd_B Y_A - (\Divd Y)\gd_{AB}, & \Curld X :=& \in^{AB}\Nd_A X_B,
\end{aligned} \label{EQdefNotationNullStructure7778} \end{align}
where $\in$ denotes the area $2$-form of $S_{u,v}$.
\item For two symmetric $S_{u,v}$-tangential $2$-tensors $V$ and $W$, let
\begin{align*}
\tr V := \gd^{AB} V_{AB}, \,\, \widehat{V} := V - \half \tr V \gd, \,\, V \wedge W := \ind^{AB} V_{AC}W^C_{\,\,\,B}.
\end{align*}
\item For a symmetric $S_{u,v}$-tangential $2$-tensor $V$ and a $1$-form $X$, let
\begin{align*} 
\begin{aligned} 
(V \cdot X)_A := V_{AB}X^B.
\end{aligned} %\label{}
\end{align*}
\item For a symmetric $S_{u,v}$-tangential $2$-tensor $V$, let
\begin{align*} 
\begin{aligned} 
\Divd V_A := \Nd^B V_{BA}.
\end{aligned} %\label{}
\end{align*}
\item For a symmetric $S_{u,v}$-tangential tensor $W$, let $\widehat{D}W$ denote the tracefree part of $DW$ with respect to $\gd$, and $\widehat{\Du}W$ the tracefree part of $\Du W$ with respect to $\gd$.
\end{itemize}

\ni Using the above notation, we can state the \emph{null structure equations}. We have the first variation equations,
\begin{align} \begin{aligned}
D \gd =& 2 \Om \chi, & \Du \gd =& 2 \Om \chib,
\end{aligned} \label{EQfirstvariation1}\end{align}
which imply in particular that
\begin{align} 
\begin{aligned} 
D \phi = \frac{\Om \tr \chi \phi}{2},
\end{aligned} \label{EQusefulDphiRELATION}
\end{align}
the Raychauduri equations, 
\begin{align} \begin{aligned}
D \trchi + \frac{\Om}{2} (\trchi)^2 - \om \trchi =& - \Om \vert \chih \vert^2_{\gd},& \Du \trchib + \frac{\Om}{2} (\trchib)^2 - \omb \trchib =& - \Om \vert \chibh \vert^2_{\gd},
\end{aligned}\label{EQRaychauduri1}\end{align}
and further
\begin{align} \begin{aligned}
D\chih =& \Om \vert \chih \vert^2 \gd + \om \chih - \Om \a, &\Du \chibh =& \Om \vert \chibh \vert^2 \gd + \omb \chibh - \Om \ab,\\
D \eta =& \Om (\chi \cdot \etab - \beta), & \Du \etab =& \Om (\chib \cdot \eta + \beb), \\
D \omb =& \Om^2(2 (\eta, \etab) - \vert \eta \vert^2 -\rh), & \Du \om=& \Om^2(2 (\eta,\etab) - \vert \etab \vert^2 - \rh), \\
\Curld \eta=& - \half \chih \wedge \chibh - \si, & \Curld \etab =& - \Curld \eta = - \Curld \zeta, \\
D(\etab) =& - \Om (\chi \cdot \etab -\be) + 2 \di \om, & \Du(\eta) =& - \Om (\chib \cdot \eta + \beb) + 2 \di \omb.
\end{aligned} \label{EQtransportEQLnullstructurenonlinear}\end{align}
Moreover, we have the Gauss equation, 
\begin{align} \label{EQGaussEquation}
K + \frac{1}{4} \tr \chi \tr \chib - \half (\chih,\chibh) = - \rh,
\end{align}
where $K$ denotes the Gauss curvature of $S_{u,v}$, the Gauss-Codazzi equations
\begin{align} \begin{aligned}
\Divd \chih -\half \di \tr \chi + \chih \cdot \zeta - \half \trchi \zeta =& - \beta,\\
\Divd \chibh - \half \di \trchib -\chibh \cdot \zeta +\half \trchib \zeta =& \beb,
\end{aligned}\label{EQgausscodazzinonlinear1}\end{align}
and
\begin{align*}\begin{aligned}
D (\Om \trchib) =& 2 \Om^2 \Divd \etab + 2 \Om^2 \vert \etab \vert^2 - \Om^2 (\chih, \chibh) - \half \Om^2 \trchi \trchib + 2 \Om^2 \rh,\\ 
D(\Om \chibh) =& \Om^2\lrpar{(\chih, \chibh) \gd + \half \trchi \chibh + \Nd \widehat{\otimes}\etab + \etab \widehat{\otimes}\etab - \half \trchib \chih},
\end{aligned} 
\end{align*}
and
\begin{align*} \begin{aligned}
\Du (\Om \trchi) =& 2 \Om^2 \Divd \eta + 2 \Om^2 \vert \eta \vert^2 - \Om^2 (\chih, \chibh) - \half \Om^2 \trchi \trchib + 2 \Om^2 \rh, \\
\Du(\Om \chih) =& \Om^2 \lrpar{(\chih, \chibh)\gd + \half \trchib \chih + \Nd \widehat{\otimes} \eta + \eta \widehat{\otimes} \eta - \half \trchi \chibh}.
\end{aligned} 
\end{align*} 

%%%%%%%%%%%%%%%%%%%%%%%%%%%%%%%%%%%%%%%%
\ni By Proposition 1.2 in \cite{ChrFormationBlackHoles}, the following \emph{null Bianchi equations} hold.
\begin{align*} \begin{aligned}
\widehat{D} \aa - \half \Om \tr \chi \aa + 2 \om \aa + \Om \left( \Nd \widehat{\otimes} \beb + (4 \etab - \zeta) \widehat{\otimes} \beb + 3 \chibh \rh -3 {}^\ast \chibh \si \right) =&0, \\
D \be + \frac{3}{2} \Om \tr \chi \be - \Om \chih \cdot \be - \om \be - \Om \left( \Divd \a + ( \etab + 2 \zeta) \cdot \a \right) =& 0, \\
D \beb + \half \Om \tr \chi \beb - \Om \chih \cdot \beb + \om \beb + \Om \left( \di \rh- {}^\ast \di \si + 3 \etab \rh - 3 {}^\ast \etab \si - 2 \chibh \cdot \be \right) =&0, \\
D \rh + \frac{3}{2} \Om \tr \chi \rh - \Om \left( \Divd \be + (2 \etab + \zeta, \be) - \half (\chibh, \a) \right)=&0, \\
D \si + \frac{3}{2} \Om \tr \chi \si + \Om \left( \Curld \be + (2 \etab+\zeta, {}^\ast \be)- \half \chibh  \wedge \a \right)=&0,
\end{aligned}
\end{align*}
and
\begin{align*} 
\begin{aligned} 
\widehat{\Du} \a - \half \Om \tr \chib \a + 2 \omb \a + \Om \left( -\Nd \widehat{\otimes} \be - (4 \eta + \zeta) \widehat{\otimes} \be + 3 \chih \rh + 3 {}^\ast \chih \si \right) =&0, \\
\Du \beb + \frac{3}{2} \Om \tr \chib \beb - \Om \chibh \cdot \beb - \omb \beb + \Om \left( \Divd \aa + (\eta-2 \zeta) \cdot \aa \right)=&0, \\
\Du \be + \half \Om \tr \chib \be - \Om \chibh \cdot \be + \omb \be - \Om \left(\di \rh + {}^\ast \di \si +  3 \eta \rh +3 {}^\ast \eta \si + 2 \chibh \cdot \beb \right)=&0, \\
\Du \rh + \frac{3}{2} \Om \tr \chib \rh+ \Om \left( \Divd \beb + (2 \eta- \zeta, \beb)+ \half (\chih,\aa) \right) =&0, \\
\Du \si + \frac{3}{2} \Om \tr \chib \si + \Om \left( \Curld \beb + (2 \eta-\zeta, {}^\ast \beb) + \half \chih \wedge \aa \right)=&0.
\end{aligned}
\end{align*}

%%%%%%%%%%%%%%%%%%%%%%%%%%%%%%%%%%%%%%%%
\subsection{Minkowski and Schwarzschild spacetimes} \label{SECMinkowskiSSspacetimes1} 
In this section we discuss the null geometry of Minkowski and the Schwarzschild spacetimes.\\

\ni \textbf{Minkowski spacetime.} The trivial solution to the Einstein equations is Minkowski spacetime $(\RRR^4, \mathbf{m})$ where $\mathbf{m}=\mathrm{diag}(-1,1,1,1)$. Defining standard spherical coordinates on $\RRR^3$ by \eqref{EQdefinSPHERICALAFcoord}, the reference double null coordinates on Minkowski are given by
\begin{align} 
\begin{aligned} 
(u,v,\th^1,\th^2) = \lrpar{\half \lrpar{t-r},\half \lrpar{t+r}, \th^1, \th^2},
\end{aligned} \label{EQdoubleNullMinkowski}
\end{align}
with respect to which
\begin{align*} 
\begin{aligned} 
\mathbf{m} = -4 du dv + (v-u)^2 \gac_{AB} d\th^A d\th^B,
\end{aligned} %\label{}
\end{align*}
where $\gac$ is defined in \eqref{EQdefRoundUnitMetric}. We note that the area radius of the sphere $S_{u,v}$ is given by $$r=v-u.$$

\ni Explicitly, with respect to the coordinates \eqref{EQdoubleNullMinkowski}, the Minkowski metric components, Ricci coefficients and null curvature components on $S_{u,v}$ are given by (with $r=v-u$)
\begin{align} 
\begin{aligned} 
\Om =& 1, & \gd=& r^2\gac, & & && &&&&\\
\trchi=&\frac{2}{r}, & \trchib=&-\frac{2}{r}, & \chih=&0, &\chibh=&0, &&&&\\
\eta=&0, & \etab=& 0, & \zeta=&0, &\underline{\zeta}=&0, &&&&\\
\om=& 0, & D\om=& 0, & \omb=& 0, & \Du\omb=& 0, &&&&\\
\a=&0, &\be=&0, &\beb=&0, & \ab=&0, & \rh=&0, & \si=&0. 
\end{aligned} \label{EQexplicitMinkowskiReferenceALL8999}
\end{align}

\ni \textbf{The family of Schwarzschild spacetimes.} For real numbers $M \geq 0$, let
\begin{align} 
\begin{aligned} 
\g^M = -\lrpar{1-\frac{2M}{r}} dt^2 + \lrpar{1-\frac{2M}{r}}^{-1} dr^2 + r^2 \lrpar{d\th^2 + \sin^2\th d\phi^2}.
\end{aligned} \label{EQdefSSmetric}
\end{align}
For $M=0$, the metric \eqref{EQdefSSmetric} is Minkowski, while for $M>0$ it yields a black hole solution with event horizon at $\{r=2M\}$. The so-called exterior region $\{r>2M\}$ can be covered by Eddington-Finkelstein double null coordinates $(u,v,\th^1,\th^2)$ with respect to which
\begin{align*} 
\begin{aligned} 
\g^M = -4 \Om_M^2 du dv + r_M(u,v)^2\gac_{CD} d\th^C d\th^D,
\end{aligned} %\label{}
\end{align*}
where 
\begin{align} 
\begin{aligned} 
\Om_M := \sqrt{1-\frac{2M}{r}},
\end{aligned} \label{EQdefOMEGAMSSdef8999}
\end{align}
and the area radius $r_M(u,v)$ is implicitly defined by (see for example (98) in \cite{DHR})
\begin{align} 
\begin{aligned} 
\frac{v-u}{2M} = \frac{r_M(u,v)}{2M} + \log \lrpar{\frac{r_M(u,v)}{2M}-1}.
\end{aligned} \label{EQdefRbyUV}
\end{align}

\ni Using \eqref{EQdefOMEGAMSSdef8999}, and that by \eqref{EQdefNullPairsINM} and \eqref{EQvectorexpr2488},
\begin{align*} 
\begin{aligned} 
\widehat{L}= \Om^{-1} \pr_v, \,\, \widehat{\Lb}= \Om^{-1} \pr_u,
\end{aligned} %\label{}
\end{align*}
together with the standard identities
\begin{align*} 
\begin{aligned} 
\pr_v \Om_M = \frac{\Om_M M}{r_M^2}, \,\, \pr_u \Om_M = -\frac{\Om_M M}{r_M^2}, \,\, \pr_v r_M = \Om_M^2, \,\, \pr_u r_M =-\Om_M^2,
\end{aligned} %\label{}
\end{align*}
it follows that in Eddington-Finkelstein coordinates, for real numbers $v>u$ such that $r_M(u,v)>2M$, the Schwarzschild metric components, Ricci coefficients and null curvature components on $S_{u,v}$ are given by
\begin{align} 
\begin{aligned} 
\Om_M =& \sqrt{1-\frac{2M}{r_M}}, & \gd=& r_M^2\gac, & & && \\
\trchi=&\frac{2\Om_M}{r_M}, & \trchib=&-\frac{2\Om_M}{r_M}, & \chih=&0, &\chibh=&0, \\
\eta=&0, & \etab=& 0, & \zeta=&0, &\underline{\zeta}=&0, \\
\om=& \frac{M}{r_M^2}, & D\om=& -\frac{2M}{r_M^3}\Om_M^2, & \omb=& -\frac{M}{r_M^2}, & \Du\omb=& -\frac{2M}{r_M^3}\Om_M^2, \\
\a=&0, &\be=&0, &\beb=&0, & \ab=&0, \\
\rh=& -\frac{2M}{r_M^3}, & \si=&0. && &&
\end{aligned} \label{EQspheredataSSM111222}
\end{align}

%%%%%%%%%%%%%%%%%%%%%%%%%%%%%%%%%%%%%%%%%
\subsection{Sphere data and null data} \label{SECspheredataDEF2333} In this section we define the notions of \emph{sphere data} and \emph{null data}.

\begin{definition}[Sphere data] \label{DEFspheredata2} \label{DEFFirstOrderDATA} For real numbers $v>u$, let $S_{u,v}$ be a $2$-sphere equipped with a round metric $\gac$ as in \eqref{EQdefRoundUnitMetric}. Sphere data $x_{u,v}$ on $S_{u,v}$ is given by the following tuple of $S_{u,v}$-tangential tensors,
\begin{align*} 
\begin{aligned} 
x = (\Om,\gd, \Om\trchi, \chih, \Om\trchib, \chibh, \eta, \om, D\om, \omb, \Du\omb, \a, \ab),
\end{aligned} 
\end{align*}
where
\begin{itemize}
\item $\Om>0$ is a positive scalar function and $\gd$ is a Riemannian metric, 
\item $\Om\trchi, \Om\trchib,\om, D\Om, \omb, \Du\omb, \rh$ and $\si$ are scalar functions,
\item $\eta, \be$ and $\beb$ are vectorfields,
\item $\chih$, $\chibh$, $\a$ and $\ab$ are $\gd$-tracefree symmetric $2$-tensors.
\end{itemize}
\end{definition}

\ni \emph{Remarks on Definition \ref{DEFspheredata2}.}
\begin{enumerate}

\item Sphere data is gauge-dependent, see \cite{ACR3,ACR1}.

\item The null structure equations and null Bianchi equations of Section \ref{SECnullstructureequations} determine from sphere data the Ricci coefficients and null curvature components
\begin{align*} 
\begin{aligned} 
(\etab, \zeta, \underline{\zeta}), \,\, (\be, \rh, \si, \beb),
\end{aligned} %\label{}
\end{align*}
as well as the derivatives
\begin{align*} 
\begin{aligned} 
&\lrpar{D\eta, D\etab, D\zeta, D\chi, D\chib, D\omb}, &
&\lrpar{ \Du \eta, \Du\etab, \Du\zeta, \Du\chi, \Du\chib, \Du\om }, \\
&\lrpar{D\beta, D\rh, D\si, D\beb, D\ab}, &
&\lrpar{\Du\beb, \Du\si,\Du\rh,\Du\be,\Du\a}.
\end{aligned} %\label{}
\end{align*}

\item In the following we denote by $(\be, \rh, \si, \beb)(x_{u,v})$ the null curvature components calculated from $x_{u,v}$ by the null structure equations  \eqref{EQtransportEQLnullstructurenonlinear}, \eqref{EQGaussEquation} and \eqref{EQgausscodazzinonlinear1}, and interpret them as part of sphere data.
\end{enumerate}

\ni \textbf{Notation.} We denote the \emph{Minkowski reference sphere data} on $S_{u,v}$ coming from \eqref{EQexplicitMinkowskiReferenceALL8999} by $\mathfrak{m}_{u,v}$, and for real numbers $M\geq0$, we denote the \emph{Schwarzschild reference sphere data} on $S_{u,v}$ coming from \eqref{EQspheredataSSM111222} by $\mathfrak{m}^M_{u,v}$. \\

\ni Along a null hypersurface, we consider the following \emph{null data}.
\begin{definition}[Ingoing and outgoing null data] \label{DEFnulldata111}\label{DEFnulldataHH} For three real numbers $u_0<v_1<v_2$, \emph{outgoing null data $x_{u_0,[v_1,v_2]}$ on $\HH_{u_0,[v_1,v_2]}$} is given by a family of sphere data
\begin{align*} 
\begin{aligned} 
\lrpar{x_{u_0,v}}_{v_1 \leq v \leq v_2} \text{ on } \HH_{u_0,[v_1,v_2]} = \bigcup\limits_{v_1\leq v \leq v_2} S_{u_0,v}.
\end{aligned} %\label{}
\end{align*}
Similarly, for three real numbers $u_1<u_2<v_0$, \emph{ingoing null data $x_{[u_1,u_2],v_0}$ on $\HHb_{[u_1,u_2],v_0}$} is given by a family of sphere data
\begin{align*} 
\begin{aligned} 
\lrpar{x_{u,v_0}}_{u_1 \leq u \leq u_2} \text{ on } \HHb_{[u_1,u_2],v_0} = \bigcup\limits_{u_1\leq u \leq u_2} S_{u,v_0}.
\end{aligned} %\label{}
\end{align*}

\end{definition}

\ni In addition to the above sphere data $x_{u,v}$ on spheres $S_{u,v}$, we also consider for integers $m\geq1$ the \emph{higher-order sphere data} on $S_{u,v}$
\begin{align} 
\begin{aligned} 
\lrpar{ x_{u,v}, \DD^{L,m}_{u,v}, \DD^{\Lb,m}_{u,v} },
\end{aligned} \label{EQdefHIGHERORDERspheredata1999091}
\end{align}
where $x_{u,v}$ denotes sphere data and $\DD^{L,m}$ and $\DD^{\Lb,m}$ are tuples of $L$- and $\Lb$-derivatives of sphere data up to order $m$; we refer to Section 2.10 in \cite{ACR1} for definitions and discussion. We denote the Schwarzschild reference higher-order sphere data of order $m$ by 
\begin{align*} 
\begin{aligned} 
\lrpar{\mathfrak{m}^{M}_{u,v}, \DD^{L,m,M}_{u,v},\DD^{\Lb,m,M}_{u,v}}.
\end{aligned} %\label{}
\end{align*}
Importantly, the gluing of higher-order sphere data implies the higher regularity (in all directions) of the constructed gluing solution. Similarly, we consider higher-order outgoing and ingoing null data on $\HH_{u_0,[v_1,v_2]}$ and $\HHb_{[u_1,u_2],v_0}$, respectively. We remark that higher derivatives are subject to the \emph{higher-order null constraint equations}, see Section 2.10 in \cite{ACR1} for details.

%%%%%%%%%%%%%%%%%%%%%%%%%%%%%%%%%%%%%%%
\subsection{Definition of charges $(\mathbf{E},\mathbf{P},\mathbf{L},\mathbf{G})$} \label{SECdefcharges99901}

\ni The following charges play an essential role for the characteristic gluing problem. In \cite{ACR3,ACR1} they are identified as geometric obstacles to characteristic gluing.

\begin{definition}[Charges] \label{DEFlocalCharges} For sphere data $x_{u,v}$ and $m=-1,0,1$ define the charges
\begin{align} 
\begin{aligned} 
\mathbf{E} :=& -\frac{1}{8\pi} \sqrt{4\pi} \lrpar{r^3\lrpar{\rho + r \Divd {\be}}}^{(0)}, \\
\mathbf{P}^m :=& -\frac{1}{8\pi} \sqrt{\frac{4\pi}{3}} \lrpar{r^3 \lrpar{\rho + r \Divd {\be}}}^{(1m)},\\
\mathbf{L}^m :=& \frac{1}{16\pi} \sqrt{\frac{8\pi}{3}} \lrpar{r^3 \lrpar{ \di \trchi + \trchi (\eta-\di\log\Om) }}_H^{(1m)},\\
\mathbf{G}^m :=&  \frac{1}{16\pi}\sqrt{\frac{8\pi}{3}} \lrpar{ r^3 \lrpar{ \di \trchi + \trchi (\eta-\di\log\Om) }}^{(1m)}_E,
\end{aligned} \label{EQdefcharges999919999}
\end{align}
where $r$ denotes the area radius calculated from $x_{u,v}$, and the spherical harmonics projections are defined with respect to the unit round metric $\gac$ on $S_{u,v}$, see Appendix \ref{SECellEstimatesSpheres}.
\end{definition}

\ni \emph{Remarks on Definition \ref{DEFlocalCharges}.}
\begin{enumerate}

\item By the Gauss-Codazzi equation \eqref{EQgausscodazzinonlinear1} it holds that
\begin{align} 
\begin{aligned} 
\mathbf{L}^m =& \frac{1}{8\pi} \sqrt{\frac{8\pi}{3}} \lrpar{r^3 \lrpar{ \be + \Divd \chih + \chih \cdot \lrpar{\eta-\di\log\Om} }}_H^{(1m)}, \\
 \mathbf{G}^m =&  \frac{1}{8\pi} \sqrt{\frac{8\pi}{3}} \lrpar{r^3 \lrpar{ \be + \Divd \chih + \chih \cdot \lrpar{\eta-\di\log\Om} }}^{(1m)}_E,
\end{aligned} \label{EQalternativeLGdef}
\end{align}
where $\be$ is defined from sphere data by the Gauss-Codazzi equations. The expressions \eqref{EQalternativeLGdef} are used in Section \ref{SECstatementConstruction}.
\item The numerical factors in the definitions of the charges are determined by comparison to the ADM asymptotic invariants, see Sections \ref{SECcomparisonENERGY}, \ref{SECcomparisonLINEAR}, \ref{SECcomparisonANGULAR} and \ref{SECcomparisonCENTER}.
\item By explicit calculation, for real numbers $M\geq0$, and $v>u$,
\begin{align} 
\begin{aligned} 
\lrpar{\mathbf{E},\mathbf{P},\mathbf{L},\mathbf{G}}(\mathfrak{m}^M_{u,v}) = (M,0,0,0).
\end{aligned} \label{EQexplicitCalcuChargesSSM8889}
\end{align}
\end{enumerate}

%%%%%%%%%%%%%%%%%%%%%%%%%%%%%%%%%%%%%%%
%%%%%%%%%%%%%%%%%%%%%%%%%%%%%%%%%%%%%%%
\subsection{Norms on spheres and null hypersurfaces} \label{SECtensorSpaces} In this section we define the norms used in this paper. They are analogous to the norms in \cite{ACR1}, but with the difference that they contain weights in $v-u$ for better scaling properties (see Lemma \ref{LEMscalingofNorms}, and also \cite{ACR3}).

\begin{definition}[Norms on $2$-spheres] \label{DEFnormsonSR} Let $v>u$ be two real numbers and let $S_{u,v}$ be a $2$-sphere equipped with a round metric $\gac$ as in \eqref{EQdefRoundUnitMetric}. For integers $m\geq0$ and $S_{u,v}$-tangent $k$-tensors $T$, define
\begin{align*} 
\begin{aligned} 
\Vert T \Vert^2_{H^m(S_{u,v})} :=   \sum\limits_{0\leq i \leq m} (v-u)^{2(i+k-1)} \left\Vert \Nd^i T \right\Vert^2_{L^2(S_{u,v})},
\end{aligned} %\label{}
\end{align*}
where the covariant derivative $\Nd$ and the volume element of the $L^2$-norm are with respect to the round metric $\ga=(v-u)^2\gac$ on $S_{u,v}$. Moreover, let 
\begin{align*} 
\begin{aligned} 
H^m(S_{u,v}) := \{ T: \Vert T \Vert_{H^m(S_{u,v})} < \infty \}.
\end{aligned} %\label{}
\end{align*}
\end{definition}

\begin{definition}[Norms on null hypersurfaces] \label{DEFnullHHspaces} \label{DEFnullHHbspaces}
We have the following.
\begin{itemize} 
\item For real numbers $u_0<v_1<v_2$, let $T$ be an $S_{u_0,v}$-tangential tensor on $\HH_{u_0,[v_1,v_2]}$. For integers $m\geq0$ and $l\geq0$, define
\begin{align*} 
\begin{aligned} 
\Vert T \Vert^2_{H^m_l\lrpar{\HH_{u_0,[v_1,v_2]}}} :=  \int\limits_{v_1}^{v_2} \,\, \sum\limits_{0\leq i\leq l} (v-u_0)^{2i-1} \left\Vert D^i T \right\Vert^2_{H^m(S_{u_0,v})} dv,
\end{aligned} %\label{}
\end{align*}
where the Lie derivative $D$ is with respect to the reference Minkowski metric on $\HH_{u_0,[v_1,v_2]}$. Let further
\begin{align*} 
\begin{aligned} 
H^m_l\lrpar{\HH_{u_0,[v_1,v_2]}} := \{ T: \Vert T \Vert_{H^m_l\lrpar{\HH_{u_0,[v_1,v_2]}}} < \infty\}.
\end{aligned} %\label{}
\end{align*}

\item For real numbers $u_1<u_2<v_0$, let $T$ be an $S_{u,v_0}$-tangential tensor on $\HHb_{[u_1,u_2],v_0}$. For integers $m\geq0$ and $l\geq0$, define
\begin{align*} 
\begin{aligned} 
\Vert T \Vert^2_{H^m_l\lrpar{\HHb_{ [u_1,u_2],v_0}}} :=  \int\limits_{u_1}^{u_2} \,\, \sum\limits_{0\leq i\leq l} (v_0-u)^{2i-1} \left\Vert \Du^i T \right\Vert^2_{H^m(S_{u,v_0})} du,
\end{aligned} %\label{}
\end{align*}
where the Lie derivative $\Du$ is with respect to the reference Minkowski metric on $\HHb_{ [u_1,u_2],v_0}$. Let further
\begin{align*} 
\begin{aligned} 
H^m_l\lrpar{\HHb_{ [u_1,u_2],v_0}} := \{ T: \Vert T \Vert_{H^m_l\lrpar{\HHb_{ [u_1,u_2],v_0}}} < \infty\}. 
\end{aligned} %\label{}
\end{align*}

\end{itemize}
\end{definition}

\ni In the following we define the norms of sphere data and null data using the above norms on spheres and null hypersurfaces. Their definition includes weights in $v-u$ to make them invariant under the scaling introduced in Section \ref{SECdefinitionScaling}, see Lemma \ref{LEMspheredataInvarianceSCale}.
\begin{definition}[Norm for sphere data] \label{DEFNORMspheredata2} \label{DEFnormFirstOrderDATA} Let $x_{u,v}$ be sphere data on the sphere $S_{u,v}$. The norm of $x_{u,v}$ is defined by
\begin{align*} 
\begin{aligned} 
\Vert x_{u,v}\Vert_{\XX(S_{u,v})} :=& \Vert \Om \Vert_{{H}^{6}(S_{u,v})}+(v-u)^{-2} \Vert \gd \Vert_{{H}^{6}(S_{u,v})} + \Vert \eta \Vert_{H^{5}(S_{u,v})}\\
&+ (v-u)\Vert \trchi \Vert_{H^{6}(S_{u,v})} + (v-u)^{-1}\Vert \chih \Vert_{H^{6}(S_{u,v})} \\
&+ (v-u)\Vert \trchib \Vert_{H^{4}(S_{u,v})} +(v-u)^{-1} \Vert \chibh \Vert_{H^{4}(S_{u,v})} \\
&+ (v-u)\Vert \om \Vert_{H^{6}(S_{u,v})}+(v-u)^2\Vert D\om \Vert_{H^{6}(S_{u,v})}\\
&+(v-u)\Vert \omb \Vert_{H^{4}(S_{u,v})} +(v-u)^2 \Vert \Du\omb \Vert_{H^{2}(S_{u,v})} \\
&+ \Vert \a \Vert_{H^{6}(S_{u,v})}+ (v-u) \Vert \be \Vert_{H^{5}(S_{u,v})}+(v-u)^2\Vert \rh \Vert_{H^{4}(S_{u,v})} \\
&+(v-u)^2 \Vert \si \Vert_{H^{4}(S_{u,v})} + (v-u) \Vert \beb \Vert_{H^{3}(S_{u,v})}+\Vert \ab \Vert_{H^{2}(S_{u,v})},
\end{aligned} %\label{}
\end{align*}
where the norms are with respect to $(v-u)^2 \gac$ on $S_{u,v}$, see Definition \ref{DEFnormsonSR}. Moreover, let
\begin{align*} 
\begin{aligned} 
\XX(S_{u,v}) := \{ x_{u,v} : \Vert x_{u,v} \Vert_{\XX(S_{u,v})} < \infty\}.
\end{aligned}
\end{align*}
\end{definition}
\ni \emph{Remark on Definition \ref{DEFNORMspheredata2}.}
\begin{itemize}
\item Definition \ref{DEFNORMspheredata2} reflects the regularity hierarchy of the null structure equations in the $L$-direction.

\item For sphere data $x_{u,v} \in \XX(S_{u,v})$, the charges $\mathbf{E}, \mathbf{P}, \mathbf{L}$ and $\mathbf{G}$ introduced in Definition \ref{DEFlocalCharges} are well-defined.
\end{itemize}

\begin{definition}[Norm for null data] \label{DEFnulldataNORM111}\label{DEFnormHH} \label{DEFspacetimeNORMdata} \label{DEFspacetimeNORM} 
Let $R\geq1$ be a real number. We have the following.
\begin{itemize}
\item Let $u_0<v_1<v_2$ be three real numbers. Let $x_R := x_{R\cdot u_0,R \cdot [ v_1, v_2]}$ be null data on $\HH_R:= \HH_{R\cdot u_0,R \cdot [ v_1, v_2]}$. The norm of $x_R$ on $\HH_R$ is defined by 
\begin{align*} 
\begin{aligned}
\Vert x_R \Vert_{\XX(\HH_R)} :=& \Vert \Om \Vert_{H^6_3(\HH_R)} +\Vert \gd \Vert_{H^6_3(\HH_R)}+ \Vert \eta \Vert_{H^5_2(\HH_R)}\\
&+ R\Vert \Om\trchi \Vert_{H^6_3(\HH_R)}+ R^{-1}\Vert \chih \Vert_{H^6_2(\HH_R)} + R \Vert \Om\trchib \Vert_{H^4_2(\HH_R)}+ R^{-1} \Vert \chibh \Vert_{H^4_3(\HH_R)}\\
&+R\Vert \om \Vert_{H^6_2(\HH_R)} +R^2 \Vert D\om \Vert_{H^6_1(\HH_R)}+R \Vert \omb \Vert_{H^4_3(\HH_R)}+R^2 \Vert \Du\omb \Vert_{H^2_3(\HH_R)}\\
&+  \Vert \a \Vert_{H^{6}_1(\HH_R)} + R \Vert \be \Vert_{H^{5}_2(\HH_R)}+R^2\Vert \rh \Vert_{H^{4}_2(\HH_R)} \\
&+R^2 \Vert \si \Vert_{H^{4}_2(\HH_R)} + R \Vert \beb \Vert_{H^{3}_2(\HH_R)}+\Vert \ab \Vert_{H^{2}_3(\HH_R)}, 
\end{aligned} %\label{}
\end{align*}
where the norms over $\HH_{-R,[R,2R]}$ are defined in Definition \ref{DEFnullHHspaces}. Moreover, let
\begin{align*} 
\begin{aligned} 
\XX(\HH_{R}) := \{ x_R : \Vert x_R \Vert_{\XX(\HH_{R})} < \infty \}.
\end{aligned} %\label{}
\end{align*}
\item Let $u_1<u_2<v_0$ be three real numbers. Let $x_R := {x}_{R\cdot[u_1,u_2],R\cdot v_0}$ be null data on $\HHb_R:= \HHb_{R\cdot[u_1,u_2],R\cdot v_0}$. The norm of ${x}_R$ on $\HHb_R$ is defined by 
\begin{align*} 
\begin{aligned}
\Vert {x}_R \Vert_{\XX(\HHb_R)} :=& \Vert \Om \Vert_{H^6_3(\HHb_R)} +\Vert \gd \Vert_{H^6_3(\HHb_R)}+ \Vert \eta \Vert_{H^5_2(\HHb_R)}\\
&+ R\Vert \Om\trchib \Vert_{H^6_3(\HHb_R)}+ R^{-1}\Vert \chibh \Vert_{H^6_2(\HHb_R)} + R \Vert \Om\trchi \Vert_{H^4_2(\HHb_R)}+ R^{-1} \Vert \chih \Vert_{H^4_3(\HHb_R)}\\
&+R\Vert \omb \Vert_{H^6_2(\HHb_R)} +R^2 \Vert \Du\omb \Vert_{H^6_1(\HHb_R)}+R \Vert \om \Vert_{H^4_3(\HHb_R)}+R^2 \Vert D\om \Vert_{H^2_3(\HHb_R)}\\
&+  \Vert \ab \Vert_{H^{6}_1(\HHb_R)} + R \Vert \beb \Vert_{H^{5}_2(\HH_R)}+R^2\Vert \si \Vert_{H^{4}_2(\HH_R)} \\
&+R^2 \Vert \rh \Vert_{H^{4}_2(\HH_R)} + R \Vert \be \Vert_{H^{3}_2(\HH_R)}+\Vert \a \Vert_{H^{2}_3(\HHb_R)}, 
\end{aligned} %\label{}
\end{align*}
where the norms over $\HHb_R$ are defined in Definition \ref{DEFnullHHspaces}. Moreover, let
\begin{align*} 
\begin{aligned} 
\XX(\HHb_R) := \{ {x}_R : \Vert {x}_R \Vert_{\XX(\HH_R)} < \infty \}.
\end{aligned} %\label{}
\end{align*}
\end{itemize}
\end{definition}

\ni In addition to the above norm $\XX(\HHb)$ for ingoing null data, we define the following higher regularity norm $\XX^+(\HHb)$. This norm is necessary for the characteristic gluing of \cite{ACR3,ACR1}.
\begin{definition}[Norm for higher-regularity ingoing null data] \label{DEFspacetimeHIGHERNORMdata} \label{DEFspacetimeHIGHERNORM} 
Let $u_1<u_2<v_0$ be three real numbers. Let ${x}_R:=x_{R\cdot[u_1,u_2],R\cdot v_0}$ be null data on $\HHb_R:= \HHb_{R\cdot[u_1,u_2],R\cdot v_0}$. The norm of ${x}$ is defined by 
\begin{align*} 
\begin{aligned}
\Vert {x}_R \Vert_{\XX^+(\HHb_R)} :=& \Vert \Om \Vert_{H^{12}_9(\HHb_R)} +\Vert \gd \Vert_{H^{12}_9(\HHb_R)}+ \Vert \eta \Vert_{H^{11}_8(\HHb_R)}\\
&+ R\Vert \Om\trchib \Vert_{H^{12}_9(\HHb_R)}+ R^{-1}\Vert \chibh \Vert_{H^{12}_8(\HHb_R)} + R \Vert \Om\trchi \Vert_{H^{10}_8(\HHb_R)}+ R^{-1} \Vert \chih \Vert_{H^{10}_9(\HHb_R)}\\
&+R\Vert \omb \Vert_{H^{12}_8(\HHb_R)} +R^2 \Vert \Du\omb \Vert_{H^{12}_7(\HHb_R)}+R \Vert \om \Vert_{H^{10}_9(\HHb_R)}+R^2 \Vert D\om \Vert_{H^8_9(\HHb_R)}\\
&+  \Vert \ab \Vert_{H^{12}_7(\HHb_R)} + R \Vert \beb \Vert_{H^{11}_8(\HH_R)}+R^2\Vert \si \Vert_{H^{10}_8(\HH_R)} \\
&+R^2 \Vert \rh \Vert_{H^{10}_8(\HH_R)} + R \Vert \be \Vert_{H^{9}_8(\HH_R)}+\Vert \a \Vert_{H^{8}_9(\HHb_R)}, 
\end{aligned} %\label{}
\end{align*}
where the norms over $\HHb_R$ are defined in Definition \ref{DEFnullHHspaces}. Moreover, let
\begin{align*} 
\begin{aligned} 
\XX^+(\HHb_R) := \{ {x} : \Vert {x} \Vert_{\XX^+(\HHb_R)} < \infty \}.
\end{aligned} %\label{}
\end{align*}
\end{definition}

\ni \textbf{Notation.} We do not explicitly denote the evaluation of Minkowski and Schwarzschild reference data in sphere data and null data. That is, given sphere data $x_{u,v}$ on $S_{u,v}$ and real numbers $M\geq0$, we write $\Vert x_{u,v}- \mathfrak{m}^M \Vert_{\XX(S_{u,v})}$ to denote $\Vert x_{u,v}- \mathfrak{m}^M_{u,v} \Vert_{\XX(S_{u,v})}$. Similarly for outgoing null data $x_{u_0,[v_1,v_2]}$ on $\HH_{u_0,[v_1,v_2]}$ and ingoing null data $x_{[u_0,u_1],v_0}$ on $\HHb_{[u_0,u_1],v_0}$.

%%%%%%%%%%%%%%%%%%%%%%%%%%%%%%%%%%%%%%%%
\subsection{Asymptotically flat families of sphere data and ingoing null data} \label{SECdefinitionAF} \label{SECdefinitionCHARGES} In this section, we introduce \emph{asymptotically flat families of sphere data} and \emph{ingoing null data}, and introduce the \emph{asymptotic charges}. 

\begin{definition}[Strongly asymptotically flat sphere data] \label{DEFadmissibleSequences} Let $v>u$ be two fixed real numbers, and let
\begin{align*} 
\begin{aligned} 
(x_{R\cdot u,R \cdot v})_{R\geq1}
\end{aligned} %\label{}
\end{align*}
be a family of sphere data. We say that $(x_{R\cdot u,R \cdot v})$ is a \emph{strongly asymptotically flat family of sphere data} if there is a real number $M\geq 0$ such that
\begin{align} 
\begin{aligned} 
\Vert x_{R\cdot u,R \cdot v} -\mathfrak{m}^M \Vert_{\XX(S_{R\cdot u,R \cdot v})} =& \smallO(R^{-3/2}), \\ 
\Vert \beta^{[1]}(x_{R\cdot u,R \cdot v})\Vert_{L^2(S_{R\cdot u,R \cdot v})} =& \OO\lrpar{R^{-3}}.
\end{aligned} \label{EQdefstronglydecayingfam123}
\end{align}
\end{definition}

\ni \emph{Remarks on Definition \ref{DEFadmissibleSequences}.}
\begin{enumerate}

\item In this paper we work with strongly asymptotically families of sphere data $(x_{0,R})$ (that is, $u=0$ and $v=1$), in Theorem \ref{THMdoubleCHARgluingTOkerr} and the proof of Corollary \ref{THMspacelikeGLUINGtoKERRv2}, and $(x_{-R,R})$ (that is, $u=-1$ and $v=1$) in Definition \ref{DEFadmissibleEXTSequences} below.

\item By Definition \ref{DEFnormFirstOrderDATA}, the decay \eqref{EQdefstronglydecayingfam123} implies in particular 
\begin{align*} 
\begin{aligned} 
R\Vert \a \Vert_{L^2(S_{R\cdot u,R \cdot v})} + R\Vert \ab \Vert_{L^2(S_{R\cdot u,R \cdot v})}+R\Vert \be \Vert_{L^2(S_{R\cdot u,R \cdot v})}  =& \smallO\lrpar{R^{-3/2}},\\
R^{3/2} \Vert \be^{[1]} \Vert_{L^2(S_{R\cdot u,R \cdot v})} =& \OO(R^{-3/2}).
\end{aligned} %\label{}
\end{align*}

\item \textbf{The decay rates above are in agreement with a sequence of spheres going to \emph{spacelike} infinity in a strongly asymptotically flat spacetime}; see also Theorem \ref{PROPconstructionStatement} and \cite{ChrKl93}.

\end{enumerate}

\ni We define the following asymptotic charges.
\begin{definition}[Asymptotic charges] \label{DEFasymptoticCharges} Let $(x_{R\cdot u,R \cdot v})$ be a strongly asymptotically flat family of sphere data. Let
\begin{align*} 
\begin{aligned} 
\mathbf{E}_\infty :=& \lim\limits_{R \to \infty}\mathbf{E}(x_{R\cdot u,R \cdot v}), & \mathbf{P}_\infty :=& \lim\limits_{R \to \infty}\mathbf{P}(x_{R\cdot u,R \cdot v}), \\
\mathbf{L}_\infty :=& \lim\limits_{R \to \infty}\mathbf{L}(x_{R\cdot u,R \cdot v}), & \mathbf{G}_\infty :=& \lim\limits_{R \to \infty}\mathbf{G}(x_{R\cdot u,R \cdot v}),
\end{aligned} 
\end{align*}
where the charges $(\mathbf{E},\mathbf{P},\mathbf{L},\mathbf{G})$ are defined in Definition \ref{DEFlocalCharges}. 
\end{definition}

\ni The following basic properties of the asymptotic charges are proved in Section \ref{SECdefinitionScaling}.
\begin{lemma}[Properties of asymptotic charges] \label{LEMBoundednessEinftyPinfty1} Let $(x_{R\cdot u,R \cdot v})$ be a strongly asymptotically flat family of sphere data. Then its asymptotic charges are well-defined, 
\begin{align*} 
\begin{aligned} 
\vert \mathbf{E}_\infty \vert + \vert \mathbf{P}_\infty \vert+ \vert \mathbf{L}_\infty \vert+ \vert \mathbf{G}_\infty \vert < \infty,
\end{aligned} %\label{}
\end{align*}
and it holds that $\mathbf{E}_\infty = M$ and $\mathbf{P}_\infty=0$, where $M$ is the real number appearing in \eqref{EQdefstronglydecayingfam123}, and
\begin{align*} 
\begin{aligned} 
\mathbf{E}\lrpar{ x_{R\cdot u,R \cdot v}} =& \mathbf{E}_\infty + \smallO(R^{-1/2}), & \mathbf{P}\lrpar{ x_{R\cdot u,R \cdot v}} =&  \smallO(R^{-1/2}), \\ 
\mathbf{L}\lrpar{ x_{R\cdot u,R \cdot v}} =& \mathbf{L}_\infty + \smallO(1), & \mathbf{G}\lrpar{ x_{R\cdot u,R \cdot v}} =& \mathbf{G}_\infty + \smallO(1).
\end{aligned} %\label{}
\end{align*}

\end{lemma}

\ni The above notion of asymptotic flatness is generalized to ingoing null data as follows.
\begin{definition}[Strongly asymptotically flat ingoing null data] \label{DEFadmissibleEXTSequences} Let $\de>0$ be a real number. Let
$$({x}_{-R+R\cdot [-\de,\de],R})_{R\geq1}$$ 
be a family of ingoing null data. We say that $({x}_{-R+R\cdot [-\de,\de],R})$ is \emph{strongly asymptotically flat} if there is a real number $M>0$ such that, as $R\to \infty$,
\begin{align} 
\begin{aligned} 
\Vert {x}_{-R+R\cdot [-\de,\de],R} - {\mathfrak{m}}^{M} \Vert_{\XX^+\lrpar{\HHb_{-R+R\cdot [-\de,\de],R}}} =& \smallO\lrpar{R^{-3/2}}, \\
\Vert \beta^{[1]}(x_{-R,R})\Vert_{L^2(S_{-R,R})} =& \OO\lrpar{R^{-3}},
\end{aligned} \label{EQdefstronglydecayingfam123EXT}
\end{align}
where the sphere data $x_{-R,R} := x \vert_{S_{-R,R}}$. Define the asymptotic charges $(\mathbf{E}_\infty,\mathbf{P}_\infty,\mathbf{L}_\infty, \mathbf{G}_\infty)$ of the family of ingoing null data $({x}_{-R+R\cdot [-\de,\de],R})$ by applying Definition \ref{DEFasymptoticCharges} to the family of sphere data $(x_{-R,R})$.
\end{definition}

\ni \emph{Remarks on Definition \ref{DEFadmissibleEXTSequences}.}
\begin{enumerate}

\item In this paper, strongly asymptotically flat families of ingoing null data $({x}_{-R+R\cdot [-\de,\de],R})$ are used in Theorem \ref{THMMAINPRECISE} and Theorem \ref{PROPconstructionStatement}.

\item For strongly asymptotically flat families of ingoing null data $({x}_{-R+R\cdot [-\de,\de],R})$, the sphere data $$(x_{-R,R}) := (x \vert_{S_{-R,R}})$$ forms a strongly asymptotically flat family of sphere data. 

\end{enumerate}

%%%%%%%%%%%%%%%%%%%%%%%%%%%%%%%%%%%%%%%%
\subsection{Scaling of Einstein equations} \label{SECdefinitionScaling} In this section we introduce the scaling used in this paper and subsequently discuss how geometric quantities change under scaling. Consider local double null coordinates $(u,v,\th^1,\th^2)$ in a spacetime $(\MM,\g)$,
\begin{align*} 
\begin{aligned} 
\g = -4 \Om^2 du dv + \gd_{AB} \lrpar{d\th^A + b^A dv}\lrpar{d\th^B + b^B dv}.
\end{aligned} %\label{}
\end{align*}
The scaling is defined in two steps.
\begin{enumerate}
\item For a real number $R\geq1$, consider the local coordinates $(\tilde{u},\tilde{v},\tth^1,\tth^2)$ defined by
\begin{align*} 
\begin{aligned} 
(R\cdot \tilde{u},R \cdot \tilde{v},\tth^1,\tth^2)=  (u, v, \th^1,\th^2).
\end{aligned} %\label{}
\end{align*}
Clearly it holds that
\begin{align*} 
\begin{aligned} 
du = R d\tilde{u}, \,\, dv= Rd\tilde{v}, \,\, d\th^1=d\tth^1, \,\, d\th^2=d\tth^2,
\end{aligned} %\label{}
\end{align*}
and thus
\begin{align*} 
\begin{aligned} 
\g =& -4 R^2\cdot \Om^2 d\tilde{u} d\tilde{v} + \gd_{AB} \lrpar{d\tth^A +R\cdot b^A d\tilde{v}}\lrpar{d\tth^B + R\cdot b^B d\tilde{v}} \\
=& R^{2} \lrpar{ -4 \cdot \Om^2 d\tilde{u} d\tilde{v} + R^{-2} \gd_{AB} \lrpar{d\tth^A +R\cdot b^A d\tilde{v}}\lrpar{d\tth^B + R\cdot b^B d\tilde{v}}}.
\end{aligned} %\label{}
\end{align*}

\item It is well-known that given a Lorentzian metric $\g$ is a solution to the Einstein equations, the conformal metric ${}^{(R)}\g:= R^{-2} \g$ is also a solution.

\end{enumerate}
\ni Expressing ${}^{(R)}\g$ in coordinates $(\tilde{u},\tilde{v},\tth^1,\tth^2)$, we get the spacetime metric
\begin{align*} 
\begin{aligned} 
{}^{(R)}\g = -4 {}^{(R)}{\Om}^2 d\tilde{u}d\tilde{v} + {}^{(R)}{\gd}_{AB}  \lrpar{d\tth^A + {}^{(R)}{b}^A d\tilde{v}}\lrpar{d\tth^B + {}^{(R)}{b}^B d\tilde{v}}
\end{aligned} %\label{}
\end{align*}
with 
\begin{align*} 
\begin{aligned} 
{}^{(R)}{\Om}(\tilde{u},\tilde{v}):=\Om(R\tilde{u}, R\tilde{v}), \,\, {}^{(R)}{\gd}(\tilde{u},\tilde{v}) := R^{-2} \gd(R\tilde{u}, R\tilde{v}), \,\, {}^{(R)}{b}(\tilde{u},\tilde{v}) := R \,b(R\tilde{u}, R\tilde{v}),
\end{aligned} %\label{}
\end{align*}

\ni \textbf{Notation.} Denote in the following the scaling
\begin{align*} 
\begin{aligned} 
\Psi_R(\tilde{u},\tilde{v},\tth^1,\tth^2) :=(R\cdot \tilde{u},R \cdot \tilde{v},\tth^1,\tth^2).
\end{aligned} %\label{}
\end{align*}

\ni The following lemma shows how the Ricci coefficients and null curvature components change under scaling; the proof is by explicit computation.

\begin{lemma}[Scaling of Ricci coefficients and null curvature components]  \label{DEFscalingNULLSTRUC} Under the above scaling, the Ricci coefficients and null curvature components transform as follows,
\begin{align*} \begin{aligned}
{}^{(R)}\chi_{AB} =& R^{-1} \lrpar{\chi_{AB}\circ \Psi_R}, & {}^{(R)}\zeta_A =&\zeta_A\circ \Psi_R,&{}^{(R)}\eta_A =& \eta_A\circ \Psi_R,&{}^{(R)}\om =& R \lrpar{\om\circ \Psi_R},\\
{}^{(R)}\chib_{AB} =& R^{-1}\lrpar{ \chib_{AB}\circ \Psi_R}, & {}^{(R)}\underline{\zeta}_A =&  \underline{\zeta}_A\circ \Psi_R, & {}^{(R)}\etab_A =& \etab_A\circ \Psi_R, & {}^{(R)}\omb =& R \lrpar{\omb\circ \Psi_R}, 
\end{aligned} \end{align*}
and
\begin{align*} 
\begin{aligned} 
{}^{(R)}(D\om) = R^2\lrpar{(D\om) \circ \Psi_R}, \,\, {}^{(R)}\Du\omb = R^2 \lrpar{(\Du\omb)\circ \Psi_R},
\end{aligned} %\label{}
\end{align*}
and
\begin{align*} \begin{aligned}
{}^{(R)}\alpha_{AB} =& \a_{AB}\circ \Psi_R, & {}^{(R)}\beta_A =& R \lrpar{\beta_{A}\circ \Psi_R}, & {}^{(R)}\rh =& R^2 \lrpar{\rh\circ \Psi_R}, \\
 {}^{(R)}\sigma =& R^2 \lrpar{\sigma\circ \Psi_R}, & {}^{(R)}\beb_A =& R \lrpar{\beb_A\circ \Psi_R}, & {}^{(R)}\ab_{AB} =& \ab_{AB}\circ \Psi_R.
\end{aligned}\end{align*}
Moreover, 
\begin{align*} \begin{aligned}
\tr_{{}^{(R)}\gd} {}^{(R)}\chi =& R \lrpar{\trchi \circ \Psi_R}, & {}^{(R)}\chih_{AB}=& R^{-1} \lrpar{ \chih_{AB}\circ \Psi_R},\\
\tr_{{}^{(R)}\gd} {}^{(R)}\chib =& R \lrpar{ \trchib\circ \Psi_R}, & {}^{(R)}\chibh_{AB} =& R^{-1} \lrpar{\chibh_{AB}\circ \Psi_R},
\end{aligned}\end{align*}
where the tracefree parts of ${}^{(R)}\chi, {}^{(R)}\chib$ and $\chi, \chib$ are calculated with respect to ${}^{(R)}\gd$ and $\gd$, respectively. Furthermore, the area radius $r$ scales as
\begin{align*} 
\begin{aligned} 
{}^{(R)}r = R^{-1} \lrpar{ r \circ \Psi_R}.
\end{aligned} %\label{}
\end{align*}
\end{lemma}

\ni \textbf{Notation.} For real numbers $R\geq1$ and sphere data $x_{u,v}$ on $S_{u,v}$,  denote the rescaled sphere data following Lemma \ref{DEFscalingNULLSTRUC} by
${}^{(R)}x_{R^-1 u, R^{-1}v}$. \\

\ni \emph{Remarks on Lemma \ref{DEFscalingNULLSTRUC}.}
\begin{enumerate}
\item By the invariance of the Einstein equations under the above scaling, it follows that the null structure equations and the null Bianchi equations of Section \ref{SECnullstructureequations} are scale-invariant under the scaling of Lemma \ref{DEFscalingNULLSTRUC}.

\item Importantly, we have the following \emph{scale-invariance of Schwarzschild}. Using the explicit formulas for $\mathfrak{m}^M$ in \eqref{EQspheredataSSM111222}, it is straight-forward to show that for real numbers $M\geq0$, $R\geq1$ and $v>u$, 
\begin{align} 
\begin{aligned} 
\lrpar{{}^{(R)} \mathfrak{m}^M}_{R^{-1}u,R^{-1} v} =  \mathfrak{m}^{M/R}_{u,v}.
\end{aligned} \label{EQdefSCALINGSSprop8999}
\end{align}
\noindent Also the Kerr reference sphere data constructed in Section \ref{SECappKerrFamilyDetails} is scale-invariant, see Remark \ref{RemarkScalingKerrData}.
\end{enumerate}

\ni In the following we analyse the scaling of the norms of tensors, sphere data and null data and the charges.
\begin{lemma}[Scaling of tensor norms] \label{LEMscalingofNorms} Let $p\in \RRR$ and $R\geq1$ be two real numbers, and let $m\geq0$ and $l\geq0$ be two integers. Let $F$ be a tensor on $S_{u,v}$ for real numbers $v>u$, and define the tensor ${}^{(R)}F$ by
$${}^{(R)}F := R^p \cdot \lrpar{ F \circ \Psi_R}.$$
Then it holds for integers $i, l\geq0$ that
\begin{align*} 
\begin{aligned} 
\Vert {}^{(R)}F \Vert_{{H}^i(S_{u,v})} =& R^p \cdot \Vert F \Vert_{{H}^i(S_{R\cdot u, R\cdot v})}, \\
\Vert {}^{(R)}F \Vert_{H^i_l\lrpar{\HH_{u_0,[v_1,v_2]}}} =& R^p\cdot  \Vert F \Vert_{H^i_l\lrpar{\HH_{r\cdot u_0,[R\cdot v_1,R \cdot v_2]}}}, \\
\Vert {}^{(R)}F \Vert_{H^i_l\lrpar{\HHb_{[u_1,u_2],v_0}}} =& R^p\cdot  \Vert F \Vert_{H^i_l\lrpar{\HHb_{R\cdot [u_1,u_2],R\cdot v_0}}}.
\end{aligned} %\label{}
\end{align*}
\end{lemma}

\begin{proof}[Proof of Lemma \ref{LEMscalingofNorms}] The proof is based on Definitions \ref{DEFnormsonSR}, \ref{DEFnullHHspaces} and \ref{DEFnullHHbspaces}, and the fact that for each $k$-tensor $F$ and integers $i\geq0$ and $j\geq0$,
\begin{align*} 
\begin{aligned} 
\Vert \Nd^i D^j {}^{(R)} F \Vert_{L^2(S_{R^{-1}u,R^{-1}v})} =& R^p R^{i+j+k-1} \Vert \Nd^i D^j F \Vert_{L^2(S_{u,v})}, \\
\Vert \Nd^i \Du^j {}^{(R)}F \Vert_{L^2(S_{R^{-1}u,R^{-1}v})} =& R^p R^{i+j+k-1} \Vert \Nd^i \Du^j F \Vert_{L^2(S_{u,v})}.
\end{aligned} %\label{}
\end{align*}
This finishes the proof of Lemma \ref{LEMscalingofNorms}.
\end{proof}

\ni The following lemma is a direct consequence of Definitions \ref{DEFnormFirstOrderDATA}, \ref{DEFnormHH}, \ref{DEFspacetimeNORM} and \ref{DEFscalingNULLSTRUC}, and Lemma \ref{LEMscalingofNorms}; its proof is omitted.
\begin{lemma}[Scale-invariance of data norms] \label{LEMspheredataInvarianceSCale} \label{LEMNORMFULLscalingR} Let $R\geq1$ be a real number. Then it holds that for sphere data $x_{u,v}$ on $S_{u,v}$,
\begin{align*} 
\begin{aligned} 
\Vert {}^{(R)}x_{R^{-1}u,R^{-1}v} \Vert_{\XX(S_{R^{-1}u,R^{-1}v})} = \Vert x_{u,v} \Vert_{\XX(S_{u,v})}, 
\end{aligned} %\label{}
\end{align*}
for outgoing null data $x_{u_0,[v_1,v_2]}$ on $\HH_{u_0,[v_1,v_2]}$, 
\begin{align*} 
\begin{aligned} 
\Vert {}^{(R)}x_{R^{-1} u_0, R^{-1} [v_1,v_2]} \Vert_{\XX(\HH_{R^{-1} u_0, R^{-1} [v_1,v_2]})}= \Vert x_{u_0,[v_1,v_2]} \Vert_{\XX(\HH_{u_0,[v_1,v_2]})},
\end{aligned} %\label{}
\end{align*}
and for ingoing null data $x_{[u_1,u_2],v_0}$ on $\HHb_{[u_1,u_2],v_0}$,
\begin{align*} 
\begin{aligned} 
\Vert {}^{(R)}x_{R^{-1} [u_1,u_2],R^{-1} v_0} \Vert_{\XX(\HHb_{R^{-1} [u_1,u_2],R^{-1} v_0})}=& \Vert x_{[u_1,u_2],v_0} \Vert_{\XX(\HHb_{[u_1,u_2],v_0})}, \\
\Vert {}^{(R)}x_{R^{-1} [u_1,u_2],R^{-1} v_0} \Vert_{\XX^+(\HHb_{R^{-1} [u_1,u_2],R^{-1} v_0})}=& \Vert x_{[u_1,u_2],v_0} \Vert_{\XX^+(\HHb_{[u_1,u_2],v_0})},
\end{aligned} %\label{}
\end{align*}
where the norms $\XX(S_{u,v}), \XX(\HH_{u_0,[v_1,v_2]}), \XX(\HHb_{[u_1,u_2],v_0})$ and $\XX^+(\HHb_{[u_1,u_2],v_0})$ are defined in Definitions \ref{DEFNORMspheredata2}, \ref{DEFspacetimeNORM} and \ref{DEFspacetimeHIGHERNORMdata}, respectively.
\end{lemma}

\ni The charges of Definition \ref{DEFlocalCharges} have the following scaling behaviour.
\begin{lemma}[Scaling of charges] \label{LEMrescaleLOCALCHARGES} Let $x_{u,v}$ be sphere data and let $R\geq1$ be a real number. Let ${}^{(R)}x_{R^{-1}u,R^{-1}v}$ denote the rescaling of $x_{u,v}$ according to Definition \ref{DEFscalingNULLSTRUC}. Then it holds that
\begin{align} 
\begin{aligned} 
\mathbf{E}\lrpar{{}^{(R)}x_{R^{-1}u,R^{-1}v}} =& R^{-1} \cdot \mathbf{E}(x_{u,v}), & \mathbf{P}\lrpar{{}^{(R)}x_{R^{-1}u,R^{-1}v}} =& R^{-1}\cdot \mathbf{P}(x_{u,v}), \\
\mathbf{L}\lrpar{{}^{(R)}x_{R^{-1}u,R^{-1}v}} =& R^{-2} \cdot \mathbf{L}(x_{u,v}), &  \mathbf{G}\lrpar{{}^{(R)}x_{R^{-1}u,R^{-1}v}} =& R^{-2} \cdot \mathbf{G}(x_{u,v}).
\end{aligned} \label{EQscalingofCHARGES}
\end{align} 
\end{lemma}

\begin{proof}[Proof of Lemma \ref{LEMrescaleLOCALCHARGES}] By Definition \ref{DEFlocalCharges} and Lemma \ref{DEFscalingNULLSTRUC} we have on the one hand,
\begin{align*} 
\begin{aligned} 
\mathbf{E}({}^{(R)}x_{R^{-1}u,R^{-1}v}) :=& -\frac{1}{8\pi} \sqrt{4\pi}\lrpar{ r^3 \lrpar{\rh+r \Divd \be}}^{(0)}({}^{(R)}x_{R^{-1}u,R^{-1}v}) \\
=& -\frac{1}{8\pi} \sqrt{4\pi} \lrpar{\frac{r^3}{R^3} \lrpar{R^2 \rh+ R r \Divd (R \be)}}^{(0)}(x_{u,v}) \\
=& R^{-1} \lrpar{-\frac{1}{8\pi} \sqrt{4\pi}\lrpar{r^3\lrpar{\rh+ r \Divd \be}}^{(0)}(x_{u,v})} \\
=& R^{-1} \cdot \mathbf{E}(x_{u,v}).
\end{aligned} %\label{}
\end{align*}
and on the other hand, for $m=-1,0,1$,
\begin{align*} 
\begin{aligned} 
\mathbf{L}^m({}^{(R)}x_{R^{-1}u,R^{-1}v}) :=& \frac{1}{16\pi} \sqrt{\frac{4\pi}{3}} \lrpar{r^3\lrpar{\di\trchi +\trchi \lrpar{\eta - \di \log\Om}}}^{(1m)}_H({}^{(R)}x_{R^{-1}u,R^{-1}v}) \\
=&\frac{1}{16\pi} \sqrt{\frac{4\pi}{3}} \lrpar{\frac{r^3}{R^3} \lrpar{R\di \trchi + R\trchi \lrpar{\eta - \di \log\Om}}}^{(1m)}_H(x_{u,v}) \\
=& R^{-2} \lrpar{\frac{1}{16\pi} \sqrt{\frac{4\pi}{3}}\lrpar{r^3 \lrpar{\di \trchi + \trchi \lrpar{\eta - \di \log\Om}}}^{(1m)}_H(x_{u,v})}\\
=& R^{-2} \cdot \mathbf{L}^m(x_{u,v}).
\end{aligned} %\label{}
\end{align*}

\ni The proof of \eqref{EQscalingofCHARGES} for $\mathbf{P}$ and $\mathbf{G}$ is analogous and omitted. This finishes the proof of Lemma \ref{LEMrescaleLOCALCHARGES}. \end{proof}

\ni In the following we apply the scaling of the Einstein equations to prove Lemma \ref{LEMBoundednessEinftyPinfty1}.
\begin{proof}[Proof of Lemma \ref{LEMBoundednessEinftyPinfty1}] 

\ni First consider $\mathbf{E}$. By Definitions \ref{DEFlocalCharges} and \ref{DEFadmissibleSequences}, \eqref{EQexplicitCalcuChargesSSM8889} and Lemma \ref{LEMrescaleLOCALCHARGES}, for $R\geq1$ sufficiently large,
\begin{align*} 
\begin{aligned} 
\mathbf{E}\lrpar{x_{R\cdot u, R\cdot v}}-M =& R \cdot \lrpar{\mathbf{E}\lrpar{{}^{(R)} x_{u,v}}-M/R} \\
=& - \frac{\sqrt{4\pi}}{8\pi}  R \cdot \lrpar{r^3 \lrpar{\rho + \frac{2M/R}{r^3} + r \Divd {\be}}}^{(0)}\lrpar{{}^{(R)} x_{u,v}} \\
=& -\frac{\sqrt{4\pi}}{8\pi}  R \cdot \lrpar{ r^3\lrpar{\rh+\frac{2M/R}{r_{M/R}(u,v)^3} +\frac{2M/R}{r^3}- \frac{2M/R}{r_{M/R}(u,v)^3} }}^{(0)}\lrpar{{}^{(R)} x_{u,v}}\\ 
&-\frac{\sqrt{4\pi}}{8\pi}  R \cdot \lrpar{ r^4\lrpar{\Divd-\frac{1}{r^2}\Divdo} \be}^{(0)}\lrpar{{}^{(R)} x_{u,v}} \\
=& R \cdot \smallO(R^{-3/2}) + R\cdot \smallO(R^{-3/2}) \cdot \smallO(R^{-3/2})\\
=&\smallO(R^{-1/2}).
\end{aligned} %\label{}
\end{align*}

\ni Second, consider $\mathbf{L}$. By Lemma \ref{LEMrescaleLOCALCHARGES}, Definitions \ref{DEFlocalCharges},  \ref{DEFadmissibleSequences} and \ref{DEFasymptoticCharges} and \eqref{EQalternativeLGdef}, for $m=-1,0,1$,
\begin{align*} 
\begin{aligned} 
\mathbf{L}^m\lrpar{x_{R\cdot u,R\cdot v}} =& R^2 \mathbf{L}^m\lrpar{{}^{(R)} x_{u,v}} \\
=&\frac{1}{8\pi}\sqrt{\frac{8\pi}{3}} R^2\cdot \lrpar{r^3 \lrpar{\beta+\Divd \chih+\chih \cdot (\eta-\di\log\Om)}}_H^{(1m)}\lrpar{{}^{(R)} x_{u,v}} \\
=& \frac{1}{8\pi}\sqrt{\frac{8\pi}{3}} R^2 \cdot \lrpar{\lrpar{r^3 \beta}_H^{(1m)}\lrpar{{}^{(R)} x_{-1,1}} + \smallO(R^{-3/2})\smallO(R^{-3/2})} \\
=& R^2 \cdot \lrpar{\OO(R^{-2}) + \smallO(R^{-3})}\\
=& \OO(1),
\end{aligned} %\label{}
\end{align*}
where we used Lemma \ref{LEMnonlinearFourier} and that $\beta_H^{(1m)}\lrpar{{}^{(R)} x_{u,v}}$ is bounded by
\begin{align*} 
\begin{aligned} 
\left\vert \beta_H^{(1m)}\right\vert = \left\vert \, \int\limits_{S_{u,v}} \gac\lrpar{\beta^{[1]},H^{(1m)}} d\mu_\gac \right\vert \les \Vert \beta^{[1]} \Vert_{L^2(S_{u,v})} = \OO(R^{-2}).
\end{aligned} %\label{}
\end{align*}

\ni The proofs for $\mathbf{P}$ and $\mathbf{G}$ are similar. This finishes the proof of Lemma \ref{LEMBoundednessEinftyPinfty1}. \end{proof}

\ni From Lemmas \ref{LEMBoundednessEinftyPinfty1}, \ref{LEMspheredataInvarianceSCale} and \ref{LEMrescaleLOCALCHARGES}, \eqref{EQdefSCALINGSSprop8999} and Definitions \ref{DEFlocalCharges}, \ref{DEFadmissibleSequences} and \ref{DEFasymptoticCharges} we directly get the following lemma.

\begin{lemma}[Rescaling strongly asymptotically flat families] \label{LEMrescaletoSMALLdata} For two fixed real numbers $v>u$, let $(x_{R\cdot u,R\cdot v})$ be strongly asymptotically flat family of sphere data with asymptotic charge $\mathbf{E}_\infty$ as defined in Definition \ref{DEFasymptoticCharges}. Then it holds that
\begin{align*} 
\begin{aligned} 
\Vert {}^{(R)}x_{u,v} -\mathfrak{m}^{\mathbf{E}_\infty/R}_{u,v} \Vert_{\mathcal{X}(S_{u,v})}=\smallO\lrpar{ R^{-3/2}}, \,\, \Vert \be^{[1]}({}^{(R)}x_{u,v}) \Vert_{L^2(S_{u,v})} = \OO\lrpar{R^{-2}}.
\end{aligned} 
\end{align*}
Moreover, for a fixed real number $\de>0$, let $(x_{-R+R\cdot [-\de,\de], R})$ be a strongly asymptotically flat family of ingoing null data with asymptotic charge $\mathbf{E}_\infty$ as defined in Definition \ref{DEFadmissibleEXTSequences}. Then it holds that
\begin{align*} 
\begin{aligned} 
\Vert {}^{(R)}{x}_{-1+[-\de,\de],1} -{\mathfrak{m}}_{-1+[-\de,\de],1}^{\mathbf{E}_\infty/R} \Vert_{\XX^+\lrpar{\HHb_{-1+[-\de,\de],1}}} 
=&\smallO\lrpar{ R^{-3/2}},\\
\Vert \be^{[1]}({}^{(R)}x \vert_{-1,1}) \Vert_{L^2(S_{-1,1})} =& \OO\lrpar{R^{-2}}.
\end{aligned} 
\end{align*}
\end{lemma}

%%%%%%%%%%%%%%%%%%%%%%%%%%%%%%%%%%%%%%%%
\section{Statement of main results} \label{SECpreciseSTATEMENTMAIN111} 

\ni The following is the precise version of the main theorem of this paper.
\begin{theorem}[Perturbative characteristic gluing to Kerr, version 2] \label{THMMAINPRECISE} Let $\de>0$ be a real number, and let $(\tilde{{x}}_{-R+R\cdot[-\de,\de],R})$ along $\tilde{\HHb}_{-R+R\cdot[-\de,\de],R}$ be a strongly asymptotically flat family of ingoing null data with asymptotic charges
\begin{align*} 
\begin{aligned} 
(\mathbf{E}_\infty, \mathbf{P}_\infty=0, \mathbf{L}_\infty, \mathbf{G}_\infty) \in I(0)\times \RRR^3 \times \RRR^3.
\end{aligned} %\label{}
\end{align*}
For sufficiently large $R\geq1$, there exist
\begin{itemize}
\item sphere data $x'_{-R,R}$ on a perturbation $S'_{-R,R}$ of the sphere $\tilde{S}_{-R,R}$ along $\tilde{\HHb}_{-R+R\cdot[-\de,\de],R}$, 
\item outgoing null data ${x}$ on a null hypersurface ${\HH}_{-R,[R,2R]}$ solving the null structure equations,
\item sphere data $x_{-R,2R}^{\Kerr}$ on a spacelike $2$-sphere $S^{\Kerr}_{-R,2R}$ in a Kerr spacetime,
\end{itemize}
such that we have full matching of sphere data on $S_{-R,R} \subset {\HH}_{-R,[R,2R]}$ and on $S_{-R,2R} \subset {\HH}_{-R,[R,2R]}$,
\begin{align} 
\begin{aligned} 
x_{-R,R} = {x}'_{-R,R}, \,\, {x}_{-R,2R}= x_{-R,2R}^{\Kerr},
\end{aligned} \label{EQpreciseMainTheoremfullgluing}
\end{align}
and the following estimates hold,
\begin{align} 
\begin{aligned} 
\Vert x-\mathfrak{m}^{\mathbf{E}_\infty} \Vert_{\mathcal{X}({\HH}_{-R,[R,2R]})} + \Vert x'_{-R,R} - \mathfrak{m}^{\mathbf{E}_\infty} \Vert_{\XX(S_{-R,R})}  =& \smallO \lrpar{R^{-3/2}}.
\end{aligned}\label{EQmainTheoremPrecisePreciseBounds}
\end{align}
Moreover, the sphere $S^{\Kerr}_{-R,2R}$ in Kerr lies in a spacelike hypersurface $\Si^{\Kerr}$ whose asymptotic invariants are bounded by
\begin{align} 
\begin{aligned} 
\mathbf{E}_{\mathrm{ADM}} =&\mathbf{E}_\infty+ \smallO(R^{-1/2}), & \mathbf{P}_{\mathrm{ADM}} =& \smallO(R^{-1/2}), \\
\mathbf{L}_{\mathrm{ADM}} =& \mathbf{L}_\infty + \smallO(1), & \mathbf{C}_{\mathrm{ADM}} =& \mathbf{G}_\infty +3R\cdot \mathbf{P}_{\mathrm{ADM}}+ \smallO(1),
\end{aligned} \label{EQcoeffEstimLAMBDA}
\end{align}
and $S^{\Kerr}_{-R,2R}$ admits a future-complete outgoing null congruence and past-complete ingoing null congruence. Moreover, if the strongly asymptotically flat family of ingoing null data $(\tilde{{x}}_{-R+R\cdot[-\de,\de],R})$ satisfies the stronger decay rates
\begin{align} 
\begin{aligned} 
\mathbf{E}(\tilde{x}_{-R,R}) = \mathbf{E}_\infty+ \OO(R^{-1}), \,\, \mathbf{P}(\tilde{x}_{-R,R}) = \OO(R^{-3/2}),
\end{aligned} \label{EQstrongDecayPconditionmainthmprecise999094}
\end{align}
then the asymptotic invariants of the spacelike hypersurface $\Si^{\Kerr}$ are bounded by
\begin{align} 
\begin{aligned} 
\mathbf{E}_{\mathrm{ADM}} =&\mathbf{E}_\infty+ \OO(R^{-1}), & \mathbf{P}_{\mathrm{ADM}} =& \OO(R^{-3/2}), \\
\mathbf{L}_{\mathrm{ADM}} =& \mathbf{L}_\infty + \smallO(1), & \mathbf{C}_{\mathrm{ADM}} =& \mathbf{G}_\infty + \smallO(1),
\end{aligned} \label{EQcoeffEstimLAMBDAstronger}
\end{align}
so in particular, $\mathbf{C}_{\mathrm{ADM}}$ is not growing in $R$.
\end{theorem}

\ni \emph{Remarks on Theorem \ref{THMMAINPRECISE}.}
\begin{enumerate}

\item The key ingredients of the proof are the perturbative characteristic gluing of \cite{ACR3,ACR1} (used as black box), the geometric interpretation of the asymptotic charges $(\mathbf{E}_\infty,\mathbf{P}_\infty,\mathbf{L}_\infty,\mathbf{G}_\infty)$ in terms of the ADM asymptotic invariants \emph{energy}, \emph{linear momentum}, \emph{angular momentum}, and \emph{center-of-mass} in Section \ref{SECstatementConstruction}, and the $10$-dimensional parametrisation of Kerr reference sphere data through asymptotic invariants in Section \ref{SECappKerrFamilyDetails}.

\item The additional convergence condition \eqref{EQstrongDecayPconditionmainthmprecise999094} is satisfied by sphere data constructed from strongly asymptotically flat spacelike initial data, see Section \ref{SECstatementConstruction}.

\item The smallness on the right-hand side of \eqref{EQmainTheoremPrecisePreciseBounds} is consistent with our definition of strong asymptotic flatness, see Definition \ref{DEFadmissibleEXTSequences}.

\item Theorem \ref{THMMAINPRECISE} is at the level of $C^2$-gluing for the metric components. It can be generalized to include higher-order derivatives \emph{tangential to the gluing hypersurface} $\HH_{-R,[R,2R]}$; see Theorem 3.2 in \cite{ACR1} for the corresponding setup. For the gluing of higher-order derivatives in \emph{all directions}, we refer to Theorem \ref{THMdoubleCHARgluingTOkerr} below. 

\item More precisely, in Theorem \ref{THMMAINPRECISE} we glue to a \emph{Kerr reference sphere} $S_{-R,2R}^{\KK(\la)}$ which we explicitly construct in Section \ref{SECappKerrFamilyDetails} for some \emph{asymptotic invariants vector} $\la \in I(0) \times \RRR^3 \times \RRR^3$.

\item In Theorem \ref{THMMAINPRECISE} it is not necessary to have a family of ingoing null data data. Indeed, one can replace this family with one fixed ingoing null datum with sufficiently strong bounds.

\end{enumerate}

\ni The argument for the matching to Kerr in Theorem \ref{THMMAINPRECISE} applies similarly to the \emph{bifurcate} characteristic gluing of \cite{ACR3,ACR1}, see Remark \ref{REMtransversalGluing}. The corresponding theorem is the following.

\begin{theorem}[Bifurcate characteristic gluing to Kerr] \label{THMdoubleCHARgluingTOkerr} Let $m\geq1$ be an integer. Let  
\begin{align*} 
\begin{aligned} 
\lrpar{x_{0,R}, \DD^{L,m}_{0,R}, \DD^{\Lb,m}_{0,R} },
\end{aligned} %\label{}
\end{align*}
be a strongly asymptotically flat family of smooth higher-order sphere data with asymptotic charges
\begin{align*} 
\begin{aligned} 
(\mathbf{E}_\infty, \mathbf{P}_\infty=0, \mathbf{L}_\infty, \mathbf{G}_\infty) \in I(0)\times \RRR^3 \times \RRR^3.
\end{aligned} %\label{}
\end{align*}
For sufficiently large $R\geq1$, there exist
\begin{itemize}
\item smooth higher-order ingoing null data $(x,{\DD}^{L,m}, {\DD}^{\Lb,m})$ on $\HHb_{[-R,0],R}$ and outgoing higher-order null data $(x,{\DD}^{L,m}, {\DD}^{\Lb,m})$ on $\HH_{-R,[R,2R]}$ solving the higher-order null structure equations and matching to order $m$ on $S_{-R,R}$,
\item smooth higher-order sphere data 
$$\lrpar{x^{\Kerr}_{-R,2R},{\DD}^{L,m, \Kerr}_{-R,2R}, {\DD}^{\Lb,m, \Kerr}_{-R,2R}}$$ 
on a smooth spacelike $2$-sphere $S^{\Kerr}_{-R,2R}$ in a Kerr spacetime,
\end{itemize}
such that 
\begin{align*} 
\begin{aligned} 
\lrpar{x,{\DD}^{L,m}, {\DD}^{\Lb,m}} \vert_{{S}_{0,R}} =& \lrpar{x_{0,R}, \DD^{L,m}_{0,R}, \DD^{\Lb,m}_{0,R} },\\
\lrpar{x,{\DD}^{L,m}, {\DD}^{\Lb,m}} \big\vert_{S_{-R,2R}} =& \lrpar{x^{\Kerr}_{-R,2R},{\DD}^{L,m, \Kerr}_{-R,2R}, {\DD}^{\Lb,m, \Kerr}_{-R,2R}}.
\end{aligned}
\end{align*}
The sphere $S^{\Kerr}_{-R,2R}$ in Kerr lies in a spacelike hypersurface with asymptotic invariants
\begin{align*} 
\begin{aligned} 
\mathbf{E}_{\mathrm{ADM}} =&\mathbf{E}_\infty+ \smallO(R^{-1/2}), & \mathbf{P}_{\mathrm{ADM}} =& \smallO(R^{-1/2}), \\
\mathbf{L}_{\mathrm{ADM}} =& \mathbf{L}_\infty + \smallO(1), & \mathbf{C}_{\mathrm{ADM}} =& \mathbf{G}_\infty +3R\cdot \mathbf{P}(\la)+ \smallO(1).
\end{aligned}
\end{align*}
amd admits a future-complete outgoing null congruence and past-complete ingoing null congruence. Moreover, if the strongly asymptotically flat family of sphere data $(x_{0,R})$ satisfies the stronger decay condition
\begin{align*} 
\begin{aligned} 
\mathbf{E}(x_{0,R}) = \mathbf{E}_\infty+ \OO(R^{-1}), \,\, \mathbf{P}(x_{0,R}) = \OO(R^{-3/2}),
\end{aligned}
\end{align*}
then the asymptotic invariants of the spacelike hypersurface $\Si^{\Kerr}$ are bounded by
\begin{align*} 
\begin{aligned} 
\mathbf{E}_{\mathrm{ADM}} =&\mathbf{E}_\infty+ \OO(R^{-1}), & \mathbf{P}_{\mathrm{ADM}} =& \OO(R^{-3/2}), \\
\mathbf{L}_{\mathrm{ADM}} =& \mathbf{L}_\infty + \smallO(1), & \mathbf{C}_{\mathrm{ADM}} =& \mathbf{G}_\infty + \smallO(1).
\end{aligned}
\end{align*}

\end{theorem}

\ni \emph{Remarks on Theorem \ref{THMdoubleCHARgluingTOkerr}.}
\begin{enumerate}
\item The strong asymptotic flatness of the family $x_{0,R}$ is consistent with decay towards spacelike infinity. In particular, the spheres $S_{0,R}$ should be interpreted as spheres on a spacelike hypersurface with radius of size $R$.
\item Theorem \ref{THMdoubleCHARgluingTOkerr} is at the level of $C^{m+2}$-gluing for the metric components; see Section 2.10 in \cite{ACR1} for the precise definition of higher-order sphere data.
\end{enumerate}

\ni As corollary of Theorem \ref{THMdoubleCHARgluingTOkerr}, we give in Section \ref{SECspacelikecorollary} an alternative proof of the Corvino--Schoen gluing to Kerr for strongly asymptotically flat spacelike initial data. We refer to Section \ref{SECCSgluingbasics} below for the definition of spacelike initial data, strong asymptotic flatness and asymptotic invariants $\mathbf{E}_{\mathrm{ADM}}$, $\mathbf{P}_{\mathrm{ADM}}$, $\mathbf{L}_{\mathrm{ADM}}$ and $\mathbf{C}_{\mathrm{ADM}}$. 

\begin{corollary}[Spacelike gluing to Kerr, version 2] \label{THMspacelikeGLUINGtoKERRv2} Consider strongly asymptotically flat spacelike initial data $(\Si,g,k)$ with asymptotic invariants
\begin{align*} 
\begin{aligned} 
\Big(\mathbf{E}_{\mathrm{ADM}}, \mathbf{P}_{\mathrm{ADM}}=0, \mathbf{L}_{\mathrm{ADM}}, \mathbf{C}_{\mathrm{ADM}}\Big) \in I(0) \times \RRR^3 \times \RRR^3.
\end{aligned} %\label{}
\end{align*}
For real numbers $R\geq1$ sufficiently large, there exists a Kerr spacetime $(\MM^{\Kerr},\g^{\Kerr})$ and a spacelike hypersurface $\Si^{\mathrm{Kerr}}$ with asymptotic invariants
\begin{align*} 
\begin{aligned} 
\lrpar{\mathbf{E}_{\mathrm{ADM}}^{\Kerr}, \mathbf{P}_{\mathrm{ADM}}^{\Kerr}, \mathbf{L}_{\mathrm{ADM}}^{\Kerr}, \mathbf{C}_{\mathrm{ADM}}^{\Kerr}},
\end{aligned}
\end{align*}
such that the spacelike initial data $(g,k)$ of $\Si$ can be glued across a spacelike annulus $A_{[R,3R]}$ to the induced spacelike initial data $(g^{\Kerr},k^{\Kerr})$ of $\Si^{\mathrm{Kerr}}$. The Kerr asymptotic invariants can bounded by
\begin{align*} 
\begin{aligned} 
\mathbf{E}_{\mathrm{ADM}}^{\Kerr}=& \mathbf{E}_{\mathrm{ADM}}+\OO(R^{-1}), & \mathbf{P}_{\mathrm{ADM}}^{\Kerr}=&  \OO(R^{-3/2}), \\
\mathbf{L}_{\mathrm{ADM}}^{\Kerr} =& \mathbf{L}_{\mathrm{ADM}}+\smallO(1), & \mathbf{C}_{\mathrm{ADM}}^{\Kerr}=& \mathbf{C}_{\mathrm{ADM}}+\smallO(1).
\end{aligned} %\label{}
\end{align*}

\end{corollary}

\ni \emph{Remarks on Corollary \ref{THMspacelikeGLUINGtoKERRv2}.}
\begin{enumerate}
\item More precisely, in Corollary \ref{THMspacelikeGLUINGtoKERRv2} we glue to a Kerr reference spacelike initial data $(g^{\KK(\la)}, k^{\KK(\la)})$ for some \emph{asymptotic invariants vector} $\la \in I(0) \times \RRR^3 \times \RRR^3$, see Section \ref{SECappKerrFamilyDetails}.
\end{enumerate}

%%%%%%%%%%%%%%%%%%%%%%%%%%%
%%%%%%%%%%%%%%%%%%%%%%%%%%%%%%%%%%%%%%%%
\section{Proof of perturbative characteristic gluing to Kerr} \label{SECproofMainTheorem1} \ni In this section we prove Theorem \ref{THMMAINPRECISE}. Let $\de>0$ be a real number and let $(\tilde{x}_{-R+R\cdot[-\de,\de],R})$ be a strongly asymptotically flat family of ingoing null data with asymptotic charges 
\begin{align*} 
\begin{aligned} 
(\mathbf{E}_\infty, \mathbf{P}_\infty=0, \mathbf{L}_\infty, \mathbf{G}_\infty) 
\end{aligned} %\label{}
\end{align*}
with $\mathbf{E}_\infty>0$. We proceed as follows.
\begin{enumerate} 

\item In Section \ref{SECsmalldataRescaling} we rescale the strongly asymptotically flat ingoing null data $$(\tilde{x}_{-R+R\cdot[-\de,\de],R})$$ to  ingoing null data 
\begin{align} 
\begin{aligned} 
({}^{(R)}\tilde{x}_{-1+[-\de,\de],1}).
\end{aligned} \label{EQextendedspheredataProofMain777}
\end{align}
\noindent For $R\geq1$ sufficiently large, the rescaled ingoing null data \eqref{EQextendedspheredataProofMain777} is close to Schwarzschild of mass $\mathbf{E}_\infty/R$, see Section \ref{SECsmalldataRescaling} below.

\item In Section \ref{SECApplicationCHARGLU333} we apply the perturbative characteristic gluing of \cite{ACR3,ACR1} to glue -- up to the $10$-dimensional space of charges $(\mathbf{E}, \mathbf{P}, \mathbf{L}, \mathbf{G})$ -- from the rescaled ingoing null data \eqref{EQextendedspheredataProofMain777} to sphere data corresponding to a sphere in a Kerr spacetime to be determined.

\item In Section \ref{SECconclusionofMainTheorem} we use a classical topological degree argument to prove that there exists a sphere in a Kerr spacetime such that, following Step (2) above, also the charges $(\mathbf{E}, \mathbf{P}, \mathbf{L}, \mathbf{G})$ are glued. 

\item In Section \ref{SECconclusionofMainTheoremREAL} we conclude the proof of Theorem \ref{THMMAINPRECISE} by writing out the explicit estimates and scaling the gluing construction from ${\HH}_{-1,[1,2]}$ to ${\HH}_{-R,[R,2R]}$.

\item The proof that the sphere $S_{-R,2R}^{\Kerr}$ in Kerr admits a future-complete outgoing null congruence and a past-complete ingoing null congruence is postponed to Appendix \ref{SECcompletenessKERR}.

\end{enumerate}

%%%%%%%%%%%%%%%%%%%%%%%%%%%%%%%%%%%
\subsection{Rescaling to small sphere data} \label{SECsmalldataRescaling} 
\ni  Using the scaling of the Einstein equations, see Definition \ref{DEFscalingNULLSTRUC}, for $R\geq1$ large we rescale $(\tilde{x}_{-R+R\cdot[-\de,\de],R})$ to ingoing null data 
 $$({}^{(R)}\tilde{x}_{-1+[-\de,\de],1}).$$ 
 By Lemma \ref{LEMrescaletoSMALLdata} and \eqref{EQdefSCALINGSSprop8999}, it holds that
\begin{align} 
\begin{aligned} 
\Vert {}^{(R)}\tilde{x}_{-1+[-\de,\de],1} -{\mathfrak{m}}^{\mathbf{E}_\infty/R} \Vert_{\XX^+\lrpar{\HHb_{-1+[-\de,\de],1}}} =&\smallO(R^{-3/2}).
\end{aligned} \label{EQfullenergyEstimates32777789}
\end{align}

%%%%%%%%%%%%%%%%%%%%%%%%%%%%%%%%%%%%%%%%
\subsection{Application of perturbative characteristic gluing of \cite{ACR3,ACR1}} 
\label{SECApplicationCHARGLU333}

\ni In this section we apply the perturbative characteristic gluing of \cite{ACR3,ACR1} to glue from the rescaled ingoing null data $({}^{(R)}\tilde{x}_{-1+[-\de,\de],1})$ to sphere data of a Kerr spacetime.

For the construction and analysis of the Kerr sphere data we refer to Section \ref{SECappKerrFamilyDetails}. The take-away is the following. We parametrize Kerr sphere data by an \emph{asymptotic invariants vector} 
$$\la=(\mathbf{E}(\la),\mathbf{P}(\la),\mathbf{L}(\la),\mathbf{C}(\la)) \in I(0) \times \RRR^3 \times \RRR^3,$$ 
where $I(0):=\{(\mathbf{E},\mathbf{P})\in \RRR \times \RRR^3: \mathbf{E} > \vert\mathbf{P}\vert\}$ and denote the corresponding Kerr sphere data by $x_{-R,2R}^{\KK(\la)}$ (for real numbers $R\geq1$ large), see Section \ref{SECKerrSphereDataDefinition}. The sphere data $x_{-R,2R}^{\KK(\la)}$ corresponds to a sphere $S_{-R,2R}$ lying in a Kerr spacelike initial data set with ADM asymptotic invariants $\la$ (see also \eqref{EQasympDECAY55540002} below). We denote by ${}^{(R)} x_{-1,2}^{\KK(\la)}$ the rescaled Kerr sphere data.

In particular, in the following we consider asymptotic invariants vector 
\begin{align} 
\begin{aligned} 
\la \in \EE_R(\mathbf{E}_\infty) \subset I(0) \times \RRR^3 \times \RRR^3,
\end{aligned} \label{EQla7707}
\end{align}
where for real numbers $R\geq1$ large and $\mathbf{E}_\infty>0$, the ellipsoid $\EE_R(\mathbf{E}_\infty)$ is defined by (see also \eqref{EQdefeereinftyset999002})
\begin{align} 
\begin{aligned} 
&\lrpar{R^{1/2} \vert \mathbf{E}(\la)-\mathbf{E}_\infty \vert}^2 + \lrpar{R^{1/2}\vert \mathbf{P}(\la) \vert}^2 \\
&+ \lrpar{R^{-1/4} \vert \mathbf{L}(\la)\vert}^2 + \lrpar{R^{-1/2}\vert \mathbf{C}(\la)\vert}^2  \leq \lrpar{\mathbf{E}_\infty}^2.
\end{aligned} \label{EQDEFeeReinfty8889}
\end{align}
Importantly, in Proposition \ref{PROPspheredataESTIM91} it is shown that for $R\geq1$ sufficiently large and for asymptotic invariants vectors $\la \in\EE_R(\mathbf{E}_\infty)$, the following estimate holds (see also Lemma \ref{LEMspheredataInvarianceSCale})
\begin{align} 
\begin{aligned} 
\Vert {}^{(R)} x_{-1,2}^{\KK(\la)} - \mathfrak{m}^{\mathbf{E}_\infty/R} \Vert_{\XX(S_{-1,2})} \les&
R^{-1} \cdot \vert \mathbf{E}(\la)-\mathbf{E}_\infty \vert + R^{-1} \cdot \vert \mathbf{P}(\la) \vert \\
&+ R^{-2} \cdot \vert \mathbf{L}(\la) \vert + R^{-2} \cdot \vert \mathbf{C}(\la) \vert \\
&+ \lrpar{\frac{ R^{-2} \cdot \vert \mathbf{L}(\la) \vert + \frac{\vert \mathbf{P}(\la) \vert}{\mathbf{E}_\infty} \cdot R^{-2} \cdot \vert \mathbf{C}(\la) \vert}{\mathbf{E}_\infty/R}}^2 \\
=&\OO\lrpar{R^{-3/2}}.
\end{aligned} \label{EQsmallnessEstimateKERR8999}
\end{align}

\ni Hence, by \eqref{EQfullenergyEstimates32777789} and \eqref{EQsmallnessEstimateKERR8999}, for $R\geq1$ sufficiently large we can apply the perturbative characteristic gluing of \cite{ACR3,ACR1} with 
\begin{align} 
\begin{aligned} 
M=\mathbf{E}_\infty/R \text{ and } \varep=\OO(R^{-3/2}),
\end{aligned} \label{EQappl7007}
\end{align}
to glue from the rescaled ingoing null data ${}^{(R)}\tilde{x}_{-1+[-\de,\de],1}$ to ${}^{(R)} x_{-1,2}^{\KK(\la)}$ for an asymptotic invariants vector $\la \in\EE_R(\mathbf{E}_\infty)$ to be determined. That is, there are
\begin{itemize}
\item sphere data ${}^{(R)}{x}_{-1,1}$ on a sphere $S_{-1,1}$ stemming from a perturbation of $\tilde{S}_{-1,1}$ in $\tilde{\HHb}_{-1+[-\de,\de],1}$,  
\item a solution $x \in \mathcal{X}(\HH_{-1,[1,2]})$ to the null structure equations on $\HH_{-1,[1,2]}$,
\end{itemize}
such that $x\vert_{S_{-1,1}} = {}^{(R)}{x}_{-1,1}$ and $x\vert_{S_{-1,2}}$ agrees with ${}^{(R)} x_{-1,2}^{\KK(\la)}$ up to the $10$-dimensional space of charges $(\mathbf{E}, \mathbf{P}, \mathbf{L}, \mathbf{G})$, that is, if we have that
\begin{align} 
\begin{aligned} 
\lrpar{\mathbf{E},\mathbf{P},\mathbf{L},\mathbf{G}}(x\vert_{S_{-1,2}}) =\lrpar{\mathbf{E},\mathbf{P},\mathbf{L},\mathbf{G}}\lrpar{{}^{(R)} x_{-1,2}^{\KK(\la)}},
\end{aligned} \label{EQchargegluingremainingtoshow}
\end{align}
then the constructed solution $x$ satisfies
\begin{align} 
\begin{aligned} 
x\vert_{S_{-1,2}} = {}^{(R)} x_{-1,2}^{\KK(\la)}.
\end{aligned}\label{EQfullgluing353453}
\end{align}
Importantly, by \eqref{EQappl7007} and the estimates proved in \cite{ACR1}, the following general charge estimate holds,
\begin{align} 
\begin{aligned} 
\lrpar{\mathbf{E},\mathbf{P},\mathbf{L},\mathbf{G}} \lrpar{x \vert_{S_{-1,2}}} - \lrpar{\mathbf{E},\mathbf{P},\mathbf{L},\mathbf{G}} \lrpar{{}^{(R)}\tilde{x}_{-1,1}}=& \OO\lrpar{M\cdot\varep} + \OO(\varep^2) \\
=& \OO\lrpar{\frac{\mathbf{E}_\infty}{R}R^{-3/2}} + \OO(R^{-3}) \\
=& \OO(R^{-5/2}).
\end{aligned} \label{EQchargeEstimate999901345}
\end{align}

%%%%%%%%%%%%%%%%%%%%%%%%%%%%%%%%%%%%%%%%
%%%%%%%%%%%%%%%%%%%%%%%%%%%%%%%%%%%%%%%%
\subsection{Choice of Kerr spacetime} \label{SECconclusionofMainTheorem} 
\ni In this section we use a classical topological degree argument to determine an asymptotic invariants vector $\la' \in \EE_R(\mathbf{E}_\infty)$ such that \eqref{EQchargegluingremainingtoshow} holds. The idea to determine the $\la'$ by a degree argument is similar to \cite{CorvinoSchoen}.

First, for asymptotic invariants vectors $\la \in \EE_R(\mathbf{E}_\infty)$, we define the error map $f_R\lrpar{\la}$ by
\begin{align} 
\begin{aligned} 
f_R\lrpar{\la} := \lrpar{R\, \mathbf{E},R \,\mathbf{P}, R^2\,\mathbf{L}, R^2\,\mathbf{G}}\lrpar{x \vert_{S_{-1,2}}} -\lrpar{R\,\mathbf{E},R\, \mathbf{P}, R^2\,\mathbf{L}, R^2 \,\mathbf{G}}\lrpar{{}^{(R)}x^{\KK(\la)}_{-1,2}}.
\end{aligned} \label{EQdeffRmap}
\end{align}
In the following we show that for $R\geq1$ sufficiently large, there is a $\la'\in \EE_R(\mathbf{E}_\infty)$ such that
\begin{align} 
\begin{aligned} 
f_R(\la')=0.
\end{aligned} \label{EQflaiszero}
\end{align}
By definition of $f_R$ in \eqref{EQdeffRmap}, the condition \eqref{EQflaiszero} is equivalent to charge matching at $S_{-1,2}$, see \eqref{EQchargegluingremainingtoshow}, which subsequently implies the matching \eqref{EQfullgluing353453}.

To prove the existence of $\la'$ satisfying \eqref{EQflaiszero}, we estimate $f_R(\la)$ and apply a topological degree argument. First, we can estimate $f_R(\la)$ by using the following estimates we have for the charges $(\mathbf{E},\mathbf{P},\mathbf{L},\mathbf{G})$.

\begin{itemize}

\item From \eqref{EQchargeEstimate999901345} we have
\begin{align} 
\begin{aligned} 
\lrpar{\mathbf{E},\mathbf{P},\mathbf{L},\mathbf{G}} \lrpar{x \vert_{S_{-1,2}}} - \lrpar{\mathbf{E},\mathbf{P},\mathbf{L},\mathbf{G}} \lrpar{{}^{(R)}\tilde{x}_{-1,1}}= \OO(R^{-5/2}).
\end{aligned} \label{EQestimJJJ1} 
\end{align}

\item By Lemma \ref{LEMBoundednessEinftyPinfty1}, the strongly asymptotically flat family of sphere data $(\tilde{x}_{-R,R})$ satisfies (with $\mathbf{P}_\infty=0$)
\begin{align} 
\begin{aligned} 
\lrpar{\mathbf{E},\mathbf{P}, \mathbf{L}, \mathbf{G}}\lrpar{\tilde{x}_{-R,R}} =& \lrpar{\mathbf{E}_\infty,\mathbf{P}_\infty, \mathbf{L}_\infty, \mathbf{G}_\infty}\\
&+\lrpar{\smallO(R^{-1/2}),\smallO(R^{-1/2}),\smallO(1),\smallO(1)}.
\end{aligned} \label{EQasympDECAY55540001}
\end{align}

\item By Proposition \ref{CORfullchargeexpressionSRkerrreference}, for $R\geq1$ large, the Kerr reference sphere data $x^{\KK(\la)}_{-R,2R}$ satisfies for $\la \in \EE_R(\mathbf{E}_\infty)$,
\begin{align} 
\begin{aligned} 
\mathbf{E}\lrpar{x^{\KK(\la)}_{-R,2R}} =& \mathbf{E}(\la)+\OO(R^{-1}), & \mathbf{P}\lrpar{x^{\KK(\la)}_{-R,2R}} =& \mathbf{P}(\la)+\OO(R^{-3/2}), \\
\mathbf{L}\lrpar{x^{\KK(\la)}_{-R,2R}} =& \mathbf{L}(\la)+\OO(R^{-1/2}), & \mathbf{G}\lrpar{x^{\KK(\la)}_{-R,2R}} =& \mathbf{C}(\la)-3R\cdot \mathbf{P}(\la) +\OO(R^{-1/4}).
\end{aligned} \label{EQasympDECAY55540002}
\end{align}
\end{itemize}

\ni Applying \eqref{EQestimJJJ1}, \eqref{EQasympDECAY55540001} and \eqref{EQasympDECAY55540002} to \eqref{EQdeffRmap}, we can estimate $f_R\lrpar{\la}$ as follows,
\begin{align} 
\begin{aligned} 
f_R\lrpar{\la} :=& \lrpar{R\, \mathbf{E},R \,\mathbf{P}, R^2\,\mathbf{L}, R^2\,\mathbf{G}}\lrpar{x \vert_{S_{-1,2}}} -\lrpar{R\,\mathbf{E},R\, \mathbf{P}, R^2\,\mathbf{L}, R^2 \,\mathbf{G}}\lrpar{{}^{(R)}x^{\KK(\la)}_{-1,2}} \\
=& \lrpar{R\,\mathbf{E},R\,\mathbf{P}, R^2\,\mathbf{L}, R^2\,\mathbf{G}}\lrpar{{}^{(R)} \tilde{x}_{-1,1}} - \lrpar{R\,\mathbf{E},R\,\mathbf{P}, R^2\,\mathbf{L}, R^2\,\mathbf{G}}\lrpar{{}^{(R)}x^{\KK(\la)}_{-1,2}} \\
& +\lrpar{\OO(R^{-3/2}), \OO(R^{-3/2}), \OO(R^{-1/2}), \OO(R^{-1/2})} \\
=& \lrpar{\mathbf{E},\mathbf{P}, \mathbf{L}, \mathbf{G}}\lrpar{\tilde{x}_{-R,R}} - \lrpar{\mathbf{E},\mathbf{P}, \mathbf{L}, \mathbf{G}}\lrpar{x^{\KK(\la)}_{-R,2R}} \\
& +\lrpar{\OO(R^{-3/2}), \OO(R^{-3/2}), \OO(R^{-1/2}), \OO(R^{-1/2})} \\
=& \lrpar{\mathbf{E}_\infty,\mathbf{P}_\infty, \mathbf{L}_\infty, \mathbf{G}_\infty} - \lrpar{\mathbf{E}(\la),\mathbf{P}(\la), \mathbf{L}(\la), \mathbf{C}(\la)-3R\cdot\mathbf{P}(\la)} \\
&+\lrpar{\smallO(R^{-1/2}), \smallO(R^{-1/2}), \smallO(1), \smallO(1)},
\end{aligned} \label{EQestimatefmap1}
\end{align}
where we underline that the error terms also depend on $\la \in \EE_R(\mathbf{E}_\infty)$.

Second, we have the following classical topological degree argument; see, for example, Chapter 1 of \cite{Nirenberg}.
\begin{lemma}[Topological degree argument] \label{DegreeLemma} Let $B \subset \RRR^{10}$ be the open unit ball. Let $f_0$ and $f_1$ be two continuous maps on $\overline{B}$ into $\RRR^{10}$ and assume that $f_0$ is a homeomorphism on $\overline{B}$ with $f_0(\la_0)=0$ for a $\la_0 \in B$. For $0\leq t \leq 1$, let $f(\la,t)$ be a homotopy on $\overline{B}$ such that 
\begin{align*} 
\begin{aligned} 
f(\la,0)= f_0(\la), \,\, f(\la,1)=f_1(\la).
\end{aligned} %\label{}
\end{align*}
If for all $0\leq t \leq 1$ it holds that
\begin{align} \begin{aligned} 
0 \notin f(\pr B,t),
\end{aligned}\label{EQhomotopyCondition} \end{align} 
then there exists $\la' \in B$ such that $f_1(\la')=0$.
\end{lemma}

\begin{remark} \label{REMtransversalGluing} The proof of \eqref{EQflaiszero} below uses only the charge estimate \eqref{EQestimJJJ1}. Given that the bifurcate characteristic gluing provides analogous charge estimates, the proof applies also to the matching to Kerr in Theorem \ref{THMdoubleCHARgluingTOkerr}.
\end{remark}

\ni We are now in position to prove the existence of $\la'$ such that \eqref{EQflaiszero} is satisfied. Based on \eqref{EQestimatefmap1}, define for $0\leq t \leq 1$ the homotopy $f_R(\la, t)$ on $\EE_R(\mathbf{E}_\infty)$ by 
\begin{align} 
\begin{aligned} 
f_R(\la,t) :=& \lrpar{\mathbf{E}_\infty,\mathbf{P}_\infty, \mathbf{L}_\infty, \mathbf{G}_\infty} - \lrpar{\mathbf{E}(\la),\mathbf{P}(\la), \mathbf{L}(\la), \mathbf{C}(\la)-3R\cdot\mathbf{P}(\la)}\\
&+t \cdot \lrpar{\smallO(R^{-1/2}), \smallO(R^{-1/2}), \smallO(1), \smallO(1)},
\end{aligned} \label{EQdefinitionHOMOTOPY}
\end{align}
such that
\begin{align} 
\begin{aligned} 
f_R(\la, 0) =& \lrpar{\mathbf{E}_\infty,\mathbf{P}_\infty, \mathbf{L}_\infty, \mathbf{G}_\infty} - \lrpar{\mathbf{E}(\la),\mathbf{P}(\la), \mathbf{L}(\la), \mathbf{C}(\la)-3R\cdot\mathbf{P}(\la)}, \\ f_R(\la,1) =& f_R(\la).
\end{aligned} \label{EQpreciseDEFhomotopyvalues8889}
\end{align}
We make the following three observations.
\begin{enumerate}

\item From \eqref{EQpreciseDEFhomotopyvalues8889} it follows that $f_R(\la,0)$ is a homeomorphism on $\EE_R(\mathbf{E}_\infty)$.

\item For $R\geq1$ large, we have that
\begin{align*} 
\begin{aligned} 
\la_0 := (\mathbf{E}_\infty, 0, \mathbf{L}_\infty, \mathbf{G}_\infty) \in \EE_R(\mathbf{E}_\infty),
\end{aligned} %\label{}
\end{align*}
and satisfies, by the definition of $f_R(\la, 0)$ in \eqref{EQpreciseDEFhomotopyvalues8889},
\begin{align*} 
\begin{aligned} 
f_R(\la_0,0) =0.
\end{aligned} %\label{}
\end{align*}

\item For $R\geq1$ sufficiently large and all $0\leq t \leq 1$, it holds that
\begin{align} 
\begin{aligned} 
0 \notin f_R(\partial \mathcal{E}_R\lrpar{\mathbf{E}_\infty}, t).
\end{aligned} \label{EQminimum444}
\end{align}

\noindent Indeed, assume by contradiction that there are $\tilde{\la} \in \partial \mathcal{E}_R\lrpar{\mathbf{E}_\infty}$ and $0\leq \tilde t \leq 1$ such that 
\begin{align} 
\begin{aligned} 
f_R(\tilde\la,\tilde t)=0.
\end{aligned} %\label{}
\end{align}
Then by definition of $f_R(\la,t)$ in \eqref{EQdefinitionHOMOTOPY}, 
\begin{align} 
\begin{aligned} 
R^{1/2} \lrpar{ \mathbf{E}(\tilde\la)-\mathbf{E}_\infty } =&\tilde t \cdot \smallO(1), \\
R^{1/2} \cdot \mathbf{P}(\tilde\la) =& \tilde t \cdot \smallO(1), \\
R^{-1/4} \lrpar{ \mathbf{L}(\tilde\la)-\mathbf{L}_\infty } =& \tilde t \cdot \smallO(R^{-1/4}), \\
\frac{R^{-1/2}}{2} \lrpar{ \mathbf{C}(\tilde\la)-\mathbf{G}_\infty }=& \frac{R^{-1/2}}{2} \lrpar{-3R \cdot \mathbf{P}(\tilde\la) + \tilde t \cdot \smallO(1)} \\
=& \frac{R^{-1/2}}{2} \lrpar{-3R \cdot \tilde t\cdot \smallO(R^{-1/2}) + \tilde t \cdot \smallO(1)} \\
=& \tilde t \cdot \smallO(1).
\end{aligned} \label{EQlambdaContradiction123}
\end{align}
The estimates \eqref{EQlambdaContradiction123} imply in particular that for $R\geq1$ sufficiently large,
\begin{align*} 
\begin{aligned} 
&\lrpar{R^{1/2} \vert \mathbf{E}(\tilde\la)-\mathbf{E}_\infty \vert}^2 + \lrpar{R^{1/2}\vert \mathbf{P}(\tilde\la) \vert}^2 \\
&+ \lrpar{R^{-1/4} \vert \mathbf{L}(\tilde\la)\vert}^2 + \lrpar{R^{-1/2}\vert \mathbf{C}(\tilde\la)\vert}^2 \les \tilde{t}\cdot \smallO(1) <\lrpar{\mathbf{E}_\infty}^2,
\end{aligned} %\label{}
\end{align*}
which implies that $\tilde \la \notin \partial \mathcal{E}_R\lrpar{\mathbf{E}_\infty}$ (see the definition of $\EE_R(\mathbf{E}_\infty)$ in \eqref{EQDEFeeReinfty8889}). This is a contradiction and hence finishes the proof of \eqref{EQminimum444}.

\end{enumerate}

\ni By the above observations and the fact that the set $\mathcal{E}_R\lrpar{\mathbf{E}_\infty}\subset \RRR^{10}$ is topologically a ball, we can apply Lemma \ref{DegreeLemma} to the homotopy $f_R(\la,t)$ for $R\geq1$ sufficiently large, and conclude the existence of a vector $\la' \in \mathcal{E}_R\lrpar{\mathbf{E}_\infty}$ such that 
\begin{align} 
\begin{aligned} 
f_R(\la', 1)=0.
\end{aligned} \label{EQmatchingcharge88}
\end{align}
This finishes the proof of \eqref{EQflaiszero}. Moreover, we deduce from \eqref{EQdefinitionHOMOTOPY} that
\begin{align} 
\begin{aligned} 
 \mathbf{E}(\la') =&\mathbf{E}_\infty+ \smallO(R^{-1/2}), & \mathbf{P}(\la') =& \smallO(R^{-1/2}), \\
\mathbf{L}(\la') =&\mathbf{L}_\infty+ \smallO(1), &  \mathbf{C}(\la') =& \mathbf{G}_\infty + 3R\cdot \mathbf{P}(\la')+ \smallO(1).
\end{aligned} \label{EQmatchingKerrparameters}
\end{align}

\ni It remains to show that in case of the stronger decay assumption \eqref{EQstrongDecayPconditionmainthmprecise999094},
\begin{align} 
\begin{aligned} 
\mathbf{E}(x_{-R,R}) = \mathbf{E}_\infty + \OO(R^{-1}), \,\, \mathbf{P}(x_{-R,R}) = \OO(R^{-3/2}),
\end{aligned}\label{EQ99901}
\end{align}
we have the improved estimate \eqref{EQcoeffEstimLAMBDAstronger} for $\la'$,
\begin{align} 
\begin{aligned}
 \mathbf{E}(\la') =&\mathbf{E}_\infty+ \OO(R^{-1}), &  \mathbf{P}(\la') =& \OO(R^{-3/2}), \\
 \mathbf{L}(\la') =& \mathbf{L}_\infty + \smallO(1), &  \mathbf{C}(\la') =& \mathbf{G}_\infty + \smallO(1).
\end{aligned} \label{EQ99902}
\end{align}
Indeed, applying \eqref{EQ99901} in the above derivation of \eqref{EQestimatefmap1}, we get that the error map $f_R\lrpar{\la}$ satisfies the improved bound
\begin{align*} 
\begin{aligned} 
f_R\lrpar{\la} :=&\lrpar{\mathbf{E}_\infty,\mathbf{P}_\infty, \mathbf{L}_\infty, \mathbf{G}_\infty} - \lrpar{\mathbf{E}(\la),\mathbf{P}(\la), \mathbf{L}(\la), \mathbf{C}(\la)-3R\cdot\mathbf{P}(\la)} \\
&+\lrpar{\OO(R^{-1}), \OO(R^{-3/2}), \smallO(1), \smallO(1)}.
\end{aligned} %\label{}
\end{align*}
This shows that $\la'$ which satisfies by construction $f_R\lrpar{\la'}=0$, see \eqref{EQflaiszero}, satisfies the improved bound \eqref{EQ99902}.
%%%%%%%%%%%%%%%%%%%%%%%%%%%%%%%%%%%%%%%%
%%%%%%%%%%%%%%%%%%%%%%%%%%%%%%%%%%%%%%%%
\subsection{Conclusion of proof} \label{SECconclusionofMainTheoremREAL} 

\ni In this section we conclude the proof of Theorem \ref{THMMAINPRECISE}. By \eqref{EQmatchingKerrparameters} with the first of \eqref{EQsmallnessEstimateKERR8999} (see also Proposition \ref{PROPspheredataESTIM91}), we have the estimate
\begin{align*} 
\begin{aligned} 
\Vert {}^{(R)} x_{-1,2}^{\KK(\la')} - \mathfrak{m}^{\mathbf{E}_\infty/R} \Vert_{\XX(S_{-1,2})} \les&
R^{-1} \cdot \vert \mathbf{E}(\la')-\mathbf{E}_\infty \vert + R^{-1} \cdot \vert \mathbf{P}(\la') \vert \\
&+ R^{-2} \cdot \vert \mathbf{L}(\la') \vert + R^{-2} \cdot \vert \mathbf{C}(\la') \vert \\
&+ \lrpar{\frac{ R^{-2} \cdot \vert \mathbf{L}(\la') \vert + \frac{\vert \mathbf{P}(\la') \vert}{\mathbf{E}_\infty} \cdot R^{-2} \cdot \vert \mathbf{C}(\la') \vert}{\mathbf{E}_\infty/R}}^2 \\
=&\smallO\lrpar{R^{-3/2}},
\end{aligned}
\end{align*}
which, together with \eqref{EQfullenergyEstimates32777789}, implies that the constructed solution $x$ on $\HH_{-1,[1,2]}$ is bounded by
\begin{align} 
\begin{aligned} 
\Vert x -\mathfrak{m}^{\mathbf{E}_\infty/R} \Vert_{\mathcal{X}(\HH_{-1,[1,2]})} + \left\Vert x \vert_{S_{-1,1}}- {}^{(R)}\tilde{x}_{-1,1} \right\Vert_{\mathcal{X}(S_{-1,1})} =&  \smallO(R^{-3/2}).
\end{aligned} \label{EQmatchingKerrfinalEstimates777}
\end{align}

\ni Applying the scaling of Section \ref{SECdefinitionScaling} with scale factor $R^{-1}$, we get by \eqref{EQdefSCALINGSSprop8999}, \eqref{EQmatchingKerrfinalEstimates777} and Lemma \ref{LEMspheredataInvarianceSCale} that
\begin{align*} 
\begin{aligned} 
\Vert {}^{(R^{-1})}x - \mathfrak{m}^{\mathbf{E}_\infty} \Vert_{\XX\lrpar{\HH_{-R,[R,2R]}}} +\Vert {}^{(R^{-1})} x  - \tilde{x}_{-R,R} \Vert_{\XX(S_{-R,R})} =& \smallO(R^{-3/2}).
\end{aligned} %\label{}
\end{align*}
It remains to show that the sphere 
\begin{align*} 
\begin{aligned} 
S_{-R,2R}^{\Kerr} := S_{-R,2R}^{\KK(\la')}
\end{aligned} %\label{}
\end{align*}
in Kerr admits a future-complete outgoing null congruence and a past-complete ingoing null congruence. As the proof is based on different methods (namely, a classical perturbation argument based on Jacobi fields), it is postponed to Appendix \ref{SECcompletenessKERR}. This finishes the proof of Theorem \ref{THMMAINPRECISE}.

%%%%%%%%%%%%%%%%%%%%%%%%%%%%%%%%%%%%%%%%
%%%%%%%%%%%%%%%%%%%%%%%%%%%%%%%%%%%%%%%%
\section{Proof of spacelike gluing to Kerr} \label{SECspacelikecorollary} \label{SECconclusionSpacelikeSequence}

\ni In this section we prove Corollary \ref{THMspacelikeGLUINGtoKERRv2}, the gluing of spacelike initial data to Kerr. Let $(\Si,g,k)$ be given smooth strongly asymptotically flat spacelike initial data with asymptotic invariants
\begin{align*} 
\begin{aligned} 
(\mathbf{E}_{\mathrm{ADM}}, \mathbf{P}_{\mathrm{ADM}},\mathbf{L}_{\mathrm{ADM}},\mathbf{C}_{\mathrm{ADM}}) \in I(0)\times \RRR^3 \times \RRR^3,
\end{aligned} %\label{}
\end{align*}
where by the strong asymptotic flatness, $\mathbf{P}_{\mathrm{ADM}}=0$. We proceed in four steps.
\begin{enumerate}

\item We apply the material of Sections \ref{SECspacelikeRescalingtoSmalldata} and \ref{SECproofExistenceConstruction} where it is shown how to construct and estimate families of higher-order sphere data from given spacelike initial data. We work in the rescaled picture, that is, we construct smooth higher-order sphere data on a sphere $S_{0,1} \subset \Si$ in the \emph{rescaled} spacelike initial data. 

\item We use the bifurcate characteristic gluing of Theorem \ref{THMdoubleCHARgluingTOkerr} to glue the constructed higher-order sphere data on $S_{0,1}$ to a sphere $S_{-1,2}^{\KK({}^{(R)}\la)}$ in a Kerr spacetime.

\item We construct a local spacetime $(\MM,\g)$ by applying local existence results for the spacelike and characteristic initial value problem for the Einstein equations, and pick a spacelike hypersurface connecting $S_{0,1}$ and $S_{-1,2}^{\KK({}^{(R)}\la)}$. We conclude the proof by rescaling.

\end{enumerate}

\ni \textbf{Notation.} For ease of presentation, we work in the following with smooth spacelike initial data and smooth higher-order sphere data $(x, \DD^{L,m}, \DD^{\Lb,m})$ for a fixed integer $m\geq1$.\\

%%%%%%%%%%%%%%%%%%%%%%%%%%%%%%%%%%%%%%%%
\ni \textbf{(1) Rescaling and construction of sphere data.} In this section we follow the construction of Sections \ref{SECspacelikeRescalingtoSmalldata} and \ref{SECproofExistenceConstruction}: We start by rescaling the given spacelike initial data by scaling factor $R$ to $({}^{(R)}g, {}^{(R)}k)$ and constructing on the sphere
\begin{align*} 
\begin{aligned} 
S_{0,1} := S_{r_{\mathbf{E}_{\mathrm{ADM}}/R}(0,1)} \subset \Si
\end{aligned} %\label{}
\end{align*}
the higher-order sphere data (see \eqref{EQdefHIGHERORDERspheredata1999091} and also Remark \ref{REMARKConstructionHigherOrder})
\begin{align} 
\begin{aligned} 
\lrpar{{}^{(R)}x_{0,1}, {}^{(R)}\DD^{L,m}_{0,1}, {}^{(R)}\DD^{\Lb,m}_{0,1}}.
\end{aligned} \label{EQcsHigherOrderSphereData}
\end{align}
In Sections \ref{SECspacelikeRescalingtoSmalldata} and \ref{SECproofExistenceConstruction} it is shown that by the strong asymptotic flatness and the scaling of spacelike initial data (see Section \ref{SECspacelikescaling}), the constructed higher-order sphere data \eqref{EQcsHigherOrderSphereData} is -- with respect to an appropriate higher-regularity norm -- $\smallO(R^{-3/2})$-close to Schwarzschild reference higher-order sphere data of order $m$ of mass $\mathbf{E}_{\mathrm{ADM}}/R$; we denote this by
\begin{align} 
\begin{aligned} 
\lrpar{{}^{(R)}x_{0,1}, {}^{(R)}\DD^{L,m}_{0,1}, {}^{(R)}\DD^{\Lb,m}_{0,1}}-\lrpar{\mathfrak{m}^{\mathbf{E}_{\mathrm{ADM}}/R}_{0,1}, \DD^{L,m,\mathbf{E}_{\mathrm{ADM}}/R}_{0,1},\DD^{\Lb,m,\mathbf{E}_{\mathrm{ADM}}/R}_{0,1}}= \smallO(R^{-3/2}).
\end{aligned} \label{EQsmallnessSMOOTHsphereDATA1}
\end{align}

\ni Moreover, in Theorem \ref{PROPconstructionStatement} it is proved that the charges $(\mathbf{E},\mathbf{P},\mathbf{L},\mathbf{G})({}^{(R)}x_{0,1})$ can be estimated by
\begin{align} 
\begin{aligned} 
&\lrpar{R \cdot \mathbf{E}\lrpar{{}^{(R)}x_{0,1}},R \cdot\mathbf{P}\lrpar{{}^{(R)}x_{0,1}},R^2 \cdot\mathbf{L}\lrpar{{}^{(R)}x_{0,1}},R^2 \cdot\mathbf{G}\lrpar{{}^{(R)}x_{0,1}}} \\
=& \lrpar{\mathbf{E}_{\mathrm{ADM}}, \mathbf{P}_{\mathrm{ADM}}, \mathbf{L}_{\mathrm{ADM}}, \mathbf{C}_{\mathrm{ADM}}} + \lrpar{\OO(R^{-1}), \OO(R^{-3/2}), \smallO(1),\smallO(1)}.
\end{aligned} \label{EQconvergenceADMchargesspacelikeAFgluing99012}
\end{align}

%%%%%%%%%%%%%%%%%%%%%%%%%%%%%%%%%%%%%%%%
\ni \textbf{(2) Application of bifurcate characteristic gluing to Kerr.} By \eqref{EQsmallnessSMOOTHsphereDATA1} and \eqref{EQconvergenceADMchargesspacelikeAFgluing99012}, we can apply Theorem \ref{THMdoubleCHARgluingTOkerr} (to be precise, the rescaled version thereof) to the higher-order sphere data \eqref{EQcsHigherOrderSphereData} to get  
\begin{itemize}
\item smooth higher-order outgoing null data $({}^{(R)}x,{}^{(R)}\DD^{L,m}, {}^{(R)}\DD^{\Lb,m})$ on $\HH_{-1,[1,2]}$ and smooth higher-order ingoing null data $({}^{(R)}{x},{}^{(R)}{\DD}^{L,m}, {}^{(R)}{\DD}^{\Lb,m})$ on $\HHb_{[-1,0],1}$ satisfying the higher-order null structure equations and matching on $S_{-1,1}$,
\item a Kerr reference sphere $S_{-1,2}^{{}^{(R)}\Kerr}$ in a Kerr spacetime $(\MM^{{}^{(R)}\Kerr},\g^{{}^{(R)}\Kerr})$ with Kerr reference higher-order sphere data 
$$(x^{{}^{(R)}\Kerr}_{-1,2}, \DD^{L,m,{}^{(R)}\Kerr}_{-1,2}, \DD^{\Lb,m,{}^{(R)}\Kerr}_{-1,2}).$$

\end{itemize}
such that we have matching up to order $m$ on $S_{-1,1}$, $S_{0,1}$ and $S_{-1,2}$,
\begin{align*} 
\begin{aligned} 
({}^{(R)}x,{}^{(R)}\DD^{L,m}, {}^{(R)}\DD^{\Lb,m}) \vert_{S_{-1,1}}=& ({}^{(R)}{x},{}^{(R)}{\DD}^{L,m}, {}^{(R)}{\DD}^{\Lb,m})\vert_{S_{-1,1}},\\
({}^{(R)}{x},{}^{(R)}{\DD}^{L,m}, {}^{(R)}{\DD}^{\Lb,m}) \vert_{S_{0,1}} =& ({}^{(R)}x_{0,1}, {}^{(R)}\DD^{L,m}_{0,1}, {}^{(R)}\DD^{\Lb,m}_{0,1}), \\
({}^{(R)}{x},{}^{(R)}{\DD}^{L,m}, {}^{(R)}{\DD}^{\Lb,m}) \vert_{S_{-1,2}} =& (x^{{}^{(R)}\Kerr}_{-1,2}, \DD^{L,m,{}^{(R)}\Kerr}_{-1,2}, \DD^{\Lb,m,{}^{(R)}\Kerr}_{-1,2}).
\end{aligned} %\label{}
\end{align*}

\ni In particular, it holds that
\begin{enumerate}
\item the null data $({}^{(R)}x,{}^{(R)}\DD^{L,m}, {}^{(R)}\DD^{\Lb,m})$ on $\HH_{-1,[1,2]}$ and $({}^{(R)}{x},{}^{(R)}{\DD}^{L,m}, {}^{(R)}{\DD}^{\Lb,m})$ on $\HHb_{[-1,0],1}$ are $\smallO(R^{-3/2})$-close to Schwarzschild reference higher-order null data of mass $M/R$, respectively, and
\item the sphere $S_{-1,2}^{{}^{(R)}\Kerr}$ lies in a Kerr reference spacelike hypersurface $\Si^{{}^{(R)}\Kerr} \subset \MM^{{}^{(R)}\Kerr}$ with asymptotic invariants (see Sections \ref{SECspacelikescaling} and \ref{SECappKerrFamilyDetails})
\begin{align} 
\begin{aligned} 
 \mathbf{E}^{{}^{(R)}\mathrm{Kerr}}_{\mathrm{ADM}} =& R^{-1} \cdot \mathbf{E}_{\mathrm{ADM}}+ \OO(R^{-2}), & \mathbf{P}^{{}^{(R)}\mathrm{Kerr}}_{\mathrm{ADM}} =& \OO(R^{-5/2}), \\
\mathbf{L}^{{}^{(R)}\mathrm{Kerr}}_{\mathrm{ADM}} =& R^{-2} \cdot \mathbf{L}_{\mathrm{ADM}}+ \smallO(R^{-2}), &  \mathbf{C}^{{}^{(R)}\mathrm{Kerr}}_{\mathrm{ADM}} =& R^{-2} \cdot \mathbf{C}_{\mathrm{ADM}} + \smallO(R^{-2}).
\end{aligned} \label{EQKerrgluing99990}
\end{align}
\end{enumerate}

%%%%%%%%%%%%%%%%%%%%%%%%%%%
\ni \textbf{(3) Construction of spacelike hypersurface and conclusion.} The constructed solutions to the higher-order null structure equations,
\begin{align*} 
\begin{aligned} 
({}^{(R)}x,{}^{(R)}\DD^{L,m}, {}^{(R)}\DD^{\Lb,m}) \text{ on } \HH_{-1,[1,2]} \text{ and } ({}^{(R)}{x},{}^{(R)}{\DD}^{L,m}, {}^{(R)}{\DD}^{\Lb,m}) \text{ on } \HHb_{[-1,0],0},
\end{aligned} %\label{}
\end{align*}
form \emph{characteristic initial data} for the Einstein vacuum equations which is $\smallO(R^{-3/2})$-close to Schwarzschild of mass $\mathbf{E}_\infty/R$. By the work of Luk on the characteristic initial value problem for the Einstein equations \cite{LukChar,LukRod1}, for $R\geq1$ sufficiently large, the associated maximal globally hyperbolic spacetime $({\MM'},{\g'})$ contains slabs of universal width along the null hypersurface $\HHb_{[-1,0],1}$ and $\HH_{-1,[1,2]}$.

\begin{center}
\includegraphics[width=12cm]{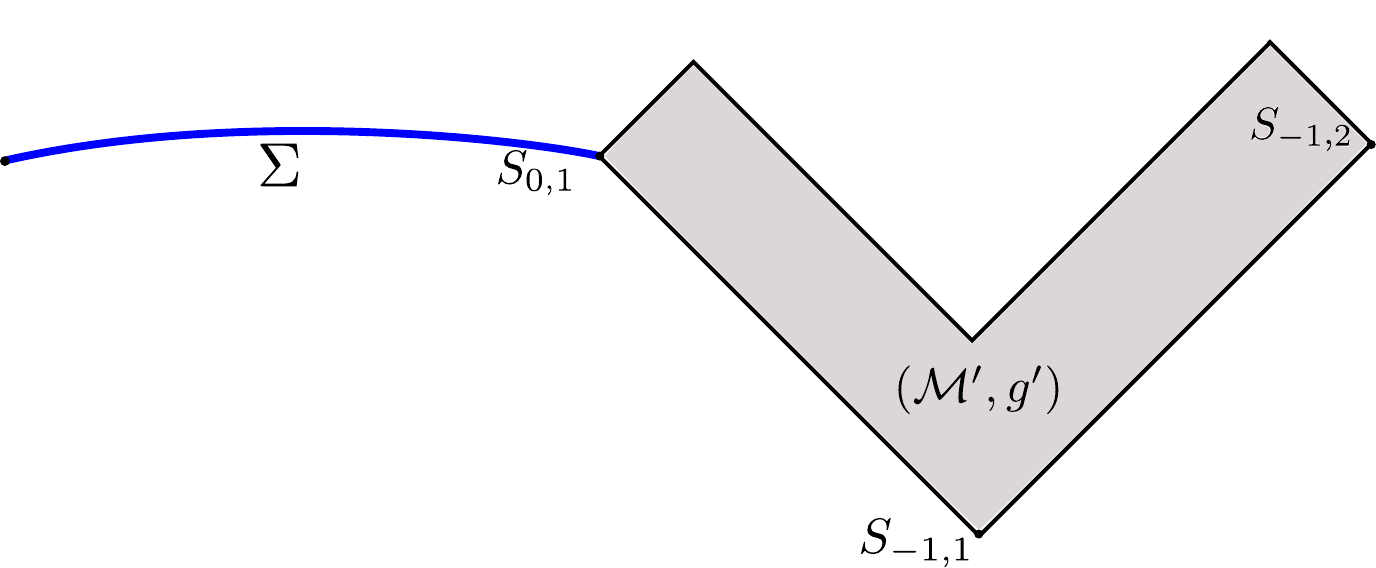} 
\captionof{figure}{The spacetime $(\MM',\g')$ is denoted as shaded region. \label{FIG1}}
\end{center}

\ni Applying local existence for the spacelike Cauchy problem defined on $\Si$ (the resulting region is shaded red in Figure \ref{FIG3}), and, subsequently, for the characteristic Cauchy problems (the resulting regions are shaded green in Figure \ref{FIG3}) defined on 
\begin{itemize} 
\item $\partial^+ \DD(\Si)$ and $\partial^+ \MM'$, 
\item $\partial^- \DD(\Si)$ and $\HHb$,
\item $\partial^+ \MM'$ and $\partial^+ \MM^{{}^{(R)}\Kerr}$,
\item $\HH$ and $\partial^- \MM^{{}^{(R)}\Kerr}$,
\end{itemize}
we construct the spacetime $(\MM'',\g'')$, see Figure \ref{FIG3}. Here $\partial^+$ and $\partial^-$ denote the future and past boundaries, and $\DD(\Si)$ the domain of dependence  of $\Si$.

\ni In $(\MM'',\g'')$ we define a spacelike hypersurface $\Si''$ as follows (see Figure \ref{FIG3})
\begin{itemize}
\item $\Si''$ agrees with $\Si$ in $(\MM,\g)$,
\item $\Si''$ is spacelike and contained in the slabs in $(\MM',\g')$,
\item $\Si''$ agrees with $\Si^{{}^{(R)}\Kerr}$ in $(\MM^{{}^{(R)}\Kerr},\g^{{}^{(R)}\Kerr})$.
\end{itemize}

\begin{center}
\includegraphics[width=14cm]{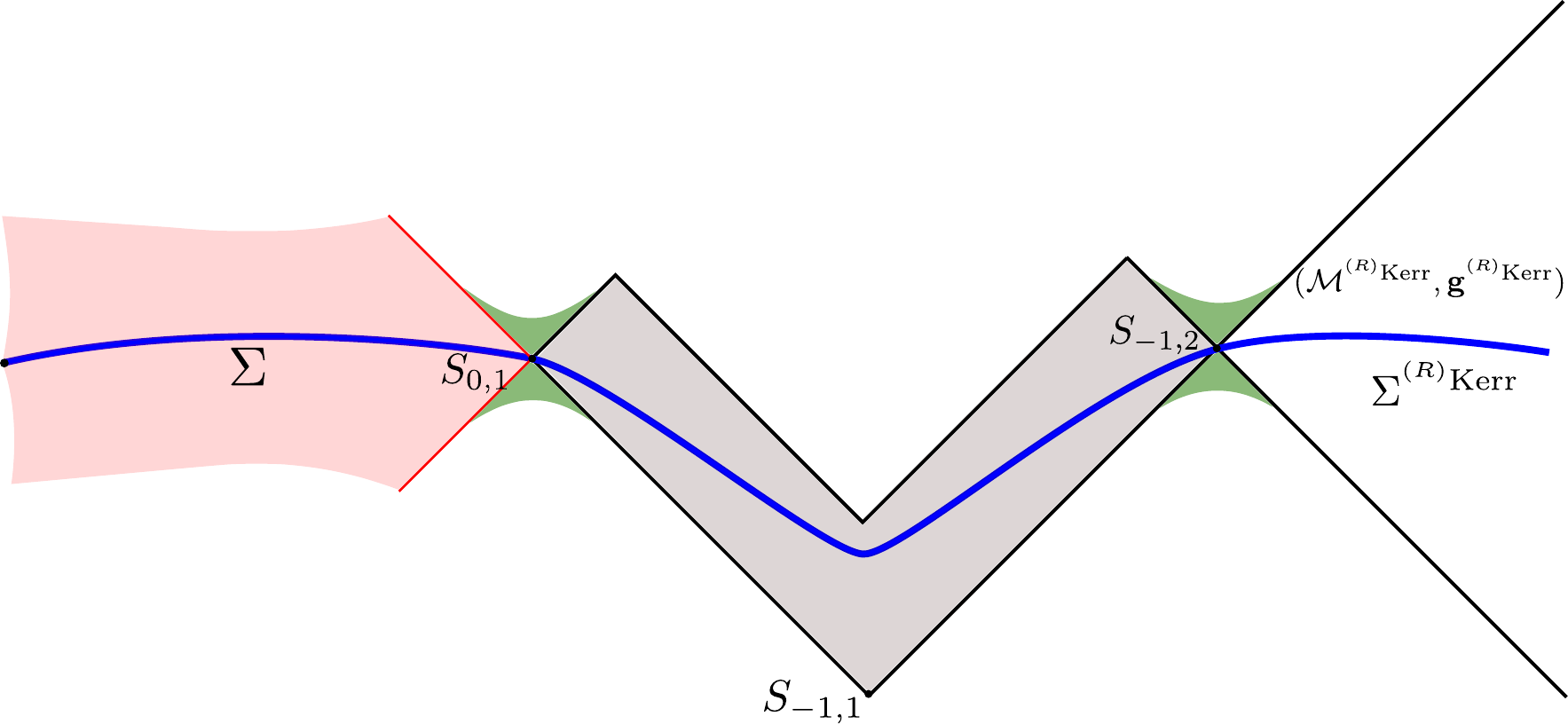} 
\captionof{figure}{The spacelike hypersurface $\Si''$ is indicated by the bold blue line. \label{FIG3}}
\end{center}

\ni By construction, the induced spacelike initial data on $\Si''$ agrees with
\begin{itemize}
\item $({}^{(R)}{g},{}^{(R)}{k})$ on $\Si$, 
\item Kerr reference spacelike initial data $(g^{{}^{(R)}\Kerr},k^{{}^{(R)}\Kerr})$ on $\Si^{{}^{(R)}\Kerr}$, see Section \ref{SECappKerrFamilyDetails} for details.
\end{itemize} 
In particular, it is a solution to the spacelike gluing problem from the rescaled spacelike initial data $({}^{(R)}{g},{}^{(R)}{k})$ to the Kerr reference spacelike initial data $(g^{{}^{(R)}\Kerr},k^{{}^{(R)}\Kerr})$.

Scaling the above spacelike initial data by factor $R^{-1}$, and using the scale-invariance of the Kerr reference spacelike initial data (see Remark \ref{RemarkScalingKerrData}), we conclude the spacelike gluing to Kerr. 

%%%%%%%%%%%%%%%%%%%%%%%%%%%%%%%%%%%%%%%%
%%%%%%%%%%%%%%%%%%%%%%%%%%%%%%%%%%%%%%%%
\section{Spacelike initial data for the Einstein equations} \label{SECCSgluingbasics} 

\ni In this section we set up spacelike initial data for the Einstein equations and introduce the notation used in Sections \ref{SECstatementConstruction} and \ref{SECconclusionSpacelikeSequence}, and Section \ref{SECappKerrFamilyDetails}. This section is organized as follows.
\begin{itemize}
\item In Section \ref{SECspacelikeIVPdefREAL} we introduce spacelike initial data and define the reference spacelike initial data of Minkowski and Schwarzschild.
\item In Section \ref{SECAFdefspacelike} we recapitulate two types of \emph{asymptotic flatness} of spacelike initial data.
\item In Section \ref{SECdefAsymptoticInvariants99901} we recall discuss the so-called \emph{asympotic invariants} of asymptotically flat spacelike initial data.
\item In Section \ref{SECfoliationGeometry} we recall the geometry of foliations of spacelike initial data by $2$-spheres.
\item In Section \ref{SECspacelikescaling} we define the scaling of spacelike initial data and introduce local norms.
\end{itemize}

%%%%%%%%%%%%%%%%%%%%%%%%%%%%%%%%%%%%%%%%
\subsection{The spacelike constraint equations and spacelike initial data} \label{SECspacelikeIVPdefREAL}
\ni Let $(\MM,\g)$ be a spacetime, and denote its Riemann curvature tensor by $\Rbf$. Let $\Si$ be a spacelike hypersurface in $\MM$ with future-directed timelike unit normal $T$. The \emph{electric-magnetic decomposition} of $\Rbf$ on $\Si$ is defined by
\begin{align} 
\begin{aligned} 
E_{ab} := \Rbf_{TaTb}, \,\, H_{ab} := {}^\ast\Rbf_{TaTb},
\end{aligned} \label{EQdefEHdecompRIC}
\end{align}
where ${}^\ast\Rbf_{\a \be \ga \de}:= \half \in_{\a \be \mu \nu}\Rbf^{\mu\nu}_{\,\,\,\,\,\,\ga \de}$ denotes the Hodge dual of $\Rbf$ with respect to the volume form $\in$ on $(\MM,\g)$. The $2$-tensors $E$ and $H$ are symmetric and tracefree, and the following decomposition holds (see, for example, \cite{ChrKl93})
\begin{align} 
\begin{aligned} 
\Rbf_{abcT}= -\in_{ab}^{\,\,\,\,\,\,s} H_{sc}, \,\, \Rbf_{abcd}=-\in_{abs}\in_{cdl} E^{sl},
\end{aligned} \label{EQlinearCombination112}
\end{align}
where $\in_{abc}:= \in_{Tabc}$ denotes the induced volume element on $\Si$.

In the following, let $g$ be the induced metric and $k$ be the second fundamental form of $\Si \subset \MM$. Denote the covariant derivative of $g$ by $\nab$ and the Ricci tensor of $g$ by $\RRRic$. The Gauss-Codazzi equations on $\Si$ are
\begin{align} 
\begin{aligned} 
\RRRic_{ij}-k_{ia}k^a_{\,\,\, j} + k_{ij} \tr k =& E_{ij},\\
\nab_i k_{jm}-\nab_j k_{im} =& \in_{ij}^{\,\,\,\,\,\,l} H_{lm},
\end{aligned} \label{EQGaussCodazziSpacelike}
\end{align} 
where $\tr k := g^{ab}k_{ab}$. We further have the Bianchi equation (see, for example, equation (7.2.2d) in \cite{ChrKl93}),
\begin{align} 
\begin{aligned} 
\mathrm{div}\, H = E \wedge k,
\end{aligned} \label{EQBianchiEQforH}
\end{align}
where for symmetric $2$-tensors $V$ and $W$ on $\Si$,
\begin{align} 
\begin{aligned} 
(\mathrm{div}\, V)_i:= \nab^j V_{ji},  \,\, (V \wedge W)_{i} := \in_{i}^{\,\,\,jl} V_{jn}W^n_{\,\,\, l}.
\end{aligned} \label{EQnotationBianchiHdefinition}
\end{align}
Taking the trace of \eqref{EQGaussCodazziSpacelike} with respect to $g$ leads to the \emph{spacelike constraint equations},
\begin{align} 
\begin{aligned}
R_{\mathrm{scal}}=&\vert k \vert^2 - (\tr k)^2, \\
\mathrm{div} k =& d (\tr k),
\end{aligned} \label{EQspacelikeConstraints1}
\end{align}
where $d$ denotes the exterior derivative on $\Si$ and $R_{\mathrm{scal}}:= g^{ab}\RRRic_{ab}$. 

Based on \eqref{EQspacelikeConstraints1}, we have the following well-known definition.
\begin{definition}[Spacelike initial data] \label{DEFspacelikeDATA} Spacelike initial data for the Einstein equations is specified by a triple $(\Si,g,k)$ where $(\Si,g)$ is a Riemannian $3$-manifold and $k$ is a symmetric $2$-tensor on $\Si$, satisfying the spacelike constraint equations \eqref{EQspacelikeConstraints1}, that is,
\begin{align*} 
\begin{aligned} 
R_{\mathrm{scal}}=&\vert k \vert^2 - (\tr k)^2, \\
\mathrm{div} k =& d (\tr k).
\end{aligned} %\label{}
\end{align*}
\end{definition}
\ni Local well-posedness of the Cauchy problem of general relativity with sufficiently regular spacelike initial data is well-known \cite{BruhatOrigin,SpacelikeLocalEx1,SpacelikeLocalEx2}.

In the following we discuss reference spacelike initial data for Minkowski and Schwarzschild. Reference spacelike initial data for Kerr is defined and analysed in Section \ref{SECappKerrFamilyDetails}. \\

\ni \textbf{Minkowski reference spacelike initial data.} Minkowski spacetime is given by $(\RRR^{3+1}, \mathfrak{m})$, see Section \ref{SECMinkowskiSSspacetimes1}. We define the reference spacelike initial data for Minkowski to be the induced spacelike initial data on the hypersurface $\{t=0\}$,
\begin{align*} 
\begin{aligned} 
(\Si,g,k)= (\RRR^3,e,0),
\end{aligned} %\label{}
\end{align*}
where $e$ denotes the standard Euclidean metric on $\RRR^3$. \\

\ni \textbf{Schwarzschild reference spacelike initial data.} The Schwarzschild metric of mass $M\geq0$ is given in Schwarzschild coordinates $(t,r,\th,\phi)$ by
\begin{align*} 
\begin{aligned} 
\g= -\lrpar{1-\frac{2M}{r}} dt^2 + \lrpar{1-\frac{2M}{r}}^{-1} dr^2 + r^2 \lrpar{d\th^2 + \sin^2\th d\phi^2}.
\end{aligned} %\label{}
\end{align*}
The induced spacelike initial data on the spacelike hypersurface $\{t=0 \} \cap \{r>2M\}$ is given by
\begin{align} 
\begin{aligned} 
(\Si,g,k) = \lrpar{ \RRR^3 \setminus \overline{B(0,2M)}, \lrpar{1-\frac{2M}{r}}^{-1} dr^2 + r^2 \lrpar{d\th^2 + \sin^2\th d\phi^2}, 0}.
\end{aligned} \label{EQssInitialSpacelike1}
\end{align}
It is well-known that the induced metric $g$ is conformally flat. Indeed, defining \emph{isotropic coordinates} $(\tilde{r},\tilde{\th},\tilde{\phi})$ the from Schwarzschild coordinates $(r,\th,\phi)$ by the relations
\begin{align} 
\begin{aligned} 
\frac{r}{\tilde{r}} = \lrpar{1+\frac{M}{2\tilde{r}}}^2, \,\, \tilde{\th}=\th, \,\, \tilde{\phi}= \phi,
\end{aligned} \label{EQcoordinatechangedef1}
\end{align}
and using the explicit relations
\begin{align} 
\begin{aligned} 
\frac{d r}{d\tilde{r}} = \frac{r}{\tilde r} \lrpar{1-\frac{2M}{r}}^{1/2}, \,\, \tilde{r}(2M)= \frac{M}{2},
\end{aligned} \label{EQcoordinatechangedef2}
\end{align}
we get that for $r>2M$ (that is, by \eqref{EQcoordinatechangedef2}, for $\tilde{r}>M/2$),
\begin{align} 
\begin{aligned} 
g=& \lrpar{1-\frac{2M}{r}}^{-1} dr^2 + r^2 \lrpar{d\th^2 + \sin^2\th d\phi^2} \\
=& \lrpar{1-\frac{2M}{r}}^{-1} \lrpar{\frac{dr}{d\tilde{r}}}^2 d\tilde{r}^2 + \lrpar{\frac{r}{\tilde{r}} }^2\tilde{r}^2 \lrpar{d\tilde{\th}^2 + \sin^2\tilde{\th} d\tilde{\phi}^2} \\
=& \lrpar{1+\frac{M}{2\tilde{r}}}^4 \lrpar{d\tilde{r}^2 + \tilde{r}^2\lrpar{d\tilde{\th}^2 + \sin^2\tilde{\th} d\tilde{\phi}^2}} \\
=& \lrpar{1+\frac{M}{2\tilde{r}}}^4 \tilde{e},
\end{aligned} \label{EQcoordiSSchange}
\end{align}
where $\tilde{e}_{ij}$ denotes the Euclidean metric in Cartesian isotropic coordinates $(\tilde{x}^1,\tilde{x}^2,\tilde{x}^3)$ defined by \eqref{EQdefinSPHERICALAFcoord} from $(\tilde{r},\tilde{\th},\tilde{\phi})$. In particular, in isotropic coordinates the metric admits the following asymptotic expansion,
\begin{align} 
\begin{aligned} 
g_{ij}(x)=\lrpar{1+\frac{M}{2{\vert {x} \vert}}}^4 {e}_{ij} = \lrpar{1+\frac{2M}{\vert {x} \vert}} {e}_{ij} + \OO\lrpar{\vert {x}\vert^{-2}}.
\end{aligned} \label{EQSSasympExpansion}
\end{align}
The expansion \eqref{EQSSasympExpansion} is the basis for the definition of \emph{strong asymptotic flatness} in Definition \ref{DEFasymptoticFLATNESS}.\\

\ni \textbf{Notation.} For real numbers $M\geq0$, we denote the metric components of the Schwarzschild reference metric $g$ in Schwarzschild Cartesian coordinates $(x^1,x^2,x^3)$ by $g^M_{ij}$, and in isotropic Cartesian coordinates by $\tilde{g}^M_{ij}$.
%%%%%%%%%%%%%%%%%%%%%%%%%%%%%%%%%%%%%%%%
\subsection{Asymptotic flatness of spacelike initial data} \label{SECAFdefspacelike} 
In this paper we consider two types of asymptotic flatness of spacelike initial data. First, the following definition of \emph{strong asymptotic flatness} corresponds to the center-of-mass frame of the isolated system under consideration, see \cite{ChrLectures,ChrKl93,KlNic}.
\begin{definition}[Strong asymptotic flatness] \label{DEFasymptoticFLATNESS} Spacelike initial data $(\Si,g,k)$ is \emph{strongly asymptotically flat} if there exist a real number $M
\geq0$, a compact set $K \subset \Si$ such that its complement $\Si \setminus K$ is diffeomorphic to the complement of the closed unit ball in $\RRR^3$, and a coordinate system $(x^1,x^2,x^3)$ defined near spacelike infinity such that, as $\vert x \vert\to \infty$, 
\begin{align} 
\begin{aligned} 
g_{ij}({x}) = \lrpar{1+\frac{2M}{\vert x \vert}} e_{ij} + \smallO\lrpar{\vert x \vert^{-3/2}}, \,\, k_{ij}({x}) = \smallO\lrpar{\vert x \vert^{-5/2}}.
\end{aligned} \label{EQAFconditions}
\end{align}
We moreover require analogous conditions on successive derivatives as needed.
\end{definition}

\ni Second, we define \emph{asymptotically flat spacelike initial data with Regge-Teitelbaum conditions} (also denoted as \emph{RT-conditions}), see for example \cite{CorvinoSchoen,ChruscielDelay}. The family of Kerr reference spacelike initial data studied in Section \ref{SECappKerrFamilyDetails} belongs to this type of asymptotic flatness.

\begin{definition}[Asymptotic flatness with RT-conditions] \label{DEFasymptoticFLATNESSWEAK} Spacelike initial data $(\Si,g,k)$ is \emph{asymptotically flat with RT-conditions} if there exists a compact set $K \subset \Si$ such that its complement $\Si \setminus K$ is diffeomorphic to the complement of the closed unit ball in $\RRR^3$ and there exists a Cartesian coordinate system $(x^1,x^2,x^3)$ defined near spacelike infinity such that, as $\vert x \vert \to \infty$, 
\begin{align*} 
\begin{aligned} 
g_{ij}(x) = e_{ij} + \OO\lrpar{r^{-1}}, \,\, k_{ij}(x) = \OO\lrpar{r^{-2}},
\end{aligned}
\end{align*}
and the so-called \emph{Regge-Teitelbaum conditions} hold,
\begin{align} 
\begin{aligned} 
g_{ij}(x)-g_{ij}(-x) =& \OO(\vert x \vert^{-2}), & k_{ij}(x)+k_{ij}(-x) =& \OO(\vert x \vert^{-3}). 
\end{aligned} \label{EQdefRTconditionsDEFDEF8889}
\end{align}
We moreover require analogous conditions on successive derivatives as needed.
\end{definition} 
\textbf{Notation.} \begin{enumerate} 
\item Defining the Regge-Teitelbaum quantities
\begin{align} 
\begin{aligned} 
\mathrm{RT}(g)_{ij} :=& g_{ij}(x) - g_{ij}(-x), \,\, \mathrm{RT}(k)_{ij} :=& k_{ij}(x) - k_{ij}(-x),
\end{aligned} \label{EQRTcondDefinition3}
\end{align}
\noindent we can write the RT-conditions \eqref{EQdefRTconditionsDEFDEF8889} as
\begin{align*} 
\begin{aligned} 
\mathrm{RT}(g)_{ij}(x) = \OO(\vert x \vert^{-2}), \,\, \mathrm{RT}(k)_{ij}(x) = \OO(\vert x \vert^{-3}).
\end{aligned} %\label{}
\end{align*}

\item In the following we call spacelike initial data \emph{asymptotically flat} if it is either strongly asymptotically flat or asymptotically flat with Regge-Teitelbaum conditions.
\end{enumerate}

%%%%%%%%%%%%%%%%%%%%%%%%%%%%%%%%%
\subsection{Asymptotic invariants of asymptotically flat spacelike initial data} \label{SECdefAsymptoticInvariants99901} 
\ni Given asymptotically flat spacelike initial data $(\Si,g,k)$ with Cartesian coordinates $(x^1,x^2,x^3)$ near spacelike infinity, define standard spherical coordinates $(r,\th^1,\th^2)$ by \eqref{EQdefinSPHERICALAFcoord}, and let the $2$-spheres $S_r \subset \Si$ be defined as the level sets of $r$. The following \emph{asymptotic invariants} are fundamental quantities in mathematical relativity, see for example \cite{ADM,CorvinoSchoen,ChrKl93,ChrLectures}.

\begin{definition}[Asymptotic invariants] \label{DEFADMquantities} Let $(\Si,g,k)$ be asymptotically flat spacelike initial data with coordinates $(x^1,x^2,x^3)$ near spacelike infinity. For $i=1,2,3$, define
\begin{align*}
\begin{aligned} 
\mathbf{E}_{\mathrm{ADM}} :=&  \lim\limits_{r \to\infty} \frac{1}{16\pi}\int\limits_{S_r} \sum\limits_{j=1,2,3} \lrpar{\pr_j g_{jl}-\pr_l g_{jj}} N^l d\mu_\gd, \\
\lrpar{\mathbf{P}_{\mathrm{ADM}}}^i :=& \lim\limits_{r \to\infty} \frac{1}{8\pi} \int\limits_{S_{r}}  \lrpar{k_{il}-\tr k \, g_{il}} N^l d\mu_\gd, \\
\lrpar{\mathbf{L}_{\mathrm{ADM}}}^i :=&  \lim\limits_{r\to\infty} \frac{1}{8\pi}\int\limits_{S_{r}} (k_{jl}-\tr k \, g_{jl}) \lrpar{Y_{(i)}}^j N^l d\mu_\gd, \\
\lrpar{\mathbf{C}_{\mathrm{ADM}}}^i :=& \lim\limits_{r\to\infty} \int\limits_{S_r} \lrpar{ x^i \sum\limits_{j=1,2,3} \lrpar{\pr_j g_{jl}- \pr_l g_{jj}}N^l - \sum\limits_{j=1,2,3} \lrpar{g_{ji} N^j - g_{jj} N^i} }d\mu_\gd,
\end{aligned}
\end{align*}
where $N$ denotes the outward-pointing unit normal to $S_r$ and $d\mu_\gd$ the induced volume element on $S_r$. Furthermore, $Y_{(i)}$, $i=1,2,3$, are the rotation fields defined by
\begin{align} 
\begin{aligned} 
(Y_{(i)})_j := \in_{ilj} x^l.
\end{aligned} \label{EQdefRotationFieldsComponents}
\end{align}
\end{definition}

\ni \emph{Remarks on Definition \ref{DEFADMquantities}.}
\begin{enumerate}

\item The asymptotic invariants can be interpreted as \emph{energy $\mathbf{E}_{\mathrm{ADM}}$, linear momentum $\mathbf{P}_{\mathrm{ADM}}$, angular momentum $\mathbf{L}_{\mathrm{ADM}}$} and \emph{center-of-mass $\mathbf{C}_{\mathrm{ADM}}$} of the spacelike initial data set.

\item The asymptotic invariants are well-defined and foliation-independent for asymptotically flat spacelike initial data, see for example \cite{ChrLectures,CorvinoSchoen} and references therein. 
\item By the positive energy theorem \cite{SchoenYau1,SchoenYau2,Witten} it holds for sufficiently regular asymptotically flat spacelike initial data that
\begin{align} 
\begin{aligned} 
\mathbf{E}_{\mathrm{ADM}}\geq0.
\end{aligned} \label{EQPMT}
\end{align} 
Moreover, if equality holds in \eqref{EQPMT}, then the initial data must be isometric to initial data for Minkowski spacetime.

\item For strongly asymptotically flat initial data, it holds that (see, for example, \cite{ChrKl93})
\begin{align*} 
\begin{aligned} 
\mathbf{E}_{\mathrm{ADM}} = M, \,\, \mathbf{P}_{\mathrm{ADM}} = 0,
\end{aligned} %\label{}
\end{align*}
where $M$ is the real number appearing in \eqref{EQAFconditions}.
\end{enumerate}

\ni It is well-known (see, for example, \cite{Miao,Herzlich,AshtekarHansen,Chrusciel,Huang,Huang2}) that the asymptotic invariants $\mathbf{E}_{\mathrm{ADM}}$ and $\mathbf{C}_{\mathrm{ADM}}$ can be calculated in terms of the Ricci tensor as follows.
\begin{theorem}[Alternative expressions for $\mathbf{E}_{\mathrm{ADM}}$ and $\mathbf{C}_{\mathrm{ADM}}$] \label{THMaltDEFinvariants} Let $(\Si,g,k)$ be asymptotically flat spacelike initial data such that $\mathbf{E}_{\mathrm{ADM}}>0$. Then it holds that for $i=1,2,3$,
\begin{align*} 
\begin{aligned} 
\mathbf{E}_{\mathrm{ADM}}=& \lim\limits_{r\to \infty} - \frac{1}{8\pi} \int\limits_{S_{r}} \lrpar{\RRRic - \half R_{\mathrm{scal}} \, g }(X,N) d\mu_\gd, \\
\lrpar{\mathbf{C}_{\mathrm{ADM}}}^i =& \lim\limits_{r\to\infty} \frac{1}{16\pi} \int\limits_{S_{r}} \lrpar{\RRRic - \half R_{\mathrm{scal}} \, g }(Z^{(i)},N) d\mu_\gd,
\end{aligned} %\label{}
\end{align*}
where the vectorfields $X$ and $Z^{(i)}$, $i=1,2,3$, are defined with respect to Cartesian coordinates $(x^1,x^2,x^3)$ by 
\begin{align} 
\begin{aligned} 
X := x^i \pr_i, \,\, Z^{(i)} := \lrpar{\vert x \vert^2 \de^{ij}-2x^i x^j} \pr_j.
\end{aligned} \label{EQdefXandZiALTADM}
\end{align}

\end{theorem}
\ni \emph{Remarks on Theorem \ref{THMaltDEFinvariants}.}
\begin{enumerate}
\item The vectorfields $X$ and $Z^{(i)}$, $i=1,2,3$, are conformal Killing vectorfields of Euclidean space. In Lemma \ref{LEMgeometricIDENTITIES}, see \eqref{EQREMdecompositionZIapp}, we show that $Z^{(i)}$, $i=1,2,3$, can be expressed in terms of spherical harmonics as follows, with $(m_1,m_2,m_3) := (1,-1,0)$,
\begin{align*} 
\begin{aligned} 
Z^{(i)} = - \vert x \vert^3 \sqrt{\frac{8\pi}{3}} E^{(1m_i)} - \vert x \vert^2 \lrpar{\sqrt{\frac{4\pi}{3}}Y^{(1m_i)}} \pr_r.
\end{aligned} %\label{}
\end{align*}

\end{enumerate}

\ni Based on Definition \ref{DEFADMquantities} and Theorem \ref{THMaltDEFinvariants} we introduce the following \emph{local integrals}. Their relations to the charges $(\mathbf{E}, \mathbf{P}, \mathbf{L}, \mathbf{G})$ of Definition \ref{DEFlocalCharges} is studied in Sections \ref{SECcomparisonENERGY}, \ref{SECcomparisonLINEAR}, \ref{SECcomparisonANGULAR} and \ref{SECcomparisonCENTER}.

\begin{definition}[Local integrals] \label{DEFlocalADMcharges} Let $(\Si,g,k)$ be asymptotically flat spacelike initial data such that $\mathbf{E}_{\mathrm{ADM}}>0$, and let $(x^1,x^2,x^3)$ be corresponding Cartesian coordinates near spacelike infinity. For real numbers $r\geq1$ sufficiently large and $i=1,2,3$, define
\begin{align} 
\begin{aligned}  
{\mathbf{E}}^{\mathrm{loc}}_{\mathrm{ADM}}(S_r,g,k) :=& -\frac{1}{8\pi} \int\limits_{S_{r}} \lrpar{\RRRic - \half R_{\mathrm{scal}} \, g}(X ,N) d\mu_\gd, \\
\lrpar{\mathbf{P}^{\mathrm{loc}}_{\mathrm{ADM}}}^i(S_r,g,k) :=& \frac{1}{8\pi} \int\limits_{S_{r}}  \lrpar{k_{il}-\tr k \, g_{il}} N^l d\mu_\gd, \\
\lrpar{\mathbf{L}^{\mathrm{loc}}_{\mathrm{ADM}}}^i(S_r,g,k) :=& \frac{1}{8\pi} \int\limits_{S_{r}} (k_{jl}-\tr k \, g_{jl}) \lrpar{Y_{(i)}}^j N^l d\mu_\gd, \\
\lrpar{{\mathbf{C}}^{\mathrm{loc}}_{\mathrm{ADM}}}^i(S_r,g,k):=& \frac{1}{16\pi} \int\limits_{S_{r}} \lrpar{\RRRic - \half R_{\mathrm{scal}} \, g }(Z^{(i)},N) d\mu_\gd.
\end{aligned} \label{EQlocalisedADMcharges}
\end{align}
\end{definition}

\begin{remark} The local integrals ${\mathbf{E}}^{\mathrm{loc}}_{\mathrm{ADM}}$ and ${\mathbf{C}}^{\mathrm{loc}}_{\mathrm{ADM}}$ are defined following Theorem \ref{THMaltDEFinvariants}. This is advantageous for us in two aspects: First, by their geometric nature, ${\mathbf{E}}^{\mathrm{loc}}_{\mathrm{ADM}}$ and ${\mathbf{C}}^{\mathrm{loc}}_{\mathrm{ADM}}$ are more natural to relate to the charges $(\mathbf{E}, \mathbf{P}, \mathbf{L}, \mathbf{G})$, see Sections \ref{SECcomparisonENERGY} and \ref{SECcomparisonCENTER}. Second, by the twice-contracted Bianchi identity
\begin{align*} 
\begin{aligned} 
\Div \lrpar{\RRRic - \half R_{\mathrm{scal}} \, g } =0,
\end{aligned} %\label{}
\end{align*}
the local integrals ${\mathbf{E}}_{\mathrm{ADM}}^{\mathrm{loc}}$ and ${\mathbf{C}}_{\mathrm{ADM}}^{\mathrm{loc}}$ are highly susceptible to the application of Stokes' theorem. This is used in Section \ref{SECappKerrFamilyDetails} to estimate the convergence of ${\mathbf{E}}_{\mathrm{ADM}}^{\mathrm{loc}}$ and ${\mathbf{C}}_{\mathrm{ADM}}^{\mathrm{loc}}$ towards ${\mathbf{E}}_{\mathrm{ADM}}$ and ${\mathbf{C}}_{\mathrm{ADM}}$ in Kerr reference spacelike initial data.
\end{remark}

\ni The following lemma analyses the convergence rates of $\mathbf{E}_{\mathrm{ADM}}^{\mathrm{loc}}$ and $\mathbf{P}_{\mathrm{ADM}}^{\mathrm{loc}}$ for strongly asymptotically flat spacelike initial data. It is applied in Section \ref{SECappProofPropositionSPHERESPACE}.

\begin{lemma}[Convergence rates for $\mathbf{E}_{\mathrm{ADM}}^{\mathrm{loc}}$ and $\mathbf{P}_{\mathrm{ADM}}^{\mathrm{loc}}$ for strongly asymptotically flat initial data] \label{LEMdecayPADMlocal} Let $(\Si,g,k)$ be strongly asymptotically flat initial data. Then it holds that
\begin{align*} 
\begin{aligned} 
\mathbf{E}^{\mathrm{loc}}_{\mathrm{ADM}}(S_r,g,k) = \mathbf{E}_{\mathrm{ADM}} +\OO(r^{-1}), \,\, \mathbf{P}^{\mathrm{loc}}_{\mathrm{ADM}}(S_r,g,k) = \OO(r^{-3/2}).
\end{aligned} %\label{}
\end{align*}
\end{lemma}
\begin{proof}[Proof of Lemma \ref{LEMdecayPADMlocal}] Consider first $\mathbf{P}_{\mathrm{ADM}}^{\mathrm{loc}}$. By Definitions \ref{DEFADMquantities} and \ref{DEFlocalADMcharges}, Stokes' theorem and the spacelike constraint equations \eqref{EQspacelikeConstraints1}, that is, $\Div \, k = d \tr k$, we have that for strongly asymptotically flat spacelike initial data,
\begin{align} 
\begin{aligned} 
\mathbf{P}^{\mathrm{loc}}_{\mathrm{ADM}}(S_r,g,k) :=& \frac{1}{8\pi} \int\limits_{S_r} \lrpar{k -\tr k \, g}(\pr_i,N) d\mu_\gd \\
=& \underbrace{\mathbf{P}_{\mathrm{ADM}}}_{=0} - \int\limits_{\RRR^3 \setminus B_r} \Div \lrpar{k_{i \cdot} - \tr k \, g_{i \cdot }} \\
=& -\int\limits_{\RRR^3 \setminus B_r} \underbrace{\lrpar{\Div k - d \tr k}}_{=0}(\pr_i) +( k^{nm} -\tr\, k \, g^{nm})\nab_n \lrpar{\pr_i}_m.
\end{aligned} \label{EQPdecayStronglyAFdata}
\end{align} 
By strong asymptotic flatness we have that
\begin{align*} 
\begin{aligned} 
\vert k(x) \vert \cdot \vert \pr g(x) \vert \les \vert x\vert^{-5/2} \cdot \vert x \vert^{-2}.
\end{aligned} %\label{}
\end{align*}
Hence we can bound the right-hand side of \eqref{EQPdecayStronglyAFdata} by
\begin{align*} 
\begin{aligned} 
\left\vert \, \int\limits_{\RRR^3 \setminus B_r} ( k^{nm} -\tr\, k \, g^{nm}) \nab_n \lrpar{\pr_i}_m \right\vert = \OO\lrpar{ r^{-3/2} }.
\end{aligned} %\label{}
\end{align*}
\ni The proof of the convergence rate of $\mathbf{E}_{\mathrm{ADM}}^{\mathrm{loc}}$ is similar and omitted, see also \eqref{EQenergyEstimateKERRreferenceEconv}. This finishes the proof of Lemma \ref{LEMdecayPADMlocal}. \end{proof}

%%%%%%%%%%%%%%%%%%%%%%%%%%%%%%%%%%%%%%%%
%%%%%%%%%%%%%%%%%%%%%%%%%%%%%%%%%%%%%%%%
\subsection{Foliation geometry in spacelike initial data} \label{SECfoliationGeometry} 

\ni In this section we set up notation for the geometry of foliations of spacelike initial data by $2$-spheres. Let $(\Si,g,k)$ be strongly asymptotically flat spacelike initial data, and let $(x^1,x^2,x^3)$ be corresponding Cartesian coordinates near spacelike infinity. Denote by $(r,\th^1,\th^2)$ the associated spherical coordinates, see \eqref{EQdefinSPHERICALAFcoord}. We have the following notation.

\begin{itemize} 

\item Let $S_r$ denote the level sets of $r$, and let $\gd$ and $\Nd$ denote the induced metric and covariant derivative. Let $K$ denote the Gauss curvature of $\gd$.
\item Let $N$ denote the outward pointing unit normal to $S_r$. The second fundamental form $\Th$ of $S_r$ is defined by
\begin{align*} 
\begin{aligned} 
\Th_{AB} := \D_A N_B,
\end{aligned} %\label{}
\end{align*}
and is composed into trace and tracefree part as follows,
\begin{align*} 
\begin{aligned} 
\tr\Th:= \gd^{AB}\Th_{AB}, \,\, \widehat{\Th}_{AB} := \Th_{AB} - \half \tr\Th \gd_{AB}.
\end{aligned} %\label{}
\end{align*}
\item Let $(e_A)_{A=1,2}$ denote a local orthonormal frame on $S_r$. We decompose the symmetric $2$-tensor $k$ into the $S_r$-tangent tensors
\begin{align} 
\begin{aligned} 
k_{NN}, \,\, \hspace{2mm}  \kdN_A := k_{NA}, \,\, \kd_{AB} := k_{AB}.
\end{aligned} \label{EQdefKprojection8889}
\end{align}
\end{itemize}
The Gauss-Codazzi equations of $S_r \subset \Si$ are (see, for example, Section 3.1 in \cite{ChrKl93})
\begin{align} 
\begin{aligned} 
\RRRic_{AN} =& \Divd \widehat{\Th}_A - \half \di \tr\th_A, \\
\RRRic_{NN} - \half R_{\mathrm{scal}} =& - K + \frac{1}{4} (\tr\Th)^2 + \half \vert \widehat{\Th} \vert^2,
\end{aligned} \label{EQchrkl123}
\end{align}
where $\di$ denotes the exterior derivative on $S_r$, and for a symmetric $2$-tenors $V$ on $S_r$,
\begin{align*} 
\begin{aligned} 
(\Divd V)_A := \Nd^{B}V_{BA}, \,\, \vert V \vert^2:= \gd^{AB}\gd^{CD} V_{AC}V_{BD}.
\end{aligned} %\label{}
\end{align*}

\begin{remark} For Schwarzschild reference spacelike initial data in Schwarzschild coordinates $(r,\th,\phi)$, see \eqref{EQssInitialSpacelike1}, we have that
\begin{align} 
\begin{aligned} 
\gd=&r^2 \gac, & \widehat{\Th}=&0, & \tr \Th =& \frac{2}{r}\lrpar{1-\frac{2M}{r}}^{1/2}, \\
k_{NN}=& 0, & \kdN_A =&0, & \kd_{AB} =& 0, \\
\RRRic_{NN}=& -\frac{2M}{r^3}, & \RRRic_{NA} =& 0, & \RRRic_{AB} =& \frac{M}{r^3} r^2 \gac_{AB}.
\end{aligned} \label{EQSchwarzschildSphericalSScoordinates}
\end{align}
\end{remark}

%%%%%%%%%%%%%%%%%%%%%%%%%%%%%%%%%%%%%%%%
%%%%%%%%%%%%%%%%%%%%%%%%%%%%%%%%%%%%%%%%
\subsection{Scaling of spacelike initial data and local norms} \label{SECspacelikescaling} In this section we define, analogous to Section \ref{SECdefinitionScaling}, the scaling of spacelike initial data, and introduce local norms.

Let $(\Si,g,k)$ be an asymptotically flat spacelike initial data set and let $(x^1,x^2,x^3)$ denote associated coordinates near spacelike infinity. We define the scaling of $(g,k)$ in two steps.  
\begin{enumerate}
\item For a real number $R\geq1$, define new coordinates $(y^1,y^2, y^3)$ by 
\begin{align} 
\begin{aligned} 
\Psi_R(y^1,y^2,y^3) := (R\cdot y^1, R \cdot y^2, R\cdot y^3) = (x^1, x^2, x^3).
\end{aligned} \label{EQdefcoordinatechangeSPACELIKE8889}
\end{align}

\item Based on the conformal scaling of spacetime metrics ${}^{(R)}\g:= R^{-2} \g$ (see Section \ref{SECdefinitionScaling}) and that $k$ is the second fundamental form of $\Si$, we define $({}^{(R)} g,{}^{(R)} k)$ by
\begin{align} 
\begin{aligned} 
{}^{(R)} g := R^{-2} g, \,\,{}^{(R)}k := R^{-1}\, k.
\end{aligned} \label{EQdefconformalchangeSPACELIKE8889}
\end{align}
\noindent By construction, $({}^{(R)} g,{}^{(R)} k)$ solve the spacelike constraint equations \eqref{EQspacelikeConstraints1}.
\end{enumerate}

By \eqref{EQdefcoordinatechangeSPACELIKE8889} and \eqref{EQdefconformalchangeSPACELIKE8889}, for all integers $l\geq0$, we have the relations
\begin{align} 
\begin{aligned} 
\pr_y^{l} \lrpar{{}^{(R)}g_{ij}}= R^l \lrpar{\pr_x^l g_{ij}} \circ \Psi_R, \,\, \pr_y^{l} \lrpar{{}^{(R)}k_{ij}} = R^{l+1} \lrpar{\pr_x^l k_{ij}} \circ \Psi_R,
\end{aligned} \label{EQscalingderivatives}
\end{align}
where we denote
\begin{align*} 
\begin{aligned} 
{}^{(R)} g_{ij} := {}^{(R)} g(\pr_{y^i},\pr_{y^j}), \,\, {}^{(R)} k_{ij} := {}^{(R)} k(\pr_{y^i},\pr_{y^j}).
\end{aligned} %\label{}
\end{align*}

\ni \emph{Remarks on the scaling of spacelike initial data.}
\begin{enumerate} 
\item Analogous to Lemma \ref{LEMrescaleLOCALCHARGES}, we deduce from \eqref{EQscalingderivatives} that the charges scale as follows. The proof is omitted.
\begin{align} 
\begin{aligned} 
\mathbf{E}_{\mathrm{ADM}}^{\mathrm{loc}}\lrpar{S_{r_0},{}^{(R)}g,{}^{(R)}k} =& R^{-1} \cdot \mathbf{E}_{\mathrm{ADM}}^{\mathrm{loc}}\lrpar{S_{R\cdot r_0},g,k}, \\
\mathbf{P}_{\mathrm{ADM}}^{\mathrm{loc}}\lrpar{S_{r_0},{}^{(R)}g,{}^{(R)}k} =& R^{-1} \cdot \mathbf{P}_{\mathrm{ADM}}^{\mathrm{loc}}\lrpar{S_{R\cdot r_0},g,k}, \\
\mathbf{L}_{\mathrm{ADM}}^{\mathrm{loc}}\lrpar{S_{r_0},{}^{(R)}g,{}^{(R)}k} =& R^{-2} \cdot \mathbf{L}_{\mathrm{ADM}}^{\mathrm{loc}}\lrpar{S_{R\cdot r_0},g,k}, \\
\mathbf{C}_{\mathrm{ADM}}^{\mathrm{loc}}\lrpar{S_{r_0},{}^{(R)}g,{}^{(R)}k} =& R^{-2} \cdot \mathbf{C}_{\mathrm{ADM}}^{\mathrm{loc}}\lrpar{S_{R\cdot r_0},g,k}.
\end{aligned} \label{EQLEMscalingADMlocal}
\end{align}

\item Applying the above scaling to Schwarzschild reference spacelike initial data, see \eqref{EQssInitialSpacelike1}, it holds that
\begin{align} 
\begin{aligned} 
{}^{(R)}g^M_{ij} = g^{M/R}_{ij}, \,\, {}^{(R)}\tilde{g}^M_{ij} = \tilde{g}^{M/R}_{ij}
\end{aligned} \label{EQschwarzschildSPACELIKEscaling}
\end{align}

\item By \eqref{EQscalingderivatives}, the property of strong asymptotic flatness is conserved under rescaling. 

\end{enumerate}

\ni We now turn to the introduction of local norms for spacelike initial data. For ease of presentation we use $C^k$-spaces.

\begin{definition}[Norms for tenors] Let $K\subset \RRR^3$ denote a compact set with smooth boundary, and let $T$ be an $j$-tensor on $K$. For integers $k\geq0$ define
\begin{align*} 
\begin{aligned} 
\Vert T \Vert_{C^k(K)} := \sum\limits_{1\leq i_1, \cdots i_j \leq 3}\sum\limits_{0\leq \vert \a \vert \leq k} \Vert \pr^\a T_{i_1\cdots i_j} \Vert_{L^{\infty}(K)},
\end{aligned} 
\end{align*}
where $\a=(\a_1,\a_2,\a_3) \in \mathbb{N}^3$, $\pr^\a = \pr_1^{\a_1}\pr_2^{\a_2}\pr_3^{\a_3}$ and $T_{i_1 \cdots i_l}$ denotes the Cartesian coordinate components of $T$. Define $C^k(K)$ to be the space of $k$-times continuously differentiable tensors $T$ on $K$ with
\begin{align} 
\begin{aligned} 
\Vert T  \Vert_{C^k(K)} < \infty. 
\end{aligned} \label{EQfiniteCKnormTspacelike}
\end{align}

\ni Moreover, let $C^k_{\mathrm{loc}}(\RRR^3 \setminus B(0,1))$ be the space of $k$-times continuously differentiable tensors $T$ on $\RRR^3 \setminus B(0,1)$ such that \eqref{EQfiniteCKnormTspacelike} holds for each compact subset $K \subset \RRR^3 \setminus B(0,1)$.
\end{definition}

\ni \textbf{Notation.} For two real numbers $0<r_1<r_2$, define the coordinate annulus $A_{[r_1,r_2]}$ by
\begin{align*} 
\begin{aligned} 
A_{[r_1,r_2]} := \left\{ x \in \RRR^3: r_1 \leq \vert x \vert \leq r_2 \right\}.
\end{aligned} %\label{}
\end{align*}

\begin{definition}[Local norm for spacelike initial data] \label{DEFnormSpacelikeInitialData} Let $0<r_1<r_2$ be two real numbers, and let $k\geq1$ be an integer. For spacelike initial data $(g,k)$ on $A_{[r_1,r_2]}$ we define
\begin{align*} 
\begin{aligned} 
\left\Vert \lrpar{g,k} \right\Vert_{C^{k}(A_{[r_1,r_2]}) \times C^{k-1}(A_{[r_1,r_2]})} :=
\Vert g \Vert_{C^{k}(A_{[r_1,r_2]})} + \Vert k \Vert_{C^{k-1}(A_{[r_1,r_2]})}.
\end{aligned} %\label{}
\end{align*}
\end{definition}

\ni \textbf{Notation.} In the following we assume that the metric $g$ is $k_0$-times and the second fundamental form $k$ is $k_0$-times continuously differentiable, where the universal integer $k_0 \geq 8$ is determined in Section \ref{SECstatementConstruction} by the condition that the ingoing null data to be constructed from spacelike initial data is sufficiently regular. \\

\ni By scaling and the definition of strong asymptotic flatness, we have the following estimates for rescaled spacelike initial data.
\begin{lemma}[Smallness of rescaled spacelike initial data] \label{LEMspacelikeRescaling} Let $(\Si,g,k)$ be strongly asymptotically flat spacelike initial data with Cartesian coordinates $(x^1,x^2,x^3)$ near spacelike infinity. For real numbers $R\geq1$ sufficiently large, the rescaled spacelike initial data
$$({}^{(R)}{g}_{ij}, {}^{(R)}{k}_{ij})$$
is well-defined on ${A}_{[1/2,7/2]}$ and
\begin{align} 
\begin{aligned} 
\Vert \lrpar{{}^{(R)}{g} -\tilde{g}^{M/R},  {}^{(R)}{k}} \Vert_{C^{k_0}({A}_{[1/2,7/2]})\times C^{k_0-1}({A}_{[1/2,7/2]})} 
= \smallO(R^{-3/2}),
\end{aligned} \label{EQestimateLEMspacelikeRescaling}
\end{align}
where $M$ is the real number appearing in \eqref{EQAFconditions}.
\end{lemma}

\begin{proof}[Proof of Lemma \ref{LEMspacelikeRescaling}] The strong asymptotic flatness condition \eqref{EQAFconditions} on $g$ can be rewritten as follows, see also \eqref{EQSSasympExpansion},
\begin{align*} 
\begin{aligned} 
g_{ij}= \lrpar{1+\frac{2M}{r}} e_{ij} +\smallO(r^{-3/2}) = \lrpar{1+\frac{M}{2r}}^4 e_{ij} + \smallO(r^{-3/2}),
\end{aligned} %\label{}
\end{align*}
which yields by definition of the Schwarzschild metric $\tilde{g}^M_{ij}$ in isotropic coordinates that
\begin{align} 
\begin{aligned} 
g_{ij} - \tilde{g}^M_{ij} = \smallO(r^{-3/2}).
\end{aligned} \label{EQscalingsmallnessFAR899}
\end{align}

\ni For real numbers $R\geq1$ sufficiently large, the rescaled spacelike initial data $({}^{(R)}g,{}^{(R)}k)$ is well-defined on the annulus $A_{[1/2, 7/2]}$. Subsequently, \eqref{EQestimateLEMspacelikeRescaling} follows from \eqref{EQscalingsmallnessFAR899} by \eqref{EQscalingderivatives} and the scaling \eqref{EQschwarzschildSPACELIKEscaling}. The estimate for ${}^{(R)}k$ is similar. This finishes the proof of Lemma \ref{LEMspacelikeRescaling}. \end{proof}

\ni Moreover, we note the following lemma. Its proof follows from \eqref{EQssInitialSpacelike1} and is omitted.
\begin{lemma}[Estimates for Schwarzschild reference metric] \label{LEMschwarzschildtrivialEstimate} For real numbers $M\geq0$ sufficiently small, it holds that
\begin{align*} 
\begin{aligned} 
\Vert g^M - e \Vert_{C^{k_0}({A}_{[1/2,7/2]})} \les M.
\end{aligned} %\label{}
\end{align*}

\end{lemma}

%%%%%%%%%%%%%%%%%%%%%%%%%%%%%%%%%%%%%%%%

\section{Construction of sphere data from spacelike initial data} \label{SECstatementConstruction} \ni In this section we construct families of ingoing null data from spacelike initial data. The following theorem is the main result of this section.

\begin{theorem}[Construction of ingoing null data from spacelike initial data] \label{PROPconstructionStatement} Let $(\Si,g,k)$ be strongly asymptotically flat spacelike initial data with asymptotic invariants
\begin{align*} 
\begin{aligned} 
(\mathbf{E}_{\mathrm{ADM}}, \mathbf{P}_{\mathrm{ADM}}, \mathbf{L}_{\mathrm{ADM}}, \mathbf{C}_{\mathrm{ADM}}),
\end{aligned} %\label{}
\end{align*}
where $\mathbf{P}_{\mathrm{ADM}}=0$ by the strong asymptotic flatness. There is a real number $\de>0$ and a strongly asymptotically flat family of ingoing data $(x_{-R+R\cdot [-\de,\de],R})$, constructed on spheres in $\Si$, such that for $m=-1,0,1$ and $(i_{-1},i_0,i_1)=(2,3,1)$,
\begin{align*} 
\begin{aligned} 
\mathbf{E}(x_{-R,R}) =& \mathbf{E}_{\mathrm{ADM}} + \OO(R^{-1}), & \mathbf{P}^m(x_{-R,R}) =& \lrpar{\mathbf{P}_{\mathrm{ADM}}}^{i_m} + \OO(R^{-3/2}), \\
\mathbf{L}^m(x_{-R,R}) =& \lrpar{\mathbf{L}_{\mathrm{ADM}}}^{i_m}+ \smallO(1), & \mathbf{G}^m(x_{-R,R}) =& \lrpar{\mathbf{C}_{\mathrm{ADM}}}^{i_m} + \smallO(1).
\end{aligned} %\label{}
\end{align*}
Moreover, if the spacelike initial data is smooth, then the constructed ingoing null data is smooth, along with all higher-order derivatives in all directions.
\end{theorem} 

\ni \emph{Remarks on Theorem \ref{PROPconstructionStatement}.}
\begin{enumerate} 
\item In the particular case of Schwarzschild reference spacelike initial data in isotropic coordinates, $\tilde{g}^{\mathbf{E}_{\mathrm{ADM}}}_{ij}$, the construction of Theorem \ref{PROPconstructionStatement} produces the Schwarzschild reference family of sphere data $\mathfrak{m}^{\mathbf{E}_{\mathrm{ADM}}}_{-R,R}$ in Eddington-Finkelstein coordinates, see \eqref{EQspheredataSSM111222}.

\item The construction of sphere data in Section \ref{SECproofExistenceConstruction} and the estimates of Sections \ref{SECcomparisonENERGY}, \ref{SECcomparisonLINEAR}, \ref{SECcomparisonANGULAR} and \ref{SECcomparisonCENTER} are moreover applied in Section \ref{SECappKerrFamilyDetails} to Kerr reference spacelike initial data. The latter is not strongly asymptotically flat but the construction still goes through with a more general smallness parameter $\varep_R>0$, see \eqref{EQREMgeneralisedEstimates}.

\end{enumerate}

\ni The proof of Theorem \ref{PROPconstructionStatement} is structured as follows.
\begin{itemize}

\item In Section \ref{SECspacelikeRescalingtoSmalldata} we rescale the strongly asymptotically flat spacelike initial data on the annulus $A_{[R/2,7R/2]}$ to spacelike initial data on $A_{[1/2,7/2]}$ and change from isotropic to Schwarzschild coordinates, to arrive at spacelike initial data on $A_{[1,3]}$ close to Schwarzschild (in Schwarzschild coordinates) of mass $M/R$.

\item In Section \ref{SECproofExistenceConstruction} we construct from the spacelike initial data on $A_{[1,3]}$ ingoing null data $$({}^{(R)}{x}_{-1+[-\de,\de],1}),$$
and prove estimates.

\item In Sections \ref{SECcomparisonENERGY}, \ref{SECcomparisonLINEAR}, \ref{SECcomparisonANGULAR} and \ref{SECcomparisonCENTER} we compare the charges $(\mathbf{E},\mathbf{P},\mathbf{L},\mathbf{G})$ of ${}^{(R)}x_{-1,1}$ with the local integrals $(\mathbf{E}_{\mathrm{ADM}}^{\mathrm{loc}},\mathbf{P}_{\mathrm{ADM}}^{\mathrm{loc}},\mathbf{L}_{\mathrm{ADM}}^{\mathrm{loc}},\mathbf{C}_{\mathrm{ADM}}^{\mathrm{loc}})$ on $S_{-1,1}\subset A_{[1,3]}$ of the spacelike initial data. 

\item In Section \ref{SECappProofPropositionSPHERESPACE} we conclude the proof of Theorem \ref{PROPconstructionStatement} by scaling the constructed ingoing null data $({}^{(R)}{x}_{-1+[-\de,\de],1})$ up to $({x}_{-R+R\cdot[-\de,\de],R})$, and analyzing the asymptotics of $(\mathbf{E},\mathbf{P},\mathbf{L},\mathbf{G})(x_{-R,R})$ by use of the estimates of Sections \ref{SECcomparisonENERGY}, \ref{SECcomparisonLINEAR}, \ref{SECcomparisonANGULAR} and \ref{SECcomparisonCENTER}.

\end{itemize}

%%%%%%%%%%%%%%%%%%%%%%%%%%%%%%%%%%%%%%%%
\subsection{Rescaling and change to Schwarzschild coordinates} \label{SECspacelikeRescalingtoSmalldata}
\ni Let $(\Si,g,k)$ be strongly asymptotically flat spacelike initial data, and let $(x^1,x^2,x^3)$ denote corresponding coordinates near spacelike infinity. In the following we first rescale to small data on an annulus $A_{[1/2, 7/2]}$ and then change from isotropic coordinates to Schwarzschild coordinates, see \eqref{EQcoordinatechangedef1} and \eqref{EQcoordinatechangedef2}, yielding spacelike initial data on the annulus $A_{[1, 3]}$.

In the particular case of Schwarzschild reference spacelike data in isotropic coordinates of mass $M$, denoted by $\tilde{g}^M_{ij}$, the following construction maps to Schwarzschild reference spacelike initial data in Schwarzschild coordinates of mass $M/R$, denoted by $g_{ij}^M/R$.

First, let $( {}^{(R)} g, {}^{(R)} k)$ denote the rescaled spacelike initial data. By Lemma \ref{LEMspacelikeRescaling} we have that
\begin{align} 
\begin{aligned} 
\left\Vert \lrpar{ {}^{(R)} g-\tilde{g}^{M/R}, {}^{(R)} k } \right\Vert_{C^{k_0}(A_{[1/2, 7/2]})\times C^{k_0-1}(A_{[1/2, 7/2]})}
= \smallO(R^{-3/2}),
\end{aligned} \label{EQsmallnesssphere11}
\end{align}

\ni Second, we apply the coordinate change $\Phi$ from isotropic coordinates $(\tilde{r},\tth^1,\tth^2)$ to Schwarzschild coordinates $(r,\th^1,\th^2)$, see \eqref{EQcoordinatechangedef1}, with $M/R$,
\begin{align*} 
\begin{aligned} 
\Phi: (\tilde{r}, \tth^1,\tth^2) \to (r,\th^1,\th^2) := \lrpar{\tilde{r} \lrpar{1+\frac{M/R}{2\tilde{r}}}^2, \tth^1,\tth^2},
\end{aligned} %\label{}
\end{align*}
\ni On the one hand, for $R\geq1$ sufficiently large, the Schwarzschild coordinates $(r,\th^1,\th^2)$ range over the coordinate domain $A_{[1,3]}$. On the other hand, by \eqref{EQcoordinatechangedef2} we can estimate for $R\geq1$ sufficiently large
\begin{align} 
\begin{aligned} 
\left\Vert \mathrm{D}\Phi-\mathrm{Id} \right\Vert_{C^{k_0}(A_{[1/2, 7/2]})} \leq C,
\end{aligned} \label{EQDPhiEstimate}
\end{align}
where $C>0$ is a universal constant. Thus by \eqref{EQsmallnesssphere11} and \eqref{EQDPhiEstimate},
\begin{align*} 
\begin{aligned} 
\left\Vert \Phi^\ast \lrpar{{}^{(R)}g} - {g}^{M/R} \right\Vert_{C^{k_0}(A_{[1,3]})} =& \left\Vert \Phi^\ast \lrpar{{}^{(R)}g - \tilde{g}^{M/R} } \right\Vert_{C^{k_0}(A_{[1,3]})} \\
\les& \left\Vert {}^{(R)}g-\tilde{g}^{M/R} \right\Vert_{C^{k_0}(A_{[1/2, 7/2]})} \\
=& \smallO(R^{-3/2}),
\end{aligned} %\label{}
\end{align*}
where we used that by \eqref{EQcoordiSSchange}, $\Phi^\ast \lrpar{\tilde{g}^{M/R}} = g^{M/R}$. Furthermore, by \eqref{EQsmallnesssphere11} and \eqref{EQDPhiEstimate} it similarly follows that
\begin{align*} 
\begin{aligned} 
\left\Vert \Phi^\ast \lrpar{{}^{(R)}k} \right\Vert_{C^{k_0-1}(A_{[1,3]})} = \smallO(R^{-3/2}).
\end{aligned} %\label{}
\end{align*}

\ni To summarize the above, for $R\geq1$ sufficiently large, we constructed from strongly asymptotically flat spacelike initial data $(\Si,g,k)$ the spacelike initial data 
\begin{align*} 
\begin{aligned} 
\lrpar{\Phi^\ast\lrpar{{}^{(R)} g}, \Phi^\ast\lrpar{{}^{(R)} k}} \text{ on } A_{[1,3]},
\end{aligned} %\label{}
\end{align*}
satisfying
\begin{align} 
\begin{aligned} 
\left\Vert \Phi^\ast \lrpar{{}^{(R)}g} - g^{M/R} \right\Vert_{C^{k_0}(A_{[1,3]})} + \left\Vert \Phi^\ast \lrpar{{}^{(R)}k} \right\Vert_{C^{k_0}(A_{[1,3]})} = 
\smallO(R^{-3/2}).
\end{aligned} \label{EQSMALLNESSrescaled130009}
\end{align}

\ni \textbf{Notation.} We use the following notation in Sections \ref{SECproofExistenceConstruction}, \ref{SECcomparisonENERGY}, \ref{SECcomparisonLINEAR}, \ref{SECcomparisonANGULAR} and \ref{SECcomparisonCENTER}.
\begin{enumerate}
\item We denote the $R$-dependent smallness on the right-hand side of \eqref{EQSMALLNESSrescaled130009} by 
\begin{align}
\varep_R:= \smallO(R^{-3/2}). \label{EQREMgeneralisedEstimates}
\end{align}
\noindent This allows us to apply the same estimates in Section \ref{SECappKerrFamilyDetails} for the study of Kerr spacelike initial data where smallness parameter is replaced by $\varep_R = \OO(R^{-3/2})$.
\item For ease of presentiation we abuse notation by denoting $\lrpar{\Phi^\ast\lrpar{{}^{(R)} g}, \Phi^\ast\lrpar{{}^{(R)} k}}$ by $(g,k)$.

\end{enumerate}

%%%%%%%%%%%%%%%%%%%%%%%%%%%%%%%%%%%%%%%%
\subsection{Construction of sphere data} \label{SECproofExistenceConstruction}

\ni Let $M\geq0$ and $R\geq1$ be two real numbers. Consider spacelike initial data $(g,k)$ on $A_{[1,3]}$ such that 
\begin{align} 
\begin{aligned} 
\left\Vert \lrpar{ g - g^{M/R},k} \right\Vert_{C^{k_0}(A_{[1,3]})\times C^{k_0-1}(A_{[1,3]})}  \leq \varep_R,
\end{aligned} \label{EQSMALLNESSrescaled13000}
\end{align}
see \eqref{EQREMgeneralisedEstimates} for the $\varep_R$-notation. In this section we construct from $(g,k)$ the ingoing null data
\begin{align*} 
\begin{aligned} 
({}^{(R)}{x}_{-1+[-\de,\de],1})
\end{aligned} %\label{}
\end{align*}
and prove that for $R\geq1$ sufficiently large,
\begin{align} 
\begin{aligned} 
\Vert {}^{(R)}{x}_{-1+[-\de,\de],1} - {\mathfrak{m}}^{M/R} \Vert_{\XX^+(\HHb_{-1+[-\de,\de],1})} \les \varep_R.
\end{aligned} \label{EQfullEstimateExtendedSpheredataCONST}
\end{align} 
In the particular case Schwarzschild reference data in Schwarzschild coordinates, $g^{M/R}$, the construction of this section produces the Schwarzschild reference ingoing null data in Eddington-Finkelstein coordinates $\mathfrak{m}^{M/R}_{-1+[-\de,\de],1}$. We remark that the universal integer $k_0\geq6$ is determined from the regularity in \eqref{EQfullEstimateExtendedSpheredataCONST}, see the notational remark after Definition \ref{DEFnormSpacelikeInitialData}.\\

\ni \textbf{Definition of $S_{-1,1}$ and gauge choices.} 
Let $(\MM,\g)$ denote the unique maximal future globally-hyperbolic development of the spacelike initial data 
$$(A_{[1,3]},g,k).$$
Let $T$ denote the future-directed timelike unit vector to $A_{[1,3]}$ in $(\MM,\g)$, and let $N$ denote the outward pointing unit normal to $S_r \subset \Si$ tangent to $A_{[1,3]}$ for $1\leq r \leq 3$.

For $R\geq1$ sufficiently large, consider the sphere 
\begin{align*} 
\begin{aligned} 
S_{r_{M/R}(-1,1)} \subset A_{[1,3]}, 
\end{aligned} %\label{}
\end{align*}
where the definition of $r_M(u,v)$ is given in \eqref{EQdefRbyUV}. On $S_{r_{M/R}(-1,1)} \subset \Sigma$ define the renormalized null vectors $(\widehat{L},\widehat{\Lb})$ by
\begin{align} 
\begin{aligned} 
\widehat{L}=T+N, \,\, \widehat{\Lb} = T-N,
\end{aligned} \label{EQdefinitionSpacelike223}
\end{align}
which satisfy by construction $\g(\widehat{L},\widehat{\Lb})=-2$. We can construct around $S_{r_{M/R}(-1,1)} \subset \Sigma$ a local double null coordinate system $(u,v,\th^1,\th^2)$ such that with respect to $(u,v)$,
\begin{align*} 
\begin{aligned} 
S_{-1,1} = S_{r_{M/R}(-1,1)},
\end{aligned} %\label{}
\end{align*}
and moreover, the following holds on $S_{-1,1}$ (which is in agreement with \eqref{EQspheredataSSM111222})
\begin{align} 
\begin{aligned} 
\Om^2:=&1-\frac{2M/R}{r_{M/R}(-1,1)}, & \om:=&\frac{M/R}{\lrpar{r_{M/R}(-1,1)}^2}, & \omb:=&-\frac{M/R}{\lrpar{r_{M/R}(-1,1)}^2}, \\ 
D\om:=&-\frac{2\Om^2M/R}{\lrpar{r_{M/R}(-1,1)}^3}, & \Du\omb:=&-\frac{2\Om^2M/R}{\lrpar{r_{M/R}(-1,1)}^3}.&&
\end{aligned} \label{EQGaugeChoicesSRscaling}
\end{align} 

\ni \textbf{Definition and analysis of $\chi$ and $\chib$.} Defining $\chi$ and $\chib$ as in Definition \ref{DEFricciCoefficients}, we have by \eqref{EQdefKprojection8889} and \eqref{EQdefinitionSpacelike223},
\begin{align} 
\begin{aligned} 
\chi_{AB} =  -\kd_{AB} + \Theta_{AB}, \,\, \chib_{AB} = -\kd_{AB}-\Theta_{AB}.
\end{aligned} \label{EQCHIFULLkrelations1222}
\end{align}

\ni Taking the tracefree part and trace with respect to $\gd$, we get
\begin{align} 
\begin{aligned} 
\trchi =&- \tr \kd + \tr \Theta, & \chih_{AB} =& -\widehat{\kd}_{AB} + \widehat{\Theta}_{AB}, \\
\trchib =& - \tr \kd - \tr \Theta, & \chibh_{AB} =& -\widehat{\kd}_{AB} - \widehat{\Theta}_{AB}.
\end{aligned} \label{EQCHIkrelations1222}
\end{align}

\ni By \eqref{EQspheredataSSM111222}, \eqref{EQSMALLNESSrescaled13000} and \eqref{EQCHIkrelations1222} we have that for $R\geq1$ sufficiently large,
\begin{align} 
\begin{aligned} 
\Vert \chih \Vert_{H^6(S_{-1,1})} + \Vert \chibh \Vert_{H^6(S_{-1,1})} \les&\varep_R, \\
\left\Vert \trchi - \frac{2\Om_M}{r_{M/R}(-1,1)} \right\Vert_{H^6(S_{-1,1})}+\left\Vert \trchib + \frac{2\Om_M}{r_{M/R}(-1,1)}  \right\Vert_{H^6(S_{-1,1})} \les& \varep_R.
\end{aligned} \label{EQfirstEstimScalingSpacelikeSphere}
\end{align}

\ni \textbf{Definition and analysis of $\zeta$ and $\eta$.} Defining $\zeta$ and $\eta$ on $S_{-1,1}$ as in Definition \ref{DEFricciCoefficients}, we have by \eqref{EQdefinitionSpacelike223} that
\begin{align} 
\begin{aligned} 
\zeta_A :=& \half \g\lrpar{\D_A \widehat{L},\widehat{\Lb}} = -\half \g\lrpar{\D_A T,N} + \half \g\lrpar{\D_A N,T} = - \g\lrpar{\D_A T,N} = \kd_{A},\\
\eta_A:=& \zeta_A +\di \log \Om = \kd_A + \di \log \Om=  \kd_A,
\end{aligned} \label{EQrelationZETAk}
\end{align}
where we used \eqref{EQdefinitionSpacelike223} and \eqref{EQGaugeChoicesSRscaling}. Subsequently, by \eqref{EQSMALLNESSrescaled13000} and \eqref{EQrelationZETAk} we have that for $R\geq1$ sufficiently large,
\begin{align} 
\begin{aligned} 
\left\Vert \zeta \right\Vert_{{H}^5(S_{-1,1})}+\left\Vert \eta \right\Vert_{{H}^5(S_{-1,1})}\les \varep_R.
\end{aligned} \label{EQscalingETAestimate}
\end{align}

\ni \textbf{Definition and analysis of $\a$ and $\ab$.} By Definition \ref{EQnullcurvatureCOMPDEF}, \eqref{EQlinearCombination112}, \eqref{EQGaussCodazziSpacelike} and \eqref{EQdefinitionSpacelike223}, we have that
\begin{align} 
\begin{aligned} 
\a_{AB} :=& \Rbf\lrpar{e_A,\widehat{L},e_B, \widehat{L}} \\
=& \Rbf_{ATBT}+\Rbf_{ATBN}+\Rbf_{ANBT}+\Rbf_{ANBN} \\
=& \Rbf_{ANBT}+\Rbf_{BNAT} \\
=& -\in_{AN}^{\,\,\,\,\,\,\,\,\,\,s} H_{sB}-\in_{BN}^{\,\,\,\,\,\,\,\,\,\,s} H_{sA} \\
=&  - \lrpar{\nab_A k_{NB}-\nab_N k_{AB}}- \lrpar{\nab_B k_{NA}-\nab_N k_{BA}} \\
=& - \lrpar{\nab_A k_{NB}+\nab_B k_{NA}}+  2\nab_N k_{BA},
\end{aligned} \label{EQalphaScalingEstimatesEQ}
\end{align}
where we used that by the Einstein vacuum equations and the algebraic properties of the Riemann curvature tensor,
\begin{align*} 
\begin{aligned} 
\Rbf_{ANBN} = \underbrace{\mathbf{Ric}}_{=0} - \Rbf_{ATBT} - \sum\limits_{C=1,2} \underbrace{\Rbf_{ACBC}}_{=0}= -\Rbf_{ATBT}.
\end{aligned}
\end{align*}

\ni Subsequently, by \eqref{EQSMALLNESSrescaled13000} and \eqref{EQalphaScalingEstimatesEQ} we have that for $R\geq1$ sufficiently large,
\begin{align} 
\begin{aligned} 
\Vert \a \Vert_{H^6(S_{-1,1})} \les \varep_R.
\end{aligned} \label{EQalphaScalingEstimates}
\end{align}

\ni Similarly, by Definition \ref{EQnullcurvatureCOMPDEF}, \eqref{EQdefinitionSpacelike223} and \eqref{EQalphaScalingEstimatesEQ} we have that
\begin{align} 
\begin{aligned} 
\ab_{AB} := \Rbf\lrpar{e_A,\widehat{\Lb},e_B, \widehat{\Lb}} =&\Rbf_{ATBT}-\Rbf_{ATBN}-\Rbf_{ANBT}+\Rbf_{ANBN} \\
=& -\a_{AB},
\end{aligned} \label{EQalphaBARScalingEstimatesEQ}
\end{align}
so we get that from \eqref{EQalphaScalingEstimates} that for $R\geq1$ sufficiently large,
\begin{align} 
\begin{aligned} 
\Vert \ab \Vert_{H^2(S_{-1,1})} \les \varep_R.
\end{aligned} \label{EQalphaBScalingEstimates}
\end{align}

\ni \textbf{Definition and analysis of ${}^{(R)}x_{-1,1}$ on $S_{-1,1}$.} Let ${}^{(R)}x_{-1,1}$ be the sphere data on $S_{-1,1}$ determined by the quantities constructed in \eqref{EQGaugeChoicesSRscaling}, \eqref{EQCHIFULLkrelations1222}, \eqref{EQCHIkrelations1222}, \eqref{EQrelationZETAk}, \eqref{EQalphaScalingEstimatesEQ} and \eqref{EQalphaBARScalingEstimatesEQ}. From the estimates \eqref{EQfirstEstimScalingSpacelikeSphere}, \eqref{EQscalingETAestimate}, \eqref{EQalphaScalingEstimates} and \eqref{EQalphaBScalingEstimates}, it follows that
\begin{align} 
\begin{aligned} 
\Vert {}^{(R)} x_{-1,1} - \mathfrak{m}^{M/R} \Vert_{\XX(S_{-1,1})} \les \varep_R.
\end{aligned} \label{EQspheredataestimateVAREPRfinal007778}
\end{align}

\ni \textbf{Definition and analysis of ${}^{(R)}{x}_{-1+[-\de,\de],1}$ on $\HHb_{-1+[-\de,\de],1}$.} Following \eqref{EQdefinitionSpacelike223} and \eqref{EQGaugeChoicesSRscaling}, define $\Lb'$ on $S_{-1,1}$ by 
\begin{align*} 
\begin{aligned} 
\Lb' := \Om_M^{-1} \widehat{\Lb} = \frac{1}{\sqrt{1-\frac{2M/R}{r_{M/R}(-1,1)} }} \lrpar{T-N},
\end{aligned} %\label{}
\end{align*}
and extend $\Lb'$ to the spacetime $(\MM,\g)$ as null geodesic vectorfield. The ingoing null hypersurface $\HHb_{1} \subset \MM$ passing through $S_{r_{M/R}(-1,1)} \subset \Si$ is ruled by $\Lb'$. We define on $\HHb_{1}$ the function $u$ by
\begin{align*} 
\begin{aligned} 
\Lb'(u) = \frac{1}{\sqrt{1-\frac{2M/R}{r_{M/R}(u,1)} }} \text{ on } \HHb_1 \text{ and } u \vert_{S_{-1,1}} = -1.
\end{aligned} %\label{}
\end{align*}
The level sets $S_{u,1}\subset \HHb_1$ of $u$ are locally well-defined and foliate $\HHb_1$ by construction with Schwarzschild reference null lapse.

By the smallness \eqref{EQSMALLNESSrescaled13000} together with the above gauge choices \eqref{EQdefinitionSpacelike223} and \eqref{EQGaugeChoicesSRscaling}, by the local existence and Cauchy stability for the spacelike Cauchy problem, see, for example, \cite{BruhatOrigin}, it follows that for $R\geq1$ and $k_0\geq6$ sufficiently large, there is a universal real number $\de>0$, such that the foliated null hypersurface 
\begin{align*} 
\begin{aligned} 
\HHb_{-1+[-\de,\de],1} := \bigcup\limits_{-\de \leq u \leq \de} S_{u,1}
\end{aligned} %\label{}
\end{align*}
is well-defined in $(\MM,\g)$ and the induced null data, denoted by ${}^{(R)}{x}_{-1+[-\de,\de],1}$, satisfies
\begin{align} 
\begin{aligned} 
\Vert {}^{(R)}{x}_{-1+[-\de,\de],1}-{\mathfrak{m}}^{M/R} \Vert_{\XX^+\lrpar{\HHb_{-1+[-\de,\de],1}}} \les \varep_R.
\end{aligned} \label{EQunderlinexspacelikeestim}
\end{align}

\ni To summarise the above, we constructed ingoing null data 
$$({}^{(R)}{x}_{-1+[-\de,\de],1})$$ 
satisfying
\begin{align} 
\begin{aligned} 
\Vert {}^{(R)}{x}_{-1+[-\de,\de],1}-{\mathfrak{m}}^{M/R} \Vert_{\XX^+\lrpar{\HHb_{-1+[-\de,\de],1}}} \les\varep_R.
\end{aligned} \label{EQfinalspaceliketospherebound}
\end{align}
This finishes the proof of \eqref{EQfullEstimateExtendedSpheredataCONST}.

\begin{remark}\label{REMARKConstructionHigherOrder} In case of higher regularity, we impose gauge conditions on $D^m\om$ and $\Du^m \omb$, for integers $m\geq2$ on $S_{-1,1}$ in accordance with the Schwarzschild reference higher-order sphere data \eqref{EQspheredataSSM111222}. Subsequently, the higher-order sphere data on $S_{-1,1}$ can be explicitly calculated and estimated by the Bianchi identities.
\end{remark}

%%%%%%%%%%%%%%%%%%%%%%%%%%%%%%%%%%%%%%%%
\subsection{Comparison of $\mathbf{E}$ and $\mathbf{E}^{\mathrm{loc}}_{\mathrm{ADM}}$} \label{SECcomparisonENERGY}

\ni In this section we prove that
\begin{align} 
\begin{aligned} 
\mathbf{E}({}^{(R)}x_{-1,1}) =\mathbf{E}_{\mathrm{ADM}}^{\mathrm{loc}}(S_{-1,1},g,k) +\OO\lrpar{\frac{M}{R} \varep_R}+ \OO(\varep_R^2),
\end{aligned} \label{EQEcomparison11}
\end{align}
where we recall from Definitions \ref{DEFlocalCharges} and \ref{DEFlocalADMcharges} that
\begin{align*} 
\begin{aligned} 
\mathbf{E}({}^{(R)}x_{-1,1}) :=& -\frac{1}{8\pi} \sqrt{4\pi} \lrpar{ r^3\lrpar{\rho + r \Divd {\be}}}^{(0)}({}^{(R)}x_{-1,1}), \\ 
\mathbf{E}_{\mathrm{ADM}}^{\mathrm{loc}}(S_{-1,1},g,k) :=& -\frac{1}{8\pi} \int\limits_{S_{-1,1}} \lrpar{\RRRic -\half R_{\mathrm{scal}} \, g}(x^i \pr_i ,N) d\mu_\gd.
\end{aligned} %\label{}
\end{align*}
In the following we rewrite $\mathbf{E}$ into $\mathbf{E}_{\mathrm{ADM}}^{\mathrm{loc}}$, where we eased notation. Using the null structure equations \eqref{EQGaussEquation} and \eqref{EQgausscodazzinonlinear1}, and the relations \eqref{EQCHIkrelations1222} and \eqref{EQrelationZETAk}, we can write
\begin{align} 
\begin{aligned} 
&\rho + r \Divd {\be}\\
 =& -\lrpar{K+ \frac{1}{4}\trchi \trchib - \half \lrpar{\chih,\chibh}} - r \Divd \lrpar{\Divd \chih -\half \di \trchi + \chih \cdot \eta - \half \trchi \eta}\\
=&-\lrpar{K+ \frac{1}{4}\lrpar{-\tr\kd+\tr\Th} \lrpar{-\tr\kd-\tr\Th} - \half \lrpar{-\widehat{\kd}+\widehat{\Th},-\widehat{\kd}-\widehat{\Th}}} \\
&- r\Divd \lrpar{\Divd \lrpar{-\widehat{\kd}+\widehat{\Th}} -\half \di \lrpar{-\tr\kd+\tr\Th} }\\
&- r\Divd \lrpar{\lrpar{-\widehat{\kd}+\widehat{\Th}} \cdot \kd - \half \lrpar{-\tr\kd+\tr\Th}  \kd}\\
=& -\lrpar{K- \frac{1}{4} \lrpar{\tr\Th}^2 + \half \vert \widehat{\Th} \vert^2 + \frac{1}{4} \lrpar{\tr\kd}^2 + \half \vert \widehat{\kd} \vert^2}\\
&- r\Divd \lrpar{\Divd \lrpar{-\widehat{\kd}+\widehat{\Th}} -\half \di \lrpar{-\tr\kd+\tr\Th} }\\
&- r\Divd \lrpar{\lrpar{-\widehat{\kd}+\widehat{\Th}} \cdot \kd - \half \lrpar{-\tr\kd+\tr\Th}  \kd}.
\end{aligned} \label{EQrhoExpression55}
\end{align}
Plugging the Gauss equation \eqref{EQchrkl123} into the right-hand side of \eqref{EQrhoExpression55} leads to
\begin{align*} 
\begin{aligned} 
\rho + r \Divd \be =& -\half \lrpar{R_{\mathrm{scal}}-2\RRRic_{NN} + \frac{1}{2} \lrpar{\tr\kd}^2 + \vert \widehat{\kd} \vert^2 + \vert \widehat{\Th}\vert^2 } \\
&- r\Divd \lrpar{\Divd \lrpar{-\widehat{\kd}+\widehat{\Th}} -\half \di \lrpar{-\tr\kd+\tr\Th} }\\
&- r\Divd \lrpar{\lrpar{-\widehat{\kd}+\widehat{\Th}} \cdot \kd - \half \lrpar{-\tr\kd+\tr\Th}  \kd}.
\end{aligned}
\end{align*}
Hence we get that
\begin{align*} 
\begin{aligned} 
-\frac{8\pi}{\sqrt{4\pi}}\mathbf{E} :=& \lrpar{r^3 \lrpar{\rh+r\Divd \be}}^{(0)}\\
=& -\lrpar{\frac{r^3}{2}\lrpar{R_{\mathrm{scal}}-2\RRRic_{NN} + \frac{1}{2} \lrpar{\tr\kd}^2 + \vert \widehat{\kd} \vert^2+ \vert \widehat{\Th}\vert^2} }^{(0)} \\
&- \lrpar{r^4 \Divd \lrpar{\Divd \lrpar{-\widehat{\kd}+\widehat{\Th}} -\half \di \lrpar{-\tr\kd+\tr\Th} }}^{(0)}\\
&-  \lrpar{r^4 \Divd \lrpar{\lrpar{-\widehat{\kd}+\widehat{\Th}} \cdot \kd - \half \lrpar{-\tr\kd+\tr\Th}  \kd}}^{(0)},
\end{aligned} %\label{}
\end{align*} 
which we can estimate by \eqref{EQSMALLNESSrescaled13000} and Lemma \ref{LEMnonlinearFourier} for $R\geq1$ sufficiently large and \eqref{EQlocalisedADMcharges} as
\begin{align*} 
\begin{aligned} 
\mathbf{E} =& \frac{\sqrt{4\pi}}{8\pi} \lrpar{\frac{r^3}{2}\lrpar{R_{\mathrm{scal}}-2\RRRic_{NN} }}^{(0)}+\OO\lrpar{\frac{M}{R} \varep_R}+\OO(\varep_R^2) \\
=& \frac{1}{16\pi} \frac{1}{r^2} \int\limits_{S_{-1,1}} r^3 \lrpar{R_{\mathrm{scal}}-2\RRRic_{NN} } d\mu_{r^2\gac}+\OO\lrpar{\frac{M}{R} \varep_R}+\OO(\varep_R^2) \\
=& -\frac{1}{8\pi} \int\limits_{S_{-1,1}} \lrpar{\RRRic - \half R_{\mathrm{scal}} \, g }(r N,N) d\mu_{\gd} +\OO\lrpar{\frac{M}{R} \varep_R}+\OO(\varep_R^2) \\
=& -\frac{1}{8\pi} \int\limits_{S_{-1,1}} \lrpar{\RRRic - \half R_{\mathrm{scal}} \, g }(x^j \pr_j,N) d\mu_{\gd} +\OO\lrpar{\frac{M}{R} \varep_R}+\OO(\varep_R^2) \\
=& \mathbf{E}^{\mathrm{loc}}_{\mathrm{ADM}}+\OO\lrpar{\frac{M}{R} \varep_R}+\OO(\varep_R^2),
\end{aligned} %\label{}
\end{align*}
where we used Lemma \ref{LEMschwarzschildtrivialEstimate}. This finishes the proof of \eqref{EQEcomparison11}.
%%%%%%%%%%%%%%%%%%%%%%%%%%%%%%%%%%%%%%%%
\subsection{Comparison of $\mathbf{P}$ and $\mathbf{P}_{\mathrm{ADM}}^{\mathrm{loc}}$.} \label{SECcomparisonLINEAR} 
\ni In this section we prove that for $i=1,2,3$ and $(m_1,m_2,m_3) = (1,-1,0)$,
\begin{align} 
\begin{aligned} 
\lrpar{\mathbf{P}^{\mathrm{loc}}_{\mathrm{ADM}}}^i(S_{-1,1},g,k)= \mathbf{P}^{m_i}({}^{(R)}x_{-1,1})+\OO\lrpar{\frac{M}{R}\varep_R}+ \OO(\varep_R^2),
\end{aligned} \label{EQPcomparison11}
\end{align}
where we recall from Definitions \ref{DEFlocalCharges} and \ref{DEFlocalADMcharges} that
\begin{align} 
\begin{aligned} 
\lrpar{\mathbf{P}^{\mathrm{loc}}_{\mathrm{ADM}}}^i(S_{-1,1},g,k) :=& \frac{1}{8\pi} \int\limits_{S_{-1,1}}  \lrpar{k_{il}-\tr k \, g_{il}} N^l d\mu_\gd, \\
\mathbf{P}^m\lrpar{{}^{(R)}x_{-1,1}} :=& -\frac{1}{8\pi} \sqrt{\frac{4\pi}{3}} \lrpar{ r^3 \lrpar{\rho + r \Divd {\be}}}^{(1m)}\lrpar{{}^{(R)}x_{-1,1}}.
\end{aligned} \label{EQdefLINEARMOMENTA7778221}
\end{align}

\ni In the following we rewrite $\lrpar{\mathbf{P}_{\mathrm{ADM}}^{\mathrm{loc}}}^i$ into $\mathbf{P}^{m_i}$, where we eased notation. First, by \eqref{EQSMALLNESSrescaled13000}, for $R\geq1$ sufficiently large, it holds that on the annulus $A_{[1,3]}$,
\begin{align*} 
\begin{aligned} 
(\pr_i)^j - (\nab x^i)^j = e^{ij} - g^{ij} = \OO\lrpar{\frac{M}{R}} +\OO(\varep_R).
\end{aligned} %\label{}
\end{align*}
Hence we can write, using that $N$ is normal to $S_{-1,1}$,
\begin{align}
\begin{aligned} 
&\int\limits_{S_{-1,1}} \lrpar{k_{iN}-\tr k \, g_{iN}} d\mu_\gd \\
=& \int\limits_{S_{-1,1}}  \lrpar{k_{jN}-\tr k\, g_{jN}} \, (\nab x^i)^j d\mu_\gd +\OO\lrpar{\frac{M}{R}\varep_R} +O\lrpar{\varep_R^2} \\
=& \int\limits_{S_{-1,1}}  \lrpar{k_{jN}-\tr k \, g_{jN}}\, \lrpar{N(x^i) N^j + \lrpar{\Nd x^i}^j} d\mu_\gd +\OO\lrpar{\frac{M}{R}\varep_R}+O\lrpar{\varep_R^2} \\
=& \int\limits_{S_{-1,1}}  \lrpar{\lrpar{ k_{NN} - \tr k} N(x^i) + k(N,\Nd x^i)} d\mu_\gd+\OO\lrpar{\frac{M}{R}\varep_R}+O\lrpar{\varep_R^2}.
\end{aligned} \label{EQPfirstidentitycomparison1}
\end{align}
Using that by \eqref{EQdefinitionSpacelike223} and \eqref{EQrelationZETAk}, for $R\geq1$ sufficiently large we have on $S_{-1,1}$,
\begin{align*} 
\begin{aligned} 
k_{NN} - \tr k =& - \tr \kd = \half \lrpar{\trchib + \trchi}, \\
k(N,\Nd x^i) =& \zeta(x^i), \\
N(x^i) =& \frac{x^i}{r} +\OO\lrpar{\frac{M}{R}}+\OO\lrpar{\varep_R},
\end{aligned} %\label{}
\end{align*}
where $r$ denotes the area radius on $(S_{-1,1}, \gd)$, we get from \eqref{EQPfirstidentitycomparison1} that
\begin{align} 
\begin{aligned} 
&\int\limits_{S_{-1,1}} \lrpar{k_{iN}-\tr k \, g_{iN}} d\mu_\gd \\
 =&\int\limits_{S_{-1,1}}  \lrpar{\half \lrpar{\trchib + \trchi} \frac{x^i}{r} + \zeta(x^i) } d\mu_\gd+\OO\lrpar{\frac{M}{R}\varep_R}+\OO\lrpar{\varep_R^2} \\
=&\int\limits_{S_{-1,1}}  \lrpar{\half \lrpar{\trchib + \trchi} -r \Divd \zeta } \frac{x^i}{r}d\mu_\gd+\OO\lrpar{\frac{M}{R}\varep_R}+\OO\lrpar{\varep_R^2}\\
=&\int\limits_{S_{-1,1}}  \lrpar{\half \lrpar{\trchib + \trchi} -r \Divd \zeta } \frac{x^i}{r}d\mu_{r^2\gac}+\OO\lrpar{\frac{M}{R}\varep_R}+\OO\lrpar{\varep_R^2} \\
=&\int\limits_{S_{-1,1}}  \lrpar{\half \lrpar{\trchib - \trchi}+\trchi -r \Divd \zeta } \lrpar{\sqrt{\frac{4\pi}{3}}Y^{(1m_i)}} d\mu_{r^2\gac}\\&+\OO\lrpar{\frac{M}{R}\varep_R}+\OO\lrpar{\varep_R^2},
\end{aligned} \label{EQlinmomrelating1}
\end{align}
where we integrated by parts in the second equality and used that by \eqref{EQREMchangeitom}, for $i=1,2,3$,
\begin{align*} 
\begin{aligned} 
\frac{x^i}{\vert x \vert} = \sqrt{\frac{4\pi}{3}}Y^{(1m_i)} \text{ for } (m_1,m_2,m_3)=(1,-1,0).
\end{aligned} %\label{}
\end{align*}

\ni In the following we use two identities to rewrite the right-hand side of \eqref{EQlinmomrelating1}. First, we can rewrite
\begin{align} 
\begin{aligned} 
\frac{1}{2} \lrpar{\trchib-\trchi} =& \frac{r}{4} \trchi \trchib - \frac{1}{r} -\frac{r}{4} \lrpar{\trchi-\frac{2}{r}} \lrpar{\trchib+\frac{2}{r}} \\
=&r \lrpar{- \rh - K + \half(\chih,\chibh)} - \frac{1}{r} -\frac{r}{4} \lrpar{\trchi-\frac{2}{r}} \lrpar{\trchib+\frac{2}{r}},
\end{aligned} \label{EQrelation1linmom1}
\end{align}
where we used the Gauss equation \eqref{EQGaussEquation}.

Second, we can express
\begin{align} 
\begin{aligned} 
r \Divd \zeta =& \frac{r^2}{2} \Divd \lrpar{\trchi \zeta - \lrpar{\trchi- \frac{2}{r}} \zeta} \\
=& r^2 \Divd \lrpar{\half \trchi \zeta} - \frac{r^2}{2} \Divd \lrpar{\lrpar{\trchi- \frac{2}{r}} \zeta}\\
=&r^2 \Divd \lrpar{\Divd \chih -\half \di \tr \chi + \chih \cdot \zeta+\be} - \frac{r^2}{2} \Divd \lrpar{\lrpar{\trchi- \frac{2}{r}} \zeta}\\
=& r^2 \Divd \be - \frac{r^2}{2} \Ld \tr\chi + r^2 \lrpar{\Divd \Divd \chih + \Divd \lrpar{\chih \cdot \zeta}} \\
&- \frac{r^2}{2} \Divd \lrpar{\lrpar{\trchi- \frac{2}{r}} \zeta}\\
=& r^2 \Divd \be - \frac{1}{2} \Ldo \tr\chi + \frac{1}{2} \lrpar{\Ldo-r^2 \Ld}\trchi \\
& + r^2 \lrpar{\Divd \Divd \chih + \Divd \lrpar{\chih \cdot \zeta}} - \frac{r^2}{2} \Divd \lrpar{\lrpar{\trchi- \frac{2}{r}} \zeta},
\end{aligned} \label{EQrelation1linmom2}
\end{align}
where we used the Gauss-Codazzi equations \eqref{EQgausscodazzinonlinear1}.

Plugging \eqref{EQrelation1linmom1} and \eqref{EQrelation1linmom2} into the right-hand side of \eqref{EQlinmomrelating1} and subsequently \eqref{EQdefLINEARMOMENTA7778221}, and using from Appendix \ref{SECellEstimatesSpheres} that for any scalar function $f$, for $m=-1,0,1$,
\begin{align*} 
\begin{aligned} 
 \half \lrpar{\Ldo f}^{(1m)} = - f^{(1m)},
\end{aligned} %\label{}
\end{align*}
we get that for $R\geq1$ sufficiently large,
\begin{align} 
\begin{aligned} 
&8 \pi \cdot \lrpar{\mathbf{P}_{\mathrm{ADM}}^{\mathrm{loc}}}^i\\
=&\int\limits_{S_{-1,1}} \lrpar{k_{iN}-\tr k \, g_{iN}} d\mu_\gd  \\
=&\int\limits_{S_{-1,1}}  \lrpar{\half \lrpar{\trchib - \trchi}+\trchi -r \Divd \zeta } {\sqrt{\frac{4\pi}{3}}Y^{(1m_i)}} d\mu_{r^2\gac}+\OO\lrpar{\frac{M}{R}\varep_R}+\OO\lrpar{\varep_R^2} \\
=& \int\limits_{S_{-1,1}} \lrpar{-r\rh -r^2 \Divd \beta +\trchi+ \frac{1}{2} \Ldo \trchi } {\sqrt{\frac{4\pi}{3}}Y^{(1m_i)}} d\mu_{r^2\gac} \\&+\OO\lrpar{\frac{M}{R}\varep_R}+\OO\lrpar{\varep_R^2}+ \RR \\
=& \int\limits_{S_{-1,1}} \lrpar{-r\rh-r^2 \Divd \beta} {\sqrt{\frac{4\pi}{3}}Y^{(1m_i)}} d\mu_{r^2\gac} +\OO\lrpar{\frac{M}{R}\varep_R}+\OO\lrpar{\varep_R^2} + \RR\\
=&-r^3\sqrt{\frac{4\pi}{3}} \lrpar{\rh+ r\Divd \beta}^{(1m_i)} +\OO\lrpar{\frac{M}{R}\varep_R}+\OO\lrpar{\varep_R^2}+ \RR\\
=& 8\pi \cdot \mathbf{P}^{m_i} +\OO\lrpar{\frac{M}{R}\varep_R}+\OO\lrpar{\varep_R^2}+ \RR,
\end{aligned} \label{EQrelationlinmom444}
\end{align}
where the remainder term $\RR$ is given by
\begin{align*} 
\begin{aligned} 
\RR=& \int\limits_{S_{-1,1}} \lrpar{- rK + \frac{r}{2} (\chih,\chibh)- \frac{1}{r} - \frac{r}{4} \lrpar{\trchi-\frac{2}{r}} \lrpar{\trchib+\frac{2}{r}}} {\sqrt{\frac{4\pi}{3}}Y^{(1m_i)}} d\mu_{r^2\gac} \\
&-\int\limits_{S_{-1,1}} \lrpar{\frac{1}{2} \lrpar{\Ldo-r^2 \Ld} \lrpar{ \tr \chi-\frac{2}{r}}+ \frac{r^2}{2} \Divd \Divd \chih} {\sqrt{\frac{4\pi}{3}}Y^{(1m_i)}} d\mu_{r^2\gac} \\
&- \int\limits_{S_{-1,1}} \lrpar{r^2 \Divd\lrpar{ \chih \cdot \zeta} - \frac{r^2}{2} \Divd\lrpar{{ \lrpar{\trchi -\frac{2}{r}} \zeta}}} {\sqrt{\frac{4\pi}{3}}Y^{(1m_i)}} d\mu_{r^2\gac}.
\end{aligned} %\label{}
\end{align*}
To conclude \eqref{EQPcomparison11}, it remains to show that for $R\geq1$ sufficiently large,
\begin{align} 
\begin{aligned} 
\RR= \OO(\varep^{2}_R)+\OO\lrpar{\frac{M}{R}\varep_R}.
\end{aligned} \label{EQcontrolRemainderTerm}
\end{align}
Indeed, we first have by Lemma \ref{LEMGaussSquareLoneEstimate} that
\begin{align*} 
\begin{aligned} 
\int\limits_{S_{-1,1}} K \cdot Y^{(1m_i)} d\mu_{r^2 \gac} = \OO(\varep_R^2).
\end{aligned} %\label{}
\end{align*}
Second,
\begin{align*} 
\begin{aligned} 
&\int\limits_{S_{-1,1}} \lrpar{\trchi-\frac{2}{r}} \lrpar{\trchib+\frac{2}{r}} Y^{(1m_i)} d\mu_{r^2\gac} \\
=& \int\limits_{S_{-1,1}} \lrpar{\trchi-\frac{2\Om_M}{r}} \lrpar{\trchib+\frac{2\Om_M}{r}} Y^{(1m_i)} d\mu_{r^2\gac} \\
&+ \int\limits_{S_{-1,1}} \lrpar{\frac{2}{r}(\Om_M-1)\lrpar{\trchib+\frac{2\Om_M}{r}} -\frac{2}{r}(\Om_M-1)\lrpar{\trchi-\frac{2\Om_M}{r}}}Y^{(1m_i)} d\mu_{r^2\gac} \\
=&\OO\lrpar{\varep_R^2}+\OO\lrpar{\frac{M}{R}\varep_R},
\end{aligned} %\label{}
\end{align*}
where we used that $\Om_M$ is spherically symmetric, and that, for $M/R$ sufficiently small, 
\begin{align*} 
\begin{aligned} 
\Om_M -1 := \sqrt{1-\frac{2M/R}{r}}-1 = \OO\lrpar{\frac{M}{R}},
\end{aligned} %\label{}
\end{align*} 
The remaining integrands in $\RR$ are estimated similarly (using also Lemma \ref{LEMnonlinearFourier}), yielding \eqref{EQcontrolRemainderTerm}. This finishes the proof of \eqref{EQPcomparison11}.

%%%%%%%%%%%%%%%%%%%%%%%%%%%%%%%%%%%%%%%%
\subsection{Comparison of $\mathbf{L}$ and $\mathbf{L}_{\mathrm{ADM}}^{\mathrm{loc}}$.} \label{SECcomparisonANGULAR}
\ni In this section we prove that for $i=1,2,3$ and $(m_1,m_2,m_3) = (1,-1,0)$,
\begin{align} 
\begin{aligned} 
\mathbf{L}^{m_i}({}^{(R)} x_{-1,1}) =&\lrpar{ \mathbf{L}^{\mathrm{loc}}_{\mathrm{ADM}}}^i(S_{-1,1},g,k)  + \OO\lrpar{\frac{M}{R} \varep_R}+ \OO(\varep_R^2),
\end{aligned} \label{EQLcomparison11}
\end{align}
where we recall from Definitions \ref{DEFlocalCharges} and \ref{DEFlocalADMcharges} that
\begin{align*} 
\begin{aligned} 
\mathbf{L}^m({}^{(R)}x_{-1,1}) :=& \frac{1}{16\pi} \sqrt{\frac{8\pi}{3}} \lrpar{ r^3  \lrpar{ \di \trchi + \trchi (\eta-\di\log\Om) }}_H^{(1m)}({}^{(R)}x_{-1,1}), \\ 
\lrpar{\mathbf{L}^{\mathrm{loc}}_{\mathrm{ADM}}}^i(S_{-1,1},g,k) :=& \frac{1}{8\pi} \int\limits_{S_{-1,1}} (k_{jl}-\tr k \, g_{jl}) \lrpar{Y_{(i)}}^j N^l d\mu_\gd.
\end{aligned} %\label{}
\end{align*}

\ni On the one hand we have, using that $\di\Om_M=0$,
\begin{align} 
\begin{aligned} 
\mathbf{L}^{m_i} :=&  \frac{1}{16\pi} \sqrt{\frac{8\pi}{3}} r^3 \lrpar{\di \trchi + \trchi \cdot \lrpar{\eta-\di\log\Om}}^{(1m_i)}_H \\
=&  \frac{1}{16\pi} \sqrt{\frac{8\pi}{3}} r^3 \lrpar{\trchi \cdot \eta }^{(1m_i)}_H \\
=&   \frac{1}{8\pi} \sqrt{\frac{8\pi}{3}} r^2 {\eta }^{(1m_i)}_H + \frac{1}{16\pi} \sqrt{\frac{8\pi}{3}} r^3\lrpar{\lrpar{\trchi-\frac{2}{r}} \cdot \eta }^{(1m_i)}_H \\
=& \frac{1}{8\pi} \sqrt{\frac{8\pi}{3}} r^2 {\eta }^{(1m_i)}_H + \OO\lrpar{\frac{M}{R}\varep_R}+ \OO(\varep_R^2),
\end{aligned} \label{EQLestimateADM123}
\end{align}
where we used that for any scalar function $f$,
\begin{align*} 
\begin{aligned} 
(\di f )_H = 0.
\end{aligned} %\label{}
\end{align*}
On the other hand,
\begin{align} 
\begin{aligned} 
\lrpar{ \mathbf{L}^{\mathrm{loc}}_{\mathrm{ADM}}}^i :=& \frac{1}{8\pi} \int\limits_{S_{-1,1}} \lrpar{k_{jl}-\tr k \, g_{jl}} \lrpar{Y_{(i)}}^j N^l d\mu_\gd \\
=& \frac{1}{8\pi} \int\limits_{S_{-1,1}} k_{jN} \lrpar{Y_{(i)}}^j d\mu_\gd-\frac{1}{8\pi} \int\limits_{S_{-1,1}} \tr k \, \underbrace{g\lrpar{Y_{(i)}, N}}_{=0} d\mu_\gd \\
=& \frac{1}{8\pi} \int\limits_{S_{-1,1}} \eta_A \lrpar{Y_{(i)}}^A d\mu_{r^2\gac} +\OO(\varep_R^2)\\
=&  \frac{1}{8\pi} \sqrt{\frac{8\pi}{3}} r^2 \eta_H^{(1m_i)} +\OO(\varep_R^2),
\end{aligned} \label{EQLestimateADM1232}
\end{align}
where we used \eqref{EQrelationZETAk} and that the rotation fields $Y_{(i)}$, $i=1,2,3$, are $S_{-1,1}$-tangential and can, by Lemma \ref{LEMgeometricIDENTITIES}, be related to $H^{(1m)}$, $m=-1,0,1$, as follows
\begin{align*} 
\begin{aligned} 
Y_{(i)} = \sqrt{\frac{8\pi}{3}}\vert x \vert^2 H^{(1m_i)} \text{ with } (m_1,m_2,m_3) = (1,-1,0).
\end{aligned} %\label{}
\end{align*}
Putting together \eqref{EQLestimateADM123} and \eqref{EQLestimateADM1232} finishes the proof of \eqref{EQLcomparison11}.
%%%%%%%%%%%%%%%%%%%%%%%%%%%%%%%%%%%%%%%%
\subsection{Expression of $\mathbf{G}$ in terms of $\mathbf{C}_{\mathrm{ADM}}^{\mathrm{loc}}$ and $\mathbf{P}_{\mathrm{ADM}}^{\mathrm{loc}}$.} \label{SECcomparisonCENTER}
In this section we prove that for $i=1,2,3$ and $(m_1,m_2,m_3) = (1,-1,0)$,
\begin{align} 
\begin{aligned} 
\mathbf{G}^m({}^{(R)} x_{-1,1}) =& \lrpar{{\mathbf{C}}_{\mathrm{ADM}}^{\mathrm{loc}}}^{i_m}(S_{-1,1},g,k) - r(S_{-1,1},g,k) \cdot \lrpar{\mathbf{P}^{\mathrm{loc}}_{\mathrm{ADM}}}^{i_m}(S_{-1,1},g,k) \\
&+ \OO\lrpar{\frac{M}{R}\varep_R} + \OO\lrpar{\varep^2_R},
\end{aligned} \label{EQCcomparison11}
\end{align}
where $r(S_{-1,1},g,k)$ denotes the area radius of $(S_{-1,1},\gd)$ and we recall from Definitions \ref{DEFlocalCharges} and \ref{DEFlocalADMcharges} that
\begin{align*} 
\begin{aligned} 
\mathbf{G}^m\lrpar{{}^{(R)}x_{-1,1}}:=&  \frac{1}{16\pi}\sqrt{\frac{4\pi}{3}} \lrpar{ r^3 \lrpar{ \di \trchi + \trchi (\eta-\di\log\Om) }}^{(1m)}_E\lrpar{{}^{(R)}x_{-1,1}}, \\
\lrpar{\mathbf{C}^{\mathrm{loc}}_{\mathrm{ADM}}}^i(S_{-1,1},g,k):=& \frac{1}{16\pi} \int\limits_{S_{-1,1}} \lrpar{\RRRic - \half R_{\mathrm{scal}} \, g}(Z^{(i)},N) d\mu_\gd,
\end{aligned} %\label{}
\end{align*}
where the vectorfields $Z^{(i)}$, $i=1,2,3$, are defined by
\begin{align*} 
\begin{aligned} 
Z^{(i)}:=& \vert x \vert^2 \pr_i - 2x^i \vert x \vert \pr_r.
\end{aligned} %\label{}
\end{align*}

\ni Consider first ${\mathbf{C}}_{\mathrm{ADM}}^{\mathrm{loc}}$. By \eqref{EQREMdecompositionZIapp} the vectorfields $Z^{(i)}$, $i=1,2,3$, can be expressed as
\begin{align*} 
\begin{aligned} 
Z^{(i)} = - \vert x \vert^3 \sqrt{\frac{8\pi}{3}} E^{(1m_i)} - \vert x \vert^2 \lrpar{\sqrt{\frac{4\pi}{3}}Y^{(1m_i)}} \pr_r,
\end{aligned} %\label{}
\end{align*}
Hence we have that, using that $g(E^{(1m)}, N)=0$,
\begin{align} 
\begin{aligned} 
16\pi \lrpar{{\mathbf{C}}_{\mathrm{ADM}}^{\mathrm{loc}}}^i =&  \int\limits_{S_{-1,1}}  \RRRic\lrpar{- \vert x\vert^3 \sqrt{\frac{8\pi}{3}} E^{(1m_i)}, N} d\mu_\gd \\
&+ \int\limits_{S_{-1,1}}  \lrpar{\RRRic -\half R_{\mathrm{scal}} \, g}\lrpar{-\vert x\vert^2 \sqrt{\frac{4\pi}{3}}Y^{(1m_i)}\pr_r, N}  d\mu_\gd,\\
=& - \sqrt{\frac{8\pi}{3}} \int\limits_{S_{-1,1}} \vert x\vert^3 \RRRic\lrpar{E^{(1m_i)}, N} d\mu_\gd \\
&-\sqrt{\frac{4\pi}{3}} \int\limits_{S_{-1,1}} \vert x\vert^2 \lrpar{\RRRic - \half R_{\mathrm{scal}} \, g}(\pr_r, N) Y^{(1m_i)} d\mu_\gd,
\end{aligned} \label{EQG1rewriting}
\end{align}
By the Gauss-Codazzi equations \eqref{EQchrkl123} the second integral on the right-hand side of \eqref{EQG1rewriting} can be expressed as
\begin{align} 
\begin{aligned} 
&\int\limits_{S_{-1,1}} \vert x\vert^2\lrpar{\RRRic - \half R_{\mathrm{scal}} \,g}(\pr_r, N) Y^{(1m_i)} d\mu_\gd \\
=& \int\limits_{S_{-1,1}}\vert x\vert^4 \lrpar{\RRRic- \half R_{\mathrm{scal}} \,g}(N, N) Y^{(1m_i)} d\mu_\gac +\OO\lrpar{\frac{M}{R}\varep_R}+ \OO(\varep_R^2) \\
=& \int\limits_{S_{-1,1}}\vert x\vert^4 \lrpar{-K+\frac{1}{4}(\tr\Th)^2 + \vert \widehat{\Th} \vert^2} Y^{(1m_i)} d\mu_\gac +\OO\lrpar{\frac{M}{R}\varep_R}+ \OO(\varep_R^2) \\
=& -\underbrace{\int\limits_{S_{-1,1}}\vert x\vert^4 K \cdot Y^{(1m_i)} d\mu_\gac}_{:=\II_1} +\frac{1}{4} \underbrace{\int\limits_{S_{2}}\vert x\vert^4 (\tr\Th)^2 \cdot Y^{(1m_i)} d\mu_\gac}_{:=\II_2} +\OO\lrpar{\frac{M}{R}\varep_R}+ \OO(\varep_R^2),
\end{aligned} \label{EQtrTheta2rewrite000100007778}
\end{align}
where we used Lemma \ref{LEMschwarzschildtrivialEstimate}. In the following we analyse $\II_1$ and $\II_2$. First, by \eqref{EQsmallnesssphere11} and Lemma \ref{LEMGaussSquareLoneEstimate} we have that for $R\geq1$ sufficiently large,
\begin{align} 
\begin{aligned} 
\II_1 = \vert x \vert^4 K^{(1m_i)} = \OO(\varep_R^2).
\end{aligned} \label{EQtrTheta2rewrite0001007778}
\end{align}
Second, by the relation $Y^{(1m)} = \frac{1}{\sqrt{2}} \Divdo E^{(1m)}$, integration by parts, and the Gauss-Codazzi equations \eqref{EQchrkl123}, we have that
\begin{align} 
\begin{aligned} 
\II_2:=& \int\limits_{S_{-1,1}} \vert x\vert^4 (\tr\Th)^2 \cdot Y^{(1m_i)} d\mu_{\gac}\\
=& -\frac{4}{\sqrt{2}} \int\limits_{S_{-1,1}}\vert x\vert^4 \cdot \tr\Th \cdot \gac\lrpar{\half\di \tr\Th, E^{(1m_i)}} d\mu_{\gac} \\ 
=& -\frac{8}{\sqrt{2}} \int\limits_{S_{-1,1}}\vert x\vert^3 \cdot\gac\lrpar{\half\di \tr\Th-\Divd \widehat{\Th}, E^{(1m_i)}} d\mu_{\gac}+\OO\lrpar{\frac{M}{R}\varep_R}+ \OO(\varep_R^2) \\
=& \frac{8}{\sqrt{2}} \int\limits_{S_{-1,1}}\vert x\vert^3 \cdot \gd\lrpar{\Divd \widehat{\Th}-\half\di \tr\Th, E^{(1m_i)}} d\mu_{\gd}+\OO\lrpar{\frac{M}{R}\varep_R}+ \OO(\varep_R^2) \\
=& \frac{8}{\sqrt{2}} \int\limits_{S_{-1,1}}\vert x\vert^3 \cdot \RRRic(E^{(1m_i)},N) d\mu_{\gd}+\OO\lrpar{\frac{M}{R}\varep_R}+ \OO(\varep_R^2), 
\end{aligned} \label{EQtrTheta2rewrite0001}
\end{align}
where we used that $\tr\Th-\frac{2}{\vert x \vert} = \OO\lrpar{\frac{M}{R}} + \OO(\varep_R)$ by \eqref{EQsmallnesssphere11} and Lemma \ref{LEMschwarzschildtrivialEstimate}. Plugging \eqref{EQtrTheta2rewrite0001007778} and \eqref{EQtrTheta2rewrite0001} into \eqref{EQtrTheta2rewrite000100007778} and subsequently into \eqref{EQG1rewriting}, we get that
\begin{align} 
\begin{aligned} 
16\pi \lrpar{{\mathbf{C}}_{\mathrm{ADM}}^{\mathrm{loc}}}^i  =& -2 \sqrt{\frac{8\pi}{3}} \int\limits_{S_{-1,1}} \vert x\vert^3 \cdot \RRRic(E^{(1m_i)}, N) d\mu_\gd+\OO\lrpar{\frac{M}{R}\varep_R} + \OO(\varep_R^2).
\end{aligned} \label{EQGcalcConcl1}
\end{align}

\ni Consider now $\mathbf{G}^m$. By Definition \ref{DEFlocalCharges}, \eqref{EQalternativeLGdef} and \eqref{EQGaugeChoicesSRscaling}, we have that
\begin{align} 
\begin{aligned} 
8\pi\sqrt{\frac{3}{8\pi}} \cdot \mathbf{G}^m = \lrpar{r^3 \lrpar{ \be + \Divd \chih + \chih \cdot \eta }}^{(1m)}_E = r^3 \be^{(1m)}_E + \OO(\varep_R^2), 
\end{aligned} \label{EQwriting1csd1}
\end{align}
where we used that by \eqref{EQsmallnesssphere11}, \eqref{EQCHIkrelations1222} and \eqref{EQrelationZETAk}, 
\begin{align*} 
\begin{aligned} 
r^3 \lrpar{\Divd \chih + \chih \cdot \eta }^{(1m)}_E = \OO(\varep_R^2).
\end{aligned}
\end{align*}

\ni Recalling the definition of $\mathbf{P}$ from Definition \ref{DEFlocalCharges} and applying \eqref{EQsmallnesssphere11}, it holds that
\begin{align*} 
\begin{aligned} 
- 8\pi\sqrt{\frac{3}{4\pi}} \frac{1}{r^3}\mathbf{P}^m :=& \lrpar{\rh+r\Divd \be}^{(1m)} \\
=&\rh^{(1m)} + \lrpar{\frac{1}{r}\Divdo \be}^{(1m)} + \OO(\varep_R^2)\\
=& \rh^{(1m)} + \lrpar{\frac{1}{r}\sqrt{2} \be_E^{(1m)}} + \OO(\varep_R^2).
\end{aligned} %\label{}
\end{align*}
In particular, $\be_E^{(1m)}$ can expressed as
\begin{align} 
\begin{aligned} 
\be_E^{(1m)} = -\frac{r}{\sqrt{2}} \lrpar{8\pi\sqrt{\frac{3}{4\pi}} \frac{1}{r^3} \mathbf{P}^m + \rh^{(1m)} }+\OO(\varep_R^2)
\end{aligned} \label{EQbetaRewriteP0001}
\end{align}
Plugging \eqref{EQbetaRewriteP0001} into the right-hand side of \eqref{EQwriting1csd1} yields
\begin{align} 
\begin{aligned} 
8\pi\sqrt{\frac{3}{8\pi}} \cdot \mathbf{G}^m =& r^3 \be^{(1m)}_E + \OO(\varep_R^2) \\
=& r^3 \lrpar{-\frac{r}{\sqrt{2}} \lrpar{8\pi\sqrt{\frac{3}{4\pi}} \frac{1}{r^3} \mathbf{P}^m + \rh^{(1m)} }} + \OO(\varep_R^2) \\
=&-8\pi\sqrt{\frac{3}{8\pi}} \cdot r \cdot \mathbf{P}^m - \frac{r^4}{\sqrt{2}} \rh^{(1m)} +\OO(\varep_R^2). 
\end{aligned} \label{EQGrewrite0002}
\end{align}
The second term on the right-hand side of \eqref{EQGrewrite0002} can be rewritten by the Gauss equation \eqref{EQGaussEquation}, application of \eqref{EQsmallnesssphere11} and \eqref{EQCHIkrelations1222}, and use of \eqref{EQtrTheta2rewrite0001} as follows,
\begin{align} 
\begin{aligned} 
\rh^{(1m)} =& \lrpar{-K-\frac{1}{4} \trchi \trchib + \half (\chih, \chibh)}^{(1m)}  \\
=& -\frac{1}{4} \lrpar{\trchi\trchib}^{(1m)} +\OO(\varep_R^2)\\
=& -\frac{1}{4} \lrpar{(-\tr \overline{k}+\tr\Th)(-\tr\Th-\tr \overline{k})}^{(1m)}+\OO(\varep_R^2)\\
=& \frac{1}{4} \lrpar{(\tr\Th)^2}^{(1m)} +\OO(\varep_R^2) \\
=& \frac{2}{\sqrt{2}} \int\limits_{S_{-1,1}}\frac{1}{\vert x\vert} \cdot \RRRic(E^{(1m_i)},N) d\mu_{\gd}+\OO\lrpar{\frac{M}{R}\varep_R}+ \OO(\varep_R^2).
\end{aligned} \label{EQGrewrite0003}
\end{align}
Plugging \eqref{EQGrewrite0003} into \eqref{EQGrewrite0002}, and using \eqref{EQPcomparison11} and \eqref{EQGcalcConcl1}, we get 
\begin{align*} 
\begin{aligned} 
8\pi \sqrt{\frac{3}{8\pi}} \mathbf{G}^m =& -8\pi\sqrt{\frac{3}{8\pi}} \cdot r \cdot \mathbf{P}^m - r^3 \int\limits_{S_2} \RRRic(N,E^{(1m)}) d\mu_\gd + \OO\lrpar{\frac{M}{R}\varep_R} + \OO\lrpar{\varep^2_R} \\
=& -8\pi\sqrt{\frac{3}{8\pi}} \cdot r \cdot \lrpar{\mathbf{P}_{\mathrm{ADM}}^{\mathrm{loc}}}^{i_m} + 8\pi\sqrt{\frac{3}{8\pi}} \lrpar{\mathbf{C}_{\mathrm{ADM}}^{\mathrm{loc}}}^{i_m} + \OO\lrpar{\frac{M}{R}\varep_R} + \OO\lrpar{\varep^2_R},
\end{aligned} 
\end{align*}
which can be rewritten as
\begin{align*} 
\begin{aligned} 
\mathbf{G}^m = \lrpar{\mathbf{C}_{\mathrm{ADM}}^{\mathrm{loc}}}^{i_m} -r \cdot \lrpar{\mathbf{P}_{\mathrm{ADM}}^{\mathrm{loc}}}^{i_m} + \OO\lrpar{\frac{M}{R}\varep_R} + \OO\lrpar{\varep^2_R}.
\end{aligned} %\label{}
\end{align*}
This finishes the proof of \eqref{EQCcomparison11}.

%%%%%%%%%%%%%%%%%%%%%%%%%%%%%%%%%%%%%%%%
\subsection{Conclusion of proof of Theorem \ref{PROPconstructionStatement}} \label{SECappProofPropositionSPHERESPACE} 

\ni In the following we conclude the proof of Theorem \ref{PROPconstructionStatement}. We first rescale the estimate \eqref{EQfullEstimateExtendedSpheredataCONST} of Section \ref{SECproofExistenceConstruction} and then calculate the asymptotic charges by using the estimates of Sections \ref{SECcomparisonENERGY}, \ref{SECcomparisonLINEAR}, \ref{SECcomparisonANGULAR} and \ref{SECcomparisonCENTER}.

We recall that in Theorem \ref{PROPconstructionStatement} we are working with strongly asymptotically flat spacelike initial data, hence we apply the estimates of Sections \ref{SECproofExistenceConstruction}, \ref{SECcomparisonENERGY}, \ref{SECcomparisonLINEAR}, \ref{SECcomparisonANGULAR} and \ref{SECcomparisonCENTER} with $\varep_R=\smallO(R^{-3/2})$.\\

\ni \textbf{Rescaling of estimates.} From \eqref{EQfullEstimateExtendedSpheredataCONST} we have that the constructed ingoing null data $({}^{(R)} {x}_{-1+[-\de,\de],1})$ satisfies, for $R\geq1$ sufficiently large, 
\begin{align*} 
\begin{aligned} 
\Vert {}^{(R)} {x}_{-1+[-\de,\de],1}-{\mathfrak{m}}^{M/R} \Vert_{\XX^+\lrpar{\HHb_{-1+[-\de,\de],1}}} = \smallO\lrpar{ R^{-3/2}}.
\end{aligned} %\label{}
\end{align*}
By Lemma \ref{LEMspheredataInvarianceSCale}, the rescaled ingoing null data 
$${x}_{-R+R[-\de,\de],R} := {}^{(R^{-1})}\lrpar{{}^{(R)} {x}_{-1+[-\de,\de],1}}$$
satisfies
\begin{align*} 
\begin{aligned} 
\Vert {x}_{-R+R[-\de,\de],R}-{\mathfrak{m}}^{M} \Vert_{\XX^+\lrpar{\HHb_{-R+R[-\de,\de],R}}} = \smallO\lrpar{ R^{-3/2}}.
\end{aligned} %\label{}
\end{align*}
To conclude that the family of ingoing null data $({x}_{-R+R[-\de,\de],R})$ is strongly asymptotically flat, it remains by Definition \ref{DEFadmissibleEXTSequences} to show that 
\begin{align} 
\begin{aligned} 
\Vert \beta^{[1]}(x_{-R,R})\Vert_{L^2(S_{-R,R})} =& \OO\lrpar{R^{-3}}.
\end{aligned} \label{EQadditionalSAFcondition1277788}
\end{align}
We claim that \eqref{EQadditionalSAFcondition1277788} follows from the finiteness of the charges $\mathbf{L}_\infty$ and $\mathbf{G}_\infty$ shown below. Indeed, by Definition \ref{DEFlocalCharges}, \eqref{EQalternativeLGdef}, Lemma \ref{LEMrescaleLOCALCHARGES} and \eqref{EQfullEstimateExtendedSpheredataCONST}, we have that for $R\geq1$ large,
\begin{align*} 
\begin{aligned} 
\mathbf{L}^m(x_{-R,R}) =& R^2 \cdot \mathbf{L}^m({}^{(R)}x_{-1,1}) \\
=& R^2 \cdot \lrpar{-r^3 \lrpar{\be + \Divd \chih + \chih \cdot \lrpar{\eta-\di \log \Om}}}^{(1m)}_H({}^{(R)}x_{-1,1}) \\
=& R^2 \cdot \lrpar{-\lrpar{r_{M/R}(-1,1)}^3 \be_H^{(1m)}({}^{(R)}x_{-1,1}) + \OO(R^{-3})}, \\
\mathbf{G}^m(x_{-R,R}) =& R^2 \cdot \mathbf{G}^m({}^{(R)}x_{-1,1}) \\
=& R^2 \cdot \lrpar{-r^3 \lrpar{\be + \Divd \chih + \chih \cdot \lrpar{\eta-\di \log \Om}}}^{(1m)}_E({}^{(R)}x_{-1,1}) \\
=& R^2 \cdot \lrpar{-\lrpar{r_{M/R}(-1,1)}^3 \be_E^{(1m)}({}^{(R)}x_{-1,1}) + \OO(R^{-3})}.
\end{aligned} %\label{}
\end{align*}
Hence by Definition \ref{DEFasymptoticCharges} and the finiteness of $\mathbf{L}_\infty$ and $\mathbf{G}_\infty$ (shown below) we deduce that
\begin{align} 
\begin{aligned} 
\vert \be^{(1m)}_H({}^{(R)}x_{-1,1}) \vert + \vert \be^{(1m)}_E({}^{(R)}x_{-1,1}) \vert = \OO(R^{-2}).
\end{aligned} \label{EQcontrolbeta1constructed777}
\end{align}
Therefore, by the expansion
\begin{align*} 
\begin{aligned} 
\be^{[1]}({}^{(R)}x_{-1,1}) = \sum\limits_{m=-1,0,1} \lrpar{\be^{(1m)}_H({}^{(R)}x_{-1,1}) \cdot H^{(1m)} + \be^{(1m)}_E({}^{(R)}x_{-1,1}) \cdot E^{(1m)}},
\end{aligned} %\label{}
\end{align*}
we get from \eqref{EQcontrolbeta1constructed777} and by the scaling of $\be$, see Lemma \ref{DEFscalingNULLSTRUC}, that
\begin{align*} 
\begin{aligned} 
\Vert \be^{[1]}(x_{-R,R}) \Vert^2_{L^2(S_{-R,R})} =& R^{-2} \cdot \Vert \be^{[1]}({}^{(R)}x_{-1,1}) \Vert^2_{L^2(S_{-1,1})} \\
\les& R^{-2}\cdot \sum\limits_{m=-1,0,1} \lrpar{\lrpar{\be^{(1m)}_H({}^{(R)}x_{-1,1})}^2 + \lrpar{\be^{(1m)}_E({}^{(R)}x_{-1,1})}^2} \\
=&\OO(R^{-6}).
\end{aligned} %\label{}
\end{align*}
This finishes the proof of \eqref{EQadditionalSAFcondition1277788}. It thus only remains to analyze the asymptotics of the charges $(\mathbf{E}, \mathbf{P}, \mathbf{L}, \mathbf{G})(x_{-R,R})$. 

%%%%%%%%%%%%%%%%%%%%%%%%%%%%%%%%%%%%%%%%
Consider first $\mathbf{E}$. By Lemmas \ref{LEMrescaleLOCALCHARGES} and \ref{LEMdecayPADMlocal}, and \eqref{EQLEMscalingADMlocal} and \eqref{EQEcomparison11} with $\varep= \smallO(R^{-3/2})$, it follows that
\begin{align*} 
\begin{aligned} 
\mathbf{E}(x_{-R,R}) =& R \cdot \mathbf{E}\lrpar{{}^{(R)}x_{-1,1}} \\
=& R \cdot \lrpar{\mathbf{E}^{\mathrm{loc}}_{\mathrm{ADM}}\lrpar{S_{-1,1},{}^{(R)}g,{}^{(R)}k} + \smallO(R^{-5/2})} \\
=& \mathbf{E}^{\mathrm{loc}}_{\mathrm{ADM}}\lrpar{S_{-R,R},g,k} + \smallO(R^{-3/2}) \\
=& \mathbf{E}_{\mathrm{ADM}} + \OO(R^{-1}).
\end{aligned} %\label{}
\end{align*}

\ni The analysis of $\mathbf{P}$ and $\mathbf{L}$ based on \eqref{EQPcomparison11}, \eqref{EQLcomparison11} and Lemma \ref{LEMdecayPADMlocal} is similar and yields for $m=-1,0,1$ and $(i_{-1},i_0,i_1)=(2,3,1)$,
\begin{align} 
\begin{aligned} 
\mathbf{P}^m(x_{-R,R}) =\OO(R^{-3/2}), \,\, \mathbf{L}^m(x_{-R,R})=\lrpar{{\mathbf{L}}_{\mathrm{ADM}}}^{i_m}+\smallO(1).
\end{aligned} \label{EQdecayPLstrongAFspacelikegluing999012322}
\end{align}

%%%%%%%%%%%%%%%%%%%%%%%%%%%%%%%%%%%%%%%%
\ni It remains to analyze $\mathbf{G}$. By Lemma \ref{LEMrescaleLOCALCHARGES} and \eqref{EQLEMscalingADMlocal}, \eqref{EQCcomparison11} and \eqref{EQdecayPLstrongAFspacelikegluing999012322}, we have, for $m=-1,0,1$ and $(i_{-1},i_0,i_1)=(2,3,1)$,
\begin{align*} 
\begin{aligned} 
\mathbf{G}^m(x_{-R,R}) =& R^2 \cdot \mathbf{G}^m\lrpar{{}^{(R)} x_{-1,1}}\\
=& R^2 \cdot \lrpar{\lrpar{\mathbf{C}_{\mathrm{ADM}}^{\mathrm{loc}}}^{i_m}\lrpar{S_{-1,1},{}^{(R)}g,{}^{(R)}k} }\\
&+ R^2 \cdot \lrpar{ -r\lrpar{S_{-1,1},{}^{(R)}g,{}^{(R)}k}  \cdot \lrpar{\mathbf{P}^{\mathrm{loc}}_{\mathrm{ADM}}}^{i_m}\lrpar{S_{-1,1},{}^{(R)}g,{}^{(R)}k}+ \smallO(R^{-5/2})} \\
=& \lrpar{\mathbf{C}_{\mathrm{ADM}}^{\mathrm{loc}}}^{i_m}\lrpar{S_{-R,R},g,k} -  r(S_{-R,R},g,k) \cdot \lrpar{\mathbf{P}^{\mathrm{loc}}_{\mathrm{ADM}}}^{i_m}\lrpar{S_{-R,R},g,k} +\smallO(R^{-1/2}) \\
=& \lrpar{\mathbf{C}_{\mathrm{ADM}}}^{i_m} + \smallO(1).
\end{aligned} %\label{}
\end{align*}
where we used that by Lemma \ref{LEMdecayPADMlocal}, for strongly asymptotically flat spacelike initial data, 
\begin{align*} 
\begin{aligned} 
\mathbf{P}_{\mathrm{ADM}}^{\mathrm{loc}}(S_r,g,k) =  \OO(r^{-3/2}).
\end{aligned} %\label{}
\end{align*}

\ni This finishes the proof of Theorem \ref{PROPconstructionStatement}.

%%%%%%%%%%%%%%%%%%%%%%%%%%%%%%%%%%%%%%%%
%%%%%%%%%%%%%%%%%%%%%%%%%%%%%%%%%%%%%%%%
\section{Kerr reference spacelike initial data and sphere data} \label{SECappKerrFamilyDetails} \ni  In this section we define Kerr reference spacelike initial data and sphere data, and prove estimates. We proceed as follows.

\begin{itemize}
\item In Section \ref{SECdefKerrParameter} we define the \emph{Kerr parameter map} $\KK_{\mathrm{par}}$ which maps Kerr parameters 
\begin{align*} 
\begin{aligned} 
\mathbf{p}:= (M,a,\La,\La', \mathbf{a}) \in \RRR \times \RRR \times \mathrm{SO(1,3)} \times \mathrm{SO(3)} \times \RRR^3,
\end{aligned} %\label{}
\end{align*}
to Kerr spacelike initial data $(g^{\mathbf{p}}, k^{\mathbf{p}})$, and prove preliminary estimates. 
\item In Section \ref{SECkerrChargeMapping} we define and prove estimates for the \emph{Kerr asymptotic invariants map} $\KK$ which maps asymptotic invariants vectors 
\begin{align*} 
\begin{aligned} 
\la := \lrpar{\mathbf{E}(\la),\mathbf{P}(\la),\mathbf{L}(\la),\mathbf{C}(\la)} \in I(0) \times \RRR^3 \times \RRR^3,
\end{aligned} %\label{}
\end{align*}
to Kerr parameters $(M,a,\La,\La', \mathbf{a}):=\KK(\la)$ such that the asymptotic invariants of the spacelike initial data $(g^{\mathbf{p}}, k^{\mathbf{p}})$ agree with the asymptotic invariants vector $\la$.

\item In Section \ref{SECKerrSphereDataDefinition} we construct Kerr reference sphere data $x_{-R,2R}^{\KK(\la)}$ from $(g^{\KK(\la)}, k^{\KK(\la)})$ by following the construction of Section \ref{SECstatementConstruction} with smallness parameter $\varep_R=\OO(R^{-3/2})$, see \eqref{EQREMgeneralisedEstimates}.

\item In Section \ref{SECchargeEstimateKERR444} we combine the estimates of Sections \ref{SECdefKerrParameter} and \ref{SECkerrChargeMapping} to relate the charges, local integrals and asymptotic invariants of the Kerr reference spacelike initial data and sphere data.

\end{itemize}

\ni The estimates proved in Section \ref{SECchargeEstimateKERR444} are essential for our characteristic gluing to Kerr in Section \ref{SECproofMainTheorem1}, see \eqref{EQasympDECAY55540002}.

%%%%%%%%%%%%%%%%%%%%%%%%%%%%%%%%%%%%%%%%
\subsection{Kerr parameter map and estimates} \label{SECdefKerrParameter} In this section we define the Kerr parameter map $\KK_{\mathrm{par}}$ and prove preliminary estimates. The construction of the map $\KK_{\mathrm{par}}$ follows Appendices E and F of \cite{ChruscielDelay}. \\

\ni \textbf{Definition of Kerr parameters.} We introduce the following \emph{Kerr parameters}.
\begin{itemize}
\item Let $M\in \RRR$ be the \emph{mass parameter} and $a \in \RRR$ be the \emph{normalized angular momentum parameter}. 
\item For a vector $v \in B(0,1) \subset \RRR^3$, define, with $v^2 := \vert v \vert^2$,
\begin{align} 
\begin{aligned} 
\ga := \frac{1}{\sqrt{1- v^2}},
\end{aligned} \label{EQdefinitionGAMMA}
\end{align}
and let $\La \in \mathrm{SO(1,3)}$ denote the \emph{Lorentz boost matrix}
\begin{align} 
\begin{aligned} 
\La = \begin{pmatrix}
\ga & -\ga v^1 & -\ga v^2  & -\ga v^3\\
-\ga v^1 & 1 +(\ga-1)\frac{(v^1)^2}{v^2} & (\ga-1)\frac{(v^1)(v^2)}{v^2} &  (\ga-1)\frac{(v^1)(v^3)}{v^2}\\
 -\ga v^2 & (\ga-1)\frac{(v^1)(v^2)}{v^2} & 1 +(\ga-1)\frac{(v^2)^2}{v^2} &  (\ga-1)\frac{(v^2)(v^3)}{v^2} \\
  -\ga v^3 & (\ga-1)\frac{(v^1)(v^3)}{v^2} & (\ga-1)\frac{(v^2)(v^3)}{v^2} & 1 +(\ga-1)\frac{(v^3)^2}{v^2} \end{pmatrix}.
\end{aligned} \label{EQdefinitionLAMBDA}
\end{align}
\item Let $\La' \in \mathrm{SO(3)}$ be a \emph{rotation matrix} of the coordinate system $(x^1,x^2,x^3)$ on $\RRR^3$.
\item Let $\mathbf{a} \in \RRR^3$ be a \emph{translation vector} of the coordinate system $(x^1,x^2,x^3)$ on $\RRR^3$.
\end{itemize}

\ni \textbf{Definition of Kerr parameter map.} The Kerr parameter map $\KK_{\mathrm{par}}$ is defined to map
\begin{align*} 
\begin{aligned} 
\KK_{\mathrm{par}}: \RRR \times \RRR \times \mathrm{SO(1,3)} \times \mathrm{SO(3)}\times \RRR^3 &\to C_{\mathrm{loc}}^{k_0}\lrpar{\RRR^3 \setminus B(0,1)} \times C_{\mathrm{loc}}^{k_0-1}\lrpar{\RRR^3 \setminus B(0,1)},\\
\mathbf{p}:=(M,a,\La, \La',\mathbf{a}) &\mapsto \lrpar{ g^{\mathbf{p}}_{ij},k^{\mathbf{p}}_{ij}},
\end{aligned} %\label{}
\end{align*}
where $\lrpar{ g^{\mathbf{p}}_{ij},k^{\mathbf{p}}_{ij}}$ denotes Kerr spacelike initial data on $\RRR^3 \setminus B(0,1)$ to be constructed below, and the tensor space $C^{k_0}_{\mathrm{loc}}(\RRR^3 \setminus B(0,1))\times C^{k_0-1}_{\mathrm{loc}}(\RRR^3 \setminus B(0,1))$ is defined in Section \ref{SECspacelikescaling}. \\

\ni \textbf{Construction of Kerr parameter map.} In the following we sketch the construction of the Kerr reference spacelike initial data $(g^{\mathbf{p}}_{ij}, k^{\mathbf{p}}_{ij})$, following Appendix F of \cite{ChruscielDelay}.

The construction is based on rotating, translating and boosting the hypersurface  $\{t=0\}$ (defined in Boyer-Lindquist coordinates), see the four steps below. In this context it is useful to define the so-called \emph{Poincar\'e charges} $p_\mu$ and $J_{\mu\nu}$ (see Appendix E in \cite{ChruscielDelay}).
\begin{definition}[Poincar\'e charges $p_\mu$ and $J_{\mu\nu}$] \label{DEFpoincareCHARGES} Given asymptotically flat spacelike initial data with asymptotic invariants
$$(\mathbf{E}_{\mathrm{ADM}}, \mathbf{P}_{\mathrm{ADM}},\mathbf{L}_{\mathrm{ADM}},\mathbf{C}_{\mathrm{ADM}}),$$
define the \emph{energy-momentum $4$-vector} $p_\mu$ and the \emph{Lorentz charges} $J_{\mu\nu}$ for $\mu,\nu=0,1,2,3$ by
\begin{align*} 
\begin{aligned} 
p_0 :=& \mathbf{E}_{\mathrm{ADM}}, & p_i :=& \lrpar{\mathbf{P}_{\mathrm{ADM}}}_i, \\ 
J_{i0} :=& \lrpar{\mathbf{C}_{\mathrm{ADM}}}_i, & J_{ij} :=& \in_{ijk} \lrpar{\mathbf{L}_{\mathrm{ADM}}}^k, & J_{\mu\nu}:=&-J_{\nu\mu}.
\end{aligned} 
\end{align*}
\end{definition}

\ni The Poincar\'e charges satisfy the following useful transformation law under translations and Lorentz boosts, see Proposition E.1 in \cite{ChruscielDelay} for a proof. 

\begin{proposition}[Transformation law for Poincar\'e charges, \cite{ChruscielDelay}] \label{REMARKtransformationlaw} The Poincar\'e charges transform under Lorentz boosts $\La$ and space translations $\mathbf{a}$ as follows,
\begin{align*} 
\begin{aligned} 
p'_\mu = \La_\mu^{\,\,\,\a} p_\a, \,\, 
J'_{\mu\nu} = \La_\mu^{\,\,\,\a} \La_\nu^{\,\,\,\be} J_{\a\be} + \mathbf{a}_\mu \La_\nu^{\,\,\,\a} p_\a- \mathbf{a}_\nu \La_\mu^{\,\,\,\a} p_\a.
\end{aligned} %\label{}
\end{align*}
\end{proposition}

\ni We are now in position to construct $(g^{\mathbf{p}}_{ij}, k^{\mathbf{p}}_{ij})$ in the following four steps.\\

\ni \textbf{Step 1.} For \emph{mass parameter} $M\in \RRR$ and \emph{normalized angular momentum parameter} $a \in \RRR$, the Kerr metric is given in Boyer-Lindquist coordinates $(t,r,\th^1,\th^2)$ by
\begin{align} 
\begin{aligned} 
\g =& - dt^2 + \Si \lrpar{\frac{1}{\Delta} dr^2 + d(\th^1)^2} + (r^2+a^2) \sin^2 \th^1 d(\th^2)^2 \\
&+ \frac{2Mr}{\Si} \lrpar{a\sin^2\th^1 \, d(\th^2)^2 - dt}^2,
\end{aligned} \label{EQKerrMetric22287777}
\end{align}
where 
\begin{align*} 
\begin{aligned} 
\Delta = r^2 -2Mr + a^2, \,\, \Si = r^2 + a^2 \cos^2\th^1.
\end{aligned} %\label{}
\end{align*}
For $\vert M \vert$ and $\vert a \vert$ sufficiently small, the hypersurface $$\Si:= \{t=0\} \cap \{ r \geq 1 \}$$ is smooth and spacelike, and its induced spacelike initial data has asymptotic invariants 
\begin{align*} 
\begin{aligned} 
\mathbf{E}_{\mathrm{ADM}}= M, \,\, \mathbf{P}_{\mathrm{ADM}}=0, \,\, \lrpar{\mathbf{L}_{\mathrm{ADM}}}^i= Ma \cdot \de^{i}_3, \,\, \mathbf{C}_{\mathrm{ADM}}=0,
\end{aligned} %\label{}
\end{align*}
where $\de$ denotes the Kronecker delta. By Definition \ref{DEFpoincareCHARGES} (see also (F.4) in \cite{ChruscielDelay}) the associated Poincar\'e charges are
\begin{align*} 
\begin{aligned} 
p_\mu = M \cdot \de_\mu^0, \,\, J_{\mu\nu} = 2Ma \cdot \de^1_{[\mu} \de^2_{\nu]}.
\end{aligned} %\label{}
\end{align*}

\ni \textbf{Step 2.} Apply the \emph{rotation matrix} $\La'$ to change the angular coordinates $(\th^1,\th^2)$ to $(\th'^1,\th'^2)$, and denote $t':=t$ and $r':=r$. The hypersurface $$\Si' := \{t'=0\} \cap \{ r' \geq 1 \}$$ with respect to coordinates $(t',r',\th'^1,\th'^2)$ is smooth and spacelike, and the induced spacelike initial data has asymptotic invariants
\begin{align*} 
\begin{aligned} 
\mathbf{E}'_{\mathrm{ADM}}= M, \,\, \mathbf{P}'_{\mathrm{ADM}}=0, \,\, \mathbf{L}'_{\mathrm{ADM}}= Ma \cdot n, \,\, \mathbf{C}'_{\mathrm{ADM}}=0,
\end{aligned} %\label{}
\end{align*}
where the \emph{direction of the angular momentum}, $n = \La'\lrpar{ \de^{i}_3} \in \RRR^3$ satisfies $\vert n \vert =1$. By Definition \ref{DEFpoincareCHARGES} (see also (F.5) in \cite{ChruscielDelay}) the associated Poincar\'e charges are
\begin{align} 
\begin{aligned} 
p'_\mu = M \cdot \de_\mu^0, \,\, J'_{i0} = 0, \,\, J'_{ij} = Ma \cdot \in_{ijk}n^k.
\end{aligned} \label{EQpcrot17778}
\end{align}
Let $(x'^1,x'^2,x'^3)$ denote the Cartesian coordinate system defined from $(r',\th'^1,\th'^2)$ by \eqref{EQdefinSPHERICALAFcoord}. \\

\ni \textbf{Step 3.} Translate $(x'^1,x'^2,x'^3)$ by the \emph{translation vector} $\mathbf{a} \in \RRR^3$ and denote the resulting Cartesian coordinates by $(x''^1,x''^2,x''^3)$. Denote the associated spherical coordinates by $(r'',\th''^1,\th''^2)$, and let $t'' := t'$. For $a^i \in \RRR^3$ sufficiently small, the hypersurface $$\Si'':= \{t''=0\} \cap \{ r'' \geq 1 \}$$ with respect to $(t'',r'',\th''^1,\th''^2)$ is smooth and spacelike. By Proposition \ref{REMARKtransformationlaw}, see also (F.6) in \cite{ChruscielDelay}, the induced spacelike initial data has Poincar\'e charges
\begin{align*} 
\begin{aligned} 
p''_\mu = M \cdot \de_\mu^0, \,\, J''_{i0} = M \mathbf{a}_i, \,\, J''_{ij} = Ma \cdot  \in_{ijk}n^k,
\end{aligned} %\label{}
\end{align*}
which corresponds by Definition \ref{DEFpoincareCHARGES} to asymptotic invariants
\begin{align*} 
\begin{aligned} 
\mathbf{E}_{\mathrm{ADM}}= M, \,\, \mathbf{P}_{\mathrm{ADM}}=0, \,\, \mathbf{L}_{\mathrm{ADM}}= M a \cdot n, \,\, \mathbf{C}_{\mathrm{ADM}} =M \cdot \mathbf{a}.
\end{aligned} %\label{}
\end{align*}

\ni \textbf{Step 4.} Boost the coordinate system $(t'', x''^1,x''^2,x''^3)$ by the Lorentz boost matrix $\La\in \mathrm{SO(1,3)}$ (see \eqref{EQdefinitionLAMBDA}) to new coordinates $(t''',x'''^1,x'''^2,x'''^3)$, and subsequently define $(t''', r''', \th'''^1, \th'''^2)$ by \eqref{EQdefinSPHERICALAFcoord}. For $\La$ sufficiently close to the identity matrix $\mathrm{Id}$, the hypersurface $$\Si''':= \{ t''' =0\} \cap \{ r'''\geq1\}$$ with respect to $(t''', r''', \th'''^1, \th'''^2)$ is smooth and spacelike. By Proposition \ref{REMARKtransformationlaw}, see also (F.6) in \cite{ChruscielDelay}, the induced spacelike initial data has Poincar\'e charges
\begin{align*} 
\begin{aligned} 
p'''_\mu = \La_\mu^{\,\, \a} p''_\a, \,\, J'''_{\a\be}= \La_{\a}^{\,\, \mu}\La_{\be}^{\,\, \nu} J''_{\mu\nu}.
\end{aligned}
\end{align*}
In particular, by Definition \ref{DEFpoincareCHARGES}, its asymptotic invariants $\mathbf{E}'''_{\mathrm{ADM}}$ and $\mathbf{P}'''_{\mathrm{ADM}}$ are given by
\begin{align} 
\begin{aligned} 
\mathbf{E}'''_{\mathrm{ADM}} =\La_0^{\,\, \a} p''_\a= \ga M, \,\,
\lrpar{\mathbf{P}'''_{\mathrm{ADM}}}_i =\La_i^{\,\, \a} p''_\a= -\ga M v_i.
\end{aligned} \label{EQasymptoticInvariantsFinal88777}
\end{align}
We denote the coordinate components of the induced spacelike initial data on $\Si'''$ with respect to $(x'''^1,x'''^2,x'''^3)$ by
$$(g_{ij}^{\mathbf{p}}, k_{ij}^{\mathbf{p}}).$$ 
This finishes the definition of the Kerr parameter map $\KK_{\mathrm{par}}$.

The following properties of $\KK_{\mathrm{par}}$ play an essential role for the derivation of quantitative estimates for $\KK_{\mathrm{par}}$ in Propositions \ref{PROPestimatesforKerr11101} and \ref{PROPKerrParameterRTestimates} below.

\begin{enumerate}
\item From the explicit construction of $(g_{ij}^{\mathbf{p}}, k_{ij}^{\mathbf{p}})$ above, it is straight-forward to verify that $\KK_{\mathrm{par}}$ is well-defined and smooth in an open neighbourhood of 
\begin{align} 
\begin{aligned} 
(M,a,\La,a_i)=(0,0,\mathrm{Id}, 0), \,\, \La' \in \mathrm{SO(3)}.
\end{aligned} \label{EQoriginEvalKerrmap}
\end{align}

\item For real numbers $\mathbf{E}\geq0$, we have that
\begin{align} 
\begin{aligned} 
(g_{ij}^{\mathbf{E}}, k_{ij}^{\mathbf{E}}) = \lrpar{g_{ij}^{(\mathbf{E},0,\mathrm{Id},\mathrm{Id},0)}, k_{ij}^{(\mathbf{E},0,\mathrm{Id},\mathrm{Id},0)}}.
\end{aligned} \label{EQrefEqualsSSreference999089}
\end{align}
where $(g_{ij}^{\mathbf{E}}, k_{ij}^{\mathbf{E}})$ denotes the Schwarzschild reference spacelike initial data in Schwarzschild Cartesian coordinates, see \eqref{EQssInitialSpacelike1}.

\item It holds that for all Lorentz boosts $\La \in \mathrm{SO(1,3)}$, rotations $\La' \in \mathrm{SO(3)}$ and translations $\mathbf{a} \in \RRR^3$,
\begin{align} 
\begin{aligned} 
D_a \KK_{\mathrm{par}} \vert_{(0,0,\La,\La', \mathbf{a})} =0.
\end{aligned} \label{EQKerrDclaim1}
\end{align}
Indeed, for $a\in \RRR$ and $M=0$, the spacetime metric \eqref{EQKerrMetric22287777} reduces to
\begin{align} 
\begin{aligned} 
\mathbf{g} = -dt^2 + \frac{r^2+a^2 \cos^2\th}{r^2 + a^2} dr^2 + (r^2+a^2\cos^2\th) d\th^2 + (r^2+a^2)\sin^2 \th d\phi^2.
\end{aligned} \label{EQaBLkerrD}
\end{align}
As the parameter $a$ only appears quadratically in \eqref{EQaBLkerrD}, the induced spacelike initial data on $\{t=0\} \cap \{r\geq1\}$ also depends quadratically on $a$, which shows that
\begin{align*} 
\begin{aligned} 
D_a \KK_{\mathrm{par}} \vert_{(0,0,\mathrm{Id},\mathrm{Id}, 0)} =0.
\end{aligned} %\label{}
\end{align*}
Moreover, boosts, rotations and translations of \eqref{EQaBLkerrD} do not change the quadratic appearance of $a$ in the spacetime metric components, hence we conclude further that \eqref{EQKerrDclaim1} holds.

\item It holds that for all Lorentz boosts $\La \in \mathrm{SO(1,3)}$, rotations $\La' \in \mathrm{SO(3)}$ and translations $\mathbf{a} \in \RRR^3$,
\begin{align} 
\begin{aligned} 
\KK_{\mathrm{par}}(0,0,\La,\La',\mathbf{a})= \KK_{\mathrm{par}}(0,0,\mathrm{Id}, \mathrm{Id},0),
\end{aligned} \label{EQKerrDclaim2}
\end{align}
and thus in particular,
\begin{align} 
\begin{aligned} 
D_\La \KK_{\mathrm{par}} \vert_{(0,0,\La,\La', \mathbf{a})} =0, \,\, D_{\mathbf{a}} \KK_{\mathrm{par}} \vert_{(0,0,\La,\La', \mathbf{a})} =0.
\end{aligned} \label{EQKerrDclaim22}
\end{align}
\noindent Indeed, for each Lorentz boost matrix $\La \in \mathrm{SO(1,3)}$ it holds that 
\begin{align*} 
\begin{aligned} 
\La^\top \mathbf{m} \La = \mathbf{m},
\end{aligned} %\label{}
\end{align*}
that is, the Minkowski spacetime metric is invariant under Lorentz boosts. In particular, this shows that Lorentz boosts leave Minkowski spacelike initial data invariant. Similarly, due to its spherical symmetry and translation invariance, the Minkowski reference spacelike initial data is invariant under space rotations and space translations. This finishes the proof of \eqref{EQKerrDclaim2} and \eqref{EQKerrDclaim22}.

\item Recall from \eqref{EQRTcondDefinition3} the Regge-Teitelbaum quantities
\begin{align} 
\begin{aligned} 
\mathrm{RT}(g)_{ij}(x) := g_{ij}(x) - g_{ij}(-x), \,\, \mathrm{RT}(k)_{ij}(x) :=& k_{ij}(x) + k_{ij}(-x),
\end{aligned} \label{EQRTcondDefinition3222777}
\end{align}
and their higher order generalisations. By the spherical symmetry of Schwarzschild spacelike initial data $(g_{ij}^{\mathbf{E}}, k_{ij}^{\mathbf{E}})$, it holds that 
\begin{align*} 
\begin{aligned} 
\mathrm{RT}(g^{\mathbf{E}})_{ij}(x) \equiv 0, \,\, \mathrm{RT}(k^{\mathbf{E}})_{ij}(x) \equiv 0. 
\end{aligned} %\label{}
\end{align*}
Further, as the anti-podal map $\AA: x \mapsto -x$ commutes with Lorentz boosts and space rotations (see, for example, Appendix F in \cite{ChruscielDelay}), that is, 
\begin{align*} 
\begin{aligned} 
{} [ \La, \AA]=0, \,\, [\La',\AA]=0,
\end{aligned} %\label{}
\end{align*}
we deduce that for all Lorentz boosts $\La \in \mathrm{SO(1,3)}$ and space rotations $\La' \in \mathrm{SO(3)}$,
\begin{align} 
\begin{aligned} 
\mathrm{RT}(g^{(M,0,\La,\La',0)})_{ij} \equiv 0, \,\, \mathrm{RT}(k^{(M,0,\La,\La',0)})_{ij} \equiv 0.
\end{aligned} \label{EQSSRTvanishing1}
\end{align}

\item Analogously to (4), by the translation invariance and spherical symmetry of Minkowski reference spacelike initial data, it holds that for all Lorentz boosts $\La \in \mathrm{SO(1,3)}$, rotations $\La' \in \mathrm{SO(3)}$ and space translations $\mathbf{a} \in \RRR^3$,
\begin{align*} 
\begin{aligned} 
\mathrm{RT}(g^{(0,0,\La,\La',\mathbf{a})})_{ij,l} \equiv 0, \,\, \mathrm{RT}(k^{(0,0,\La,\La',\mathbf{a})})_{ij} \equiv 0,
\end{aligned} %\label{}
\end{align*}
which implies in particular that
\begin{align} 
\begin{aligned} 
D_{\mathbf{a}} \lrpar{\mathrm{RT}(g^{(0,0,\La,\La',\mathbf{a})})_{ij}} \equiv 0, \,\, D_{\mathbf{a}} \lrpar{\mathrm{RT}(k^{(0,0,\La,\La',\mathbf{a})})_{ij}} \equiv 0.
\end{aligned} \label{EQSSRTvanishing22}
\end{align}

\item From \eqref{EQKerrDclaim1} and \eqref{EQRTcondDefinition3222777} it follows that for all Lorentz boosts $\La \in \mathrm{SO(1,3)}$, rotations $\La' \in \mathrm{SO(3)}$ and space translation $\mathbf{a} \in \RRR^3$,
\begin{align} 
\begin{aligned} 
D_a \lrpar{\mathrm{RT}(g^{(0,a,\La,\La',\mathbf{a})})_{ij}} \Big\vert_{a=0} \equiv0, \,\,
D_a \lrpar{\mathrm{RT}(k^{(0,a,\La,\La',\mathbf{a})})_{ij}} \Big\vert_{a=0} \equiv0.
\end{aligned} \label{EQRTzeroDerivative12}
\end{align}

\end{enumerate}

\ni Based on the above properties of $\KK_{\mathrm{par}}$, we now prove quantitative estimates.
First, we have the following proposition.
\begin{proposition}[Parameter estimates I] \label{PROPestimatesforKerr11101} Let $\mathbf{p} := (M,a,\Lambda,\Lambda', \mathbf{a})$ be a Kerr parameter tuple, and let $\mathbf{E}\geq0$ be a real number. For $\mathbf{E}\geq0$ sufficiently small and $(M,a,\Lambda, \abf)$ sufficiently close to 
\begin{align*} 
\begin{aligned} 
(M,a,\Lambda, \abf)= (0,0,\mathrm{Id}, 0),
\end{aligned} %\label{}
\end{align*}
it holds that
\begin{align*} 
\begin{aligned} 
&\Vert (g^{\mathbf{p}},k^{\mathbf{p}})- (g^{\mathbf{E}},k^{\mathbf{E}}) \Vert_{C^{k_0}(A_{[1,3]}) \times C^{k_0-1}(A_{[1,3]})}\\ 
\les& \vert M-\mathbf{E} \vert + \lrpar{\vert \mathbf{E}\vert + \vert a \vert} \cdot \vert a \vert + \vert \mathbf{E} \vert \cdot \lrpar{ \vert \La-\mathrm{Id}\vert+ \vert \abf \vert}.
\end{aligned} 
\end{align*}

\end{proposition}

\begin{proof}[Proof of Proposition \ref{PROPestimatesforKerr11101}] By \eqref{EQrefEqualsSSreference999089} and the spherical symmetry of the Schwarzschild reference initial data, for all $\La' \in \mathrm{SO(3)}$, 
\begin{align*} 
\begin{aligned} 
(g^{\mathbf{E}},k^{\mathbf{E}})= (g^{(\mathbf{E},0,\mathrm{Id},\mathrm{Id},0)}, k^{(\mathbf{E},0,\mathrm{Id},\mathrm{Id},0)}) =\KK_{\mathrm{par}}(\mathbf{E},0,\mathrm{Id},\mathrm{Id},0) = \KK_{\mathrm{par}}(\mathbf{E},0,\mathrm{Id},\La',0),
\end{aligned} %\label{}
\end{align*}
so that we get
\begin{align} 
\begin{aligned} 
&\Vert (g^{\mathbf{p}},k^{\mathbf{p}})- (g^{\mathbf{E}},k^{\mathbf{E}}) \Vert_{C^{k_0}(A_{[1,3]}) \times C^{k_0-1}(A_{[1,3]})}\\ 
=& \Vert \KK_{\mathrm{par}}(M,a,\La,\La',\abf) - \KK_{\mathrm{par}}(\mathbf{E},0,\mathrm{Id},\La',0) \Vert_{C^{k_0}(A_{[1,3]}) \times C^{k_0-1}(A_{[1,3]})}.
\end{aligned} \label{EQparameterestimation7771}
\end{align}
We rewrite the right-hand side of \eqref{EQparameterestimation7771} as
\begin{align} 
\begin{aligned} 
&\KK_{\mathrm{par}}(M,a,\La,\La',\abf) - \KK_{\mathrm{par}}(\mathbf{E},0,\mathrm{Id},\La',0) \\
=& \KK_{\mathrm{par}}(M,a,\La,\La',\abf) - \KK_{\mathrm{par}}(\mathbf{E},a,\La,\La',\abf)\\
&+ \KK_{\mathrm{par}}(\mathbf{E},a,\La,\La',\abf)- \KK_{\mathrm{par}}(\mathbf{E},0,\La,\La',\abf) \\
&+ \KK_{\mathrm{par}}(\mathbf{E},0,\La,\La',\abf) - \KK_{\mathrm{par}}(\mathbf{E},0,\mathrm{Id},\La',0).
\end{aligned} \label{EQKKdecomposition9991}
\end{align}

\ni First, by the smoothness of $\KK_{\mathrm{par}}$ around 
\begin{align} 
\begin{aligned} 
(M,a,\Lambda, \abf)= (0,0,\mathrm{Id}, 0), \,\, \La' \in \mathrm{SO(3)},
\end{aligned} \label{EQtrivialdatapoint999}
\end{align}
we get that for real numbers $\mathbf{E}\geq0$ sufficiently small and $\mathbf{p}=(M,a,\La,a^i)$ close to \eqref{EQtrivialdatapoint999},
\begin{align} 
\begin{aligned} 
\Vert \KK_{\mathrm{par}}(M,a,\La,\La',\abf) - \KK_{\mathrm{par}}(\mathbf{E},a,\La,\La',\abf) \Vert_{C^{k_0}(A_{[1,3]}) \times C^{k_0-1}(A_{[1,3]})} \les& \vert M-\mathbf{E} \vert.
\end{aligned} \label{EQEQKKdecomposition9991estim1}
\end{align}
Second, by the smoothness of $\KK_{\mathrm{par}}$ and \eqref{EQKerrDclaim1},
\begin{align} 
\begin{aligned} 
\Vert \KK_{\mathrm{par}}(\mathbf{E},a,\La,\La',\abf) - \KK_{\mathrm{par}}(\mathbf{E},0,\La,\La',\abf) \Vert_{C^{k_0}(A_{[1,3]}) \times C^{k_0-1}(A_{[1,3]})}
\les \lrpar{\vert \mathbf{E}\vert + \vert a \vert} \cdot \vert a \vert.
\end{aligned} \label{EQEQKKdecomposition9991estim2}
\end{align}
Third, by the smoothness of $\KK_{\mathrm{par}}$ and \eqref{EQKerrDclaim22},
\begin{align} 
\begin{aligned} 
&\Vert \KK_{\mathrm{par}}(\mathbf{E},0,\La,\La',\abf) - \KK_{\mathrm{par}}(\mathbf{E},0,\mathrm{Id},\La',0) \Vert_{C^{k_0}(A_{[1,3]}) \times C^{k_0-1}(A_{[1,3]})} \\
\les& \vert \mathbf{E}\vert \cdot \lrpar{ \vert \La-\mathrm{Id} \vert + \vert \abf \vert}.
\end{aligned} \label{EQEQKKdecomposition9991estim3}
\end{align}
Plugging \eqref{EQKKdecomposition9991}, \eqref{EQEQKKdecomposition9991estim1}, \eqref{EQEQKKdecomposition9991estim2} and \eqref{EQEQKKdecomposition9991estim3} into \eqref{EQparameterestimation7771} finishes the proof of Proposition \ref{PROPestimatesforKerr11101}. \end{proof}

\ni Second, we have the following proposition for the Regge-Teitelbaum quantities.
\begin{proposition}[Parameter estimates II] \label{PROPKerrParameterRTestimates} Let $\mathbf{p} := (M,a,\Lambda,\Lambda', a_i)$ be a Kerr parameter tuple. For $(M,a,\Lambda, \abf)$ sufficiently close to 
\begin{align} 
\begin{aligned} 
(M,a,\Lambda, \abf)= (0,0,\mathrm{Id}, 0),
\end{aligned} \label{EQsmoothnesspointrecall2}
\end{align}
it holds that
\begin{align*} 
\begin{aligned} 
\Vert \mathrm{RT}(g^{\mathbf{p}})_{ij} \Vert_{C^{k_0}(A_{[1,3]})}+\Vert \mathrm{RT}(k^{\mathbf{p}})_{ij} \Vert_{C^{k_0-1}(A_{[1,3]})}
\les \lrpar{\vert M \vert + \vert a \vert} \lrpar{\vert \abf \vert + \vert a \vert}.
\end{aligned}
\end{align*}
\end{proposition}

\begin{proof}[Proof of Proposition \ref{PROPKerrParameterRTestimates}] We first consider $\mathrm{RT}(g^{\mathbf{p}})_{ij,l}$. Using \eqref{EQSSRTvanishing1} we have that
\begin{align} 
\begin{aligned} 
\mathrm{RT}(g^{(M,a,\La,\La',\mathbf{a})})_{ij} =& \mathrm{RT}(g^{(M,a,\La,\La',\mathbf{a})})_{ij} - \mathrm{RT}(g^{(M,a,\La,\La',0)})_{ij} \\
& + \mathrm{RT}(g^{(M,a,\La,\La',0)})_{ij} - \underbrace{\mathrm{RT}(g^{(M,0,\La,\La',0)})_{ij}}_{\equiv0}.
\end{aligned} \label{EQKerrRTparameterbounds2}
\end{align}
By the smoothness of $\KK_{\mathrm{par}}$ around \eqref{EQsmoothnesspointrecall2}, and \eqref{EQSSRTvanishing22} and \eqref{EQRTzeroDerivative12}, we get that
\begin{align} 
\begin{aligned} 
\left\Vert \mathrm{RT}(g^{(M,a,\La,\La',\mathbf{a})})_{ij} - \mathrm{RT}(g^{(M,a,\La,\La',0)})_{ij} \right\Vert_{C^{k_0}(A_{[1,3]})} \les& \lrpar{\vert M \vert + \vert a \vert} \cdot \vert \mathbf{a} \vert, \\
\left\Vert \mathrm{RT}(g^{(M,a,\La,\La',0)})_{ij} - \mathrm{RT}(g^{(M,0,\La,\La',0)})_{ij} \right\Vert_{C^{k_0}(A_{[1,3]})}
\les& \lrpar{\vert M \vert + \vert a \vert} \cdot \vert a \vert.
\end{aligned} \label{EQestimates2777RTkerrreference}
\end{align}
Combining \eqref{EQKerrRTparameterbounds2} and \eqref{EQestimates2777RTkerrreference} proves the estimate for $\mathrm{RT}(g^{\mathbf{p}})_{ij}$. By the same argument, \eqref{EQSSRTvanishing1}, \eqref{EQSSRTvanishing22} and \eqref{EQRTzeroDerivative12} lead to the estimate for $\mathrm{RT}(k^{\mathbf{p}})_{ij}$. This finishes the proof of Proposition \ref{PROPKerrParameterRTestimates}. \end{proof}

%%%%%%%%%%%%%%%%%%%%%%%%%%%%%%%%%%%%%%%%
\subsection{Kerr asymptotic invariants map and estimates} \label{SECkerrChargeMapping} \ni In this section we define and prove estimates for the Kerr asymptotic invariants mapping $\KK$. First define the set $I(0)$ of timelike energy-momentum $4$-vectors by
\begin{align*} 
\begin{aligned} 
I(0) := \{(\mathbf{E}, \mathbf{P})\in \RRR \times \RRR^3: \mathbf{E}^2 - \vert \mathbf{P} \vert^2 >0 \} \subset \RRR^4. 
\end{aligned} %\label{}
\end{align*}
and define \emph{asymptotic invariants vectors} $\la$ to be elements of the set
\begin{align*} 
\begin{aligned} 
\la \in I(0) \times \RRR^3 \times \RRR^3.
\end{aligned} %\label{}
\end{align*}
\textbf{Notation.} We denote the components of $\la\in I(0) \times \RRR^3 \times \RRR^3$ by
\begin{align*} 
\begin{aligned} 
\la=(\mathbf{E}(\la),\mathbf{P}(\la),\mathbf{L}(\la),\mathbf{C}(\la)).
\end{aligned} %\label{}
\end{align*}

\ni The Kerr asymptotic invariants mapping $\KK$ is then defined as map for asymptotic invariants vectors to Kerr parameters,
\begin{align} 
\begin{aligned} 
\KK: I(0) \times \RRR^3 \times \RRR^3 &\to \RRR \times \RRR \times \mathrm{SO(1,3)} \times \mathrm{SO(3)} \times \RRR^3, \\
\la &\mapsto \KK(\la):= (M,a,\La,\La',\mathbf{a}),
\end{aligned} \label{EQgeneralWritingKKdef}
\end{align}
such that 
\begin{align*} 
\begin{aligned} 
\lrpar{\mathbf{E}_{\mathrm{ADM}},\mathbf{P}_{\mathrm{ADM}},\mathbf{L}_{\mathrm{ADM}},\mathbf{C}_{\mathrm{ADM}}}\lrpar{\KK_{\mathrm{par}} \lrpar{\KK(\la)}} = \la,
\end{aligned} %\label{}
\end{align*}
where $\KK_{\mathrm{par}}$ denotes the Kerr parameter map defined in Section \ref{SECdefKerrParameter}. It is shown in Appendix F of \cite{ChruscielDelay} that the map $\KK(\la)$ is well-defined on $I(0)\times \RRR^3 \times \RRR^3$. In particular, the Kerr parameter map $\KK_{\mathrm{par}}$ is surjective onto the asymptotic invariants $(\mathbf{E}_{\mathrm{ADM}},\mathbf{P}_{\mathrm{ADM}},\mathbf{L}_{\mathrm{ADM}},\mathbf{C}_{\mathrm{ADM}})$ of $(g^{\mathbf{p}},k^{\mathbf{p}})$.\\

\ni \textbf{Notation.} We use the following.
\begin{enumerate}
\item For $\la \in I(0)\times \RRR^3 \times \RRR^3$ and real numbers $r\geq1$, denote
\begin{align*} 
\begin{aligned} 
(\mathbf{E}_{\mathrm{ADM}}^{\mathrm{loc}},\mathbf{P}_{\mathrm{ADM}}^{\mathrm{loc}},\mathbf{L}_{\mathrm{ADM}}^{\mathrm{loc}},\mathbf{C}_{\mathrm{ADM}}^{\mathrm{loc}})(\la,r) := (\mathbf{E}_{\mathrm{ADM}}^{\mathrm{loc}},\mathbf{P}_{\mathrm{ADM}}^{\mathrm{loc}},\mathbf{L}_{\mathrm{ADM}}^{\mathrm{loc}},\mathbf{C}_{\mathrm{ADM}}^{\mathrm{loc}})(S_r,g_{ij}^{\KK(\la)}, k_{ij}^{\KK(\la)}),
\end{aligned} %\label{}
\end{align*}
and for real numbers $r'\geq1$,
\begin{align*} 
\begin{aligned} 
r(\la, r') := r\lrpar{S_{r'},g^{\KK(\la)}_{ij},k^{\KK(\la)}_{ij}}
\end{aligned} %\label{}
\end{align*}
where the right-hand side denotes the area radius of $S_r$.
\end{enumerate}

\begin{remark}[Scaling of the asymptotic invariants map] \label{RemarkScalingKerrData} Under scaling, see Section \ref{SECspacelikescaling}, it holds that 
\begin{align*} 
\begin{aligned} 
({}^{(R)}g^{\KK(\la)}, {}^{(R)}k^{\KK(\la)}) =&(g^{\KK({}^{(R)}\la)}, k^{\KK({}^{(R)}\la)}),\,\, r({}^{(R)}\la,r')=&R^{-1} \cdot r(\la,R \cdot r'),
\end{aligned} %\label{}
\end{align*}
where the rescaled asymptotic invariants vector ${}^{(R)}\la$ is defined by (see \eqref{EQLEMscalingADMlocal})
\begin{align} 
\begin{aligned} 
{}^{(R)}\la = (R^{-1} \mathbf{E}(\la),R^{-1} \mathbf{P}(\la),R^{-2} \mathbf{L}(\la),R^{-2} \mathbf{C}(\la) ).
\end{aligned} \label{EQrescaledLAMBDAKerr}
\end{align}
It is straight-forward to verify that the Kerr metric $\g^{M,a}$ in Boyer-Lindquist coordinates $(t,r,\th^1,\th^2)$, see \eqref{EQKerrMetric22287777}, changes to $\g^{M/R,a/R}$ under the scaling transformation
\begin{align*} 
\begin{aligned} 
(t,r,\th^1,\th^2) \to (R^{-1} t, R^{-1} r,\th^1,\th^2), \,\, \g \to R^{-2} \g.
\end{aligned} %\label{}
\end{align*}

\end{remark}

\ni The following estimates for $\KK$ are essential for the charge estimates in Section \ref{SECchargeEstimateKERR444}, see also Corollary \ref{CORcombinedEstimates222} below.
\begin{proposition}[Asymptotic invariants estimates I] \label{PROPparameterEstimatesKerr} Let $\la \in I(0) \times \RRR^3 \times \RRR^3$ be an asymptotic invariants vector such that for two real numbers $\varep_\infty>0$ and $\mathbf{E}>0$,
\begin{align} 
\begin{aligned} 
\vert \mathbf{E}(\la)-\mathbf{E} \vert^2 + \vert \mathbf{P}(\la) \vert^2 \leq \lrpar{\varep_\infty \mathbf{E}}^2,
\end{aligned} \label{EQsmallnesscondparameterEstimates}
\end{align}
and denote moreover $(M,a,\La,\La',\abf) := \KK(\la)$. For $\varep_\infty>0$ and $\mathbf{E}>0$ sufficiently small, it holds that
\begin{align*} 
\begin{aligned} 
\vert M-\mathbf{E} \vert \les& \vert \mathbf{E}(\la)-\mathbf{E} \vert  + \varep_\infty \vert \mathbf{P}(\la) \vert, &
\vert \La-\mathrm{Id} \vert \les& \frac{\vert \mathbf{P}(\la) \vert}{\mathbf{E}}, \\
\vert \abf \vert \les& \frac{1}{\mathbf{E}} \lrpar{\frac{\vert \mathbf{P}(\la)\vert}{\mathbf{E}} \vert \mathbf{L}(\la) \vert+\vert \mathbf{C}(\la) \vert}, &
\vert a \vert \les& \frac{1}{\mathbf{E}} \lrpar{\vert \mathbf{L}(\la) \vert + \frac{\vert \mathbf{P}(\la) \vert }{\mathbf{E}}\vert \mathbf{C}(\la) \vert}.
\end{aligned} 
\end{align*}
\end{proposition}

\begin{proof}[Proof of Proposition \ref{PROPparameterEstimatesKerr}] The proof is based on the explicit identities of Section \ref{SECdefKerrParameter}. \\

\ni \textbf{Control of $\vert M-\mathbf{E}\vert$.} First, from \eqref{EQasymptoticInvariantsFinal88777} we have that
\begin{align} 
\begin{aligned} 
\mathbf{E}(\la) = \ga M, \,\, \mathbf{P}(\la) = -\ga M v^i.
\end{aligned} \label{EQADMRelationsParameterEP77788}
\end{align} 
Using the definition of $\ga$ in \eqref{EQdefinitionGAMMA}, we get that
\begin{align*} 
\begin{aligned} 
\sqrt{\mathbf{E}(\la)^2 - \vert \mathbf{P}(\la)\vert^2} =& \sqrt{\ga^2 M^2 (1- v^2)}=\sqrt{M^2}=M.
\end{aligned} %\label{}
\end{align*}
Hence by \eqref{EQsmallnesscondparameterEstimates}, for $\varep_\infty>0$ sufficiently small,
\begin{align} 
\begin{aligned} 
M- \mathbf{E} =& \sqrt{\mathbf{E}(\la)^2 - \vert \mathbf{P}(\la) \vert^2} -\mathbf{E} \\
=& \mathbf{E} \lrpar{\sqrt{1+ 2 \frac{\vert \mathbf{E}(\la)-\mathbf{E} \vert}{\mathbf{E}} + \lrpar{\frac{\vert \mathbf{E}(\la)-\mathbf{E} \vert}{\mathbf{E}}}^2 - \lrpar{\frac{\vert \mathbf{P}(\la) \vert}{\mathbf{E}}}^2}-1} \\
\les& \mathbf{E} \lrpar{(1+ \varep_\infty) \frac{\vert \mathbf{E}(\la)-\mathbf{E} \vert}{\mathbf{E}} + \varep_\infty \frac{\vert \mathbf{P}(\la)\vert }{\mathbf{E}} } \\
\les& \vert \mathbf{E}(\la)-\mathbf{E} \vert + \varep_\infty \vert \mathbf{P}(\la)\vert.
\end{aligned} \label{EQMminusEestimateParamKerr77888}
\end{align}

\ni \textbf{Control of $\vert \La -\mathrm{Id}\vert$.} Recall from \eqref{EQdefinitionLAMBDA} that the Lorentz boost matrix $\La$ is given by
\begin{align} 
\begin{aligned} 
\La = \begin{pmatrix}
\ga & -\ga v^1 & -\ga v^2  & -\ga v^3\\
-\ga v^1 & 1 +(\ga-1)\frac{(v^1)^2}{v^2} & (\ga-1)\frac{(v^1)(v^2)}{v^2} &  (\ga-1)\frac{(v^1)(v^3)}{v^2}\\
 -\ga v^2 & (\ga-1)\frac{(v^1)(v^2)}{v^2} & 1 +(\ga-1)\frac{(v^2)^2}{v^2} &  (\ga-1)\frac{(v^2)(v^3)}{v^2} \\
  -\ga v^3 & (\ga-1)\frac{(v^1)(v^3)}{v^2} & (\ga-1)\frac{(v^2)(v^3)}{v^2} & 1 +(\ga-1)\frac{(v^3)^2}{v^2}
\end{pmatrix}.
\end{aligned} \label{EQdefinitionLAMBDASO13boostParamEstimrecall}
\end{align}
In the following we estimate the components $\ga-1$, $-\ga v^i$ and $(\ga-1)\frac{v^iv^j}{v^2}$ of $\La -\mathrm{Id}$.

From \eqref{EQdefinitionGAMMA} and \eqref{EQADMRelationsParameterEP77788} we get 
\begin{align*} 
\begin{aligned} 
1- v^2 = 1- \frac{\vert \mathbf{P}(\la)\vert^2}{M^2 \ga^2} = 1- \frac{\vert \mathbf{P}(\la) \vert^2}{M^2} (1-v^2), 
\end{aligned} %\label{}
\end{align*}
which implies that
\begin{align} 
\begin{aligned} 
\frac{1}{1-v^2} = 1+ \frac{\vert \mathbf{P}(\la)\vert^2}{M^2}.
\end{aligned} \label{EQvPrelation77788}
\end{align}
Using \eqref{EQdefinitionGAMMA}, \eqref{EQsmallnesscondparameterEstimates} and \eqref{EQvPrelation77788}, we thus get that for $\varep_\infty>0$ sufficiently small,
\begin{align} 
\begin{aligned} 
\vert \ga -1\vert  = \left\vert \sqrt{1+ \frac{\vert \mathbf{P}(\la) \vert^2}{M^2}} -1 \right\vert =\left\vert \sqrt{1+ \lrpar{\frac{\vert \mathbf{P}(\la) \vert}{\mathbf{E}}}^2 \lrpar{\frac{\mathbf{E}}{M}}^2}
-1 \right\vert \les \varep_\infty \frac{\vert \mathbf{P}(\la) \vert}{\mathbf{E}},
\end{aligned} \label{EQgaestimateSRboost1}
\end{align}
where we used that by \eqref{EQsmallnesscondparameterEstimates} and \eqref{EQMminusEestimateParamKerr77888}, for $\varep_\infty>0$ sufficiently small,
\begin{align} 
\begin{aligned} 
\frac{\mathbf{E}}{M} = \frac{1}{1+ \frac{M-\mathbf{E}}{\mathbf{E}}} = 1 + \OO\lrpar{\frac{\vert M-\mathbf{E}\vert }{\mathbf{E}}} = 1+\OO(\varep_\infty).
\end{aligned} \label{EQEbyMestimateKerrparam}
\end{align}

\ni Moreover, by \eqref{EQADMRelationsParameterEP77788} and \eqref{EQEbyMestimateKerrparam}, for $\varep_\infty>0$ sufficiently small,
\begin{align} 
\begin{aligned} 
\vert -\ga v^i \vert = \left\vert \frac{\mathbf{P}(\la)^i}{M} \right\vert = \left\vert \frac{\mathbf{P}(\la)^i}{\mathbf{E}} \frac{\mathbf{E}}{M} \right\vert \les \frac{\vert \mathbf{P}(\la)\vert}{\mathbf{E}}.
\end{aligned} \label{EQgaestimateSRboost2}
\end{align}
Furthermore, by \eqref{EQgaestimateSRboost1},
\begin{align} 
\begin{aligned} 
\left\vert (\ga-1)\frac{v^i v^j}{v^2} \right\vert \les \vert \ga-1 \vert \les \varep_\infty \frac{\vert \mathbf{P}(\la)\vert}{\mathbf{E}}.
\end{aligned} \label{EQgaestimateSRboost277778}
\end{align}
Applying \eqref{EQgaestimateSRboost1}, \eqref{EQgaestimateSRboost2} and \eqref{EQgaestimateSRboost277778} to \eqref{EQdefinitionLAMBDASO13boostParamEstimrecall}, it follows that for $\varep_\infty>0$ sufficiently small,
\begin{align} 
\begin{aligned} 
\left\vert \La-\mathrm{Id} \right\vert \les \frac{\vert \mathbf{P}\vert}{\mathbf{E}}.
\end{aligned} \label{EQLambdaControlParameterKerrestimProof}
\end{align}

\ni \textbf{Control of $\mathbf{a} \in \RRR^3$.} The translation vector $\mathbf{a} \in \RRR^3$ can be calculated as follows. By Definition \ref{DEFpoincareCHARGES} and the transformation law for Poincar\'e charges of Proposition \ref{REMARKtransformationlaw}, it holds that
\begin{align*} 
\begin{aligned} 
0 =& \lrpar{\mathbf{C}_{\mathrm{ADM}}}_i\lrpar{g^{(M,a,\mathrm{Id},\La',0)},k^{(M,a,\mathrm{Id},\La',0)}} \\
=& J_{i0}\lrpar{g^{(M,a,\mathrm{Id},\La',0)},k^{(M,a,\mathrm{Id},\La',0)}} \\
=& (\La^{-1})_i^{\,\, \a}(\La^{-1})_0^{\,\, \be} J_{\a\be}\lrpar{g^{(M,a,\La,\La',\mathbf{a})},k^{(M,a,\La,\La',\mathbf{a})}} \\
&+ (-\mathbf{a})_i (\La^{-1})_0^{\,\, \a} p_\a\lrpar{g^{(M,a,\La,\La',\mathbf{a})},k^{(M,a,\La,\La',\mathbf{a})}},
\end{aligned} %\label{}
\end{align*}
where $\La^{-1}$ denotes the inverse Lorentz boost to $\La$ and we used that $\abf_0=0$. Hence we get the following expression for $\mathbf{a}$,
\begin{align} 
\begin{aligned} 
\mathbf{a}_i = \frac{(\La^{-1})_i^{\,\, \a}(\La^{-1})_0^{\,\, \be} J_{\a\be}}{(\La^{-1})_0^{\,\, \a} p_\a},
\end{aligned} \label{EQmathbfaexpressionExplicit}
\end{align}
with
\begin{align*} 
\begin{aligned} 
p_\a= p_\a\lrpar{g^{(M,a,\La,\La',\mathbf{a})},k^{(M,a,\La,\La',\mathbf{a})}}, \,\, J_{\a\be}=J_{\a\be}\lrpar{g^{(M,a,\La,\La',\mathbf{a})},k^{(M,a,\La,\La',\mathbf{a})}}.
\end{aligned} %\label{}
\end{align*}
By Definition \ref{DEFpoincareCHARGES} and definition of the map $\KK$ we have the relations
\begin{align} 
\begin{aligned} 
p_0=\mathbf{E}(\la), \,\, p_i=\mathbf{P}(\la)_i, \,\, J_{i0}= \mathbf{C}(\la)_i, \,\, J_{ij} = \in_{ijk} \mathbf{L}(\la)^k,
\end{aligned} \label{EQchargeRelationskerrparamuseful1778}
\end{align}
so that on the right-hand side of \eqref{EQmathbfaexpressionExplicit} we can calculate, using \eqref{EQdefinitionLAMBDASO13boostParamEstimrecall} and that the inverse of a Lorentz boost equals the Lorentz boost in the opposite direction,
\begin{align} 
\begin{aligned} 
&(\La^{-1})_i^{\,\, \a}(\La^{-1})_0^{\,\, \be} J_{\a\be} \\
=&\lrpar{\La^{-1}}_i^{\,\,\,j} \lrpar{\La^{-1}}_0^{\,\,\,0} J_{j0} + \lrpar{\La^{-1}}_i^{\,\,\,0} \lrpar{\La^{-1}}_0^{\,\,\,j} J_{0j} + \lrpar{\La^{-1}}_i^{\,\,\,j} \lrpar{\La^{-1}}_0^{\,\,\,l} J_{jl} \\
=& \lrpar{\de_i^{\,\,\,j} + (\ga-1) \frac{v_i v^j}{v^2}} \ga\, \mathbf{C}(\la)_j - \lrpar{\ga v_i} \lrpar{\ga v^j} \mathbf{C}(\la)_j \\
&+ \lrpar{\de_i^{\,\,\,j} + (\ga-1) \frac{v_i v^j}{v^2} } \lrpar{\ga v^l} \in_{jlk} \mathbf{L}(\la)^k
\end{aligned} \label{EQdetailedabfcalculation99901}
\end{align}
Plugging \eqref{EQdetailedabfcalculation99901} into \eqref{EQmathbfaexpressionExplicit}, and applying \eqref{EQsmallnesscondparameterEstimates}, \eqref{EQMminusEestimateParamKerr77888}, \eqref{EQgaestimateSRboost1}, \eqref{EQgaestimateSRboost2} and \eqref{EQchargeRelationskerrparamuseful1778}, we get that for $\varep_\infty>0$ sufficiently small,
\begin{align} 
\begin{aligned} 
\vert \mathbf{a} \vert \les \frac{1}{\mathbf{E}} \lrpar{\frac{\vert \mathbf{P}(\la) \vert }{\mathbf{E}}\vert \mathbf{L}(\la) \vert+ \vert \mathbf{C}(\la) \vert}.
\end{aligned} \label{EQparameterEstimatemathbfa77789}
\end{align}

\ni \textbf{Control of $a\in \RRR$.} The renormalized angular momentum parameter $a \in \RRR$ can be calculated as follows. By \eqref{EQpcrot17778} and the transformation law of Proposition \ref{REMARKtransformationlaw},
\begin{align*} 
\begin{aligned} 
Ma \in_{ijk} n^k =& J_{ij}\lrpar{g^{(M,a,\mathrm{Id},\La',0)},k^{(M,a,\mathrm{Id},\La',0)}} \\
=& (\La^{-1})_i^{\,\, \a}(\La^{-1})_j^{\,\, \be} J_{\a\be}\lrpar{g^{(M,a,\La,\La',\mathbf{a})},k^{(M,a,\La,\La',\mathbf{a})}} \\
&+ (-\mathbf{a})_i (\La^{-1})_j^{\,\, \a} p_\a\lrpar{g^{(M,a,\La,\La',\mathbf{a})},k^{(M,a,\La,\La',\mathbf{a})}} \\
&- (-\mathbf{a})_j (\La^{-1})_i^{\,\, \a} p_\a\lrpar{g^{(M,a,\La,\La',\mathbf{a})},k^{(M,a,\La,\La',\mathbf{a})}},
\end{aligned}
\end{align*}
which yields that
\begin{align} 
\begin{aligned} 
Ma \cdot n^l =& \half \in^{lij} Ma \in_{ijk}n^k \\
=& \half \in^{lij} (\La^{-1})_i^{\,\, \a}(\La^{-1})_j^{\,\, \be} J_{\a\be} - \in^{lij} \abf_i \lrpar{\La^{-1}}_j^{\,\,\,\a} p_\a.
\end{aligned} \label{EQManlexpr999002}
\end{align}
The first term on the right-hand side of \eqref{EQManlexpr999002} can be expressed as
\begin{align} 
\begin{aligned} 
&\half \in^{lij} (\La^{-1})_i^{\,\, \a}(\La^{-1})_j^{\,\, \be} J_{\a\be} \\
=& \in^{lij} (\La^{-1})_i^{\,\, 0}(\La^{-1})_j^{\,\, k} J_{0k} +\half \in^{lij} (\La^{-1})_i^{\,\, a}(\La^{-1})_j^{\,\, b} J_{ab} \\
=& \in^{lij} \lrpar{\ga v_i} \lrpar{\de_j^k + (\ga-1) \frac{v_i v^k}{v^2}} \lrpar{-\mathbf{C}(\la)_k} + \half \in^{lij} (\La^{-1})_i^{\,\, a}(\La^{-1})_j^{\,\, b} \in_{abs}\mathbf{L}(\la)^s.
\end{aligned} \label{EQcalculationakerr99901}
\end{align}
The first term on the right-hand side of \eqref{EQManlexpr999002} can be expressed as
\begin{align} 
\begin{aligned} 
&- \in^{lij} \abf_i \lrpar{\La^{-1}}_j^{\,\,\,\a} p_\a \\
=& - \in^{lij} \abf_i \lrpar{\lrpar{\La^{-1}}_j^{\,\,\,0} \mathbf{E}(\la) + \lrpar{\La^{-1}}_j^{\,\,\,k} \mathbf{P}(\la)_k} \\
=& - \in^{lij} \abf_i \lrpar{\lrpar{\ga v_j} \mathbf{E}(\la) +\lrpar{\de_j^k +(\ga-1) \frac{v_jv^k}{v^2}} \frac{\mathbf{P}(\la)_k}{\mathbf{E}} \mathbf{E} }.
\end{aligned} \label{EQcalculationakerr999012}
\end{align}

\ni Plugging \eqref{EQcalculationakerr99901} and \eqref{EQcalculationakerr999012} into \eqref{EQManlexpr999002}, and applying \eqref{EQsmallnesscondparameterEstimates}, \eqref{EQMminusEestimateParamKerr77888}, \eqref{EQgaestimateSRboost1}, \eqref{EQgaestimateSRboost2}, \eqref{EQchargeRelationskerrparamuseful1778} and \eqref{EQparameterEstimatemathbfa77789}, we get that for $\varep_\infty>0$ sufficiently small,
\begin{align*} 
\begin{aligned} 
\vert a \vert \les& \frac{1}{\mathbf{E}} \lrpar{\vert \mathbf{L}(\la) \vert+ \frac{\vert \mathbf{P}(\la) \vert }{\mathbf{E}}\vert \mathbf{C}(\la) \vert}.
\end{aligned}
\end{align*}

\ni This finishes the proof of Proposition \ref{PROPparameterEstimatesKerr}. \end{proof}

\ni Combining Propositions \ref{PROPestimatesforKerr11101} and \ref{PROPparameterEstimatesKerr}, we get the following.
\begin{proposition}[Asymptotic invariants estimates II] \label{CORcombinedEstimates222} Let $\la \in I(0) \times \RRR^3 \times \RRR^3$ be an asymptotic invariants vector and let $\mathbf{E}>0$ and $\varep_\infty>0$ be two real numbers such that
\begin{align*} 
\begin{aligned} 
\vert \mathbf{E}(\la)-\mathbf{E} \vert^2 + \vert \mathbf{P}(\la) \vert^2+ \vert \mathbf{L}(\la) \vert^2+ \vert \mathbf{C}(\la) \vert^2 \les \lrpar{\varep_\infty \mathbf{E}}^2.
\end{aligned} %\label{}
\end{align*}
For $\varep_\infty>0$ and $\mathbf{E}>0$ sufficiently small, it holds that
\begin{align*} 
\begin{aligned} 
&\Vert (g^{\KK(\la)},k^{\KK(\la)}) -(g^{\mathbf{E}}, k^{\mathbf{E}}) \Vert_{C^{k_0}(A_{[1,3]}) \times C^{k_0-1}(A_{[1,3]})} \\
\les& \vert \mathbf{E}(\la)-\mathbf{E} \vert + \vert \mathbf{P}(\la) \vert + \vert \mathbf{L}(\la) \vert + \vert \mathbf{C}(\la) \vert + \lrpar{\frac{\vert \mathbf{L}(\la) \vert + \frac{\vert \mathbf{P}(\la) \vert}{\mathbf{E}} \vert \mathbf{C}(\la) \vert}{\mathbf{E}}}^2.
\end{aligned} %\label{}
\end{align*}
Moreover, it holds that
\begin{align*} 
\begin{aligned} 
&\Vert \mathrm{RT}(g^{\KK(\la)})_{ij} \Vert_{C^{k_0}(A_{[1,3]})}+\Vert \mathrm{RT}(k^{\KK(\la)})_{ij} \Vert_{C^{k_0-1}(A_{[1,3]})} \\
\les& \vert \mathbf{L}(\la) \vert + \vert \mathbf{C}(\la) \vert +  \frac{\vert \mathbf{L}(\la) \vert + \frac{\vert \mathbf{P}(\la) \vert }{\mathbf{E}}\vert \mathbf{C}(\la) \vert}{\mathbf{E}} \cdot \frac{\vert \mathbf{L}(\la) \vert + \vert \mathbf{C}(\la) \vert}{\mathbf{E}}.
\end{aligned}
\end{align*}
\end{proposition}
\begin{proof}[Proof of Proposition \ref{CORcombinedEstimates222}] By Proposition \ref{PROPparameterEstimatesKerr}, for $\varep_\infty>0$ and $\mathbf{E}>0$ sufficiently small, the parameters $(M,a,\La,\La',\abf)$ determined by $\KK(\la)$ are close to 
\begin{align*} 
\begin{aligned} 
(M,a,\La,\abf) = (0,0,\mathrm{Id},0), \,\, \La' \in \mathrm{SO(3)}.
\end{aligned} %\label{}
\end{align*}
Thus we can combine Proposition \ref{PROPparameterEstimatesKerr} with Proposition \ref{PROPestimatesforKerr11101} to get
\begin{align*} 
\begin{aligned} 
&\Vert (g^{\KK(\la)},k^{\KK(\la)}) -(g^{\mathbf{E}}, k^{\mathbf{E}}) \Vert_{C^{k_0}(A_{[1,3]}) \times C^{k_0-1}(A_{[1,3]})} \\
\les& \vert M-\mathbf{E} \vert + (\mathbf{E} + \vert a \vert) \cdot \vert a \vert + \mathbf{E} \cdot \lrpar{ \vert \La - \mathrm{Id} \vert + \vert \abf \vert} \\
\les& \vert \mathbf{E}(\la)-\mathbf{E} \vert + \varep_\infty \vert \mathbf{P}(\la) \vert + \lrpar{\mathbf{E} + \frac{\vert \mathbf{L}(\la) \vert + \frac{\vert \mathbf{P}(\la) \vert }{\mathbf{E}} \vert \mathbf{C}(\la) \vert}{\mathbf{E}}} \cdot \frac{\vert \mathbf{L}(\la) \vert + \frac{\vert \mathbf{P}(\la) \vert }{\mathbf{E}}\vert \mathbf{C}(\la) \vert}{\mathbf{E}} \\
&+ \mathbf{E} \cdot \lrpar{\frac{\vert \mathbf{P}(\la) \vert }{\mathbf{E}}+ \frac{\frac{\vert \mathbf{P} \vert}{\mathbf{E}}\vert \mathbf{L}(\la) \vert + \vert \mathbf{C}(\la) \vert}{\mathbf{E}}} \\
\les& \vert \mathbf{E}(\la)-\mathbf{E} \vert + \vert \mathbf{P}(\la) \vert + \vert \mathbf{L}(\la) \vert + \vert \mathbf{C}(\la) \vert + \lrpar{\frac{\vert \mathbf{L}(\la) \vert + \frac{\vert \mathbf{P}(\la) \vert }{\mathbf{E}}\vert \mathbf{C}(\la) \vert}{\mathbf{E}}}^2,
\end{aligned} %\label{}
\end{align*}
and
\begin{align*} 
\begin{aligned} 
& \Vert \mathrm{RT}(g^{\KK(\la)})_{ij} \Vert_{C^{k}(A_{[1,3]})}+\Vert \mathrm{RT}(k^{\KK(\la)})_{ij} \Vert_{C^{k_0-1}(A_{[1,3]})}\\
\les& \lrpar{\vert M \vert + \vert a \vert } \cdot \lrpar{\vert a \vert + \vert \abf \vert} \\
\les& \lrpar{\mathbf{E} + \vert \mathbf{E}(\la)-\mathbf{E} \vert + \varep_\infty \vert \mathbf{P}(\la) \vert +  \frac{\vert \mathbf{L}(\la) \vert + \frac{\vert \mathbf{P}(\la) \vert }{\mathbf{E}}\vert \mathbf{C}(\la) \vert}{\mathbf{E}}}  \cdot  \frac{\vert \mathbf{L}(\la) \vert + \vert \mathbf{C}(\la) \vert}{\mathbf{E}} \\
\les& \vert \mathbf{L}(\la) \vert + \vert \mathbf{C}(\la) \vert +  \frac{\vert \mathbf{L}(\la) \vert + \frac{\vert \mathbf{P}(\la) \vert }{\mathbf{E}}\vert \mathbf{C}(\la) \vert}{\mathbf{E}} \cdot \frac{\vert \mathbf{L}(\la) \vert + \vert \mathbf{C}(\la) \vert}{\mathbf{E}}.
\end{aligned} %\label{}
\end{align*}

\ni This finishes the proof of Proposition \ref{CORcombinedEstimates222}. \end{proof}

%%%%%%%%%%%%%%%%%%%%%%%%%%%%%%%%%%%%%%%%
\subsection{Kerr reference sphere data}\label{SECKerrSphereDataDefinition} \ni In this section we define, for given asymptotic invariants vectors $\la$, Kerr reference sphere data and prove estimates. We assume that the asymptotic invariants vector $\la$ is contained in the following set.\\ 

\ni \textbf{Notation.} Let $R\geq1$ and $\mathbf{E}_\infty>0$ be two real numbers. Then $\EE_R(\mathbf{E}_\infty)$ is defined as the set of all asymptotic invariants vectors $\la$ such that
\begin{align} 
\begin{aligned} 
&\lrpar{R^{1/2} \vert \mathbf{E}(\la)- \mathbf{E}_\infty \vert}^2 + \lrpar{R^{1/2} \vert \mathbf{P}(\la) \vert}^2\\
&+\lrpar{R^{-1/4} \vert \mathbf{L}(\la) \vert}^2+\lrpar{R^{-1/2} \vert \mathbf{C}(\la) \vert}^2 \leq \lrpar{\mathbf{E}_\infty}^2.
\end{aligned}\label{EQdefeereinftyset999002}
\end{align}

\ni It is clear that for fixed $R\geq1$ large, the set $\EE_R(\mathbf{E}_\infty)$ is non-empty and a subset of $I(0)\times \RRR^3 \times \RRR^3 \times \RRR^3$. In particular, for $R\geq1$ large, for each $\la \in \EE_R(\mathbf{E}_\infty)$ there is a Kerr spacelike initial data set with asymptotic invariants equal to $\la$.

We define Kerr reference sphere data as follows.
\begin{definition}[Kerr reference sphere data] \label{DEFKerrreferenceSPHEREDATA} Let $R\geq1$ and $\mathbf{E}_\infty>0$ be two real numbers, and let $\la \in \EE_R(\mathbf{E}_\infty)$ be an asymptotics invariants vector. We define the Kerr reference sphere data $x_{-R,2R}^{\KK(\la)}$ by
\begin{align*} 
\begin{aligned} 
x_{-R,2R}^{\KK(\la)}:= {}^{(R^{-1})}\lrpar{x_{-1,2}^{\KK({}^{(R)}\la)}},
\end{aligned} %\label{}
\end{align*}
where the sphere data $x_{-1,2}^{\KK({}^{(R)}\la)}$ is defined to be constructed on $S_{r_{\mathbf{E}_\infty/R}(-1,2)}$ from the spacelike initial data $(g^{\KK({}^{(R)}\la)},k^{\KK({}^{(R)}\la)})$ on $A_{[1,3]}$ as in Section \ref{SECproofExistenceConstruction}, and ${}^{(R)}\la$ is defined in \eqref{EQLEMscalingADMlocal}.
\end{definition}

\ni \emph{Remarks on Definition \ref{DEFKerrreferenceSPHEREDATA}.}
\begin{enumerate}
\item Definition \ref{DEFKerrreferenceSPHEREDATA} generalizes in a straight-forward manner to yield Kerr reference higher-order sphere data of order $m\geq1$,
\begin{align*} 
\begin{aligned} 
\lrpar{x^{\Kerr}_{-R,2R},{\DD}^{L,m, \Kerr}_{-R,2R}, {\DD}^{\Lb,m, \Kerr}_{-R,2R}},
\end{aligned} %\label{}
\end{align*}
which is used in Theorem \ref{THMdoubleCHARgluingTOkerr}, see also Remark \ref{REMARKConstructionHigherOrder}.

\end{enumerate}

\ni The following estimates for Kerr reference sphere data are essential for the characteristic gluing to Kerr.
\begin{proposition}[Estimates for Kerr reference sphere data] \label{PROPspheredataESTIM91} Let $R\geq1$ and $\mathbf{E}_\infty>0$ be two real numbers. For $R\geq1$ sufficiently large the Kerr reference sphere data
\begin{align*} 
\begin{aligned} 
x_{-R,2R}^{\KK(\la)} \text{ is well-defined for }\la \in \EE_R(\mathbf{E}_\infty),
\end{aligned} %\label{}
\end{align*}
and it holds that
\begin{align} 
\begin{aligned} 
\Vert x_{-R,2R}^{\KK(\la)} - \mathfrak{m}^{\mathbf{E}_\infty} \Vert_{\XX(S_{-R,2R})} \les& R^{-1} \cdot \vert \mathbf{E}(\la)-\mathbf{E}_\infty \vert + R^{-1} \cdot \vert \mathbf{P}(\la) \vert \\
&+ R^{-2} \cdot \vert \mathbf{L}(\la) \vert + R^{-2} \cdot \vert \mathbf{C}(\la) \vert \\
&+ \lrpar{\frac{ R^{-2} \cdot \vert \mathbf{L}(\la) \vert + \frac{\vert \mathbf{P}(\la) \vert}{\mathbf{E}_\infty} \cdot R^{-2} \cdot \vert \mathbf{C}(\la) \vert}{\mathbf{E}_\infty/R}}^2.
\end{aligned} \label{PROPspheredataESTIM91estim1}
\end{align}
Moreover, we have the estimates, for $m=-1,0,1$ and $(i_{-1},i_0,i_1)=(2,3,1)$,
\begin{align} 
\begin{aligned} 
\mathbf{E}\lrpar{x_{-R,2R}^{\KK(\la)}} =& \mathbf{E}^{\mathrm{loc}}_{\mathrm{ADM}}\lrpar{S_{-R,2R},\la}+ \OO(R^{-3/2}), \\
\mathbf{P}^m\lrpar{x_{-R,2R}^{\KK(\la)}} =& \lrpar{\mathbf{P}^{\mathrm{loc}}_{\mathrm{ADM}}}^{i_m}\lrpar{S_{-R,R},\la}+\OO(R^{-3/2}),\\
\mathbf{L}^m\lrpar{x_{-R,2R}^{\KK(\la)}} =& \lrpar{\mathbf{L}^{\mathrm{loc}}_{\mathrm{ADM}}}^{i_m}\lrpar{S_{-R,2R},\la}+\OO(R^{-1/2}),\\
\mathbf{G}^m\lrpar{x_{-R,2R}^{\KK(\la)}} =& \lrpar{\mathbf{C}^{\mathrm{loc}}_{\mathrm{ADM}}}^{i_m}\lrpar{S_{-R,2R},\la} \\
&- r\lrpar{S_{-R,2R},\la} \cdot \lrpar{\mathbf{P}^{\mathrm{loc}}_{\mathrm{ADM}}}^{i_m}\lrpar{S_{-R,2R},\la} + \OO(R^{-1/2}).
\end{aligned} \label{PROPspheredataESTIM91estim2}
\end{align}

\end{proposition}

\begin{proof}[Proof of Proposition \ref{PROPspheredataESTIM91}] As $\la \in \EE_R(\mathbf{E}_\infty)$, see its definition in \eqref{EQdefeereinftyset999002}, the rescaled asymptotic invariants vector
\begin{align} 
\begin{aligned} 
{}^{(R)}\la = \lrpar{R^{-1}\mathbf{E}(\la),R^{-1}\mathbf{P}(\la),R^{-2}\mathbf{L}(\la),R^{-2}\mathbf{C}(\la)}
\end{aligned} \label{EQrescaledla9991}
\end{align}
satisfies for $R\geq1$ sufficiently large,
\begin{align*} 
\begin{aligned} 
&\left\vert \mathbf{E}({}^{(R)}\la)-\frac{\mathbf{E}_\infty}{R} \right\vert^2 + \vert \mathbf{P}({}^{(R)}\la) \vert^2+ \vert \mathbf{L}({}^{(R)}\la) \vert^2+ \vert \mathbf{C}({}^{(R)}\la) \vert^2 \\
=& \frac{1}{R^2 \cdot R} \lrpar{ R^{1/2} \left\vert \mathbf{E}(\la)-\mathbf{E}_\infty \right\vert}^2+\frac{1}{R^2 \cdot R}\lrpar{R^{1/2}\vert \mathbf{P}(\la) \vert}^2\\
&+\frac{1}{R^4 \cdot R^{-1/2} }\lrpar{R^{-1/4} \vert \mathbf{L}(\la) \vert}^2+\frac{1}{R^4 \cdot R^{-1} }\lrpar{R^{-1/2}\vert \mathbf{C}(\la) \vert}^2\\
\leq& \frac{1}{R^3} \lrpar{\mathbf{E}_\infty}^2 \\
\les& \lrpar{\varep_\infty \frac{\mathbf{E}_\infty}{R}}^2.
\end{aligned} %\label{}
\end{align*}

\ni Hence for $R\geq1$ sufficiently we can apply Proposition \ref{CORcombinedEstimates222} with $\mathbf{E}=\mathbf{E}_\infty/R$ to get that
\begin{align} 
\begin{aligned} 
&\Vert (g^{\KK({}^{(R)} \la)},k^{\KK({}^{(R)}\la)}) -(g^{\mathbf{E}_\infty/R}, k^{\mathbf{E}_\infty/R}) \Vert_{C^{k_0}(A_{[1,3]}) \times C^{k_0-1}(A_{[1,3]})} \\
\les& \vert \mathbf{E}({}^{(R)}\la)-\mathbf{E}_\infty/R \vert + \vert \mathbf{P}({}^{(R)}\la) \vert + \vert \mathbf{L}({}^{(R)}\la) \vert + \vert \mathbf{C}({}^{(R)}\la) \vert \\
&+ \lrpar{\frac{\vert \mathbf{L}({}^{(R)}\la) \vert + \frac{\vert \mathbf{P}({}^{(R)}\la) \vert}{\mathbf{E}_\infty/R} \vert \mathbf{C}({}^{(R)}\la) \vert}{\mathbf{E}_\infty/R}}^2 \\
=& R^{-1} \cdot \vert \mathbf{E}(\la)-\mathbf{E}_\infty \vert + R^{-1} \cdot \vert \mathbf{P}(\la) \vert + R^{-2} \cdot \vert \mathbf{L}(\la) \vert + R^{-2} \cdot \vert \mathbf{C}(\la) \vert \\
&+ \lrpar{\frac{ R^{-2} \cdot \vert \mathbf{L}(\la) \vert + \frac{\vert \mathbf{P}(\la) \vert}{\mathbf{E}_\infty} \cdot R^{-2} \cdot \vert \mathbf{C}(\la) \vert}{\mathbf{E}_\infty/R}}^2 \\
=&\OO(R^{-3/2}),
\end{aligned} \label{EQsmallnessEstimatesParameter1777778}
\end{align}
where we used \eqref{EQrescaledla9991} and that $\la \in \EE_R(\mathbf{E}_\infty)$, see \eqref{EQdefeereinftyset999002}. For $R\geq1$ large, this shows that the sphere data $x_{-1,2}^{\KK({}^{(R)}\la)}$ is well-defined, and by the estimates of Section \ref{SECproofExistenceConstruction},
\begin{align*} 
\begin{aligned} 
\Vert x_{-1,2}^{\KK({}^{(R)}\la)} - \mathfrak{m}^{\mathbf{E}_\infty} \Vert_{\XX(S_{-1,2})} \les& R^{-1} \cdot \vert \mathbf{E}(\la)-\mathbf{E}_\infty \vert + R^{-1} \cdot \vert \mathbf{P}(\la) \vert \\
& + R^{-2} \cdot \vert \mathbf{L}(\la) \vert + R^{-2} \cdot \vert \mathbf{C}(\la) \vert \\
&+ \lrpar{\frac{ R^{-2} \cdot \vert \mathbf{L}(\la) \vert + \frac{\vert \mathbf{P}(\la) \vert}{\mathbf{E}_\infty} \cdot R^{-2} \cdot \vert \mathbf{C}(\la) \vert}{\mathbf{E}_\infty/R}}^2.
\end{aligned} %\label{}
\end{align*}
By the scale-invariance of the sphere data norm, see Lemma \ref{LEMspheredataInvarianceSCale}, and \eqref{EQdefSCALINGSSprop8999},  \eqref{PROPspheredataESTIM91estim1} follows.

By \eqref{EQsmallnessEstimatesParameter1777778} we can moreover apply the estimates of Sections \ref{SECcomparisonENERGY}, \ref{SECcomparisonLINEAR}, \ref{SECcomparisonANGULAR}, \ref{SECcomparisonCENTER} with $\varep_R = \OO(R^{-3/2})$, see also \eqref{EQREMgeneralisedEstimates}, and get for $m=-1,0,1$ and $(i_{-1},i_0,i_1)=(2,3,1)$,
\begin{align*} 
\begin{aligned} 
\mathbf{E}\lrpar{{}^{(R)}x_{-1,2}^{\KK(\la)}} =& \lrpar{\mathbf{E}^{\mathrm{loc}}_{\mathrm{ADM}}}^{i_m}\lrpar{S_{r_{\mathbf{E}_\infty/R}(-1,2)},{}^{(R)}\la}+ \OO(R^{-5/2}), \\
\mathbf{P}^m\lrpar{{}^{(R)}x_{-1,2}^{\KK(\la)}} =& \lrpar{\mathbf{P}^{\mathrm{loc}}_{\mathrm{ADM}}}^{i_m}\lrpar{S_{r_{\mathbf{E}_\infty/R}(-1,2)},{}^{(R)}\la}+\OO(R^{-5/2}),\\
\mathbf{L}^m\lrpar{{}^{(R)}x_{-1,2}^{\KK(\la)}} =& \lrpar{\mathbf{L}^{\mathrm{loc}}_{\mathrm{ADM}}}^{i_m}\lrpar{S_{r_{\mathbf{E}_\infty/R}(-1,2)},{}^{(R)}\la}+\OO(R^{-5/2}),\\
\mathbf{G}^m\lrpar{{}^{(R)}x_{-1,2}^{\KK(\la)}} =& \lrpar{\mathbf{C}^{\mathrm{loc}}_{\mathrm{ADM}}}^{i_m}\lrpar{S_{r_{\mathbf{E}_\infty/R}(-1,2)},{}^{(R)}\la} \\
&- r\lrpar{S_{r_{\mathbf{E}_\infty/R}(-1,2)},{}^{(R)}\la} \cdot \lrpar{\mathbf{P}^{\mathrm{loc}}_{\mathrm{ADM}}}^{i_m}\lrpar{S_{r_{\mathbf{E}_\infty/R}(-1,2)},{}^{(R)}\la} + O(R^{-5/2}).
\end{aligned} %\label{}
\end{align*}
Rescaling, see Lemmas \ref{LEMrescaleLOCALCHARGES} and \eqref{EQLEMscalingADMlocal}, yields \eqref{PROPspheredataESTIM91estim2}. This finishes the proof of Proposition \ref{PROPspheredataESTIM91}. \end{proof}

%%%%%%%%%%%%%%%%%%%%%%%%%%%%%%%%%%%%%%%%
\subsection{Asymptotics of Kerr spacelike initial data and sphere data} \label{SECchargeEstimateKERR444} 

\ni In this section we discuss the convergence of local integrals in Kerr reference spacelike data and of charges of Kerr reference sphere data. The convergence rates summarized in Proposition \ref{CORfullchargeexpressionSRkerrreference} below are essential for the matching to Kerr in Section \ref{SECconclusionofMainTheorem}, see \eqref{EQasympDECAY55540002}.

The following is the main result of this section.
\begin{proposition}[Convergence of local integrals to asymptotic invariants] \label{PROPKerrchargesLOCALglobal10101}
Let $R\geq1$ and $\mathbf{E}_\infty>0$ be two real numbers. For $R\geq1$ sufficiently large it holds that for all $\la \in \EE_R(\mathbf{E}_\infty)$, 
\begin{align} 
\begin{aligned} 
\mathbf{E}_{\mathrm{ADM}}^{\mathrm{loc}}(\la,R) =& \mathbf{E}(\la) + \OO(R^{-1}), &
\mathbf{P}_{\mathrm{ADM}}^{\mathrm{loc}}(\la,R) =& \mathbf{P}(\la) + \OO(R^{-3/2}), \\
\mathbf{L}_{\mathrm{ADM}}^{\mathrm{loc}}(\la,R) =& \mathbf{L}(\la) + \OO(R^{-1/2}), &
\mathbf{C}_{\mathrm{ADM}}^{\mathrm{loc}}(\la,R) =& \mathbf{C}(\la) + \OO(R^{-1/2}).
\end{aligned} \label{EQquantEstimKERR79999}
\end{align}
In particular,
\begin{align*} 
\begin{aligned} 
\mathbf{C}_{\mathrm{ADM}}^{\mathrm{loc}}(\la,R) - r(\la, R) \cdot \mathbf{P}_{\mathrm{ADM}}^{\mathrm{loc}}(\la, R) = \mathbf{C}(\la) -R\cdot \mathbf{P}(\la) + \OO(R^{-1/2}).
\end{aligned} %\label{}
\end{align*}
\end{proposition}

%%%%%%%%%%%%%%%%%%%%%%%%%%%%%%%%%%%%%%%%
\begin{proof}[Proof of Proposition \ref{PROPKerrchargesLOCALglobal10101}] We first derive decay estimates and then consider each estimate separately.\\

\ni \textbf{Derivation of decay estimates.} In the following we derive decay estimates by using Proposition \ref{CORcombinedEstimates222} and the scaling of Kerr spacelike initial data, see Remark \ref{RemarkScalingKerrData}.

Let $R\geq1$ be a real number and let $\la \in \EE_R(\mathbf{E}_\infty)$. First we claim that for $R\geq1$ sufficiently large, $x\in \RRR^3 \setminus B(0,1)$ and $i,j=1,2,3$,
\begin{align} 
\begin{aligned} 
\left\vert k_{ij}^{\KK({}^{(R)}\la)}(x) \right\vert \les \frac{1}{\vert x\vert^2} R^{-3/2},
\end{aligned} \label{EQfirstEstimateRescaleKerrk}
\end{align}
where we recall from \eqref{EQrescaledLAMBDAKerr} that
\begin{align*} 
\begin{aligned} 
{}^{(R)}\la = \lrpar{R^{-1} \mathbf{E}(\la), R^{-1} \mathbf{P}(\la), R^{-2} \mathbf{L}(\la), R^{-2} \mathbf{C}(\la)}.
\end{aligned}
\end{align*}
Indeed, let $x_0 \in \RRR^3 \setminus B(0,1)$ and let $r_0 := \vert x_0 \vert$. On the one hand, by the scaling of Kerr spacelike initial data, see \eqref{EQscalingderivatives} and Remark \ref{RemarkScalingKerrData}, we have that
\begin{align} 
\begin{aligned} 
\left\vert k_{ij}^{\KK({}^{(R)}\la)}(x_0) \right\vert =& \frac{1}{r_0} \left\vert k_{ij}^{\KK({}^{(R\cdot r_0)}\la)}\lrpar{\frac{x_0}{r_0}} \right\vert,
\end{aligned} \label{EQestimatekKerr1777}
\end{align}
where we underline that $x_0/r_0 \in S_1 \subset \RRR^3$.

On the other hand, as $\la \in \EE_R(\mathbf{E}_\infty)$, see its definition in \eqref{EQdefeereinftyset999002}, the asymptotic invariants vector ${}^{(R\cdot r_0)}\la$ satisfies for $R\geq1$ sufficiently large
\begin{align*} 
\begin{aligned} 
&\left\vert \mathbf{E}({}^{(R\cdot r_0)}\la)-\frac{\mathbf{E}_\infty}{R\cdot r_0} \right\vert^2 + \vert \mathbf{P}({}^{(R\cdot r_0)}\la) \vert^2+ \vert \mathbf{L}({}^{(R\cdot r_0)}\la) \vert^2+ \vert \mathbf{C}({}^{(R\cdot r_0)}\la) \vert^2 \\
=& \frac{\left\vert \mathbf{E}(\la)-\mathbf{E}_\infty \right\vert^2}{(R\cdot r_0)^2} + \frac{\vert \mathbf{P}(\la) \vert^2}{(R\cdot r_0)^2}+ \frac{\vert \mathbf{L}(\la) \vert^2}{(R\cdot r_0)^4}+ \frac{\vert \mathbf{C}(\la) \vert^2}{(R\cdot r_0)^4}\\
=& \frac{\lrpar{R^{1/2} \left\vert \mathbf{E}(\la)-\mathbf{E}_\infty \right\vert}^2}{(R\cdot r_0)^2 \cdot R} + \frac{\lrpar{R^{1/2}\vert \mathbf{P}(\la) \vert}^2}{(R\cdot r_0)^2\cdot R}+ \frac{\lrpar{ R^{-1/4} \vert \mathbf{L}(\la) \vert}^2}{(R\cdot r_0)^4\cdot R^{-1/2}}+ \frac{\lrpar{R^{-1/2}\vert \mathbf{C}(\la) \vert}^2}{(R\cdot r_0)^4 \cdot R^{-1}}\\
\leq& \frac{1}{(R\cdot r_0)^2} \lrpar{\frac{\lrpar{\mathbf{E}_\infty}^2}{R} } \\
\les& \lrpar{\varep_\infty \frac{\mathbf{E}_\infty}{R\cdot r_0}}^2,
\end{aligned} %\label{}
\end{align*}
where $\varep_\infty>0$ is the universal constant of Proposition \ref{CORcombinedEstimates222}. Hence we can apply Proposition \ref{CORcombinedEstimates222} with $\la={}^{(R\cdot r_0)}\la$ and $\mathbf{E}=\mathbf{E}_\infty/(R\cdot r_0)$ to get, using that $\la \in \EE_R(\mathbf{E}_\infty)$, the estimate
\begin{align} 
\begin{aligned} 
\left\vert k_{ij}^{\KK({}^{(R\cdot r_0)}\la)}\lrpar{\frac{x_0}{r_0}} \right\vert \les& \vert \mathbf{E}({}^{(R\cdot r_0)}\la)-\lrpar{\mathbf{E}_\infty/(R\cdot r_0)} \vert + \vert \mathbf{P}({}^{(R\cdot r_0)}\la) \vert \\
&+ \vert \mathbf{L}({}^{(R\cdot r_0)}\la) \vert + \vert \mathbf{C}({}^{(R\cdot r_0)}\la) \vert \\
&+ \lrpar{\frac{\vert \mathbf{L}({}^{(R\cdot r_0)}\la) \vert + \frac{\vert \mathbf{P}({}^{(R\cdot r_0)}\la) \vert}{\mathbf{E}_\infty/(R\cdot r_0)} \vert \mathbf{C}({}^{(R\cdot r_0)}\la) \vert}{\mathbf{E}_\infty/(R\cdot r_0)}}^2 \\
=& \frac{1}{r_0} \OO(R^{-3/2}).
\end{aligned} \label{EQPlocalsmallnessKerrIngredient2}
\end{align}
Combining \eqref{EQestimatekKerr1777} and \eqref{EQPlocalsmallnessKerrIngredient2}, we get \eqref{EQfirstEstimateRescaleKerrk}, that is, for $R\geq1$ sufficiently large and $x \in \RRR^3 \setminus B(0,1)$,
\begin{align*} 
\begin{aligned} 
\left\vert k_{ij}^{\KK({}^{(R)}\la)}(x) \right\vert \les \frac{1}{\vert x \vert^2}R^{-3/2}.
\end{aligned}
\end{align*}
By similar argumentation, we get the following estimates from Proposition \ref{CORcombinedEstimates222} and scaling,
\begin{align} 
\begin{aligned} 
\left\vert \pr^m \lrpar{g_{ij}^{\KK({}^{(R)}\la)} -g_{ij}^{\mathbf{E}_\infty/R}}(x) \right\vert \les&  \frac{1}{\vert x \vert^{1+m}}R^{-3/2} \text{ for } m=0,1,2, \\
\left\vert \pr_l k_{ij}^{\KK({}^{(R)}\la)}(x) \right\vert \les&  \frac{1}{\vert x \vert^3}R^{-3/2},
\end{aligned} \label{EQscalingDecayestimates1}
\end{align}
where we denoted the multi-index $\pr^m := \pr_{a_1} \dots \pr_{a_m}$ for $a_1, \dots a_m=1,2,3$. Furthermore, by the same argument based on scaling and the estimates for the RT-quantities in Proposition \ref{CORcombinedEstimates222}, we have that
\begin{align} 
\begin{aligned} 
\left\vert \pr^m \lrpar{\mathrm{RT}\lrpar{g^{\KK({}^{(R)}\la)}}_{ij}}(x) \right\vert \les& \frac{1}{\vert x \vert^{2+m}}R^{-3/2} & &\text{ for } m=0,1,2, \\
\left\vert \pr^m \lrpar{\mathrm{RT}\lrpar{k^{\KK({}^{(R)}\la)}}_{ij}}(x) \right\vert \les& \frac{1}{\vert x \vert^{3+m}}R^{-3/2} & &\text{ for } m=0,1.
\end{aligned} \label{EQscalingDecayestimates2}
\end{align}
Moreover, from scaling \eqref{EQschwarzschildSPACELIKEscaling} and Lemma \ref{LEMschwarzschildtrivialEstimate}, we similarly have
\begin{align} 
\begin{aligned} 
\left\vert \pr^m \lrpar{ g^{\mathbf{E}_\infty/R}_{ij}-e_{ij}}(x) \right\vert \les \frac{1}{\vert x\vert^{1+m}} R^{-1}.
\end{aligned} \label{EQschwarzschilddecayrates}
\end{align}

\ni \textbf{Notation.} In the following, we denote $(g^{\KK({}^{(R)}\la)}, k^{\KK({}^{(R)}\la)})$ by $(g,k)$, and the associated covariant derivative associated by $\nab$.\\

\ni In the following, we prove the estimates of Proposition \ref{PROPKerrchargesLOCALglobal10101}. We proceed in the order of 
\begin{align*} 
\begin{aligned} 
\mathbf{P}_{\mathrm{ADM}}^{\mathrm{loc}}, \mathbf{L}_{\mathrm{ADM}}^{\mathrm{loc}}, \mathbf{E}_{\mathrm{ADM}}^{\mathrm{loc}} \text{ and } \mathbf{C}_{\mathrm{ADM}}^{\mathrm{loc}}.
\end{aligned} %\label{}
\end{align*}

%%%%%%%%%%%%%%%%%%%%%%%%%%%%%%%%%%%%%%%%
\ni \textbf{Convergence of $\mathbf{P}_{\mathrm{ADM}}^{\mathrm{loc}}$.} First, by \eqref{EQLEMscalingADMlocal},
\begin{align} 
\begin{aligned} 
\mathbf{P}^{\mathrm{loc}}_{\mathrm{ADM}}(\la,R) - \mathbf{P}(\la) =& R \cdot \lrpar{\mathbf{P}^{\mathrm{loc}}_{\mathrm{ADM}}({}^{(R)}\la,1) - \mathbf{P}({}^{(R)}\la)}.
\end{aligned} \label{EQPrewritingKestimateConvergence777}
\end{align}
By Definition \ref{DEFlocalADMcharges}, Stokes' theorem and the spacelike constraint equations \eqref{EQspacelikeConstraints1}, we have that
\begin{align} 
\begin{aligned} 
&\mathbf{P}^{\mathrm{loc}}_{\mathrm{ADM}}({}^{(R)}\la,1) - \mathbf{P}({}^{(R)}\la)\\
 =&-\frac{1}{8\pi} \int\limits_{\RRR^3 \setminus B(0,1)} \nab^j \lrpar{k_{i j} - \tr k \, g_{i j }} d\mu_g\\
=&-\frac{1}{8\pi}\int\limits_{\RRR^3 \setminus B(0,1)} \lrpar{\underbrace{\lrpar{\Div k - d \tr k}}_{=0}(\pr_i) + (k_{jm}-\tr k \, g_{jm}) \nab^j \lrpar{\pr_i}^m } d\mu_g \\
=& -\frac{1}{8\pi} \int\limits_{\RRR^3 \setminus B(0,1)} (k_{jm}-\tr k \, g_{jm}) \nab^j \lrpar{\pr_i}^m d\mu_g.
\end{aligned} \label{EQPlocalsmallnessKerrIngredient0}
\end{align}
In the following we estimate the integrands on the right-hand side of \eqref{EQPlocalsmallnessKerrIngredient0}. On the one hand, by \eqref{EQfirstEstimateRescaleKerrk}, \eqref{EQscalingDecayestimates1} and \eqref{EQschwarzschilddecayrates} we have that for $x \in \RRR^3 \setminus B(0,1)$,
\begin{align} 
\begin{aligned} 
\left\vert \lrpar{k_{jm}-\tr k \, g_{jm}}(x) \right\vert \les \frac{1}{\vert x \vert^2} R^{-3/2}.
\end{aligned} \label{EQPestimateKerrparamKerr1118887788}
\end{align}

\ni On the other hand, by \eqref{EQscalingDecayestimates1}, \eqref{EQschwarzschilddecayrates} and the standard formula for Christoffel symbols,
\begin{align*} 
\begin{aligned} 
\Ga_{ji}^m := \half g^{ms} \lrpar{\pr_{j}g_{is}+ \pr_i g_{js}-\pr_s g_{ij} },
\end{aligned} %\label{}
\end{align*}
we have that for $R\geq1$ sufficiently large,
\begin{align} 
\begin{aligned} 
\left\vert \nab_j \lrpar{\pr_i}^m(x) \right\vert = \left\vert \Ga_{ji}^m \right\vert 
= \left\vert \half g^{ms} \lrpar{\pr_{j}g_{is}+ \pr_i g_{js}-\pr_s g_{ij} } \right\vert 
\les \frac{1}{\vert x\vert^2} R^{-1},
\end{aligned} \label{EQPlocalsmallnessKerrIngredient1}
\end{align}
where we estimated, by \eqref{EQscalingDecayestimates1} and \eqref{EQschwarzschilddecayrates},
\begin{align*}
\begin{aligned}
\pr_{j}g_{is} = (\underbrace{\pr_{j}g_{is}- \pr_{j}g^{\mathbf{E}_\infty/R}_{is}}_{\les \frac{1}{\vert x \vert^2} R^{-3/2} }) + (\underbrace{\pr_{j}g^{\mathbf{E}_\infty/R}_{is} - \pr_{j}e_{is}}_{\les \frac{1}{\vert x \vert^2} R^{-1}}) 
\les \frac{1}{\vert x \vert^2} R^{-1}.
\end{aligned} %\label{}
\end{align*}

\ni By \eqref{EQscalingDecayestimates1}, \eqref{EQschwarzschilddecayrates}, \eqref{EQPestimateKerrparamKerr1118887788} and \eqref{EQPlocalsmallnessKerrIngredient1}, the integral on the right-hand side of \eqref{EQPlocalsmallnessKerrIngredient0} can thus be bounded by
\begin{align*} 
\begin{aligned} 
\int\limits_{\RRR^3 \setminus B(0,1)} (k_{jm}-\tr k \, g_{jm}) \nab^j \lrpar{\pr_i}^m d\mu_g = \OO(R^{-5/2}).
\end{aligned} %\label{}
\end{align*}
Plugging this into \eqref{EQPlocalsmallnessKerrIngredient0} and subsequently into \eqref{EQPrewritingKestimateConvergence777} yields the desired estimate
\begin{align*} 
\begin{aligned} 
\mathbf{P}_{\mathrm{ADM}}^{\mathrm{loc}}(\la, R) - \mathbf{P}(\la)= \OO(R^{-3/2}).
\end{aligned} %\label{}
\end{align*}

%%%%%%%%%%%%%%%%%%%%%%%%%%%%%%%%%%%%%%%%
\ni \textbf{Convergence of $\mathbf{L}_{\mathrm{ADM}}^{\mathrm{loc}}$.} First, by \eqref{EQLEMscalingADMlocal},
\begin{align} 
\begin{aligned} 
\mathbf{L}^{\mathrm{loc}}_{\mathrm{ADM}}(\la,R) - \mathbf{L}(\la) =& R^2 \cdot \lrpar{\mathbf{L}^{\mathrm{loc}}_{\mathrm{ADM}}({}^{(R)}\la,1) - \mathbf{L}({}^{(R)}\la)}.
\end{aligned} \label{EQLrewritingKestimateConvergence777}
\end{align}

\ni By Definition \ref{DEFlocalADMcharges}, Stokes' theorem and the spacelike constraint equations \eqref{EQspacelikeConstraints1}, we have
\begin{align} 
\begin{aligned} 
&\lrpar{\mathbf{L}^{\mathrm{loc}}_{\mathrm{ADM}}}^i({}^{(R)}\la,1) - \mathbf{L}^i({}^{(R)}\la) \\
=& -\frac{1}{8\pi} \int\limits_{\RRR^3 \setminus B(0,1)} \nab^l \lrpar{ \lrpar{k_{l m} - \tr k \, g_{lm }} (Y_{(i)})^m }d\mu_g\\
=&-\frac{1}{8\pi}\int\limits_{\RRR^3 \setminus B(0,1)} \lrpar{\underbrace{\lrpar{\Div k - d \tr k}}_{=0}(Y_{(i)}) + \lrpar{k_{l m} - \tr k \, g_{lm }} \nab^l \lrpar{Y_{(i)}}^m } d\mu_g \\
=& -\frac{1}{16\pi } \int\limits_{\RRR^3 \setminus B(0,1)} \lrpar{k_{l m} - \tr k \, g_{lm }} \lrpar{ \nab^l \lrpar{Y_{(i)}}^m + \nab^m \lrpar{Y_{(i)}}^l } d\mu_g.
\end{aligned} \label{EQLlocalsmallnessKerrIngredient0}
\end{align}
The rotation fields $Y_{(i)}$, $i=1,2,3$, are Killing with respect to the Euclidean metric, that is,
\begin{align*} 
\begin{aligned} 
\pr_l \lrpar{Y_{(i)}}_m + \pr_m \lrpar{Y_{(i)}}_l =0,
\end{aligned} %\label{}
\end{align*}
and hence, with respect to $g$,
\begin{align*} 
\begin{aligned} 
\nab_l \lrpar{Y_{(i)}}_m + \nab_m \lrpar{Y_{(i)}}_l = \pr_l \lrpar{Y_{(i)}}_m + \pr_m \lrpar{Y_{(i)}}_l - 2 \Ga^j_{lm} \lrpar{Y_{(i)}}_j = - 2 \Ga^j_{lm} \in_{isj} \cdot x^s,
\end{aligned} %\label{}
\end{align*}
where we used that $(Y_{(i)})_j := \in_{ilj} \cdot x^l$ by \eqref{EQdefRotationFieldsComponents}. Hence we can express the integral on the right-hand side of \eqref{EQLlocalsmallnessKerrIngredient0} as
\begin{align} 
\begin{aligned} 
&\int\limits_{\RRR^3 \setminus B(0,1)} \lrpar{k_{l m} - \tr k \, g_{lm }} \lrpar{ \nab^l \lrpar{Y_{(i)}}^m + \nab^m \lrpar{Y_{(i)}}^l } d\mu_g \\
=& -2 \int\limits_{\RRR^3 \setminus B(0,1)} \lrpar{k^{l m} - \tr k \, g^{lm }} \lrpar{ \Ga^j_{lm} \in_{isj} \cdot x^s } d\mu_g.
\end{aligned} \label{EQenergyestimateintegraldecompLLLL7778}
\end{align}
The integral on the right-hand side of \eqref{EQenergyestimateintegraldecompLLLL7778} is, up to lower-order terms which are straight-forward to estimate, of the form
\begin{align} 
\begin{aligned} 
&\int\limits_{\RRR^3 \setminus B(0,1)} k_{mn} \cdot \pr_l g_{ij} \cdot x^s\,  d\mu_e \\
=& \frac{1}{4}\int\limits_{\RRR^3 \setminus B(0,1)} \underbrace{\mathrm{RT}\lrpar{k}_{mn}(x)}_{\les \frac{1}{\vert x \vert^3} R^{-3/2}} \cdot (\underbrace{\pr_l g_{ij}(x) - \pr_l g_{ij}(-x)}_{\les \frac{1}{\vert x \vert^2} R^{-1}}) \cdot x^s\, d\mu_e\\
&+\frac{1}{4}\int\limits_{\RRR^3 \setminus B(0,1)} (\underbrace{k_{mn}(x)-k_{mn}(-x)}_{\les \frac{1}{\vert x \vert^2} R^{-3/2}}) \cdot \underbrace{\pr_l \lrpar{ \mathrm{RT}\lrpar{g}_{ij}}(x)}_{\les \frac{1}{\vert x \vert^3} R^{-3/2}} \cdot x^s\, d\mu_e \\
=& \OO(R^{-5/2}),
\end{aligned} \label{EQdefIntegralToBoundLKerrParameterCharge7778}
\end{align}
where we used that the coordinate function $x^s$ is odd, and applied \eqref{EQfirstEstimateRescaleKerrk}, \eqref{EQscalingDecayestimates1}, \eqref{EQscalingDecayestimates2} and \eqref{EQschwarzschilddecayrates}. From \eqref{EQdefIntegralToBoundLKerrParameterCharge7778} we conclude by \eqref{EQenergyestimateintegraldecompLLLL7778} that 
\begin{align*} 
\begin{aligned} 
\int\limits_{\RRR^3 \setminus B(0,1)} \half \lrpar{k_{l m} - \tr k \, g_{lm }} \lrpar{ \nab^l \lrpar{Y_{(i)}}^m + \nab^m \lrpar{Y_{(i)}}^l } d\mu_g = \OO(R^{-5/2}),
\end{aligned} %\label{}
\end{align*}
and subsequently, by \eqref{EQLlocalsmallnessKerrIngredient0} and \eqref{EQLrewritingKestimateConvergence777}, we arrive at the desired estimate
\begin{align*} 
\begin{aligned} 
\mathbf{L}^{\mathrm{loc}}_{\mathrm{ADM}}(\la,R) - \mathbf{L}(\la) = \OO(R^{-1/2}).
\end{aligned} %\label{}
\end{align*}

%%%%%%%%%%%%%%%%%%%%%%%%%%%%%%%%%%%%%%%%
\ni \textbf{Convergence of $\mathbf{E}_{\mathrm{ADM}}^{\mathrm{loc}}$.} First, by \eqref{EQLEMscalingADMlocal},
\begin{align} 
\begin{aligned} 
\mathbf{E}^{\mathrm{loc}}_{\mathrm{ADM}}(\la,R) - \mathbf{E}(\la) =& R \cdot \lrpar{\mathbf{E}^{\mathrm{loc}}_{\mathrm{ADM}}({}^{(R)}\la,1) - \mathbf{E}({}^{(R)}\la)}.
\end{aligned} \label{EQErewritingKestimateConvergence777}
\end{align}
By Definition \ref{DEFlocalADMcharges}, Stokes' theorem and the twice-traced Bianchi identity
\begin{align*} 
\begin{aligned} 
\Div \lrpar{\RRRic - \half R_{\mathrm{scal}} \, g }=0,
\end{aligned} %\label{}
\end{align*}
we have that, with $X:= x^s \pr_s$,
\begin{align} 
\begin{aligned} 
&\mathbf{E}^{\mathrm{loc}}_{\mathrm{ADM}}({}^{(R)}\la,1) - \mathbf{E}({}^{(R)}\la)\\
 =& \frac{1}{8\pi} \int\limits_{\RRR^3 \setminus B(0,1)} \Div\lrpar{ \lrpar{\RRRic_{\cdot j}- \half R_{\mathrm{scal}} \,g_{\cdot j} }X^j} d\mu_g \\
=& \frac{1}{8\pi}\int\limits_{\RRR^3 \setminus B(0,1)} \lrpar{\RRRic_{i j}- \half R_{\mathrm{scal}} \,g_{i j} } \nab^i X^j d\mu_g \\
=& \frac{1}{16\pi}\int\limits_{\RRR^3 \setminus B(0,1)} \lrpar{\RRRic_{i j}- \half R_{\mathrm{scal}} \,g_{i j} } \lrpar{ \nab^i X^j+\nab^j X^i}  d\mu_g.
\end{aligned} \label{EQElocalsmallnessKerrIngredient0}
\end{align}
In the following we estimate the integral on the right-hand side of \eqref{EQElocalsmallnessKerrIngredient0} based on the following two observations.
\begin{enumerate}
\item As remarked after Theorem \ref{THMaltDEFinvariants}, $X:= x^s \pr_s$ is a conformal Killing field of Euclidean space $(\RRR^3,e)$,
\begin{align*} 
\begin{aligned} 
\pr_i X_j+\pr_j X_i = 2 e_{ij},
\end{aligned} %\label{}
\end{align*}
where we used that $\mathrm{div}_e X := \pr_i X^i = 3$, and hence we have that
\begin{align} 
\begin{aligned} 
\nab_i X_j+\nab_j X_i =\pr_i X_j+\pr_j X_i -2 \Ga^m_{ij} X_m =  2 e_{ij} -2 \Ga^m_{ij} X_m,
\end{aligned} \label{EQconformalKillingXid99902}
\end{align}
where $e_{ij}=\de_{ij}$ denote the Cartesian components of the Euclidean metric.
\item By the spacelike constraint equations \eqref{EQspacelikeConstraints1}, that is, $R_{\mathrm{scal}}=\vert k \vert^2 - (\tr k)^2$, we have
\begin{align} 
\begin{aligned} 
g^{ij} \lrpar{\RRRic_{ij} - \half R_{\mathrm{scal}} \,g_{ij} } =-\half R_{\mathrm{scal}} = -\half \lrpar{ \vert k \vert^2 - (\tr k)^2}.
\end{aligned} \label{EQtraceoftensorKerrparameterEstimate222777}
\end{align}
\end{enumerate}
We are now in position to bound the integral on the right-hand side of \eqref{EQElocalsmallnessKerrIngredient0}. By \eqref{EQconformalKillingXid99902} and \eqref{EQtraceoftensorKerrparameterEstimate222777},
\begin{align} 
\begin{aligned} 
&\int\limits_{\RRR^3 \setminus B(0,1)} \lrpar{\RRRic_{ij} - \half R_{\mathrm{scal}} \,g_{ij} } \lrpar{ \nab^i X^j+\nab^j X^i}  d\mu_g \\
=& \int\limits_{\RRR^3 \setminus B(0,1)} \lrpar{\RRRic_{ij} - \half R_{\mathrm{scal}} \,g_{ij} } \lrpar{2 e^{ij}}  d\mu_g \\
&-2 \int\limits_{\RRR^3 \setminus B(0,1)} \lrpar{\RRRic^{ij} - \half R_{\mathrm{scal}} \,g^{ij} } \cdot \Ga^m_{ij} X_m \, d\mu_g\\
=& -\underbrace{\int\limits_{\RRR^3 \setminus B(0,1)} \lrpar{\vert k \vert^2 - (\tr k)^2}   d\mu_g}_{:= \II_1} \\
&- 2\underbrace{\int\limits_{\RRR^3 \setminus B(0,1)} \lrpar{\RRRic^{ij} - \half R_{\mathrm{scal}} \,g^{ij} } \lrpar{ g^{ij}- e^{ij}}  d\mu_g}_{:= \II_2}\\
&-2 \underbrace{\int\limits_{\RRR^3 \setminus B(0,1)} \lrpar{\RRRic_{ij} - \half R_{\mathrm{scal}} \,g_{ij} } \cdot \Ga^m_{ij} X_m \, d\mu_g}_{:=\II_3}.
\end{aligned} \label{EQenergyestimateintegraldecomp17778}
\end{align}
In the following we estimate the integrals $\II_1, \II_2$ and $\II_3$ on the right-hand side of \eqref{EQenergyestimateintegraldecomp17778}. First, using \eqref{EQfirstEstimateRescaleKerrk}, \eqref{EQscalingDecayestimates1} and \eqref{EQschwarzschilddecayrates}, we can estimate $\II_1$ as follows (where we only write out the highest-order term),
\begin{align} 
\begin{aligned}
\II_1= \int\limits_{\RRR^3 \setminus B(0,1)} \underbrace{k_{ij}(x)}_{\les \frac{1}{\vert x \vert^2} R^{-3/2} } \cdot \underbrace{k_{ab}(x)}_{\les \frac{1}{\vert x \vert^2} R^{-3/2}} \, d\mu_e + \OO(R^{-3})
= \OO(R^{-3}).
\end{aligned} \label{EQKerrparamIntegral17778}
\end{align}

\ni Second, using \eqref{EQfirstEstimateRescaleKerrk}, \eqref{EQscalingDecayestimates1} and \eqref{EQschwarzschilddecayrates}, we can estimate $\II_2$ as follows (where we only write out the highest-order term),
\begin{align} 
\begin{aligned}
\II_2 = \int\limits_{\RRR^3 \setminus B(0,1)} \underbrace{\pr_s\pr_c g_{mn}(x)}_{\les \frac{1}{\vert x \vert^3} R^{-1} }\cdot \underbrace{(g_{ij}- e_{ij})(x)}_{\les \frac{1}{\vert x \vert} R^{-1} }  d\mu_e + \OO(R^{-2}) = \OO(R^{-2}).
\end{aligned} \label{EQKerrparamIntegral27778}
\end{align}

\ni Third, using \eqref{EQfirstEstimateRescaleKerrk}, \eqref{EQscalingDecayestimates1} and \eqref{EQschwarzschilddecayrates} and that $X:= x^s \pr_s$, we can estimate $\II_3$ as follows (where we only write out the highest-order term),
\begin{align} 
\begin{aligned}
\II_3= \int\limits_{\RRR^3 \setminus B(0,1)} \underbrace{\pr_s\pr_c g_{mn}(x)}_{\les \frac{1}{\vert x \vert^3} R^{-1}} \cdot \underbrace{\pr_l g_{ij}(x)}_{\les \frac{1}{\vert x \vert^2} R^{-1}} \cdot x^s \, d\mu_e +\OO(R^{-2}) = \OO(R^{-2}).
\end{aligned} \label{EQKerrparamIntegral37778}
\end{align}

\ni Plugging \eqref{EQKerrparamIntegral17778}, \eqref{EQKerrparamIntegral27778} and \eqref{EQKerrparamIntegral37778} into \eqref{EQenergyestimateintegraldecomp17778} shows that
\begin{align*} 
\begin{aligned} 
\int\limits_{\RRR^3 \setminus B(0,1)} \lrpar{\RRRic_{ij} -\half R_{\mathrm{scal}} \,g_{ij} } \lrpar{ \nab^i X^j+\nab^j X^i}  d\mu_g = \OO(R^{-2}),
\end{aligned} %\label{}
\end{align*}
and hence, by \eqref{EQElocalsmallnessKerrIngredient0} and \eqref{EQErewritingKestimateConvergence777}, this proves the desired estimate
\begin{align} 
\begin{aligned} 
\mathbf{E}^{\mathrm{loc}}_{\mathrm{ADM}}(\la,R) - \mathbf{E}(\la) = \OO\lrpar{R^{-1}}.
\end{aligned} \label{EQenergyEstimateKERRreferenceEconv}
\end{align}
%%%%%%%%%%%%%%%%%%%%%%%%%%%%%%%%%%%%%%%%
\ni \textbf{Convergence of $\mathbf{C}_{\mathrm{ADM}}^{\mathrm{loc}}$.} First, by \eqref{EQLEMscalingADMlocal},
\begin{align} 
\begin{aligned} 
\mathbf{C}^{\mathrm{loc}}_{\mathrm{ADM}}(\la,R) - \mathbf{C}(\la) =& R^2 \cdot \lrpar{\mathbf{C}^{\mathrm{loc}}_{\mathrm{ADM}}({}^{(R)}\la,1) - \mathbf{C}({}^{(R)}\la)}.
\end{aligned} \label{EQCrewritingKestimateConvergence777}
\end{align}
By Definition \ref{DEFlocalADMcharges}, Stokes' theorem and the twice-traced Bianchi identity
\begin{align*} 
\begin{aligned} 
\Div \lrpar{\RRRic - \half R_{\mathrm{scal}} \, g }=0,
\end{aligned} %\label{}
\end{align*}
we have that
\begin{align} 
\begin{aligned} 
&\mathbf{C}^{\mathrm{loc}}_{\mathrm{ADM}}({}^{(R)}\la,1) - \mathbf{C}({}^{(R)}\la)\\
 =&- \frac{1}{16\pi} \int\limits_{\RRR^3 \setminus B(0,1)} \Div\lrpar{\lrpar{\RRRic_{\cdot j} -\half R_{\mathrm{scal}} \,g_{\cdot j} }\lrpar{Z^{(i)}}^j} d\mu_g \\
=&-\frac{1}{16\pi}\int\limits_{\RRR^3 \setminus B(0,1)} \lrpar{\RRRic_{sj} -\half R_{\mathrm{scal}} \,g_{sj} } \nab^s\lrpar{Z^{(i)}}^j \mu_g\\
=&\frac{1}{32\pi} \int\limits_{\RRR^3 \setminus B(0,1)} \lrpar{\RRRic_{sj} - \half R_{\mathrm{scal}} \,g_{sj}  } \lrpar{\nab^s\lrpar{Z^{(i)}}^j+\nab^j\lrpar{Z^{(i)}}^s} d\mu_g.
\end{aligned} \label{EQClocalsmallnessKerrIngredient0}
\end{align}
In the following we estimate the integral on the right-hand side of \eqref{EQClocalsmallnessKerrIngredient0}. We note that the vectorfields
\begin{align} 
\begin{aligned} 
Z^{(i)} := \lrpar{\vert x \vert^2 \de^{ij}-2x^i x^j} \pr_j \text{ for } i=1,2,3,
\end{aligned} \label{EQrecallZidefinition99909}
\end{align}
are conformal Killing fields of $(\RRR^3,e)$, that is,
\begin{align*} 
\begin{aligned} 
\pr_s\lrpar{Z^{(i)}}_j+\pr_j\lrpar{Z^{(i)}}_s = \frac{2}{3} e_{sj} \, \mathrm{div}_e(Z^{(i)}) = - 4 e_{sj} \cdot x^i,
\end{aligned} %\label{}
\end{align*}
where the Euclidean divergence operator is defined for vectorfields $X$ by $\mathrm{div}_e(X):= \pr_i X^i$ and we calculated that for $i=1,2,3$, 
\begin{align*} 
\begin{aligned} 
\mathrm{div}_e\lrpar{Z^{(i)}}=-6x^i.
\end{aligned} %\label{}
\end{align*}
Hence it follows that
\begin{align} 
\begin{aligned} 
\nab_s\lrpar{Z^{(i)}}_j+\nab_j\lrpar{Z^{(i)}}_s = - 4 e_{sj} \cdot x^i - 2 \Ga_{sj}^m \lrpar{Z^{(i)}}_m.
\end{aligned} \label{EQZiConformalKillingUsage}
\end{align}
Plugging \eqref{EQZiConformalKillingUsage} into \eqref{EQClocalsmallnessKerrIngredient0} and using \eqref{EQtraceoftensorKerrparameterEstimate222777}, we get that
\begin{align} 
\begin{aligned} 
&\int\limits_{\RRR^3 \setminus B(0,1)} \lrpar{\RRRic_{sj} - \half R_{\mathrm{scal}} \,g_{sj} } \lrpar{ \nab^s \lrpar{Z^{(i)}}^j+\nab^j \lrpar{Z^{(i)}}^s}  d\mu_g \\
=& \int\limits_{\RRR^3 \setminus B(0,1)} \lrpar{\RRRic^{sj} - \half R_{\mathrm{scal}} \,g^{sj} } \lrpar{-4 e_{sj} \cdot x^i}  d\mu_g \\
&-2  \int\limits_{\RRR^3 \setminus B(0,1)} \lrpar{\RRRic^{sj} - \half R_{\mathrm{scal}} \,g^{sj} } \cdot \Ga^m_{sj} \lrpar{Z^{(i)}}_m \, d\mu_g\\
=& 2 \underbrace{\int\limits_{\RRR^3 \setminus B(0,1)}  \lrpar{\vert k \vert^2 - (\tr k)^2}   \cdot x^i \, d\mu_g}_{:=\II_1}\\
&+4 \underbrace{\int\limits_{\RRR^3 \setminus B(0,1)} \lrpar{\RRRic^{sj} - \half R_{\mathrm{scal}} \,g^{sj} } \lrpar{g^{sj}- e^{sj}}\cdot x^i\,  d\mu_g}_{:=\II_2}\\
&-2 \underbrace{\int\limits_{\RRR^3 \setminus B(0,1)}  \lrpar{\RRRic^{sj} - \half R_{\mathrm{scal}} \,g^{sj} } \cdot \Ga^m_{sj} \lrpar{Z^{(i)}}_m \, d\mu_g}_{:=\II_3}.
\end{aligned} \label{EQCCCestimateintegraldecomp17778}
\end{align}
First, using \eqref{EQfirstEstimateRescaleKerrk}, \eqref{EQscalingDecayestimates1}, \eqref{EQscalingDecayestimates2}, \eqref{EQschwarzschilddecayrates} and using that $x^s$ is an odd function, we can estimate $\II_1$ as follows (where we only write out the highest-order term), 
\begin{align} 
\begin{aligned}
\II_1 =&\int\limits_{\RRR^3 \setminus B(0,1)} k_{ij} \cdot k_{ab} \cdot x^s \, d\mu_e + \OO(R^{-3}) \\
=& \frac{1}{4} \int\limits_{\RRR^3 \setminus B(0,1)} \underbrace{\mathrm{RT}\lrpar{k}_{ij}(x)}_{\les \frac{1}{\vert x \vert^3}R^{-3/2}} \cdot (\underbrace{k_{ab}(x)-k_{ab}(-x)}_{\les \frac{1}{\vert x \vert^2}R^{-3/2}}) \cdot x^s d\mu_e\\
&+ \frac{1}{4} \int\limits_{\RRR^3 \setminus B(0,1)} (\underbrace{k_{ij}(x)-k_{ij}(-x)}_{\les \frac{1}{\vert x \vert^2}R^{-3/2}}) \cdot \underbrace{\mathrm{RT}\lrpar{k}_{ab}(x)}_{\les \frac{1}{\vert x \vert^3}R^{-3/2}} \cdot x^s d\mu_e+ \OO(R^{-3}) \\
=& \OO(R^{-3}).
\end{aligned} \label{EQCCCKerrparamIntegral17778}
\end{align}
Second, using \eqref{EQfirstEstimateRescaleKerrk}, \eqref{EQscalingDecayestimates1}, \eqref{EQscalingDecayestimates2}, \eqref{EQschwarzschilddecayrates} and using that $x^s$ is an odd function and the components of the Euclidean metric $e_{ij}$ are even, we can estimate $\II_2$ as follows (where we only write out the highest-order term),
\begin{align} 
\begin{aligned}
\II_2 =&\int\limits_{\RRR^3 \setminus B(0,1)} \pr_l\pr_m g_{ij} \cdot (g_{ab}- e_{ab}) \cdot x^s\,  d\mu_e + \OO(R^{-5/2})\\
=& \frac{1}{4} \int\limits_{\RRR^3 \setminus B(0,1)} \underbrace{\pr_l\pr_m\mathrm{RT}(g)_{ij}(x)}_{\les \frac{1}{\vert x \vert^4} R^{-3/2}} \cdot \Big( \underbrace{(g_{ab}- e_{ab})(x) + (g_{ab}- e_{ab})(-x)}_{\les \frac{1}{\vert x \vert}R^{-1}} \Big) \cdot x^s\,  d\mu_e \\
&+ \frac{1}{4} \int\limits_{\RRR^3 \setminus B(0,1)} \lrpar{ \underbrace{\pr_l\pr_m g_{ij}(x)+\pr_l\pr_m g_{ij}(-x)}_{\les \frac{1}{\vert x \vert^3} R^{-1}}} \cdot \underbrace{\mathrm{RT}(g)_{ab}(x)}_{\les \frac{1}{\vert x \vert^2} R^{-3/2}} \cdot x^s\,  d\mu_e+ \OO(R^{-5/2}) \\
=& \OO(R^{-5/2}).
\end{aligned} \label{EQCCCKerrparamIntegral27778}
\end{align}

\ni Third, using \eqref{EQfirstEstimateRescaleKerrk}, \eqref{EQscalingDecayestimates1}, \eqref{EQscalingDecayestimates2}, \eqref{EQschwarzschilddecayrates} and using that the components of $Z^{(i)}$ are of the even form $x^a \cdot x^d$, see \eqref{EQrecallZidefinition99909}, we can estimate $\II_3$ as follows (where we only write out the highest-order term),
\begin{align} 
\begin{aligned}
\II_3 =&\int\limits_{\RRR^3 \setminus B(0,1)}  \pr_s\pr_c g_{mn} \cdot \pr_l g_{ij} \cdot x^a \cdot x^d \, d\mu_e + \OO(R^{-5/2}) \\
=&\frac{1}{4} \int\limits_{\RRR^3 \setminus B(0,1)} \underbrace{\pr_s\pr_c \mathrm{RT}(g)_{mn}(x)}_{\les \frac{1}{\vert x \vert^4}R^{-3/2}} \cdot (\underbrace{\pr_l g_{ij}(x)-\pr_l g_{ij}(-x)}_{\les \frac{1}{\vert x \vert^2}R^{-1}} ) \cdot x^a \cdot x^d \, d\mu_e + \OO(R^{-5/2}) \\
&+ \frac{1}{4} \int\limits_{\RRR^3 \setminus B(0,1)}  (\underbrace{\pr_s\pr_c g_{mn}(x)+\pr_s\pr_c g_{mn}(-x)}_{\les \frac{1}{\vert x \vert^3}R^{-1}} ) \cdot \underbrace{\pr_l \mathrm{RT}(g)_{ij}(x)}_{\les \frac{1}{\vert x \vert^3}R^{-3/2}}  \cdot x^a \cdot x^d \, d\mu_e \\
=& \OO(R^{-5/2}).
\end{aligned} \label{EQCCCKerrparamIntegral37778}
\end{align}

\ni Plugging \eqref{EQCCCKerrparamIntegral17778}, \eqref{EQCCCKerrparamIntegral27778} and \eqref{EQCCCKerrparamIntegral37778} into \eqref{EQCCCestimateintegraldecomp17778} shows that
\begin{align*} 
\begin{aligned} 
\int\limits_{\RRR^3 \setminus B(0,1)} \lrpar{R_{\mathrm{scal}} \,g_{ij} - 2 \RRRic_{ij} } \lrpar{ \nab^i X^j+\nab^j X^i}  d\mu_g = \OO(R^{-5/2}),
\end{aligned} %\label{}
\end{align*}
and hence, by \eqref{EQClocalsmallnessKerrIngredient0} and \eqref{EQCrewritingKestimateConvergence777}, proves the desired estimate
\begin{align*} 
\begin{aligned} 
\mathbf{C}^{\mathrm{loc}}_{\mathrm{ADM}}(\la,R) - \mathbf{C}(\la) = \OO\lrpar{R^{-1/2}}.
\end{aligned} %\label{}
\end{align*}

\ni This finishes the proof of \eqref{EQquantEstimKERR79999} of Proposition \ref{PROPKerrchargesLOCALglobal10101}. \\

%%%%%%%%%%%%%%%%%%%%%%%%%%%%%%%%%
\ni \textbf{Conclusion of proof of Proposition \ref{PROPKerrchargesLOCALglobal10101}.} \ni It remains to show that for $\la\in\EE_R(\mathbf{E}_\infty)$ and for $R\geq1$ sufficiently large,
\begin{align*} 
\begin{aligned} 
\mathbf{C}_{\mathrm{ADM}}^{\mathrm{loc}}(\la,R) - r(\la,R) \cdot \mathbf{P}_{\mathrm{ADM}}^{\mathrm{loc}}(\la,R) =& \mathbf{C}(\la) - R\cdot \mathbf{P}(\la)+ \OO\lrpar{R^{-1/4}}.
\end{aligned} %\label{}
\end{align*}

\ni Indeed, by scaling, see Section \ref{SECspacelikescaling}, and Corollary \ref{CORcombinedEstimates222} we have that for $R\geq1$ large and $\la \in \EE_R\lrpar{\mathbf{E}_\infty}$,
\begin{align} 
\begin{aligned} 
r(\la,R) -R =& R \cdot \lrpar{ r\lrpar{{}^{(R)}\la,1} -1} \\
 \les& R\cdot \lrpar{\vert \mathbf{E}({}^{(R)}\la)-\lrpar{\mathbf{E}_\infty/(R)} \vert + \vert \mathbf{P}({}^{(R)}\la) \vert + \vert \mathbf{L}({}^{(R)}\la) \vert + \vert \mathbf{C}({}^{(R)}\la) \vert} \\
&+ R\cdot \lrpar{\frac{\vert \mathbf{L}({}^{(R)}\la) \vert + \frac{\vert \mathbf{P}({}^{(R)}\la) \vert}{\mathbf{E}_\infty/(R\cdot r_0)} \vert \mathbf{C}({}^{(R)}\la) \vert}{\mathbf{E}_\infty/R}}^2\\
=& \OO\lrpar{R^{-1/2}}.
\end{aligned} \label{EQestimateforarearadius77788}
\end{align}
By \eqref{EQquantEstimKERR79999} and \eqref{EQestimateforarearadius77788}, we have that for $\la \in \EE_R\lrpar{\mathbf{E}_\infty}$ and $R\geq1$ large, 
\begin{align*} 
\begin{aligned} 
&\mathbf{C}_{\mathrm{ADM}}^{\mathrm{loc}}(\la,R) - r(\la,R) \cdot \mathbf{P}_{\mathrm{ADM}}^{\mathrm{loc}}(\la,R) \\
=&\lrpar{\mathbf{C}(\la) + \OO(R^{-1/2})} -\lrpar{R+\OO\lrpar{R^{-1/2}}} \cdot \lrpar{\mathbf{P}(\la) + \OO(R^{-3/2})} \\
=&  \mathbf{C}(\la) - R\cdot \mathbf{P}(\la)+ \OO\lrpar{R^{-1/2}}.
\end{aligned}
\end{align*}
This finishes the proof of Proposition \ref{PROPKerrchargesLOCALglobal10101}. \end{proof}

\ni The above proof of Proposition \ref{PROPKerrchargesLOCALglobal10101} generalizes in a straight-forward way to yield \eqref{EQquantEstimKERR79999} with the local integrals on the left-hand side evaluated on the spheres of radius $r_{\mathbf{E}_\infty}(-R,2R)$, where $\mathbf{E}_\infty>0$ is a real number. Indeed, this is based on the fact that for $R\geq1$ large, it is well-known that (see the definition of $r_M(u,v)$ in \eqref{EQdefRbyUV})
$$2R \leq r_{\mathbf{E}_\infty}(-R,2R) \leq 3R.$$

\ni Thus we deduce from Proposition \ref{PROPKerrchargesLOCALglobal10101} the following estimates which are directly applied in \eqref{EQasympDECAY55540002} in Section \ref{SECconclusionofMainTheorem}.
\begin{proposition}[Convergence of charges to asymptotic invariants] \label{CORfullchargeexpressionSRkerrreference}
Let $R\geq1$ and $\mathbf{E}_\infty>0$ be two real numbers. For $R\geq1$ sufficiently large it holds that for all $\la\in \EE_R(\mathbf{E}_\infty)$, for $m=-1,0,1$ and $(i_{-1},i_0,i_1)=(2,3,1)$, 
\begin{align*} 
\begin{aligned} 
\mathbf{E}\lrpar{x_{-R,2R}^{\KK(\la)}} =& \mathbf{E}(\la) + \OO(R^{-1}), & \mathbf{P}^m\lrpar{x_{-R,2R}^{\KK(\la)}} =&  \mathbf{P}(\la)^{i_m} + \mathcal{O}(R^{-3/2}),\\
\mathbf{L}^m\lrpar{x_{-R,2R}^{\KK(\la)}} =& \mathbf{L}(\la)^{i_m} + \OO(R^{-1/2}), &
\mathbf{G}^m\lrpar{x_{-R,2R}^{\KK(\la)}} =& \mathbf{C}(\la)^{i_m}- 3R \cdot \mathbf{P}(\la)^{i_m} + \OO(R^{-1/4}).
\end{aligned}
\end{align*}
\end{proposition}

\begin{proof}[Proof of Proposition \ref{CORfullchargeexpressionSRkerrreference}] By combining Proposition \ref{PROPspheredataESTIM91} and Proposition \ref{PROPKerrchargesLOCALglobal10101}, and using \eqref{EQLEMscalingADMlocal}, we have that
\begin{align*} 
\begin{aligned} 
\mathbf{E}\lrpar{x_{-R,2R}^{\KK(\la)}} =& \mathbf{E}^{\mathrm{loc}}_{\mathrm{ADM}}\lrpar{\la,r_{\mathbf{E}_\infty}(-R,2R)} + \OO(R^{-3/2}) \\
=& \lrpar{\mathbf{E}(\la) + \OO(R^{-1})} + \OO(R^{-3/2}).
\end{aligned} %\label{}
\end{align*}
The estimates for $\mathbf{P}$ and $\mathbf{L}$ are proved similarly. Next, consider $\mathbf{G}$. We have by Propositions \ref{PROPspheredataESTIM91} and \ref{PROPKerrchargesLOCALglobal10101} that
\begin{align} 
\begin{aligned} 
 \mathbf{G}^m\lrpar{x_{-R,2R}^{\KK(\la)}} =& \lrpar{\mathbf{C}^{\mathrm{loc}}_{\mathrm{ADM}}}^{i_m}\lrpar{\la,r_{\mathbf{E}_\infty}(-R,2R)} \\
&- r\lrpar{\la, r_{\mathbf{E}_\infty}(-R,2R)} \cdot \lrpar{\mathbf{P}^{\mathrm{loc}}_{\mathrm{ADM}}}^{i_m}\lrpar{\la, r_{\mathbf{E}_\infty}(-R,2R)} + \OO(R^{-1/2}) \\
=& \mathbf{C}(\la)^{i_m} + \OO(R^{-1/2})  \\
&- \lrpar{ r_{\mathbf{E}_\infty}(-R,2R) + \OO(R^{-1/2})} \cdot \lrpar{\mathbf{P}(\la)^{i_m} +\OO(R^{-3/2}) }\\
=&\mathbf{C}(\la)^{i_m} - r_{\mathbf{E}_\infty}(-R,2R) \cdot \mathbf{P}(\la)^{i_m} + \OO(R^{-1/2}),
\end{aligned} \label{EQconvergenceG99904}
\end{align}
where we used that, similarly as in \eqref{EQestimateforarearadius77788},
\begin{align*} 
\begin{aligned} 
r\lrpar{\la, r_{\mathbf{E}_\infty}(-R,2R)} = r_{\mathbf{E}_\infty}(-R,2R) + \OO(R^{-1/2}).
\end{aligned} %\label{}
\end{align*}
It is well-known that for large $R\geq1$, by \eqref{EQdefRbyUV},
\begin{align*} 
\begin{aligned} 
r_{\mathbf{E}_\infty}(-R,2R) = 3R -2 \mathbf{E}_\infty \log \lrpar{\frac{3R}{2\mathbf{E}_\infty}-1} + \smallO(1),
\end{aligned} %\label{}
\end{align*}
so that, together with $\vert \mathbf{P}(\la) \vert =\OO(R^{-1/2})$ for $\la \in \EE_R(\mathbf{E}_\infty)$, we get from \eqref{EQconvergenceG99904} that
\begin{align*} 
\begin{aligned} 
\mathbf{G}^m\lrpar{x_{-R,2R}^{\KK(\la)}} = \mathbf{C}(\la)^{i_m} - 3R \cdot \mathbf{P}(\la)^{i_m} + \OO(R^{-1/4}).
\end{aligned} %\label{}
\end{align*}
This finishes the proof of Proposition \ref{CORfullchargeexpressionSRkerrreference}. \end{proof}

%%%%%%%%%%%%%%%%%%%%%%%%%%%%%%%%%%%%%%%%
%%%%%%%%%%%%%%%%%%%%%%%%%%%%%%%%%%%%%%%%
\appendix
%%%%%%%%%%%%%%%%%%%%%%%%%%%%%%%%%%%%%%%%
%%%%%%%%%%%%%%%%%%%%%%%%%%%%%%%%%%%%%%%%
\section{Completeness of perturbed null hypersurfaces in Kerr} \label{SECcompletenessKERR} 

\ni In this section we show that for real numbers $\mathbf{E}_\infty>0$ and real numbers $R\geq1$ large, for all $\la \in \EE_R(\mathbf{E}_\infty)$ the Kerr reference sphere $S_{-R,2R}^{\KK(\la)}$ in Kerr admits a future-complete outgoing null congruence (and similarly, past-complete ingoing null congruence). Here we recall from \eqref{EQdefeereinftyset999002} that $\EE_R(\mathbf{E}_\infty)$ is the set of all asymptotic invariants vectors such that
\begin{align} 
\begin{aligned} 
&\lrpar{R^{1/2} \vert \mathbf{E}(\la)- \mathbf{E}_\infty \vert}^2 + \lrpar{R^{1/2} \vert \mathbf{P}(\la) \vert}^2\\
&+\lrpar{R^{-1/4} \vert \mathbf{L}(\la) \vert}^2+\lrpar{R^{-1/2} \vert \mathbf{C}(\la) \vert}^2 \leq \lrpar{\mathbf{E}_\infty}^2.
\end{aligned}\label{EQdefeereinftyset999002later}
\end{align}

\ni In particular, as in the proof of Theorem \ref{THMMAINPRECISE} we glue to
\begin{align*} 
\begin{aligned} 
S_{-R,2R}^{\Kerr} := S_{-R,2R}^{\KK(\la')}
\end{aligned} %\label{}
\end{align*}
for a $\la' \in \EE_R(\mathbf{E}_\infty)$, this finishes the proof of Theorem \ref{THMMAINPRECISE}. This section is structured as follows, see also Figure \ref{FIGcompleteness} below.
\begin{enumerate}

\item We rescale the sphere $S_{-R,2R}^{\KK(\la)}$ to $S_{-1,2}^{\KK({}^{(R)}\la)}$ and show that for $R\geq1$ sufficiently large, it is a small smooth perturbation of the Boyer-Lindquist sphere 
\begin{align} 
\begin{aligned} 
S^{\mathrm{BL}}_{-1,2}:= \{t=0,r=3\}.
\end{aligned} \label{EQminkowksisphere123COMPLETE}
\end{align}

\item From \cite{FransIsrael,DHRScattering} we recapitulate the double null coordinates $(\tilde{u},\tilde{v},\tth^1,\tth^2)$ on the exterior region of Kerr spacetime, and show that for sufficiently small mass parameter $M\geq0$ and angular momentum parameter $a$, the Israel-Pretorius sphere 
\begin{align*} 
\begin{aligned} 
\tilde{S}^{\mathrm{IP}}_{-1,2} := \{\tilde{u}=-1,\tilde{v}=2\}
\end{aligned} %\label{}
\end{align*}
is a smooth perturbation of \eqref{EQminkowksisphere123COMPLETE}.

\item We use the double null coordinates  $(\tilde{u},\tilde{v},\tth^1,\tth^2)$ together with the property that Kerr is \emph{weakly asymptotically simple} to define a smooth conformal compactification.

\item Based on the smooth conformal compactification, we show that sufficiently small smooth perturbations of $\tilde{S}^{\mathrm{IP}}_{-1,2}$ admit future-complete outgoing null congruences. As (1) and (2) imply that for $R\geq1$ sufficiently large, $S_{-1,2}^{\KK({}^{(R)}\la)}$ is a small smooth perturbation of $\tilde{S}^{\mathrm{IP}}_{-1,2}$, we conclude the proof of Theorem \ref{THMMAINPRECISE}.
\end{enumerate}

\begin{figure}[H]
\begin{center}
\includegraphics[width=15cm]{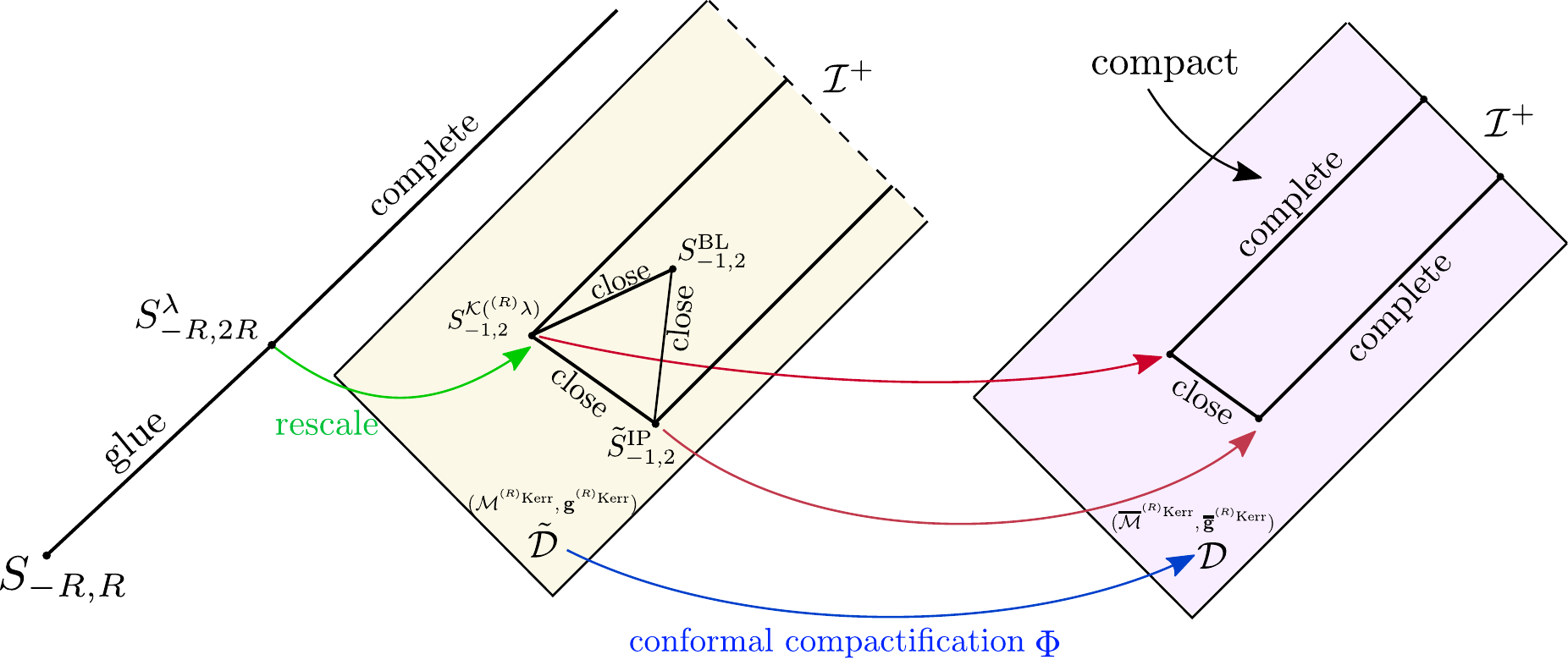} 
\caption{An illustration of the proof that the sphere $S^{\KK(\la)}_{-R,2R}$ admits a future-complete outgoing null congruence in Kerr spacetime.} \label{FIGcompleteness}

\end{center}
\end{figure}

\ni \textbf{(1) Rescaling of $S_{-R,2R}^{\KK(\la)}$ to $S_{-1,2}^{\KK({}^{(R)}\la)}$.} By construction of the Kerr reference spacelike initial data in Section \ref{SECappKerrFamilyDetails}, the sphere $S_{-R,2R}^{\KK(\la)}$ lies in a spacelike hypersurface $\{t'''=0\}$ which is constructed by rotating, spacelike translating and Lorentz boosting the spacelike hypersurface $\{t=0\}$ in Boyer-Lindquist coordinates $(t,r,\th,\phi)$, see \eqref{EQKerrMetric22287777}, with some mass parameter $M$ and angular momentum parameter $a$.

By applying the scaling of the Einstein equations, see Section \ref{SECspacelikescaling}, we can rescale $S_{-R,2R}^{\KK(\la)}$ to the sphere $S_{-1,2}^{\KK({}^{(R)}\la)}$, where the rescaled asymptotic invariants vector ${}^{(R)}\la$ is given by
\begin{align} 
\begin{aligned} 
{}^{(R)}\la = \lrpar{R^{-1} \cdot \mathbf{E}(\la), R^{-1} \cdot \mathbf{P}(\la), R^{-2} \cdot \mathbf{L}(\la), R^{-2} \cdot \mathbf{C}(\la)}.
\end{aligned} \label{EQrescaledLambdaCOMPLETENESS}
\end{align}
By the Kerr parameter estimates of Proposition \ref{PROPparameterEstimatesKerr}, in the rescaled picture for $R\geq1$ sufficiently large, the sphere $S_{-1,2}^{\KK({}^{(R)}\la)}$ is constructed by applying a small boost and small translation to the $\{t=0\}$-hypersurface in Boyer-Lindquist coordinates, where in the rescaled Kerr spacetime we have mass parameter $M/R$ and normalized angular momentum parameter $a/R$. 

Importantly, we conclude that the sphere $S_{-1,2}^{\KK({}^{(R)}\la)}$ is a small smooth perturbation of the Boyer-Lindquist sphere
\begin{align} 
\begin{aligned} 
S^{\mathrm{BL}}_{-1,2} := \{ t=0, \, r=3 \}
\end{aligned} \label{EQminkowskisphereCOMPLETENESS}
\end{align} 
in Boyer-Lindquist coordinates $(t,r,\th,\phi)$ on the rescaled Kerr. We note that for $R\geq1$ sufficiently large, both $S_{-1,2}^{\KK({}^{(R)}\la)}$ and $S^{\mathrm{BL}}_{-1,2}$ clearly lie in the exterior region of Kerr.\\

\ni \textbf{(2) Double null coordinates on the exterior region of Kerr.} In \cite{FransIsrael,DHRScattering} it is shown that, starting from Boyer-Lindquist coordinates $(t,r,\th,\phi)$, one can construct double null coordinates $(\tilde{u},\tilde{v},\tth^1,\tth^2)$ covering the exterior region of Kerr such that the outgoing null hypersurfaces $\tilde{H}_{\tilde{u}_0}$ are \emph{quasi-spherical}, that is, their spatial sections are asymptotically spherical at future null infinity; analogously for the ingoing null hypersurfaces $\HHb_{\tilde{v}_0}$ towards past null infinity. In particular, for sufficiently small mass parameter $M/R$ and angular momentum parameter $a/R$, the Israel-Pretorius sphere 
\begin{align*} 
\begin{aligned} 
\tilde{S}^{\mathrm{IP}}_{-1,2} := \{ \tilde u = -1, \, \tilde v =2 \},
\end{aligned} %\label{}
\end{align*}
is well-defined and a small smooth perturbation of the Boyer-Lindquist sphere \eqref{EQminkowskisphereCOMPLETENESS}.

Importantly, we conclude from \textbf{(1)} and \textbf{(2)} above that \emph{for $R\geq1$ sufficiently large,}
\begin{align*} 
\begin{aligned} 
S_{-1,2}^{\KK({}^{(R)}\la)} \text{ \emph{is a small smooth perturbation of} } \tilde{S}^{\mathrm{IP}}_{-1,2}.
\end{aligned} %\label{}
\end{align*}

\ni \textbf{(3) Conformal compactification of Kerr.} Let $\de>0$ be a real number and consider the outgoing slab
\begin{align*} 
\begin{aligned} 
\tilde{\DD}_{-1 +[-\de,\de],2+[-\de,\infty)} := \left\{ (\tilde{u},\tilde{v},\tth^1,\tth^2): -1-\de \leq \tilde u \leq -1+\de, \, 2-\de \leq \tilde v <\infty \right\},
\end{aligned} %\label{}
\end{align*}
where $(\tilde{u},\tilde{v},\tth^1,\tth^2)$ are the double null coordinates of \cite{FransIsrael,DHRScattering} on the exterior region of Kerr. Using the completeness in $\tilde{v}$, future null infinity is well-defined as
\begin{align*} 
\begin{aligned} 
\II^+_{-1 +[-\de,\de]}:= \{ -1-\de \leq \tilde{u}\leq -1+\de, \, \tilde{v}=\infty\} = [-1-\de, -1+\de] \times \SSS^2.
\end{aligned} %\label{}
\end{align*}
It is well-known that Kerr is \emph{weakly asymptotically simple}; for definitions and more background on conformal compactifications see, for example, \cite{HawkingEllis,Wald}. Applied to our situation, this implies that there is a smooth imbedding
\begin{align*} 
\begin{aligned} 
\Phi: \lrpar{\tilde{\DD}_{-1 +[-\de,\de],2+[-\de,\infty)}, \mathbf{g}^{{}^{(R)}\Kerr }} \to \lrpar{\DD_{-1+[-\de,\de],2+[-\de,1]},\overline{\mathbf{g}}^{{}^{(R)}\Kerr } },
\end{aligned} %\label{}
\end{align*}
where
\begin{itemize}
\item the domain
$$\DD_{-1+[-\de,\de],2+[-\de,1]}= [-1-\de,-1+\de] \times [2-\de,3] \times \SSS^2$$ is covered by local coordinates $(u,v,\th^1,\th^2)$ according to
\begin{align*} 
\begin{aligned} 
{\DD}_{-1 +[-\de,\de],2+[-\de,1]} := \left\{ ({u},{v},\th^1,\th^2): -1-\de \leq  u \leq -1+\de, \, 2-\de \leq  v \leq 3 \right\},
\end{aligned} %\label{}
\end{align*}

\item the metric $\overline{\mathbf{g}}^{{}^{(R)}\Kerr }$ is Lorentzian and smooth on $\DD_{-1+[-\de,\de],2+[-\de,1]}$ and conformal to ${\mathbf{g}}^{{}^{(R)}\Kerr }$ on $\DD_{-1+[-\de,\de],2+[-\de,1)}$, that is,
\begin{align*} 
\begin{aligned} 
\overline{\mathbf{g}}^{{}^{(R)}\Kerr } = \psi^2 {\mathbf{g}}^{{}^{(R)}\Kerr } \text{ on } \DD_{-1+[-\de,\de],2+[-\de,1)},
\end{aligned} %\label{}
\end{align*}
where we abused notation and denoted the push-forward of ${\mathbf{g}}^{{}^{(R)}\Kerr }$ under $\Phi$ by ${\mathbf{g}}^{{}^{(R)}\Kerr }$.
\item the coordinate $u$ is such that for $-1-\de\leq u \leq -1+\de$,
\begin{align*} 
\begin{aligned} 
\Phi\lrpar{\tilde{\HH}_{\tilde u}}= \HH_u
\end{aligned} %\label{}
\end{align*}
and we have the identification
\begin{align*} 
\begin{aligned} 
\{-1-\de \leq u \leq 1+ \de, \, v=3 \} = \II^+_{-1 +[-\de,\de]}.
\end{aligned} %\label{}
\end{align*}
\end{itemize}

\ni Importantly, as $\overline{\mathbf{g}}^{{}^{(R)}\Kerr }$ and ${\mathbf{g}}^{{}^{(R)}\Kerr }$ are conformally related, null geodesics of ${\mathbf{g}}^{{}^{(R)}\Kerr }$ are mapped to null geodesics of $\overline{\mathbf{g}}^{{}^{(R)}\Kerr }$ under the imbedding $\Phi$, and the hypersurfaces $\HH_u$ are null with respect to $\overline{\mathbf{g}}^{{}^{(R)}\Kerr }$. In particular, we conclude that for $-1-\de \leq u \leq -1+\de$,
\begin{align*} 
\begin{aligned} 
\text{\emph{there are no conjugate points along} } \HH_u.
\end{aligned} %\label{}
\end{align*}

\ni \textbf{(4) Future-completeness of outgoing null congruence.} In the following we show that the outgoing null congruence of small smooth perturbations of $\tilde{S}^{\mathrm{IP}}_{-1,2}$ is future-complete. Indeed, consider a small smooth perturbation $\tilde{S}$ of $\tilde{S}^{\mathrm{IP}}_{-1,2}$ such that for a real number $\de>0$, 
\begin{align*} 
\begin{aligned} 
\tilde{S} \subset \tilde{\DD}_{-1 +[-\de,\de],2+[-\de,\de]},
\end{aligned} %\label{}
\end{align*}
where, with respect to the double null coordinates $(\tilde{u}, \tilde{v},\tth^1,\tth^2)$,
\begin{align*} 
\begin{aligned} 
\tilde{\DD}_{-1 +[-\de,\de],2+[-\de,\de]} := \left\{ (\tilde{u},\tilde{v},\tth^1,\tth^2): -1-\de \leq u \leq -1+\de, \, 2-\de \leq v \leq \de \right\}.
\end{aligned} %\label{}
\end{align*}

\ni Using the smooth imbedding $\Phi$ constructed in \textbf{(4)} above, we have that the sphere
\begin{align*} 
\begin{aligned} 
\tilde{S}' := \Phi(\tilde{S})
\end{aligned} %\label{}
\end{align*} 
is a smooth small perturbation of the sphere $\lrpar{\tilde{S}^{\mathrm{IP}}_{-1,2}}' := \Phi(\tilde{S}^{\mathrm{IP}}_{-1,2})$. The latter lies in in the null hypersurface $\HH_{-1}$ which has by \textbf{(4)} no conjugate points. By a classical perturbation argument based on the analysis of Jacobi fields, we have that for $\tilde{S}'$ sufficiently close to $\lrpar{\tilde{S}^{\mathrm{IP}}_{-1,2}}'$, the outgoing null congruence $\tilde{\HH}'(\tilde{S}')$ with respect to $\overline{\mathbf{g}}^{{}^{(R)}\Kerr }$ emanating from $\tilde{S}'$,
\begin{align*} 
\begin{aligned} 
\tilde{\HH}'(\tilde{S}') \subset \DD_{-1+[-\de,\de],2+[-\de,1]}
\end{aligned} %\label{}
\end{align*}
has no conjugate points. Using the imbedding $\Phi$ and the conformal equivalence of ${\mathbf{g}}^{{}^{(R)}\Kerr }$ and $\overline{\mathbf{g}}^{{}^{(R)}\Kerr }$, we thus conclude that for $\tilde{S}$ sufficiently close to $\tilde{S}^{\mathrm{IP}}_{-1,2}$, the outgoing null congruence $\tilde{\HH}(\tilde{S})$ with respect to ${\mathbf{g}}^{{}^{(R)}\Kerr }$ emanating from $\tilde{S}$,
\begin{align*} 
\begin{aligned} 
\tilde{\HH}(\tilde{S}) \subset \tilde{\DD}_{-1+[-\de,\de],2+[-\de,\infty)}
\end{aligned} %\label{}
\end{align*}
has no conjugate points and is thus future-complete.

%%%%%%%%%%%%%%%%%%%%%%%%%%%%%%%%%%%%%%%%
\section{Hodge systems on $2$-spheres and vector spherical harmonics} \label{SECellEstimatesSpheres} 

\ni In this section we recall the theory for Hodge systems on $2$-spheres and the definition of vector spherical harmonics on round spheres; we refer to \cite{Czimek1} and \cite{ACR3,ACR1} for full details. Moreover, we prove geometric identities and estimates for spherical harmonics decompositions.\\

\ni \textbf{Definition of Hodge operators.} On a Riemannian $2$-sphere $(S, \gd)$ we define the following Hodge operators.
\begin{enumerate}
\item For $1$-forms $X_A$ define
\begin{align*}
\DDd_1(X) := (\Divd X, \Curld X),
\end{align*}
where $\Divd X:= \Nd_A X^A$ and $\Curld X := \in^{AB} \Nd_A X_B$. 
\item For symmetric $2$-tensors $W_{AB}$ define
\begin{align*}
\DDd_2(W)_C := (\Divd W)_C,
\end{align*}
where $(\Divd W)_C := \Nd^A W_{AC}$.
\item For pairs of functions $(f,g)$ define
\begin{align*}
\DDd_1^{\ast}(f,g):= -\Nd f + (\Nd g)^\ast.
\end{align*}
\end{enumerate}
We denote by $\Divdo$ and $\Curldo$ the operators $\Divd$ and $\Curld$ with respect to the round unit metric $\gac$. \\

\ni \textbf{Vector spherical harmonics on the round sphere.} We have the following standard definition.
\begin{definition}[Vector spherical harmonics] \label{DEFtensorHarmonics} Define the following.
\begin{enumerate}
\item For integers $l\geq0$, $-l \leq m \leq l$, let $Y^{(lm)}$ be the normalized (real-valued) spherical harmonics on the round unit sphere $S_1$. 
\item For $l\geq1$, $-l \leq m \leq l$, define the vectorfields
\begin{align} 
\begin{aligned} 
E^{(lm)} := \frac{1}{\sqrt{l(l+1)}} \DDd_1^\ast (Y^{(lm)},0), \,\, H^{(lm)} :=\frac{1}{\sqrt{l(l+1)}} \DDd_1^\ast (0, Y^{(lm)}).
\end{aligned} \label{EQdefVECTORspherical2333}
\end{align}
The vectorfields $E^{(lm)}$ and $H^{(lm)}$ are called \emph{electric} and \emph{magnetic}, respectively.

\end{enumerate}

\end{definition}
We remark that on the round sphere, the $L^2$-range of $\DDd_2$ consists of all $L^2$-integrable $1$-forms on $S$ which are orthogonal to the vectorfields of mode $l=1$.

Moreover, on the round sphere it holds that for modes $l\geq1$,
\begin{align} 
\begin{aligned} 
(\di f)_E^{(lm)}  =& - \sqrt{l(l+1)} f^{(lm)}, &(\di f)_H^{(lm)}  =& 0,\\
(\DDd_1^\ast(0,f))_E^{(lm)}  =& 0,&
(\DDd_1^\ast(0,f))_H^{(lm)}  =& \sqrt{l(l+1)} f^{(lm)},\\
(\Divdo X)^{(lm)} =& \sqrt{l(l+1)} X_E^{(lm)}. &
\end{aligned} \label{EQLEMfourierbasic1}
\end{align}

\ni \textbf{Geometric identities.} The following geometric identities are used to relate the charges $(\mathbf{E},\mathbf{P},\mathbf{L}, \mathbf{G})$ to the local integrals of spacelike initial data in Section \ref{SECstatementConstruction}.

\begin{lemma}[Geometric identities] \label{LEMgeometricIDENTITIES} The following holds.
\begin{enumerate}
\item Let $Y^{(1m)}$, $m=-1,0,1$, be the normalized (real-valued) spherical harmonics on the unit sphere $S_1 \subset \RRR^3$, and let $x=(x^1,x^2,x^3) \in \RRR^3$. We have that
\begin{align} 
\begin{aligned} 
\frac{x^1}{\vert x \vert}=\sqrt{\frac{4\pi}{3}}Y^{(11)}, \,\, \frac{x^2}{\vert x \vert}=\sqrt{\frac{4\pi}{3}}Y^{(1-1)}, \,\, \frac{x^3}{\vert x \vert}=\sqrt{\frac{4\pi}{3}}Y^{(10)},
\end{aligned} \label{EQREMchangeitom}
\end{align}
where $\vert x \vert :=\sqrt{\sum\limits_{i=1,2,3} (x^i)^2}$. 

\item Let $H^{(1m)}$, $m=-1,0,1$, be the magnetic vectorfields of mode $l=1$, and let $Y_{(i)}$, $i=1,2,3$, be the rotation fields on $\RRR^3$ defined in \eqref{EQdefRotationFieldsComponents}. Then it holds that
\begin{align*} 
\begin{aligned} 
H^{(1m)} = \sqrt{\frac{3}{8\pi}}\frac{1}{\vert x \vert^2} Y_{(i_m)}.
\end{aligned}
\end{align*}

\item Let $Z^{(i)}$, $i=1,2,3$, be the vector field defined in \eqref{EQdefXandZiALTADM}. Then it holds that
\begin{align} 
\begin{aligned} 
Z^{(i)} = - \vert x \vert^3 \sqrt{\frac{8\pi}{3}} E^{(1m_i)} - \vert x \vert^2 \lrpar{\sqrt{\frac{4\pi}{3}}Y^{(1m_i)}} \pr_r.
\end{aligned} \label{EQREMdecompositionZIapp}
\end{align}

\end{enumerate}
\end{lemma}

\begin{proof}[Proof of Lemma \ref{LEMgeometricIDENTITIES}] First, \eqref{EQREMchangeitom} follows from \eqref{EQdefinSPHERICALAFcoord} and the following well-known expressions for the spherical harmonics of mode $l=1$,
\begin{align*} 
\begin{aligned} 
Y^{(1-1)}= \sqrt{\frac{3}{4\pi}} \sin\th^1\sin\th^2, \,\, Y^{(10)}= \sqrt{\frac{3}{4\pi}} \cos\th^1,\,\, Y^{(10)}= \sqrt{\frac{3}{4\pi}} \sin\th^1\cos\th^2.
\end{aligned} %\label{}
\end{align*}

\ni Second, by definition of the left Hodge dual on spheres in\eqref{EQdefNotationNullStructure7778}, and Definition \ref{DEFtensorHarmonics}, we have that for $m=-1,0,1$,
\begin{align*} 
\begin{aligned} 
(H^{(1m)})_A = \frac{1}{\sqrt{2}} \ind_A^{\,\,\,\,\, B} (\di Y^{(1m)})_B,
\end{aligned} %\label{}
\end{align*}
so that by \eqref{EQREMchangeitom} and $N=\frac{x^l}{r}\pr_l$, for $(i_{-1},i_0,i_1)=(2,3,1)$ and a local orthonormal frame $(e_1,e_2)$,
\begin{align*} 
\begin{aligned} 
(H^{(1m)})_A
=&\frac{1}{\sqrt{2}} \sum\limits_{B=1,2} \ind_{AB} \Nd_B \lrpar{\sqrt{\frac{3}{4\pi}} \frac{x^{i_m}}{\vert x \vert}} \\
=& \sqrt{\frac{3}{8\pi}}\frac{1}{\vert x \vert}\sum\limits_{B=1,2} \ind_{AB} \Nd_B x^{i_m}  \\
=& \sqrt{\frac{3}{8\pi}}\frac{1}{\vert x \vert} \lrpar{\sum\limits_{B=1,2} \in_{NAB} \Nd_B x^{i_m} +\underbrace{\in_{NAN}}_{=0} N(x^{i_m})} \\
=& \sqrt{\frac{3}{8\pi}}\frac{1}{\vert x \vert} \in_{NAj}e^{jk} \underbrace{\nab_k x^{i_m}}_{=\de_k^{i_m}}\\
=&\sqrt{\frac{3}{8\pi}}\frac{1}{\vert x \vert^2} x^l \in_{lA i_m} \\
=& \sqrt{\frac{3}{8\pi}}\frac{1}{\vert x \vert^2}  \in_{i_m lA} x^l\\ 
=& \sqrt{\frac{3}{8\pi}}\frac{1}{\vert x \vert^2} (Y_{(i_m)})_A,
\end{aligned} 
\end{align*}
where we used that the rotation fields $Y_{(i)}$, $i=1,2,3$, on $\RRR^3$ are defined in \eqref{EQdefRotationFieldsComponents} by
\begin{align*} 
\begin{aligned} 
(Y_{(i)})_j := \in_{ilj} x^l.
\end{aligned} %\label{}
\end{align*}

\ni Third, we note that by \eqref{EQREMchangeitom} and $E^{(1m)}:=- \frac{1}{\sqrt{2}}\Nd Y^{(1m)}$, see Definition \ref{DEFtensorHarmonics}, 
\begin{align*} 
\begin{aligned} 
\pr_i - \frac{x^i}{ \vert x \vert} \pr_r = \nab x^i - N(x^i) N = \Nd x^i = r \Nd \lrpar{\frac{x^i}{r}}
=r\sqrt{\frac{4\pi}{3}} \Nd Y^{(1m_i)} = -r \sqrt{\frac{8\pi}{3}} E^{(1m_i)}.
\end{aligned} %\label{}
\end{align*}
Hence, using \eqref{EQdefXandZiALTADM} and \eqref{EQREMchangeitom}, the conformal Killing vectorfields $Z^{(i)}$, $i=1,2,3$, can be decomposed as follows,
\begin{align*} 
\begin{aligned} 
Z^{(i)} :=& \lrpar{\vert x \vert^2 \de^{ij}-2x^i x^j} \pr_j \\
=& \vert x \vert^2 \lrpar{ \pr_i - \frac{x^i}{ \vert x \vert} \pr_r} - x^i \vert x \vert \pr_r\\
=& - \vert x \vert^3 \sqrt{\frac{8\pi}{3}} E^{(1m_i)} - \vert x \vert^2 \lrpar{\sqrt{\frac{4\pi}{3}}Y^{(1m_i)}} \pr_r.
\end{aligned} 
\end{align*}
This finishes the proof of Lemma \ref{LEMgeometricIDENTITIES}. \end{proof}

\ni \textbf{Estimates for Fourier decompositions.} The following two lemmas are practical for estimates. First, we estimate the mode decomposition.

\begin{lemma}[Estimates for mode decomposition] \label{LEMnonlinearFourier} Let $(S,\gd)$ be a Riemannian $2$-sphere equipped with a round metric $\gac$ as defined in \eqref{EQdefRoundUnitMetric}. Let $X$ be a $1$-form and $W$ a $\gd$-tracefree symmetric $2$-tensor on $S$. Assume that for a real number $\varep>0$, it holds that
\begin{align} 
\begin{aligned} 
\Vert \gd - \gac \Vert_{H^6(S)} \leq \varep.
\end{aligned} \label{EQfouriernonlinearproof00assumption}
\end{align}
There exists a universal real number $\varep_0>0$ such that if $0<\varep<\varep_0$, then 
\begin{align} 
\begin{aligned} 
\left\vert (\Divd_\gd X)^{[0]} \right\vert \les& \, \Vert \gd - \gac \Vert_{H^6(S)} \cdot \Vert X \Vert_{H^2(S)}, \\
\left\vert \lrpar{\Divd_\gd W}^{[1]} \right\vert  \les& \, \Vert \gd - \gac \Vert_{H^6(S)} \cdot \Vert W \Vert_{H^2(S)}.
\end{aligned}\label{EQlemtoprovenonlinearFourier}
\end{align}
\end{lemma}

\begin{proof}[Proof of Lemma \ref{LEMnonlinearFourier}] Consider the first of \eqref{EQlemtoprovenonlinearFourier}. On the one hand, by \eqref{EQLEMfourierbasic1},
\begin{align} 
\begin{aligned} 
\lrpar{\Divd_\gd X}^{[0]} =& \lrpar{\lrpar{\Divd_\gd-\Divd_{\gac}}{X}}^{[0]} - \underbrace{\lrpar{\Divd_\gac X}^{[0]}}_{=0}.
\end{aligned} \label{EQfouriernonlinearproof0}
\end{align}

\ni On the other hand,
\begin{align} 
\begin{aligned} 
\lrpar{\Divd_\gd-\Divd_{\gac}}{X} =& \gd^{AB} \Nd_A X_B -  \gac^{AB} \overset{\circ}{\Nd}_A X_B\\
=&\lrpar{\gd^{AB}-  \gac^{AB}} \Nd_A X_B -\gac^{AB} \lrpar{ \overset{\circ}{\Nd}_A X_B - \Nd_A X_B}\\
=&\lrpar{\gd^{AB}-  \gac^{AB}} \Nd_A X_B \\
&-\gac^{AB} \lrpar{ \pr_A X_B - \overset{\circ}{\Ga}_{AB}^C X_C - \pr_A X_B + \Ga_{AB}^CX_C} \\
=& \lrpar{\gd^{AB}-  \gac^{AB}} \Nd_A X_B -\gac^{AB} \lrpar{  \Ga_{AB}^C   - \overset{\circ}{\Ga}_{AB}^C}X_C,
\end{aligned} \label{EQfouriernonlinearproof1}
\end{align}
where $\overset{\circ}{\Nd}$ and $\overset{\circ}{\Ga}$ denote the covariant derivative and Christoffel symbols of $\gac$, respectively. From \eqref{EQfouriernonlinearproof1} it is straight-forward to deduce that
\begin{align} 
\begin{aligned} 
\left\vert \lrpar{\Divd_\gd-\Divd_{\gac}}{X} \right\vert \les \Vert \gd-\gac \Vert_{H^6(S)} \Vert X \Vert_{H^3(S)}.
\end{aligned} \label{EQfouriernonlinearproof2}
\end{align}
Combining \eqref{EQfouriernonlinearproof00assumption}, \eqref{EQfouriernonlinearproof0} and \eqref{EQfouriernonlinearproof2}, the proof of the first estimate of \eqref{EQlemtoprovenonlinearFourier} follows. The second estimate of \eqref{EQlemtoprovenonlinearFourier} is proved similarly and omitted. This finishes the proof of Lemma \ref{LEMnonlinearFourier}. \end{proof}

\ni Second, we have the following estimate for the Gauss curvature $K$ of a sphere.
\begin{lemma}[Control of the mode $l=1$ of Gauss curvature] \label{LEMGaussSquareLoneEstimate} Let $(S,\gd)$ be a Riemannian $2$-sphere equipped with a round metric $\gac$ as defined in \eqref{EQdefRoundUnitMetric}. Let $K$ denote the Gauss curvature of $\gd$ on $S$. Assume that for a real number $\varep>0$,
\begin{align*} 
\begin{aligned} 
\Vert \gd - \gac \Vert_{H^6(S)} \leq \varep.
\end{aligned} 
\end{align*}
There exists a universal real number $\varep_0>0$ such that if $0<\varep<\varep_0$, then for $m=-1,0,1$,
\begin{align} 
\begin{aligned} 
\left\vert K^{(1m)} \right\vert  \les& \, \lrpar{\Vert \gd - \gac \Vert_{H^6(S)}}^2,
\end{aligned}\label{EQlemtoproveGAUSSFourier}
\end{align}
where the spherical harmonics $Y^{(1m)}$ are defined with respect to $\gac$.
\end{lemma}
\begin{proof}[Proof of Lemma \ref{LEMGaussSquareLoneEstimate}] We recall from \eqref{EQdefdecomposition1} that the metric $\gd$ on $S$ can be written as 
\begin{align*} 
\begin{aligned} 
\gd= \phi^2 \gd_c,
\end{aligned} %\label{}
\end{align*}
where $\phi$ is a conformal factor and $\gd_c$ is such that
\begin{align*} 
\begin{aligned} 
\det \gd_c = \det \gac \text{ on } S,
\end{aligned} %\label{}
\end{align*}
It is well-known (see, for example, equation (242) in \cite{DHR}) that linearizing the Gauss curvature $K$ with respect to variations of $\gd$ at the unit round metric $\gac$ yields
\begin{align*} 
\begin{aligned} 
\dot{K} = \half \Divdo \Divdo \gdcd - (\Ldo +2) \phid,
\end{aligned} %\label{}
\end{align*}
where we expressed the variation of $\gd$ in terms of the variation of $\phi$ and $\gd_c$. Projecting the above onto the $l=1$-mode shows that
\begin{align*} 
\begin{aligned} 
\dot{K}^{[1]} = \lrpar{\frac{1}{2} \Divdo \Divdo \gdcd - (\Ldo +2) \phid}^{[1]}=0.
\end{aligned} %\label{}
\end{align*}
Hence we conclude that $K^{[1]}$ is quadratic in variations of $\gd$ from the round metric $\gac$. This finishes the proof of Lemma \ref{LEMGaussSquareLoneEstimate}. \end{proof}

%%%%%%%%%%%%%%%%%%%%%%%%%%%%%%%%%%%%%%%%

\end{document}